\documentclass[12pt]{report}
\usepackage{amssymb,amsmath,amsthm,euscript}
\usepackage{amsfonts}
\usepackage[matrix,arrow]{xy}
\usepackage{epic}
\usepackage{eepic}
\usepackage{hyperref}

\newcommand{\ctg}{\mathop{\mathrm{ctg}}\nolimits}
\newcommand{\cth}{\mathop{\mathrm{cth}}\nolimits}
\newcommand{\End}{\mathop{\mathrm{End}}\nolimits}
\newcommand{\Hom}{\mathop{\mathrm{Hom}}\nolimits}
\newcommand{\Aut}{\mathop{\mathrm{Aut}}\nolimits}

\newcommand{\tr}{\mathop{\mathrm{tr}}\nolimits}

\newcommand{\id}{\mathop{\mathrm{id}}\nolimits}
\newcommand{\Span}{\mathop{\mathrm{Span}}\nolimits}
\newcommand{\res}{\mathop{\mathrm{res}}\limits}

\renewcommand{\Im}{\mathop{\mathrm{Im}}\nolimits}
\renewcommand{\Re}{\mathop{\mathrm{Re}}\nolimits}

\numberwithin{equation}{section}
\usepackage[utf8]{inputenc}
\usepackage[french,english]{babel}

\theoremstyle{plain}
\newtheorem{Th}{Théorème}[chapter]
\newtheorem{theor}[Th]{Theorem}
\newtheorem{de}{Définition}
\newtheorem{prop}[Th]{Proposition}

\newtheorem{lem}[Th]{Lemma}

\newtheorem{rem}{Remarque}
\newtheorem{Rem}[rem]{Remark}

\newcommand*{\hm}[1]{#1\nobreak\discretionary{}%
            {\hbox{$\mathsurround=0pt #1$}}{}}

\newcommand{\slt}{\mathfrak{sl}_2}
\newcommand{\glt}{\mathfrak{gl}_n}
\newcommand{\gln}{\mathfrak{gl}_n}
\newcommand{\la}{\langle}
\newcommand{\ra}{\rangle}
\newcommand{\La}{\big\langle}
\newcommand{\Ra}{\big\rangle}
\newcommand{\e}{e_\lambda}
\newcommand{\f}{f_\lambda}
\newcommand{\Cyl}{\mathrm{Cyl}}
\newcommand{\A}{{\mathcal A}}
\newcommand{\R}{{\mathcal R}}
\newcommand{\rr}{{\mathfrak r}}
\newcommand{\HH}{{\mathcal H}}

\newcommand{\Sym}{\mathop{\mathrm{Sym}}\nolimits}
\newcommand{\lfK}{{\cal K}}
\newcommand{\T}{\mathbb T}
\newcommand{\F}{{\cal F}}
\renewcommand{\L}{{\cal L}}
\newcommand{\cEF}{\mathfrak{e}_\tau(\widehat{\mathfrak{sl}}_2)}
\newcommand{\utau}{\mathfrak{u}_\tau(\widehat{\mathfrak{sl}}_2)}

\renewcommand{\L}{L}
\newcommand{\LL}{{\cal L}}
\newcommand{\G}{G_{\lambda}^}
\newcommand{\Gg}{{\cal G}_{\lambda}^}

\newcommand{\ph}{\varphi^}
\newcommand{\Ph}{\Phi^}

\newcommand{\bs}[1]{\backslash\{#1\}}

\newcommand{\bU}{\textbf{U}}
\newcommand{\vpint}{\mathop{\hspace{-1pt}{-}\hspace{-12pt}\int}}
\newcommand{\vpints}{\mathop{\hspace{-1pt}{-}\hspace{-10pt}\int}}
\def\fK{\mathfrak{K}}
\def\CC{\mathbb{C}}
\def\frg{\mathfrak{g}}
\def\hg{\hat{\mathfrak{g}}}
\def\r#1{(\ref{#1})}
\def\hx{\hat{x}}

\def\nn{\nonumber}
\newcounter{bbcount}[subsection]
\newcommand{\bb}[1]{\addtocounter{bbcount}{1}{\bf {\thesection}.\arabic{subsection}.\arabic{bbcount}.}
{\bfseries #1}}

\def\sigmap{\mathcal{P}}

\title{ Groupes quantiques associés aux courbes rationnelles et elliptiques et leurs applications}
\author{Alexey Viktorovich SILANTYEV}
\date{17 décembre 2008}

\begin{document}

\selectlanguage{french}

%

\thispagestyle{empty}
\begin{tabular}{lr}
 UNIVERSITÉ D'ANGERS \hspace{7cm} & Année 2008 \\
 UFR SCIENCES                       &  N$^o$ d'ordre {\bf 911} \\
 \end{tabular}
 \vspace{9mm}
\begin{center}
 \bf \Huge Groupes quantiques associés

aux courbes rationnelles et

elliptiques et leurs applications

\end{center}

 \vspace{9mm}
 \centerline{ THÈSE DE DOCTORAT}

 \vspace{0,3cm}
 \centerline{Spécialité : Mathématiques }

 \vspace{0,5cm}
 \centerline{\bf ÉCOLE DOCTORALE D'ANGERS}

 \vspace{1cm}
 \centerline{ Présentée et soutenue publiquement} 

 \centerline{ }

 \centerline{ à l'Université d'Angers}

 \centerline{ par}

 \centerline{\sc  Alexey SILANTYEV}
\vspace{3mm}
 \centerline{ devant le Jury ci-dessous}

 \vspace{7mm}
 \begin{tabular}{ll}
 \underline{Rapporteurs} :               & \\
 {\sc Jean AVAN }   &  Directeur de Recherche du CNRS, Cergy-Pontoise \\
 {\sc Junichi SHIRAISHI}    & Professeur à l'Université de Tokyo, (Japon) \\
 \underline{Examinateurs} :              & \\
 {\sc Eric DELABAERE} & Professeur à l'Université d'Angers \\
 {\sc Nikolai KITANINE} & MdC HDR à l'Université de Cergy-Pontoise \\
 {\sc Eric RAGOUCY} & Directeur de Recherche du CNRS, Annecy \\ 
 \underline{Directeurs de thèse} :     & \\
 {\sc Vladimir ROUBTSOV} & Professeur à l'Université d'Angers \\
 {\sc Stanislav PAKULIAK} & Directeur de Recherche à JINR, Dubna, (Russie)
 \end{tabular}

 \vspace{1cm}
 \centerline{LAREMA, U.M.R 6093 associée au CNRS}
 \centerline{2 Bd Lavoisier, 49045 Angers cedex 01, France}
 
 \hspace{12cm} {\bf ED 503}

%
%
%
%
%
%
%
%
%
%
%
%
%
%
%
%

%
%
%
%
%
%
%
%
%
%
%
%
%
%
%
%
%
%


\maketitle

\chapter*{Remerciements}

Cette thèse a été principalement financée par le projet de cotutelle Franco-Russe. Je remercie tout les gestionnaires de ce projet pour toute leur disponibilité et leurs encouragements. Je remercie principalement le directeur de ce projet, Jean Michel Maillet. \\

Je tiens à témoigner ma profonde reconnaissance à Stanislav Pakuliak. Merci d'avoir accepté de diriger ce travail. Ma gratitude va ensuite à Vladimir Roubtsov. Il a été pour moi un formidable superviseur. Ses suggestions, ses encouragements, sa disponibilité et ses qualités sur le plan humain m'ont été d'une grande utilité. C'est une chance immense de l'avoir eu comme Directeur de Thèse. \\

Jean Avan et Junichi Shiraishi m'ont fait l'honneur d'être les rapporteurs de cette thèse. Je leur remercie très sincèrement. \\

Eric Delabaere, Nikolai Kitanine et Eric Ragoucy ont accepté d'être examinateurs. Je leur suis très reconnaissant. \\

Mes remerciements vont également à tous les membres du Laboratoire de Recherches en Mathématiques (LAREMA) de l'Université d'Angers ainsi qu'à tous mes amis aussi bien en Russie qu'en France et particulièrement ceux d'Angers: Andrey, Roman, Vladimir, Suzanne, Rémi,  Paulo, Alexandre, Fabien, Jean-Marc,\ldots \\

Mon merci particulier à Serge Pelap, Joseph Dongho et Jean Avan qui m'ont aidé à corriger et à améliorer le texte. \\

Je voudrais aussi remercier tous mes collaborateurs de la Chaire de Théorie de Relativité et Gravité de l'Université de Kazan, le Laboratoire de Physique Mathématique de l'Institut de Physique Théorique et Expérimentale (ITEP), le Laboratoire de Physique Théorique de l'Institut de Recherches Nucléaires International (JINR), le Département de Mathématique de l'Université de Glasgow. J'exprime ma reconnaissance principale à Asya Aminova, à Sergey Kharchev et Mikhail Feigin. \\

Un grand Merci à mes parents, ma soeurs, mon épouse et toute ma famille pour leur soutien et leurs encouragements.

\tableofcontents

%

\chapter*{Introduction}

Cette thèse est consacrée aux groupes quantiques certains avec leurs applications aux systèmes intégrables et aux modèles statistiques sur des réseaux. Ces groupes quantiques peuvent être décrits en utilisant une solution de l'Équation de Yang-Baxter -- une matrice $R$ dépendant de paramètres appelés les paramètres spectraux. Les relations de commutation de ces groupes quantiques sont les relations $RLL$ avec l'opérateur de Lax dont éléments engendrent tout un groupe quantique donné. Ils peuvent être décrits en termes des courants.

Les groupes quantiques sont apparus à début des années quatre-vingt dans les travails de l'école de Leningrad consacrés à la méthode de dispersion inverse quantique~\cite{KulResh,Skl82,FadTakh}. Cette dernière est une généralisation de la méthode de dispersion inverse (classique) au cas des systèmes intégrables quantiques. Ce fait explique le nom "groupes quantiques". La méthode de dispersion inverse quantique est réduite à la construction d'un opérateur de Lax satisfaisant à la relation $RLL$ avec une matrice $R$ -- cette relation entraîne l'intégrabilité du système donné. La groupe quantique sont un résultat de l'abstraction algébrique de la relation $RLL$ pour une matrice $R$ donnée.

Les premières matrices $R$ sont introduites à début des années soixante-dix par Baxter en théorie des modèles statistiques sur des réseaux~\cite{Bax1,Bax2,Baxter}. Ceux ont été des matrices des poids de Boltzmann pour la sommet d'un réseau. La représentation de ces poids de Boltzmann à la forme matricielle permet d'écrire la fonction de partition du modèle comme une trace d'un produis ces matrices agissant dans des espaces tensorielles certains. La satisfaction de l'Équation de Yang-Baxter pour ces matrices entraîne la diagonalisation simultanée des facteurs dans ce produit correspondant aux lignes horizontales (ou verticales). Dans ce cas la fonction de partition est calculée explicitement.

Il y a trois types principaux des matrice $R$ qui jouent le rôle important pour les systèmes intégrables et pour les modèles statistiques sur des réseaux : les matrices $R$ rationnelles, trigonométriques et elliptiques. Elles dépendent des paramètres spectraux via des fonctions rationnelles, trigonométriques et elliptiques respectivement. Les groupes quantiques correspondants peuvent être associés avec les courbes complexe (surfaces de Riemann) du type correspondant. Les groupes quantiques liés aux matrices $R$ qui ne dépendent pas des paramètres spectraux ne sont pas considérés ici.

Dans~\cite{Bax2} Baxter a introduit aussi un modèle statistiques dont matrice de poids de Boltzmann ne satisfait pas à l'équation de Yang-Baxter. Elle satisfait à une équation plus générale -- Équation de Yang-Baxter Dynamique~\cite{F1,ABRR97}. Les solutions de cette équation ont été appelées matrices $R$ dynamiques et elles définissent aussi les objets algébriques appelés groupes quantiques dynamiques~\cite{F2,BBB}.

Au point de vue de l'algèbre les groupes quantiques a une structure des bigèbres. Celles sont des algèbres associatives, où hormis la multiplication et la unité les opérations duales (la comultiplication et la counité) sont définies. D'habitude on considère les groupes quantiques qui permit d'une plus opération -- un antipode. Elles sont appelées algèbre de Hopf. Les groupes quantiques dynamiques peuvent être décrits comme des objets plus générales -- les quasi-bigèbres et les quasi-algèbres de Hopf~\cite{D90}.

A la fin des années quatre-vingt Drinfeld a introduit la nouvelle méthode pour décrire les groupes quantiques dans le cas rationnel et trigonométrique~\cite{D88}. Pour cela il a introduit les courants -- des séries formelle dont coefficients (dans la décomposition en puissances du paramètre spectral) engendrent l'algèbre considérée. Quelques courants sont représentés comme une différence d'autres courants -- les demi-courants. Ces derniers peuvent être identifiés avec les coordonnées de Gauss pour la décomposition des opérateurs de Lax~\cite{DF}. Les groupes quantiques liées avec une courbe d'un genre arbitraire ont été construits en termes des courants dans les travails de Enriquez et Roubtsov~\cite{ER1,ER2,ER3}. 

La thèse est basée sur quatre articles~\cite{S1,S21,S22,S3} joints comme des appendices. Dans~\cite{S1} nous construisons la fonction de transition pour la chaîne de Toda périodique en termes de la méthode de dispersion inverse quantique. Cette fonction réalise une transition à telles variables que les fonctions propres de ce système sont factorisées en produit de fonctions d'une variable. Ces variables sont appelées séparées. Les premières idées de la Séparation de Variables pour chaîne de Toda sont proposées par Gutzwiller~\cite{Gutzwiller} et sont développées par Sklyanin en termes de méthode de dispersion inverse quantique~\cite{Sklyanin}. L'idée principal est l'utilisation des fonctions propres de la chaîne de Toda ouverte (multipliée à un facteur) comme une fonction de transition pour la chaîne périodique.

Dans~\cite{Sklyanin} Sklyanin a aussi proposé chercher les fonctions propres de la chaîne ouverte de $N$ particules comme une transformation intégrale des fonctions propres de la chaîne ouverte de $N-1$ particules sur les valeurs propres. La réalisation de cette idée a donné la représentation intégrale de ces fonctions appelée représentation de Mellin-Barns~\cite{Kharchev_O}. Quelques autres méthodes donnent une plus représentation intégrale pour les fonctions propres de la chaîne ouverte appelée représentation de Gauss-Givental~\cite{Givental,Kharchev_GG}. Elle permit aussi réécrire ces fonctions comme une transformation intégrale des fonctions propres de $N-1$ particules mais sur les arguments de fonctions -- les coordonnées du système. La sens de cette transformation et, par conséquence, de la représentation de Gauss-Givental dans la théorie de groupes quantiques (liée à la méthode de dispersion inverse quantique) a été inconnu.

Le travail~\cite{S1} a été inspiré par l'article~\cite{Derkachov} consacré à la Séparation de Variables pour le modèle $XXX$ en termes de la méthode de dispersion inverse quantique. Pour ce cas les auteurs construisent récursivement une fonction de transition comme une transformation intégrale sur coordonnées. Le point principal de cette construction est une triangularisation de l'opérateur de Lax du modèle. En modifiant la méthode de triangularisation pour le cas de la chaîne de Toda nous obtenons une représentation de Gauss-Givental. En suivant~\cite{Derkachov} nous utilisons la représentation graphique de transformations intégrales pour simplifier les calculs. Afin d'obtenir un tableau complet de Séparation de Variable pour la chaîne périodique en termes de la méthode de dispersion inverse quantique nous calculons aussi une mesure d'intégration sur variables séparées et des coefficients dans l'équation de Baxter.

Le travail~\cite{S21} est consacré à la comparaison de deux groupes quantiques dynamique définis par la même matrice $R$ dynamique elliptique -- matrice $R$ de Felder~\cite{F2} -- dans le niveau classique. Le premier groupe quantique a été obtenu dans~\cite{EF} en utilisant l'approche de Enriquez et Roubtsov~\cite{ER1,ER2,ER3}. Le second est introduit dans le travail~\cite{K98} comme une généralisation de le groupe quantique connu $U_q(\widehat{\slt})$ au cas elliptique. Les deux groupes quantiques peuvent être présentés par les relations $RLL$ dynamiques avec la même matrice $R$ et c'est pourquoi leur différence n'était pas apparente.

Le premier travail sur des groupes quantiques différentes définies par la même matrice $R$ a été consacré au cas rationnel~\cite{KLP99}. Les auteurs expliquent la différence entre les algèbres $\widehat{DY(\slt)}$ et $A_\hbar(\widehat{\slt})$ en mentionnant que les demi-courants de ces algèbres possèdent des propriétés analytiques différentes (comme des fonctions de la paramètre spectral). Dans le cas elliptique on a d'abord mentionné que les algèbres possèdent des extensions centrales différentes~\cite{JKOS2}. Après, on a montré que les demi-courants de ces algèbres possèdent des propriétés analytiques différentes~\cite{EPR}. Le dernier fait signifie que ces algèbres elliptiques sont différentes même sans les extensions centrales.

Dans le travail~\cite{S21} nous décrivons la différence de ces groupes quantiques elliptiques en détails. Pour clarifier les résultats nous les considérons au niveau classique, id est nous considérons les quasi-bigèbres de Lie correspondants. Nous proposons les schéma de la comparaison suivante. Nous mentions que il existe deux choix essentiellement différents d'un contour dans la courbe elliptique et il existe deux couvertures différentes de la courbe elliptique correspondants aux ces contours. Elles définissent des espaces de fonctions différents dont éléments sont définis sur ces couvertures. Nous construisons deux quasi-bigèbres de Lie en termes de demi-courants agissant sur ces espaces comme des distributions. Nous obtenons des crochets de Lie et un cocrochet en voyant que ces quasi-bigèbres de Lie coïncider avec les limites classiques des quasi-algèbres de Hopf considérées dans~\cite{EF} et \cite{K98}. Pour compacité des formules nous réunissons les demi-courants à deux matrices dépendent linéairement de ces demi-courants -- les opérateurs de Lax classiques. Nous clarifions la dérivation des formules différentes pour les extensions centrales. Cette différence découle du fait que ces extensions sont définies par les mêmes formules mais avec les contours d'intégration différents. Nous montrons aussi que les propriétés analytiques différentes des demi-courants entraînent des propriétés différentes algébriques. 

Nous continuons la comparaison des limites classiques des groupes quantiques dans le travail~\cite{S22}. Ici nous considérons leurs dégénérescences rationnelles et trigonométriques différentes. La première quasi-bigèbre de Lie possède un dégénérescence rationnelle et un dégénérescence trigonométrique. Dans le second cas il existe un dégénérescence rationnelle et un dégénérescence trigonométrique avec les relations de commutations ressemblantes. Les différences entre des algèbres de Hopf correspondant à ces dégénérescence rationnelles ont été constatée dans~\cite{KLP99}. Dans le second cas il existe aussi une famille des dégénérescences trigonométriques. Ce fait est expliqué par l'asymétrie du contour correspondant par rapport à l'échange des périodes elliptiques. La différence des dégénérescences correspondantes et la présence de la famille additionnelle de dégénérescences trigonométriques prouvent de nouveau la différence des algèbres considérées et aident comprendre plus profondément cette différence. Enfin nous généralisons la méthode de moyennisation de la matrice $r$ classique au cas elliptique dynamique et à la famille des bigèbres de Lie trigonométriques obtenue.

Le dernier article~\cite{S3} est consacré au calcul des projections considérées dans~\cite{EF} et à l'application de la théorie des projections aux modèles statistiques elliptiques sur des réseaux. Les projections des (quasi-)algèbres de Hopf sont apparues dans~\cite{ER2,ER3} comme une méthode de construction des groupes quantiques au genre arbitraire. Dans~\cite{EF} cette méthode a été appliquée au cas elliptique dynamique. Puis, en analysant les expressions de la matrice $R$ universelle pour le cas trigonométrique, Khoroshkin et Pakuliak a proposé une méthode du calcul des projections des courants en termes des demi-courants~\cite{KhP}. Ils ont mentionné que le noyau intégral dans l'expression de ces projections coïncide avec la fonction de partition pour le modèle 6-vertex avec des conditions frontières certaines. Nous généralisons cette remarque au cas elliptique. Le modèle statistique correspondant à ce cas est le modèle Solid-On-Solid (SOS). Nous considérons l'algèbre construite dans~\cite{EF} (dont la limite classique nous avons analysé dans~\cite{S21,S22}) et les projections décrites là. En généralisant la méthode du calcul des projections proposée dans~\cite{KhP} nous obtenons une expression pour les projections du produit des courants et la représentons comme une intégral du même produit. Ensuite, nous extrayons le noyau intégral en calculant un produit scalaire de cette intégral avec le produit de courants duales. Enfin nous montrons que l'expression obtenue satisfait les conditions analytiques définissant uniquement la fonction de partition pour le modèle SOS avec les conditions au bord de type parois de domaines.

La thèse est organisée comme suit. Dans le Chapitre~\ref{sec1} les conceptions générales sont introduites. D'abord on introduit ici les notions basiques des systèmes intégrables classiques et quantiques (sec.~\ref{sec11}). Ensuite, on considère les relations $RLL$ linéaires et quadratiques et leur signification pour les systèmes intégrables (sec.~\ref{sec12}, \ref{sec13}). On présente quelque exemples des systèmes intégrables décrivant par les relations $RLL$ quadratiques qui sont liés à la thèse (sec.~\ref{sec14}, \ref{sec15}). Puis, on définit les modèles statistiques sur des réseaux liées aux relations $RLL$ quadratiques et explique leur relation avec des systèmes intégrables (sec.~\ref{sec16}). Dans la section~\ref{sec17} on définit les notions principales de la théorie des algèbres de Hopf. Ensuite, on introduit les groupes quantiques en termes de relations $RLL$ quadratiques et définit la structures de bigèbre et d'algèbre de Hopf sur ces groupes quantiques (sec.~\ref{sec18}). Puis, on considère la construction du double de Drinfeld et l'applique aux groupes quantiques introduit à la section~\ref{sec18} en obtenant les algèbres de Hopf quasi-triangulaires décrits par les relations $RLL$ (sec.~\ref{sec19}). Enfin, on introduit les relations $RLL$ dynamiques et les quasi-algèbres de Hopf quasi-triangulaires correspondantes (sec.~\ref{sec110}).

Dans le Chapitre~\ref{sec2} on décrit la séparation de variables pour la chaîne de Toda périodique et présente les résultats de l'article~\cite{S1}. On discute les conceptions générales de la méthode de séparation de variables (sec.~\ref{sec21}). Ensuite, on les applique à la chaîne de Toda périodique (sec.~\ref{sec22}). Puis, on décrit les résultats principaux (sec.~\ref{sec23}). La section~\ref{sec24} contient des résultats complétant le tableau de la séparation de variable pour la chaîne de Toda périodique et des remarques additionnelles.

Dans le Chapitre~\ref{sec3} on explique les résultats principaux des articles~\cite{S21,S22}. D'abord on donne la notion générale des courants en termes des distributions (sec.~\ref{sec31}). Ensuite, on discute le rôle des courants pour la quantification des groupes quantiques (sec.~\ref{sec32}). Puis, on explique la relation de (quasi-)algèbres de Hopf décrites par les courants avec les courbes complexes (sec.~\ref{sec33}). Dans les sections~\ref{sec34} et \ref{sec35} on décrit les résultats principaux des travails~\cite{S21} et \cite{S22} respectivement.

Le Chapitre~\ref{sec4} est basé sur le travail~\cite{S3}. Il se débute de la section~\ref{sec41}, où on introduit la notion des projections pour les algèbres de Hopf. Ensuite, on considère le modèle 6-vertex et sa relation avec les projections d'une algèbre trigonométrique correspondante (sec.~\ref{sec42}). Puis, on généralise schéma de cette relation au cas elliptique. On décrit l'algèbre associée avec la courbe elliptique au sens de section~\ref{sec33} et introduit les projections "dynamiques" pour cette algèbre (sec.~\ref{sec43} et \ref{sec44}). Enfin, on décrit la fonction de partition pour un modèle statistique correspondant -- le modèle SOS (sec.~\ref{sec45}). Hormis les résultats contenus dans l'article~\cite{S3}, dont matériel est focalisé à l'application aux modèles statistiques, on formule dans cette Chapitre quelques fait importants pour la méthode de projection "dynamique" décrite dans~\cite{EF}. En plus, dans la section~\ref{sec42} on présente la relation des projections avec la fonction de partition du modèle 6-vertex en termes algébriques sans utilisant les propriétés analytiques de la fonction de partitions.

\chapter[Groupes quantiques, systèmes intégrables et modèles statistiques]{Groupes quantiques en théorie des systèmes inté\-grables et des modèles sur des réseaux}
\label{sec1}

Les groupes quantiques sont apparus pour la première fois dans la théorie des systèmes intégrables et la théorie des modèles statistiques exactement solubles sur des réseaux. L'outil principal de la théorie des systèmes intégrables est ce qu'on appelle un opérateur de Lax $L$. Sa propriété la plus importante provient du fait que les relations de commutation pour ses éléments matriciels peuvent être écrites comme une relation $RLL$, avec une matrice $R$ donnée. Un groupe quantique est l'algèbre "la plus universelle" définie par une relation $RLL$ avec une matrice $R$ donnée. C'est-à-dire, chaque algèbre décrite par un opérateur de Lax satisfaisant à ces relations $RLL$ est une représentation de cette algèbre. La structure de relation $RLL$ permet d'introduire un produit tensoriel entre les représentations de cette algèbre. Ce qui permet ainsi de la munir d'une structure de la bigèbre. De plus, si notre opérateur de Lax est inversible, cette bigèbre est une algèbre de Hopf.

En théorie des modèles statistiques sur des réseaux, les matrices $R$ sont apparues comme les matrices de poids de Boltzmann pour lesquelles le modèle correspondant est exactement résoluble. L'opérateur de Lax du modèle, construit par le coproduit du groupe quantique correspondant, est un outil principal de cette théorie. Par conséquent, à tout modèle statistique sur un réseau, on peut associer un système intégrable décrit par l'opérateur de Lax correspondant à ce modèle statistique.

Nous expliquerons le rôle des relations $RLL$ dans la théorie des systèmes intégrables et les relations entre ces systèmes et les modèles statistiques sur des réseaux et puis, nous considérons les groupes quantiques. Nous verrons aussi comment le coproduit participe à la description des systèmes intégrables et des modèles statistiques sur des réseaux.

\section{Les systèmes intégrables}
\label{sec11}

Nous rappelons tout d'abord la notion de système intégrable classique. Soit $A$ une algèbre associative commutative sur $\mathbb C$. On dit que une application linéaire $\{\cdot,\cdot\}\colon A\otimes A\to A$ définit une structure de Poisson sur $A$ si elle satisfait aux propriétés suivantes :
\begin{enumerate}
 \item la propriété d'antisymétrie: $\{a,b\}=-\{b,a\}$;
 \item l'identité de Jacobi: $\{\{a,b\},c\}+\{\{b,c\},a\}+\{\{c,a\},b\}=0$;
 \item la règle de Leibniz: $\{ab,c\}=a\{b,c\}+\{a,c\}b$;
\end{enumerate}
pour tout $a,b,c\in A$. L'algèbre $A$ munie d'une structure de Poisson est appelée une algèbre de Poisson. L'application $\{\cdot,\cdot\}$ est appelée crochets de Poisson de cette algèbre de Poisson.

On dit que des éléments $a_1,\ldots,a_m$ de $\mathbb C$ sont fonctionnellement indépendants si $F(a_1,\ldots,a_m)=0$ entraîne $F=const$,  $F\colon \mathbb C^m\to\mathbb C$ est une fonction à plusieurs variables telles que la substitution $F(b_1,\ldots,b_m)$ est définie pour tous $b_1,\ldots,b_m\in A$. Supposons que le nombre maximal d'éléments fonctionnellement indépendants est égal à $2N$.

\begin{de}
Un ensemble d'éléments fonctionnellement indépendants $I_1,\ldots,I_N\in A$ est appelé un système intégrable classique si ces éléments sont en involution :
\begin{align}
 \{I_i,I_j\}=0.
\end{align}
\end{de}

Soit ${\cal I}$ une sous-algèbre engendrée par les éléments $I_1,\ldots,I_N$. Grâce à la règle de Leibniz, la sous-algèbre ${\cal I}$ est involutive : $\{a,b\}=0$, $\forall a,b\in{\cal I}$. On dit encore, par abus de langage, que ${\cal I}$ est un système intégrable et les éléments $I_1,\ldots,I_N$ sont appelés intégrales de mouvement du système $I\in{\cal I}$.

 Un problème principal de la mécanique classique est celui de trouver des intégrales de mouvement pour un élément donné $I\in A$ -- des éléments satisfaisant la condition $\{I,I_i\}=0$ -- qui complètent $I$ jusqu'à un système intégrable ${\cal I}$. L'élément $I$ correspond à l'énergie du système mécanique correspondant et il est appelé l'hamiltonien du système.

On passe aux systèmes quantiques. Soit $\A$ une algèbre associative non-commutative sur $\mathbb C$ telle que le nombre maximal d'éléments fonctionnellement indépendants est égal à $2N$.

\begin{de}
Un ensemble d'éléments fonctionnellement indépendants $Q_1$, \ldots,$Q_N\in\A$ est appelé un système intégrable quantique si ces éléments commutent :
\begin{align}
 [Q_i,Q_j]=0.
\end{align}
\end{de}

La sous-algèbre ${\cal Q}$ engendrée par les éléments $Q_1,\ldots,Q_N$ est commutative et est appelée un système intégrable dont les intégrales de mouvement sont $Q_1,\ldots,Q_N$. Les systèmes intégrables quantiques sont liés aux systèmes intégrables classiques. En effet, considérons $\A_\hbar$ une famille d'algèbres associatives à un paramètre telle que les $\A_\hbar$ sont toutes isomorphes comme espaces vectoriels, les algèbres $\A_0$ et $\A_1$ sont isomorphes à $A$ et $\A$ respectivement comme algèbres et $\varphi^{-1}_\hbar\Big(\big[\varphi_\hbar(a),\varphi_\hbar(b)\big]\Big)\hm=\hbar\{a,b\}+o(\hbar)$, $\forall a,b\in A$, où nous avons fixé les isomorphismes $\varphi_\hbar\colon A\to\A_\hbar$. L'algèbre $\A_\hbar$ avec $\hbar\ne0$, en particulier, l'algèbre $\A=\A_1$ est appelée une quantification de l'algèbre de Poisson $A$. Un système intégrable quantique $Q_1,\ldots,Q_N$ est appelé une quantification d'un système intégrable classique $I_1,\ldots,I_N$ si on peut choisir les isomorphismes $\varphi_\hbar\colon\A_\hbar\to A$ tels que $\varphi_1(Q_i)=I_i$. Notons que si une algèbre $\A$ est une quantification d'une algèbre de Poisson $A$ et $\{Q_i\}\subset\A$ est un système intégrable quantique, alors $\{I_i=\varphi_1(Q_i)\}\subset A$ est un système intégrable classique. Le problème inverse -- trouver une quantification $\{Q_i\}$ du système intégrable classique $\{I_i\}$ donné -- est un problème important de la théorie des systèmes intégrables.

\section{Relations $RLL$ linéaires}
\label{sec12}

La méthode la plus répandue pour décrire les systèmes intégrables est le formalisme des opérateurs de Lax. Nous considérerons deux types d'opérateurs de Lax. Les opérateurs de Lax du premier type sont des matrices satisfaisant aux relations $RLL$ linéaires, alors que les opérateurs de Lax de second type satisfont aux relations $RLL$ quadratiques. Nous décrivons tout d'abord le premier cas.

On rappelle d'abord des notations matricielles utilisées dans le formalisme des opérateurs de Lax. Soit $X$ une matrice sur une algèbre associative $\A$ : $X=\sum\limits_{ij=1}^n E_{ij}\otimes X_{ij}\in\End\mathbb C^n\otimes\A$, où $X_{ij}\in\A$ et $E_{ij}$ est une matrice avec $1$ à $(i,j)$-ème place et $0$ aux autres, et soit $Y=\sum\limits_{ijkl=1}^n E_{ik}\otimes E_{jl}\in\A\otimes\End\mathbb C^n\otimes\End\mathbb C^n$, où $Y_{ij,kl}\in\A$. Alors, $X^{(a)}$ et $Y^{(bc)}$, où $b\ne c$, se désignent les éléments suivants de l'espace $\End\mathbb C^n\otimes\ldots\otimes\End\mathbb C^n\otimes\A$ :
\begin{align}
 X^{(a)}&=\sum_{ij=1}^n 1\otimes\ldots\otimes1\otimes E_{ij}\otimes1\otimes\ldots\otimes1\otimes X_{ij}, \label{Xa_not} \\
 Y^{(bc)}&=\sum_{ijkl=1}^n 1\otimes\ldots\otimes1\otimes E_{ik}\otimes1\otimes\ldots\otimes1\otimes E_{jl}\otimes\ldots\otimes1\otimes Y_{ij,kl}, \label{Ybc_not}
\end{align}
où la matrice $E_{ij}$ s'installe au $a$-ème espace $\End\mathbb C^n$, la matrice $E_{ik}$ -- au $b$-ème espace $\End\mathbb C^n$ et la matrice $E_{jl}$ -- au $c$-ème espace $\End\mathbb C^n$. Parfois les éléments $X^{(a)}$ et $Y^{(bc)}$ sont désignés par $X_{a}$ et $Y_{bc}$. Les notations~\eqref{Xa_not}, \eqref{Ybc_not} restent aptes pour $X\in\End\mathbb C^n$ et $Y\in\End\mathbb C^n\otimes\End\mathbb C^n$. Formellement, c'est un cas $\A=\mathbb C$. Nous pouvons encore définir de manière analogue un élément $Z^{(a_1,\ldots,a_m)}$ pour $Z\in\big(\End\mathbb C^n\big)^{\otimes m}\otimes\A$. Quelquefois, ces notations sont utilisées dans le cas plus général -- pour $Z\in{\cal B}^{\otimes m}\otimes\A$, où ${\cal B}$ est une algèbre associative unitaire arbitraire.

Soit $A$ une algèbre de Poisson. Une matrice $\L(u)\in\End\mathbb C^n\otimes A$ dépendant d'une variable $u$ est appelée un opérateur de Lax classique si elle satisfait à la relation $RLL$ linéaire classique
\begin{align}
 \{\L^{(1)}(u),\L^{(2)}(v)\}=[r^{(12)}(u,v),\L^{(1)}(u)+\L^{(2)}(v)], \label{rLL}
\end{align}
où $r(u,v)\in\End\mathbb C^n\otimes\End\mathbb C^n$ est appelée matrice $r$ classique. La partie gauche de~\eqref{rLL} est interprétée comme
\begin{align}
 \{\L^{(1)}(u),\L^{(2)}(v)\}=\sum_{ijkl=1}^n E_{ij}\otimes E_{kl}\otimes\{\L_{ij}(u),\L_{kl}(v)\}, \label{rLLgauche}
\end{align}
où $\L_{ij}(u)$ sont éléments de la matrice $\L(u)=E_{ij}\otimes \L_{ij}(u)$. La variable $u$ est appelée un paramètre spectral. L'antisymétrie des crochets de Poisson $\{\cdot,\cdot\}$ présume que la matrice $r$ doit satisfaire la condition
\begin{align}
 r^{(12)}(u,v)+r^{(21)}(v,u)=C(u,v), \label{rrC}
\end{align}
où $C(u,v)\in\End\mathbb C^n\otimes\End\mathbb C^n$ commute avec les éléments du type $x\otimes1+1\otimes x$, $\forall x\in\End\mathbb C^n$. La condition
\begin{multline}
 [r^{(12)}(u_1,u_2),r^{(13)}(u_1,u_3)]+[r^{(12)}(u_1,u_2),r^{(23)}(u_2,u_3)]+ \\
 +[r^{(13)}(u_1,u_3),r^{(23)}(u_2,u_3)]=0, \label{EYBC}
\end{multline}
est appelée Équation de Yang-Baxter Classique. Elle garantit l'identité de Jacobi pour les crochets de Poisson.

Il existe toujours une algèbre de Poisson $A$ pour une matrice $r$ donnée, telle qu'on peut trouver une matrice $\L(u)\in\End\mathbb C^n\otimes A$ satisfaisant~\eqref{rLL}. Par exemple, on peut définir l'algèbre $A$ comme l'algèbre associative commutative engendrée par les coefficients $\L^k_{ij}$ de la décomposition $\L_{ij}(u)=\sum\limits_k\L^k_{ij}\phi_k(u)$, où $\phi_k(u)$ sont des fonctions linéairement indépendantes~\footnote{On peut considérer aussi une décomposition continuelle $\L_{ij}(u)=\int\limits_{x\in\mathfrak X}\L_{ij}(x)\phi(u;x)d\mu(x)$, où $\mu$ est une mesure sur $\mathfrak X$.}, et munir $A$ des crochets de Poisson, qui sont définis par la formule~\eqref{rLL} et par la règle de Leibniz. Si les conditions~\eqref{rrC} et \eqref{EYBC} sont satisfaites, les crochets sont antisymétriques et satisfont à l'identité de Jacobi, et par conséquent définissent une structure de Poisson sur $A$.

La relation $RLL$ conduit à l'involutivité de traces de puissances de l'opérateur de Lax :
\begin{align}
 \{I_i(u),I_j(v)\}&=0, & I_k(u)&=\frac1k\tr \L(u)^k. \label{hI_htrL}
\end{align}
Ainsi, $I_1(u),\ldots,I_n(u)$ représentent un système intégrable. Les méthodes de démonstra\-tion de complétude des systèmes intégrables occupent une place spéciale dans la théorie des systèmes intégrables et nous ne les concernerons pas ici.

Un exemple d'un tel système est le modèle de Gaudin rationnel. Il est défini par l'opérateur de Lax rationnel
\begin{align} \label{LGaudinR}
 \L(u)=\sum_{s=1}^m\frac1{u-u_s}\sum_{i,j=1}^n E_{ji}\otimes e^{(s)}_{ij},
\end{align}
où $e^{(s)}_{ij}$ sont des générateurs de $A$ et $u_1,\ldots,u_m$ sont des nombres complexes distincts. La structure de Poisson sur $A$ est définie par la relation $RLL$ linéaire~\eqref{rLL} avec la matrice $r$ rationnelle
\begin{align} \label{r_rat}
 r(u,v)=\frac1{u-v}\sum_{i,j=1}^n E_{ij}\otimes E_{ji}.
\end{align}
Autrement dit :
\begin{align}
 \{e^{(s)}_{ij},e^{(r)}_{kl}\}=\delta^{sr}\big(\delta_{kj}e^{(s)}_{il}-\delta_{il}e^{(s)}_{kj}\big).
\end{align}
Pour $n=2$, par exemple, nous avons
\begin{align}
 I_1(u)&=\tr \L(u)=\sum_{s=1}^m\frac1{u-u_s}\sum_{i=1}e^{(s)}_{ii}\in \sum_{s=1}^m\frac1{u-u_s}Z(A), \\
 I_2(u)&=\frac12\tr \L(u)^2=\sum_{s=1}^m\frac1{u-u_s}\sum_{i=1}e^{(s)}_{ii}\in\sum_{s=1}^m\sum_{ij=1}^n\frac1{u-u_s}I_{2,s}+ \frac12\sum_{s=1}^m\frac1{(u-u_s)^2}Z(A),
\end{align}
où $Z(A)=\{a\in A\mid \{a,b\}=0, \forall b\in A\}$ est le centre de Poisson et
\begin{align} \label{I_Gaudin}
 I_{2,s}=\sum_{r\ne s}\frac{1}{u_s-u_r}\sum_{ij=1}^n e^{(s)}_{ij}e^{(r)}_{ji}.
\end{align}
Donc, $I_{2,1}$, \ldots, $I_{2,m}$ sont des intégrales de mouvement pour le modèle de Gaudin rationnel. Pour obtenir une algèbre de Poisson avec le nombre maximal d'éléments fonctionnellement indépendants nécessaire, il faut quotienter l'algèbre $A$ par l'idéal engendré par $x-\chi(x)$, pour tous $x\in Z(A)$, où $\chi\colon Z(A)\to\mathbb C$ est un caractère du centre fixé de Poisson $Z(A)$.

Dans le cas quantique, l'opérateur de Lax $\L(u)\in\End\mathbb C^n\otimes\A$ satisfait la relation
\begin{align}
 [\L^{(1)}(u),\L^{(2)}(v)]=[r^{(12)}(u,v),\L^{(1)}(u)+\L^{(2)}(v)], \label{rLLQ}
\end{align}
avec la même matrice $r$. La relation~\eqref{rLLQ} est appelée une relation $RLL$ linéaire quantique. L'algèbre $\A$ correspondant à la relation~\eqref{rLLQ} quantifie l'algèbre de Poisson $A$ construit par~\eqref{rLL}. La famille $\A_\hbar$ correspondante est définie par la relation
\begin{align} \label{rLLQhbar}
 [\L^{(1)}(u),\L^{(2)}(v)]=\hbar[r^{(12)}(u,v),\L^{(1)}(u)+\L^{(2)}(v)].
\end{align}

En général, les traces de puissances de l'opérateur de Lax ne commutent pas. Le problème de quantification du système~\eqref{hI_htrL} dépend de la spécification de l'opérateur de Lax et de la matrice $r$. Pour le modèle de Gaudin, cas $n=2,$ les traces $Q_1(u)=\tr \L(u)$, $Q_2(u)=\frac12\tr \L(u)^2$ commutent et, par conséquent, quantifient le système~\eqref{I_Gaudin}. Pour le cas $n>2,$ la quantification du système de Gaudin a été obtenue récemment en utilisant la notion de la courbe spectrale quantique~\cite{ChT}.

La relation $RLL$ linéaire a une propriété importante. Si les matrices $\L_1(u)\in\End\mathbb C^n\otimes\A_1$ et $\L_2(u)\in\End\mathbb C^n\otimes\A_2$ satisfont la relation $RLL$ avec la même matrice $r$, leur somme $\L_1(u)+\L_2(u)\in\End\mathbb C^n\otimes\A_1\otimes\A_2$ satisfait aussi la relation $RLL$ avec la même matrice $r$, où nous avons désigné par $L_1(u)$ et $L_2(u)$ dans cette somme les images de ces opérateurs de Lax par plongements canoniques $\A_1\to\A_1\otimes\A_2$ et $\A_2\to\A_1\otimes\A_2$. Par exemple, l'opérateur de Lax~\eqref{LGaudinR} est la somme des opérateurs de Lax
\begin{align} \label{LGaudinRs}
 \L_s(u)&=\frac1{u-u_s}\sum_{i,j=1}^n E_{ji}\otimes e^{(s)}_{ij} & & \in \End C^n\otimes\A_s,
\end{align}
et l'algèbre correspondante est représentée comme $\A=\A_1\otimes\ldots\otimes\A_m$, où $\A_s$ est une sous-algèbre engendrée par $e^{(s)}_{ij}$, $i,j=1,\ldots,n$, car les générateurs $e^{(s)}_{ij}$ commutent avec $e^{(r)}_{kl}$ pour $s\ne r$.

\section{Relations $RLL$ quadratiques}
\label{sec13}

Le second type des opérateurs de Lax classiques correspond aux matrices $L(u)\in\End\mathbb C^n\otimes A$ satisfaisant la relation $RLL$ quadratique classique
\begin{align}
 \{L^{(1)}(u),L^{(2)}(v)\}=[r^{(12)}_{quad}(u,v),L^{(1)}(u)L^{(2)}(v)], \label{rLLq}
\end{align}
où $r_{quad}(u,v)$ est une matrice $r$ satisfaisant l'équation~\eqref{EYBC} et la condition
\begin{align}
 r_{quad}^{(12)}(u,v)+r_{quad}^{(21)}(v,u)=\tilde C(u,v), \label{rrC_tilde}
\end{align}
où $\tilde C(u,v)\in\End\mathbb C^n\otimes\End\mathbb C^n$ commute avec des éléments du type $x\otimes x$, $\forall x\in\End\mathbb C^n$. L'équation~\eqref{EYBC} et la condition~\eqref{rrC_tilde} assurent l'antisymétrie des crochets et l'identité de Jacobi.

Les relations $RLL$ linéaires et quadratiques sont liées. En effet, supposons que l'opéra\-teur de Lax $L(u)$ et la matrice $r$ $r_{quad}(u,v)$ dépendent d'un paramètre $\hbar_D$ comme
\begin{align} \label{LlLqrlrq}
 L(u)&=1+\hbar_D \L(u)+o(\hbar_D), & r_{quad}(u,v)&=\hbar_D r(u,v)+o(\hbar_D),
\end{align}
où nous supposons aussi $\tilde C(u,v)=\hbar_D C(u,v)+o(\hbar_D)$. Alors, en substituant~\eqref{LlLqrlrq} à la relation $RLL$ quadratique~\eqref{rLLq} on obtient la relation $RLL$ linéaire~\eqref{rLL}. Le passage des relations $RLL$ linéaires aux relations $RLL$ quadratiques correspond à la quantification de Drinfeld dans la théorie des groupes quantiques. Comme la relation $RLL$ linéaire, la relation $RLL$ quadratique entraîne l'involutivité des traces $I_k(u)=\frac1k\tr L(u)^k$.

Une quantification des relations $RLL$ quadratiques est plus compliquée que celle des relations $RLL$ linéaires. En effet, en remplaçant les crochets de Poisson par le commutateur dans la partie gauche de~\eqref{rLLq}, on obtient une algèbre avec les générateurs $L_{ij}(u)$, qui ne satisfont pas à la condition de Poincaré–Birkhoff–Witt (PBW). Par conséquent, cette algèbre ne peut pas être une quantification de l'algèbre de Poisson $A$ définie par~\eqref{rLLq}. Une relation $RLL$ quadratique quantique pour l'opérateur de Lax quantique correspondant à $L(u)\in\End\mathbb C^n\otimes\A_\hbar$ est écrite comme
\begin{align} \label{RLL}
 R^{(12)}(u,v)L^{(1)}(u)L^{(2)}(v)=L^{(2)}(v)L^{(1)}(u)R^{(12)}(u,v),
\end{align}
où $R(u,v)\in\End\mathbb C^n\otimes\End\mathbb C^n$ satisfait l'Équation de Yang-Baxter (Quantique)
\begin{align} \label{EYB}
  R^{(12)}(u_1,u_2)R^{(13)}(u_1,u_3)R^{(23)}(u_2,u_3)=R^{(23)}(u_2,u_3)R^{(13)}(u_1,u_3)R^{(12)}(u_1,u_2)
\end{align}
et est appelé une matrice $R$ quantique. Nous supposerons que la $R$-matrice $R(u,v)$ est inversible pour $u$ et $v$ génériques. Il entraîne que $R(u,v)=R_0(u,v)+O(\hbar)$, où $R_0(u,v)$ est une matrice inversible ne dépendant pas de $\hbar$ et, par conséquence, elle peut être choisie telle que $R(u,v)=1+O(\hbar)$. On désigne par $r_{quad}(u,v)$ un premier coefficient de décomposition à $\hbar$ :
\begin{align}
 R(u,v)=1+\hbar r_{quad}(u,v)+o(\hbar). \label{R_r_quad}
\end{align}
En substituant~\eqref{R_r_quad} à \eqref{EYB}, on obtient l'Équation de Yang-Baxter Classique~\eqref{EYBC}. En plus, en substituant~\eqref{R_r_quad} à \eqref{RLL}, nous voyons que l'algèbre définie par~\eqref{RLL} quantifie l'algèbre commutative $A$ munie des crochets de Poisson~\eqref{rLLq}.

En prémultipliant~\eqref{RLL} par $R(u,v)^{-1}$ et en prenant la trace, on obtient que la fonction $t(u)=\tr L(u)$ commute avec elle même :
\begin{align}
 [t(u),t(v)]=0.
\end{align}
Dans cas $n>2,$ les intégrales de mouvement engendrées par $t(u)$ ne sont pas suffisants pour intégrabilité, et alors le problème de quantification du système classique défini par la relation $RLL$ quadratique est résolu particulièrement pour chaque opérateur de Lax. Dans quelques cas, par exemple, on peut utiliser la méthode de fusion.

La relation $RLL$ quantique~\eqref{RLL} quantifie la relation $RLL$ quantique linéaire~\eqref{rLLQ} au sens de Drinfeld. En substituant
\begin{align}
 L(u)&=1+\hbar\hbar_D \L(u,v)+o(\hbar_D),  & R(u,v)&=1+\hbar\hbar_D  r(u,v)+o(\hbar_D). \label{R_r_L_L}
\end{align}
à \eqref{RLL}, on obtient~\eqref{rLLQ}. La version quantique de la matrice~\eqref{r_rat} est
\begin{align} \label{R_rat}
 R(u,v)=1+\frac{\hbar\hbar_D}{u-v}\sum_{i,j=1}^n E_{ij}\otimes E_{ji}.
\end{align}
La matrice~\eqref{R_rat} est appelée une matrice $R$ rationnelle. Par ailleurs, l'équation~\eqref{EYB} a des solutions trigonométriques et elliptiques.

Nous pouvons construire un opérateur de Lax $L(u)\in\End\mathbb C^n\otimes\A_1\otimes\A_2$ par des opérateurs de Lax $L_1(u)\in\End\mathbb C^n\otimes\A_1$ et $L_2(u)\in\End\mathbb C^n\otimes\A_2$ donnés, comme dans le cas de la relation $RLL$ linéaire. En effet, si $L_1(u)$ et $L_2(u)$ satisfont à~\eqref{RLL}, le produit $L(u)=L_2(u)L_1(u)\in\End\mathbb C^n\otimes\A_1\otimes\A_2$ satisfait aussi~\eqref{RLL} avec la même matrice $R$. Par exemple, grâce à~\eqref{EYB}, la matrice $R(u,v)$ avec $v$ fixée est un opérateur de Lax pour l'algèbre $\End\mathbb C^n$ et, par conséquent, le produit des matrices $R$
\begin{align} \label{LRRR}
 L(u)=R^{(0N)}(u,v_N)\cdots R^{(02)}(u,v_2)R^{(01)}(u,v_1)
\end{align}
est un opérateur de Lax pour une algèbre $\big(\End\mathbb C^n\big)^{\otimes N}$. Le cas où la matrice $R$ de~\eqref{LRRR} est la matrice $R$ rationnelle~\eqref{R_rat}, le système donné par l'opérateur de Lax~\eqref{LRRR} est une quantification du système de Gaudin rationnel au sens de Drinfeld.

\section{Chaîne de Toda}
\label{sec14}

Dans~\cite{Toda}, le physicien Morikazu Toda a proposé un modèle non-linéaire intégrable du solide. Ce modèle ainsi que des modèles similaires sont appelés des chaînes de Toda. Nous considérerons les chaînes finies périodiques et ouvertes. La chaîne périodique de Toda de $N$ particules est un système à une dimension avec l'hamiltonien
\begin{align} \label{hamilTodaP}
 H_{\text{pér}}=\frac12\sum_{s=1}^N p^2_s+\sum_{s=1}^{N}e^{x_s-x_{s+1}},
\end{align}
où $x_s$ est la coordonnée de la $s$-ème particule et $p_s$ est son impulsion, et où nous avons posé $x_{N+1}=x_1$. Outre la chaîne périodique, on distingue une chaîne ouverte de Toda. Cette dernière est donnée par l'hamiltonien
\begin{align} \label{hamilTodaO}
 H_{\text{ouv}}=\frac12\sum_{s=1}^N p^2_s+\sum_{s=1}^{N-1}e^{x_s-x_{s+1}}.
\end{align}
On peut interpréter les chaînes périodiques et ouvertes comme le même système avec les conditions différentes à la frontière. Pour la chaîne périodique, ce sont des conditions périodiques $x_{N+1}=x_1,$ ou plus généralement $x_{s+N}=x_s$. La chaîne ouverte correspond aux conditions ouvertes $x_0\to-\infty$, $x_{N+1}\to+\infty$. Ce qui explique les noms des chaînes.

L'algèbre de Poisson $A$ pour ces systèmes est la même. C'est une algèbre de fonctions infiniment différentiables de $2N$ variables $p_s$ et $x_s$ avec les crochets de Poisson standards
\begin{align}
 \{f(p,x),g(p,x)\}=\sum_{s=1}^N\Big(\frac{\partial f(p,x)}{\partial p_s}\frac{\partial g(p,x)}{\partial x_s}-\frac{\partial f(p,x)}{\partial x_s}\frac{\partial g(p,x)}{\partial p_s}\Big),
\end{align}
où $p=(p_1,\ldots,p_N)$ et $x=(x_1,\ldots,x_N)$. Néanmoins, puisque les hamiltoniens~\eqref{hamilTodaP} et \eqref{hamilTodaO} ne sont pas en involution, ils définissent des systèmes différents. Comme nous le verrons, chacun de ces systèmes possède un ensemble nécessaire d'intégrales de mouvement qui engendrent une sous-algèbre involutive contenant l'hamiltonien correspondant.

L'hamiltonien du système de Toda périodique/ouverte quantique est défini par l'expression~\eqref{hamilTodaP}/\eqref{hamilTodaO}, où $p_s=-i\hbar \frac{\partial}{\partial x_s}$ et $x_s$ est un opérateur de multiplication par $x_s$. Dans le cas quantique, le rôle de l'algèbre $\A_{\hbar/i}$ est joué par l'algèbre engendrée par les opérateurs $p_s$ et $x_s$. Chaque terme des expressions de ces hamiltoniens ainsi que des intégrales de mouvement correspondantes ne contient pas de produits d'éléments non-involutifs et, par conséquent, les chaînes de Toda considérées sont quantifiées directement. C'est pourquoi nous ne nous arrêtons pas au cas classique.

On considère les matrices suivantes
\begin{equation} \label{Ls}
 L_s(u)=
  \begin{pmatrix}
   u-p_s & e^{-x_s} \\
   -e^{x_s} & 0
  \end{pmatrix}.
\end{equation}
Elles satisfont la relation $RLL$ quadratique quantique~\eqref{RLL} avec la matrice $R$ rationnelle~\eqref{R_rat}, où $\hbar\hbar_D=i\hbar$, $n=2$. Elles sont appelées les opérateurs de Lax d'une particule de la chaîne de Toda. En utilisant la construction de la fin de la sous-section précédente, nous obtenons l'opérateur de Lax suivant
\begin{equation} \label{TNLL}
 T_{[N]}(u)=L_N(u)\cdots L_1(u)=
  \begin{pmatrix}
   A_N(u) & B_N(u) \\
   C_N(u) & D_N(u)
  \end{pmatrix},
\end{equation}
qui est appelé la matrice de monodromie de la chaîne de Toda aux $N$-particules.
Comme nous le savons déjà, la trace de l'opérateur de Lax, en particulier, la fonction
\begin{align}
 \hat t(u)=\tr T_{[N]}(u)=A_N(u)+D_N(u)
\end{align}
commute avec elle même :
\begin{align}
 [\hat t(u),\hat t(v)]=0. \label{ttcomm}
\end{align}
Cela signifie que les coefficients $Q_k$ de la décomposition
\begin{align}
 \hat t(u)=\sum_{k=0}^N (-1)^k u^{N-k} Q_k \label{hattuQ}
\end{align}
commutent:
\begin{align}
 [Q_k,Q_j]=0, \label{commQQ}
\end{align}
En plus, l'hamiltonien de la chaîne de Toda périodique est exprimé par les coefficient : $H_{\text{pér}}=\frac12Q_1^2-Q_2$. Par conséquent, $Q_1$, \ldots ,$Q_N$ sont des intégrales de mouvement pour la chaîne de Toda périodique et $\hat t(u)$ est leur fonction génératrice. Chaque $Q_k$ est un polynôme en $p_1$, \ldots, $p_N$ de degré $k$ tel que $Q_k=\sum\limits_{i_1<\ldots<i_k}p_{i_1}\cdots p_{i_k}+$ un polynôme de degré $k-1$, et c'est pourquoi ils sont fonctionnellement indépendants.  En particulier, $Q_0=1$ et $Q_1=\sum\limits_{k=1}^N p_k$ -- opérateur de l'impulsion totale, Puisque le nombre des intégrales de mouvement fonctionnellement indépendantes est égal à $N,$ la chaîne de Toda est un système intégrable.

La relation $RLL$~\eqref{RLL} pour la matrice de monodromie~\eqref{TNLL} avec la matrice $R$ rationnelle~\eqref{R_rat} présume la relation
\begin{align}
 [A_N(u),A_N(v)]=0, \label{AAcomm}
\end{align}
dont on peut trouver, de manière analogique, que $A_N(u)$ est une fonction génératrice pour la chaîne de Toda ouverte :
\begin{gather}
 A_N(u)=\sum_{k=0}^N (-1)^k u^{N-k} H_k, \\
 H_0=1, \quad H_1=Q_1=\sum_{k=1}^N p_k, \quad H_{\text{ouv}}=\frac12 H_1^2-H_2, \\
 [H_k,H_j]=0, \label{commHH}
\end{gather}
et qu'elle est aussi un système intégrable.

\section{Modèles $XXX$ et $XXZ$}
\label{sec15}

Dans la sous-section précédente, nous avons décrit la chaîne de Toda en utilisant la relation $RLL$ quadratique pour un opérateur de Lax dégénéré (non-inversible) avec une matrice $R$ rationnelle. Le modèle $XXX$ est un système qu'on peut décrire en terme d'une relation $RLL$ quadratique et une matrice $R$ rationnelle, mais avec un opérateur de Lax non-dégénéré.

Le système associé a l'opérateur de Lax~\eqref{LRRR} avec la matrice $R$ rationnelle~\eqref{R_rat} (où $n=2$, $\hbar\hbar_D=1$) est appelé un modèle $XXX$ à spin $\frac12$. L'opérateur de Lax~\eqref{LRRR} joue le rôle de la matrice de  monodromie pour ce système. Si les paramètres $v_s$ sont égaux, le modèle est dit homogène, autrement il est  hétérogène. La matrice de monodromie d'un modèle $XXX$ du spin quelconque est définie par le produit des opérateurs de Lax d'une particule suivant
\begin{equation} \label{LXXXs}
 L_s(u)=
  \begin{pmatrix}
   u-v_s+h^{(s)}/2 & f^{(s)} \\
   e^{(s)} & u-v_s-h^{(s)}/2
  \end{pmatrix},
\end{equation}
où $h^{(s)}$, $e^{(s)}$ et $f^{(s)}$ sont des opérateurs dans $V_1\otimes\ldots\otimes V_N$ qui représentent les éléments de $\slt$ correspondants dans les $\slt$-modules $V_s$ :
\begin{align} \label{xxhef_sr}
 [h^{(s)},e^{(r)}]&=2\delta^{sr}e^{(s)}, &[h^{(s)},f^{(r)}]&=-2\delta^{sr}f^{(s)}, &[e^{(s)},f^{(r)}]&=\delta^{sr}h^{(s)}, \\
 [h^{(s)},h^{(r)}]&=0, & [e^{(s)},e^{(r)}]&=0, & [f^{(s)},f^{(r)}]&=0.
\end{align}
Les relations $RLL$ pour les opérateurs de Lax~\eqref{LXXXs} sont équivalentes aux relations de commutation~\eqref{xxhef_sr} pour $s=r$. La fonction génératrice des intégrales de mouvement du modèle $XXX$ est la trace de la matrice de monodromie. Dans le cas homogène ($v_1=\ldots=v_N$) elle est décomposée comme
\begin{align} \label{trLLXXX}
 \tr\big(L_N(u)\cdots L_1(u)\big)=u^N+u^{N-2}H_{XXX}+\ldots,
\end{align}
où
\begin{align} \label{HXXX}
 H_{XXX}=\sum_{1\le s<r\le N}\Big(\frac12h^{(s)}h^{(r)}+e^{(s)}f^{(r)}+f^{(s)}e^{(r)}\Big).
\end{align}
L'opérateur~\eqref{HXXX} est l'hamiltonien du modèle $XXX$. En termes d'opérateurs $S^{(s)}_x=\frac12(e^{(s)}+f^{(s)})$, $S^{(s)}_y=\frac1{2i}(e^{(s)}-f^{(s)})$, $S^{(s)}_z=\frac12h^{(s)},$ il est écrit comme
\begin{align} \label{HIxIxIx}
 H_{XXX}=\sum_{1\le s<r\le N}\Big(J_xS_x^{(s)}S_x^{(r)}+J_xS_y^{(s)}S_y^{(r)}+J_xS_z^{(s)}S_z^{(r)}\Big),
\end{align}
où $J_x=2$. Ce qui explique le nom du système. Puisque la trace~\eqref{trLLXXX} commute avec elle même, les autres coefficients dans la composition~\eqref{trLLXXX} nous donnent le reste des intégrales de mouvement. Le modèle $XXX$ du spin $\frac12$ correspond au cas où les $V_s$ sont des $\slt$-modules du spin $\frac12$ :
\begin{align}
 h^{(s)}&=H^{(s)}, &e^{(s)}&=E^{(s)}, &f^{(s)}&=F^{(s)},
\end{align}
avec
\begin{align}
 H&=
 \begin{pmatrix}
  1 & 0 \\
  0  & -1
 \end{pmatrix}, &
 E&=
 \begin{pmatrix}
  0 & 1 \\
  0  & 0
 \end{pmatrix}, &
 F&=
 \begin{pmatrix}
  0 & 0 \\
  1  & 0
 \end{pmatrix}.
\end{align}

Le modèle $XXZ$ homogène est défini par l'hamiltonien suivant :
\begin{align} \label{HIxIxIz}
 H_{XXZ}=\sum_{1\le s<r\le N}\Big(J_xS_x^{(s)}S_x^{(r)}+J_xS_y^{(s)}S_y^{(r)}+J_zS_z^{(s)}S_z^{(r)}\Big).
\end{align}
La matrice de monodromie pour le modèle $XXZ$ du spin $\frac12$ est donnée par l'expression~\eqref{LRRR} avec la matrice $R$ trigonométrique
\begin{multline} \label{R_trig}
 R(z,w)=(qz-q^{-1}w)\sum_{i=1}^2 E_{ii}\otimes E_{ii}+(z-w)\big(E_{11}\otimes E_{22}+E_{22}\otimes E_{11}\big)+ \\
 +(q-q^{-1})w E_{12}\otimes E_{21}+(q-q^{-1})z E_{21}\otimes E_{12},
\end{multline}
où nous utilisons les variables multiplicatives $z=e^u$ et $w=e^v$ au lieu des variables additives $u$ et $v$. Le paramètre $q$ définit l'anisotropie du modèle : $J_x/J_z=q+q^{-1}+2$. En substituant les paramètres $w_1$, \ldots, $w_N$ tous différents dans la formule~\eqref{LRRR}, nous obtenons un modèle $XXZ$ hétérogène.

Hormis les modèles $XXX$ et $XXZ$ on considère leur généralisation elliptique -- modèle $XYZ$. Le hamiltonien de ce modèle contient trois constante différentes $J_x$, $J_y$ et $J_z$. La matrice $R$ correspondante consiste en fonctions elliptiques et est appelée {\itshape matrice de Baxter-Belavin}. Pour la première fois, cette matrice $R$ pour le cas $\slt$ est écrit par Baxter dans une article~\cite{Bax1} comme matrice des poids de Bolzmann d'un modèle 8-vertex (voir section~\ref{sec18}). Après, Belavin a généralisé cette matrice aux cas des algèbres de Lie de rangs plus hauts.

\section{Modèles statistiques sur des réseaux}
\label{sec16}

Le modèle statistique est un modèle d'un système qui peut se trouver dans tel ou tel état avec une probabilité donnée. Les observables d'un système sont représentées comme des fonctions de l'état. Un problème posé dans la théorie des modèles statistiques est celui du calcul des valeurs moyennes des ces observables, leurs fluctuations et leurs corrélations. La loi la plus répandue de la distribution de la probabilité pour les systèmes physiques statistiques est la loi de Boltzmann. Elle postule que les probabilités d'états sont proportionnelles à $e^{-E/(kT)}$, où $E$ est l'énergie de l'état correspondant, $k$ est la constante de Boltzmann et $T$ est la température absolue du système -- une quantité constante pour un système en équilibre. La quantité $e^{-E/(kT)}$ est appelée le poids de Boltzmann de l'état. La probabilité d'un état $\omega$ avec énergie $E_\omega$ est égal à $\frac1{Z}e^{-E_\omega/(kT)}$, où le facteur de normalisation 
\begin{align}
 Z=\sum_{\omega}e^{-E_\omega/(kT)} \label{Z_gen_def}
\end{align}
est appelé une fonction de partition (ou la somme statistique) du modèle. La somme en~\eqref{Z_gen_def} est prise sur tous les états possibles. En prenant la dérivée logarithmique de~\eqref{Z_gen_def} par $-\frac1{kT}$ on obtient la valeur moyenne de l'énergie comme une fonction de la température -- la fonction thermodynamique principale. C'est pourquoi le problème principal de la théorie des modèles statistiques revient à calculer la fonction de partition.

On considère un réseau carré sur un plan qui possède $N^2$ sommets. Supposons que chaque sommet peut prendre des états appelés des configurations du sommet. On numérote les lignes verticales et horizontales par des nombres de $1$ à $N$ et on considère le sommet se trouvant au croisement de la $i$-ème ligne verticale avec la $j$-ème ligne horizontale. Désignons l'ensemble de ses configurations par $\Omega_{ij}$ et associons à chaque configuration $\omega_{ij}\in\Omega_{ij}$ un nombre $W_{ij}(\omega_{ij})$ appelé le poids de Boltzmann du sommet. Les états de tout le système sont appelés configurations du réseau et forment un ensemble $\Omega\subset\prod\limits_{ij=1}^N\Omega_{ij}$. En générale, il ne coïncide pas à tout $\prod\limits_{ij=1}^N\Omega_{ij}$, puisqu'en général, les configurations des sommets ne sont pas indépendantes. Nous postulons qu'un poids de Boltzmann de tout le réseau correspondant à une configuration $\omega=(\omega_{ij})\in\Omega$ est égal à
\begin{align}
 W(\omega)=\prod_{ij=1}^N W_{ij}(\omega_{ij}).
\end{align}
En terme de l'énergie cela signifie que l'énergie du système peut être représentée comme $E\hm=\sum\limits_{ij=1}^N E_{ij}$, où $E_{ij}$ est une partie de l'énergie ne dépendant que de la configuration du $(i,j)$-ème sommet.

\begin{figure}[h]
\begin{center}
\begin{picture}(40,40)
\put(23,28){$\alpha$}
\put(0,14){$\delta$}
\put(33,14){$\beta$}
\put(23,-10){$\gamma$}
\put(40,10){\line(-1,0){40}}
\put(40,10){\line(-1,0){3}}
\put(20,-10){\line(0,1){40}}
\put(20,-10){\line(0,1){3}}
\end{picture}
\end{center}
\caption{\footnotesize  Le poids de Boltzmann $R^{\alpha\beta}_{\gamma\delta}(i,j)$.}
\label{fig2}
\end{figure}
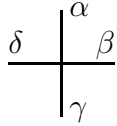

Supposons que chaque arête -- une partie d'une ligne verticale ou horizontale entre deux sommet voisin -- peut être à deux états et que la configuration de chaque sommet est définie par les états des arêtes contigus. De tels modèles sont appelés des modèles des sommets. Pour des raisons de simplicité, nous estimerons que les $16$ configurations de chaque sommet sont possibles, en supposant, s'il le faut, que les valeurs des poids de Boltzmann pour les configurations impossibles égalent zéro. Ce qui nous permet représenter les poids de Boltzmann d'un sommet comme des éléments d'une matrice : $W_{ij}(\omega_{ij})=R^{\alpha\beta}_{\gamma\delta}(i,j)$, où $\omega_{ij}$ est une configuration du $(i,j)$-ème sommet correspondante aux états des arêtes $\alpha$, $\beta$, $\gamma$ et $\delta$, où $\alpha,\beta,\gamma,\delta\hm=1,2$ sont des valeurs des états des arêtes (voir FIG~\ref{fig2}). La matrice $R(i,j)\in\End\mathbb C^2\otimes\End\mathbb C^2$ est appelée une matrice des poids de Boltzmann de $(i,j)$-ème sommet. La fonction de partition
\begin{align}
 Z=\sum_{\omega}W(\omega)=\sum_{\omega}\prod_{ij=1}^N W_{ij}(\omega_{ij})
\end{align}
pour un modèle des sommets peut être présentée comme une somme des éléments matricielles de la matrice
\begin{align}
 \mathbb R=\prod_{i=1}^N\prod_{j=1}^N R^{(i,N+j)}(i,j)\in\big(\End\mathbb C^2\big)^{2N},
\end{align}
où l'ordre de multiplication est défini par l'ordre de numération des lignes. Les indices de cette matrice correspondent aux états des arêtes frontières. Ainsi, la fonction de partition est la somme des éléments matriciels avec des indices satisfaisant les conditions à la frontière du modèle :
\begin{align}
 Z=\sum\mathbb R^{\alpha_1,\ldots,\alpha_N;\beta_1,\ldots,\beta_N}_{\gamma_1,\ldots,\gamma_N;\delta_1,\ldots,\delta_N}. \label{ZsumRabgd}
\end{align}

En particulier, pour les conditions périodiques à la frontière, on a $\alpha_i=\gamma_i$, $\beta_j=\delta_j$, et par conséquent la fonction de partition d'un modèle des sommets est une trace sur tous les $2N$ espaces
\begin{align}
 Z=\tr_{1,\ldots,2N}\mathbb R. \label{Z_trR}
\end{align}
On peut récrire l'expression~\eqref{Z_trR} comme
\begin{align}
 Z=\tr\prod_{i=1}^N \hat t(i), \label{Z_trT}
\end{align}
où la matrice $\hat t(i)\in\End\mathbb C^2$ est appelée une matrice de transfert de la $i$-ème ligne horizontale et est égale à
\begin{align}
 \hat t(i)=\tr_{1,\ldots,N}\prod_{j=1}^N R^{(0,j)}(i,j).
\end{align}
Le problème de calcul de la fonction de partition~\eqref{Z_trT} se réduit au problème spectral pour les matrices $\hat t(i)$. En effet, en diagonalisant simultanément ces matrices nous avons
\begin{align}
 Z=\lambda_{1,1}\lambda_{2,1}\cdots\lambda_{N,1}+\lambda_{1,2}\lambda_{2,2}\cdots\lambda_{N,2},
\end{align}
où $\lambda_{i,1}$ et $\lambda_{i,2}$ sont valeurs propres des matrices $\hat t(i)$.

Baxter a monté que ce problème est soluble si les matrices de poids de Boltzmann sont paramétrées par $R(i,j)=R(u_i,v_j)$ et $R(u,v)$ satisfait l'équation~\eqref{EYB}. Ce qui signifie que nous avons une correspondance entre les modèles des sommets exactement solubles et les systèmes intégrables avec la matrice de monodromie~\eqref{LRRR}. Les matrices de transfert d'un modèle des sommets sont exprimées par la fonction génératrice des intégrales de mouvement $\hat t(u)$ du système intégrable correspondant comme
\begin{align}
 \hat t(i)=\hat t(u_i).
\end{align}
Par conséquent, le problème consistant à trouver la fonction de partition est réduit au problème spectral pour $\hat t(u)$, qui est équivalent au problème spectral pour les intégrales de mouvement du système.

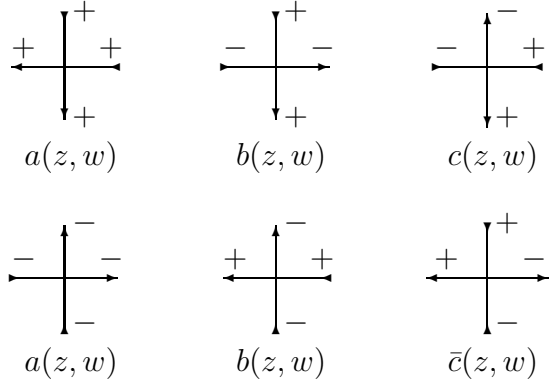
\begin{figure}[t]
\begin{center}
\begin{picture}(180,130)
\put(05,63){$a(z,w)$}
\put(05,-15){$a(z,w)$}
\put(23,118){$+$}
\put(0,104){$+$}
\put(33,104){$+$}
\put(23,78){$+$}
\put(40,100){\vector(-1,0){40}}
\put(40,100){\vector(-1,0){3}}
\put(20,120){\vector(0,-1){40}}
\put(20,120){\vector(0,-1){3}}

\put(85,63){$b(z,w)$}
\put(85,-15){$b(z,w)$}
\put(103,118){$+$}
\put(113,104){$-$}
\put(80,104){$-$}
\put(103,78){$+$}
\put(80,100){\vector(1,0){40}}
\put(80,100){\vector(1,0){3}}
\put(100,120){\vector(0,-1){40}}
\put(100,120){\vector(0,-1){3}}

\put(165,63){$c(z,w)$}
\put(165,-15){$\bar c(z,w)$}
\put(183,118){$-$}
\put(193,104){$+$}
\put(160,104){$-$}
\put(183,78){$+$}
\put(160,100){\line(1,0){40}}
\put(200,100){\vector(-1,0){3}}
\put(160,100){\vector(1,0){3}}
\put(180,80){\vector(0,1){40}}
\put(180,80){\vector(0,-1){3}}

\put(23,38){$-$}
\put(0,24){$-$}
\put(33,24){$-$}
\put(23,00){$-$}
\put(0,20){\vector(1,0){40}}
\put(0,20){\vector(1,0){3}}
\put(20,00){\vector(0,1){40}}
\put(20,00){\vector(0,1){3}}

\put(103,38){$-$}
\put(113,24){$+$}
\put(80,24){$+$}
\put(103,00){$-$}
\put(120,20){\vector(-1,0){40}}
\put(120,20){\vector(-1,0){3}}
\put(100,00){\vector(0,1){40}}
\put(100,00){\vector(0,1){3}}

\put(183,38){$+$}
\put(193,24){$-$}
\put(160,24){$+$}
\put(183,00){$-$}
\put(200,20){\line(-1,0){40}}
\put(200,20){\vector(1,0){3}}
\put(160,20){\vector(-1,0){3}}
\put(180,00){\line(0,1){40}}
\put(180,00){\vector(0,1){3}}
\put(180,40){\vector(0,-1){3}}
\end{picture}
\end{center}
\caption{\footnotesize La représentation graphique des poids de Boltzmann pour le modèle 6-vertex.}
\label{fig1}
\end{figure}

L'exemple du modèle des sommets est un modèle 6-vertex. C'est un modèle avec une matrice des poids de Boltzmann $R(i,j)=R(u_i,v_j)$, où $R(u,v)$ est une matrice $R$ trigonométrique~\eqref{R_trig}. Le paramètre $q$ définit l'anisotropie du modèle 6-vertex. Chaque sommet a six configurations avec une probabilité non-zéro, id est six configurations possibles en réalité. Ces configurations sont représentées en Fig.~\ref{fig1} avec les poids de Boltzmann
\begin{align}
a(z,w)&=qz-q^{-1}w, & b(z,w)&=z-w,  \label{Bw} \\
c(z,w)&=(q-q^{-1})z, & \bar c(z,w)&=(q-q^{-1})w. \label{Bwc}
\end{align}
Pour chaque configuration possible deux flèches entrent à chaque sommet et deux flèches quittent chaque sommet. Le système intégrable correspondant est le modèle $XXZ$ de spin $\frac12$. Le modèle 6-vertex est appelé homogène si $u_1=\ldots=u_N$ et $v_1=\ldots=v_N$, id est si la matrice des poids est même pour tous les sommets. Il correspond au modèle $XXZ$ homogène. Autrement, il est appelé hétérogène et correspond au modèle $XXZ$ hétérogène.

Le modèle $XYZ$ correspond à un autre modèle des sommets. Il généralise le modèle 6-vertex et il est appelé le modèle 8-vertex. Chaque configuration d'un sommet possible pour ce modèle a un nombre pair de flèches entrantes et un nombre pair de flèches sortantes. Ce sont six configurations présentées à Fig.~\ref{fig1} (mais avec d'autres poids de Boltzmann) et plus deux configurations, pour lesquelles toutes les flèches entrent au sommet considéré et toute les flèches quittent ce sommet. La limite trigonométrique du modèle 8-vertex nous donne le modèle 6-vertex. 

L'autre type des modèles sur des réseaux est les modèles des faces. Dans les articles initiaux ces modèles sont décrits comme suit. Chaque configuration du réseau est définie par des nombres associés à chaque sommet et appelés altitudes. A chaque face on associe un poids de Boltzmann dépendant des hauteurs autour de cette face. Exigeons que les différences des hauteurs voisins égalent à $\;\pm1$. Alors, nous avons six configurations pour chaque face. Le poids de Boltzmann pour tout le réseau est égal au produit des poids de toutes les faces du réseau. 

Cependant, on peut regarder autrement les modèles des faces. On échange les faces et les sommets, id est on considère le réseau nouvel dont sommets se trouvent aux centres des faces du réseau dernier. Alors, les hauteurs sont associées aux faces et les poids -- aux sommets. Ainsi, nous arrivons à des modèles décrits comme plus haut. En plus, si on dit qu'une arête est un état correspondant à la différence entre les hauteurs des faces contiguës, nous obtenons un modèle ressemblant à un modèle des sommets, mais les poids de chaque sommet pour ces modèles dépendent, en plus, d'une des hauteurs placée autour de ce sommet. Par exemple, on peut considérer qu'il dépend de la hauteur de la face gauche-haut relativement au sommet considéré. La matrice de poids pour ce sommet satisfait l'Équation de Yang-Baxter Dynamique et elle est appelée matrice $R$ dynamique. La valeur de la hauteur, dont cette matrice $R$ dépend, est appelée le paramètre dynamique.

\section{Bigèbres et algèbres de Hopf}
\label{sec17}

Comme nous l'avons déjà écrit les groupes quantiques sont des algèbres décrites par les relations $RLL$. Les relations $RLL$ quantiques quadratiques décrivent des bigèbres et des algèbres de Hopf, les relation $RLL$ quantiques linéaire -- des bigèbre de Lie, id est des "groupes classiques". Les relations $RLL$ classiques définissent les algèbres de Poisson obtenues comme des dégénérations de Poisson correspondantes.

On commence avec les définitions de bigèbres et d'algèbre de Hopf. Sois $\HH$ une algèbre associative unitaire sur un corps commutatif $\mathbb K$. On dit qu'on a une structure de {\itshape bigèbre} dans $\HH$ si elle est munie des homomorphismes $\Delta\colon\HH\to\HH\otimes\HH$ et $\varepsilon\colon\HH\to\mathbb K$ tels que les diagrammes
\begin{align} \label{diag_coass}
 \xymatrix{
  \HH\ar@{->}[r]^{\Delta}\ar@{->}[d]_{\Delta} & \HH\otimes\HH\ar@{->}[d]^{\id\otimes\Delta} \\
  \HH\otimes\HH\ar@{->}[r]^{\Delta\otimes\id} & \HH\otimes\HH\otimes\HH
 } \qquad \qquad
 \xymatrix{
  \HH\otimes\HH\ar@{->}[d]_{\varepsilon\otimes\id} & \HH\ar@{->}[r]^{\Delta}\ar@{->}[l]_{\Delta}\ar@{->}[dl]\ar@{->}[dr] & \HH\otimes\HH\ar@{->}[d]^{\id\otimes\varepsilon} \\
  \mathbb K\otimes\HH &  & \HH\otimes\mathbb K
 }
\end{align}
sont commutatifs, où $\HH\to\mathbb K\otimes\HH$ et $\HH\to\HH\otimes\mathbb K$ sont les isomorphismes canoniques. Les applications $\Delta$ et $\varepsilon$ sont appelées une comultiplication (ou un coproduit) et une counité de la bigèbre $\HH$. Une bigèbre $\HH$ est appelée commutative si $\HH$ est commutative comme une algèbre; elle est appelée {\itshape cocommutative} si $\Delta=\Delta^{op}\stackrel{df}{=}\hat{\cal P}\circ\Delta$, où une application $\hat{\cal P}\colon\HH\otimes\HH\to\HH\otimes\HH$ est définie comme $\hat{\cal P}(a\otimes b)=b\otimes a$ pour tous $a,b\in\HH$.

Le sens de la comultiplication en théorie des représentations est le suivant. Si on a deux représentations de la bigèbre $\HH$ : $\pi_1\colon\HH\to\End V_1$ et $\pi_2\colon\HH\to\End V_2$, on peut construire une représentation $\pi\colon\HH\to\End(V_1\otimes V_2)$ dans un produit tensoriel $V_1\otimes V_2$ comme la composition $\pi=(\pi_1\otimes\pi_2)\circ\Delta$. La counité donne la représentation triviale. La représentation $\varepsilon\colon\HH\to\mathbb K$ est une unité par rapport au "produit tensoriel" des représentations de $\HH$ :
\begin{align}
 (\pi\otimes\varepsilon)\circ\Delta=(\varepsilon\otimes\pi)\circ\Delta=\pi.
\end{align}

\begin{de}
Une bigèbre $\HH$ est appelée une algèbre de Hopf s'il est donné une application linéaire $S\colon\HH\to\HH$ telle que 
\begin{align}
 \mu\circ(S\otimes\id)\circ\Delta=\mu\circ(\id\otimes S)\circ\Delta= 1\cdot\varepsilon,
\end{align}
où $\mu\colon\HH\otimes\HH\to\HH$ est une multiplication dans $\HH$, soit $\mu(a\otimes b)=ab$.
\end{de}

En fait, l'application $S\colon\HH\to\HH$ est un antihomomorphisme inversible et il est appelé un antipode de l'algèbre de Hopf $\HH$. Il permet de construire la représentation duale à la représentation donnée. 

L'exemple de la bigèbre commutative est une algèbre de fonctions sur un monoïde $G$. La multiplication pour cette algèbre est une multiplication par points. La structure de bigèbre est construite par la structure du monoïde $G$, à savoir, la comultiplication $\Delta$ d'une fonction $f\colon G\to \mathbb K$ est une fonction de deux variables $\Delta(f)\colon G\times G\to\mathbb K$ définie comme $\Delta(f)(x,y)=f(xy)$,  $\forall x,y\in G$. La counité est définie comme $\varepsilon(f)=f(e)$, où $e$ est l'unité du monoïde $G$. En plus, si $G$ est un groupe, une application définie comme $S(f)(x)=f(x^{-1})$ est un antipode, et, par conséquent, nous avons une algèbre de Hopf commutative.

L'exemple de l'algèbre de Hopf cocommutative est une algèbre enveloppante universelle $U(\mathfrak g)$ d'une algèbre de Lie $\mathfrak g$. La structure d'algèbre de Hopf est définie comme
\begin{align}
 \Delta(x)&=x\otimes1+1\otimes x, & \varepsilon(x)&=0, & S(x)=-x, \label{DeltaLie}
\end{align}
où $x\in\mathfrak g$.

\begin{de}
 Une algèbre de Hopf $\HH$ est appelée quasi-triangulaire, s'il existe un élément inversible $\R$ appartenant à l'espace $H\otimes H$ ou à une quelconque complétion, tel que
\begin{align}
 \Delta^{op}(a)&=\R\Delta(a)\R^{-1}, \label{cmop_RcmR} \\
 (\id\otimes\Delta)\R&=\R^{(13)}\R^{(12)}, \label{cmR_R13R12} \\
 (\Delta\otimes\id)\R&=\R^{(13)}\R^{(23)}. \label{cmR_R13R23}
\end{align}
L'élément $\R$ est appelé une matrice $R$ universelle. Si, en plus, $\R^{-1}=\R^{(21)}$, l'algèbre de Hopf $\HH$ est appelée triangulaire.
\end{de}

La condition~\eqref{cmop_RcmR} est plus faible que la condition de la cocommutativité. Le cas cocommutatif correspond à la matrice $R$ universelle $\R=1$. Toutes algèbres de Hopf, aux quelles nous nous intéressons, peuvent être présentées comme des algèbres de Hopf quasi-triangulaires ou comme leurs sous-algèbres. Les formules~\eqref{cmop_RcmR} et \eqref{cmR_R13R12} (ou \eqref{cmR_R13R23}) entraînent le fait suivant.

\begin{prop}
 Si $\R$ est une matrice $R$ universelle pour une algèbre de Hopf quasi-tri\-an\-gu\-laire, elle satisfait:
\begin{gather}
 \R^{(12)}\R^{(13)}\R^{(23)}=\R^{(23)}\R^{(13)}\R^{(12)}, \label{EYBU} \\
 (\varepsilon\otimes\id)\R=(\id\otimes\varepsilon)\R=1, \\
 (S\otimes\id)\R=(\id\otimes S^{-1})\R=\R^{-1}, \\
 (S\otimes S)\R=\R.
\end{gather}
En particulier, la relation~\eqref{EYBU} est appelée Équation de Yang-Baxter Universelle.
\end{prop}

Dans la plupart des cas, les algèbres de Hopf non-cocommutatives, surtout les algèbres de Hopf quasi-triangulaires, sont apparues comme des déformations des algèbres cocommutatives -- des algèbres enveloppantes universelles. Nous ne considérerons que des déformations à un paramètre appelées des quantifications de Drinfeld. Pour ces dernières nous avons besoin d'une notion des bigèbres de Lie.

\begin{de}
Une algèbre de Lie $\mathfrak g$ (sur $\mathbb K$) est appelée bigèbre de Lie si elle est munie d'une application linéaire $\delta\colon\mathfrak g\to\mathfrak g\otimes\mathfrak g$ tel que
\begin{gather}
 \delta^{op}=-\delta, \label{cocr_asym} \\
 \delta([x,y])=[\delta(x),\Delta(y)]+[\Delta(x),\delta(y)], \label{uncocycle} \\
 Alt\circ(\delta\otimes\id)\circ\delta=0, \label{cocom_Lie}
\end{gather}
où $Alt$ est un alternateur de trois éléments :
\begin{align}
Alt(x_1\otimes x_2\otimes x_2)=\sum_{\sigma\in S_3}(-1)^\sigma (x_{\sigma(1)}\otimes x_{\sigma(2)}\otimes x_{\sigma(3)}),
\end{align}
les éléments $\Delta(x)$ et $\Delta(y)$ sont définis par la formule~\eqref{DeltaLie}.  L'application $\delta$ est appelée cocrochet de la bigèbre de Lie $\mathfrak g$.
\end{de}

La condition~\eqref{cocr_asym} et \eqref{uncocycle} signifient une antisymétrie du cocrochet et une compatibilité entre les crochets et le cocrochet, la condition~\eqref{cocom_Lie} est un analogue de l'identité de Jacobi pour le cocrochet.

La condition~\eqref{uncocycle} peut être reformulée en termes des complexes. Considérons un complexe gradué
\begin{align} \label{complexeC}
 0\to C_0\stackrel{\partial}{\to} C_1\stackrel{\partial}{\to}C_2\stackrel{\partial}{\to}C_3\stackrel{\partial}{\to}\ldots,
\end{align}
où $C_0=V=\Hom_{\mathbb K}(\{0\},V)$, $C_k=\Hom_{\mathbb K}(\mathfrak g^{\wedge k},V)$, $V$ est un $\mathfrak g$-module et $\partial\colon C_k\to C_{k+1}$ est une différentielle définie comme
\begin{multline} \label{complexeCdiff}
 (\partial f)(x_0\wedge\ldots\wedge x_k)=\sum_{i=0}^k(-1)^i x_i\cdot f(x_1\wedge\ldots\wedge\check x_i\wedge\ldots\wedge x_k)+ \\
 +\sum_{i<j}(-1)^{i+j} f([x_i,x_j]\wedge x_1\wedge\ldots\wedge\check x_i\wedge\ldots\wedge\check x_j\wedge\ldots\wedge x_k).
\end{multline}
Si $V$ est le $\mathfrak g$-module $\mathfrak g\otimes\mathfrak g$ ou une complétion de ce module, alors $\delta\in C_1$ et la condition~\eqref{uncocycle} est une condition de cocyclicité. C'est pourquoi les applications $\delta$ satisfaisant~\eqref{uncocycle}, en particulier les cocrochets, sont appelées des {\itshape un-cocycles}.

\begin{de}
 Une algèbre de Hopf $\HH$ est appelée une algèbre enveloppante universelle quantique de bigèbre $(\mathfrak g,\delta)$ ou quantification de Drinfeld de la bigèbre $(\mathfrak g,\delta)$ s'il existe une famille d'algèbres de Hopf $\HH_\hbar$ et le nombre $\hbar_0\in\mathbb C$ tels que
 \begin{enumerate}
  \item $\HH_\hbar$ sont isomorphes en sens des espaces vectoriels;
	\item $\HH_0$ et $\HH_{\hbar_0}$ sont isomorphes à $U(\mathfrak g)$ et à $\HH$ respectivement comme des algèbres de Hopf;
	\item on peut choisir des isomorphismes $\varphi_{\hbar}\colon U(\mathfrak g)\to\HH_\hbar$ tels que
\begin{align} \label{Delta_delta}
 (\varphi^{-1}_\hbar\otimes\varphi^{-1}_\hbar) \Big(\Delta_\hbar\big(\varphi_\hbar(x)\big)-\Delta^{op}_\hbar\big(\varphi_\hbar(x)\big)\Big)=\hbar\delta(x)+o(\hbar),
\end{align}
 \end{enumerate} 
$\forall x\in\mathfrak g$, où $\Delta_\hbar\colon\HH_\hbar\to\HH_\hbar\otimes\HH_\hbar$ est une comultiplication dans $\HH_\hbar$.
\end{de}

La formule~\eqref{Delta_delta} et les propriétés de la comultiplication~\eqref{diag_coass} entraînent les conditions~\eqref{cocr_asym}, \eqref{uncocycle} et \eqref{cocom_Lie}.

Revenons aux algèbres quasi-triangulaires. Si $\HH_\hbar$ est une famille d'algèbres de Hopf quasi-triangulaires quantifiant une bigèbre de Lie $(\mathfrak g,\delta)$ et $\R_\hbar$ sont leurs matrices $R$ universelles, on a
\begin{align}
 \R_\hbar=1+\hbar(\varphi_\hbar\otimes\varphi_\hbar)\rr+o(\hbar),
\end{align}
où $\rr$ est un élément de l'espace $\mathfrak g\otimes\mathfrak g$ ou sa complétion tel que
\begin{align} \label{delta_r_x}
 \delta(x)=[\Delta(x),\rr].
\end{align}
En termes de la différentielle~\eqref{complexeCdiff} la relation~\eqref{delta_r_x} est écrite comme
\begin{align}
 \delta=\partial(\rr). \label{delta_par_rr}
\end{align}

La relation~\eqref{Delta_delta} entraîne la contition~\eqref{cocr_asym}, qui entraine le fait que ${\cal C}=\rr^{(12)}+\rr^{(21)}$ est un invariant d'action de l'algèbre de Lie $\mathfrak g$, id est $[\Delta(x),{\cal C}]=0$ pour tous $x\in\mathfrak g$. En plus, l'équation~\eqref{EYBU} entraîne l'Équation de Yang-Baxter Universelle Classique
\begin{align}
 [\rr^{(12)},\rr^{(13)}]+[\rr^{(12)},\rr^{(23)}]+[\rr^{(13)},\rr^{(23)}]=0. \label{EYBUC}
\end{align}

\begin{de}
 Une bigèbre de Lie $(\mathfrak g,\delta)$ est appelée quasi-triangulaire si le un-cocyle $\delta$ est cobord : $\delta=\partial(\rr)$. Si, en plus, $\rr^{(12)}+\rr^{(21)}=0$ elle est appelé triangulaire.
\end{de}

Afin que $(\mathfrak g,\delta)$ soit une bigèbre quasi-triangulaire, il est nécessaire et suffisant qu'il existe $\rr$ tel que $\rr^{(12)}+\rr^{(21)}$ est un invariant de l'action de l'algèbre de Lie $\mathfrak g$ et les équations~\eqref{delta_par_rr} et \eqref{EYBUC} sont satisfaites. En particulier, l'équation~\eqref{EYBUC} entraîne~\eqref{cocom_Lie}. Une algèbre enveloppante universelle quantique de bigèbre quasi-triangulaire (triangulaire) est une algèbre de Hopf quasi-triangulaire (triangulaire).

\section{Groupes quantiques et relation $RLL$}
\label{sec18}

La définition exacte de groupes quantiques dépend de la source. D'habitude on définit les groupes quantiques comme toutes les algèbres de Hopf ou comme tel ou tel leur classe. Quelquefois notion de groupe quantique est utilisée pour des bigèbres ou des quasi-bigèbres, en particulier pour les quasi-algèbre de Hopf, auxquelles on réfère les groupes quantiques dynamiques~\cite{F2} (voir la section~\ref{sec110}).

La notion des groupes quantiques est apparue comme une algèbration de la notion de l'opérateur de Lax. Soit $R(u,v)$ une $R$-matrice $n^2\times n^2$ satisfaisant à~\eqref{EYB} et soit $\hat\phi^{\alpha}_{ij}(u)$ un système de fonctions, où $\alpha=1,\ldots,\hat N_{ij}$, $\hat N_{ij}\in\mathbb Z_{>0}\cup\{\infty\}$. Soit $\hat L^{\alpha}_{ij}$ un système d'opérateurs agissant dans un espace $V$. On suppose que l'opérateur de Lax suivant satisfait~\eqref{RLL}:
\begin{align} \label{LK__}
 \hat L(u)=K(u)+\sum_{ij=1}^n E_{ij}\otimes\sum_{\alpha=1}^{\hat N_{ij}}\hat L^{\alpha}_{ij}\hat\phi^{\alpha}_{ij}(u),
\end{align}
où $K(u)$ est une matrice sur $\mathbb C$. Si nous avons un autre opérateur de Lax
\begin{align}
 \tilde L(u)=\tilde K(u)+\sum_{ij=1}^n E_{ij}\otimes\sum_{\alpha=1}^{\tilde N_{ij}}\tilde L^{\alpha}_{ij}\tilde \phi^{\alpha}_{ij}(u),
\end{align}
où $\tilde L^{\alpha}_{ij}$ agissent dans $\tilde V$ et $\tilde\phi^{\alpha}_{ij}(u)$ est un autre système de fonctions, alors $\tilde L(u)\hat L(u)=\tilde L_2(u)\hat L_1(u)$ est un opérateur de Lax pour l'espace $V\otimes\tilde V$, où les indices $1$, $2$ soulignent que les opérateurs $\hat L^{\alpha}_{ij}$ présents dans l'expression $\tilde L_2(u)L_1(u)$ agissent dans $V$ et $\tilde L^{\alpha}_{ij}$ -- dans $\tilde V$. Cet opérateur de Lax est de la forme~\eqref{LK__} avec une matrice constante $\tilde K(u)K(u)$, un système de fonctions et un système d'opérateurs correspondants.

Les relations $RLL$ permettent de construire une algèbre $\hat\A$ correspondant à l'opérateur de Lax~\eqref{LK__}. Définissons l'algèbre $\hat\A$ comme une algèbre engendrée par symboles $L^{\alpha}_{ij}$ avec des relations de commutation obtenues de la relation $RLL$~\eqref{RLL} pour l'opérateur de Lax $L(u)\in\End\mathbb C^n\otimes\hat\A$ défini par
\begin{align} \label{LK_}
 L(u)=K(u)+\sum_{ij=1}^n E_{ij}\otimes\sum_{\alpha=1}^{\hat N_{ij}}L^{\alpha}_{ij}\hat\phi^{\alpha}_{ij}(u),
\end{align}
Alors, l'application $\pi\colon L^{\alpha}_{ij}\mapsto\hat L^{\alpha}_{ij}$ donne une représentation de cette algèbre $\hat\A$ dans l'espace $V$. On peut écrire la définition de $\pi$ comme $\pi(L(u))=\hat L(u)$. L'opérateur de Lax $L_2(u)L_1(u)\in\End\mathbb C^n\otimes\hat\A\otimes\hat\A$ satisfait aussi à la relation $RLL$~\eqref{RLL}, et, ainsi, il apparaît une question naturelle : peut-on définir la comultiplication comme
\begin{align}
 \Delta L(u)=L_2(u)L_1(u). \label{Delta_L}
\end{align}
La formule~\eqref{Delta_L} signifie
\begin{align}
 \Delta L_{ij}(u)=\sum_{k=1}^n L_{kj}(u)\otimes L_{ik}(u), \label{Delta_L_el}
\end{align}
où $L_{ij}(u)$ sont des éléments matriciels de la matrice $L(u)$. Cette définition a un sens si et seulement si l'opérateur de Lax~\eqref{Delta_L} est représenté à la même forme~\eqref{LK_} (avec les mêmes $K(u)$, $\hat N_{ij}$ et $\hat\phi^{\alpha}_{ij}(u)$, mais des autres $L^{\alpha}_{ij}$). Mais pourtant, en générale, l'opérateur de Lax~\eqref{Delta_L} n'est pas représenté à la forme~\eqref{LK_}. Supposons qu'il existe une matrice $B(u)$, inversible sur $\mathbb C$, telle que $[R(u,v),B(u)\otimes B(u)]=0$ et $K(u)^2=B(u)K(u)$. Alors, on peut remplacer $L(u)$ à $B(u)^{-1}L(u)$. Donc, on obtient un opérateur de Lax de la forme~\eqref{LK_}, où la matrice $K(u)$ est un idempotent. Si $K(u)$ est inversible, on peut choisir $B(u)=K(u)$. En suite nous supposerons toujours que $K(u)^2=K(u)$. En particulier, pour $K(u)$ inversible cela signifie $K(u)=1$. Alors, l'opérateur de Lax~\eqref{Delta_L} est représenté à la forme
\begin{align} \label{LLK_}
 L_2(u)L_1(u)=K(u)+\sum_{ij=1}^n E_{ij}\otimes\sum_{\alpha=1}^{\check N_{ij}}\check L^{\alpha}_{ij}\check\phi^{\alpha}_{ij}(u),
\end{align}
où $\check L^{\alpha}_{ij}\in\hat\A\otimes\hat\A$ et $\check\phi^{\alpha}_{ij}(u)$ est un autre système de fonctions. Si pour $k,l,\beta$ certains $\check L^{\beta}_{kl}\ne0$ et la fonction $\check\phi^{\beta}_{kl}(u)$ est linéairement indépendante des fonctions $\hat\phi^{\alpha}_{ij}(u)$, alors, \eqref{Delta_L} nous donne une inégalité $\Delta(0)\ne0$, qui est contraire à la linéarité de la comultiplication. Pour que la formule~\eqref{Delta_L} définisse une comultiplication il faut construire une algèbre plus universelle que l'algèbre $\hat\A$.

On dit que le système des fonctions $\hat\phi^{\alpha}_{ij}(u)$ est complet si un opérateur de Lax~\eqref{LK_} construit par ce système est tel que le produit $L_2(u)L_1(u)$ est représenté à la même forme, id est $\check N_{ij}=\hat N_{ij}$ et $\check\phi^{\alpha}_{ij}(u)=\hat\phi^{\alpha}_{ij}(u)$. Dans le cas $K(u)=1$ la condition de complétude du système $\hat\phi^{\alpha}_{ij}(u)$ peut être écrite
\begin{align}
 \hat\phi^{\alpha}_{ik}(u)\hat\phi^{\beta}_{kj}(u)=\sum_{\gamma=1}^{\hat N_{ij}}c^{\alpha\beta}_{\gamma}(i,j;k)\hat\phi^{\gamma}_{ij}(u),
\end{align}
où $c^{\alpha\beta}_{\gamma}(i,j;k)\in\mathbb C$ sont des coefficients quelconques.

Soit un opérateur de Lax donné de la forme~\eqref{LK__}, où $K(u)^2=K(u)$. Considérons l'algèbre des fonctions $\mathfrak O$ engendrée par les fonctions $\hat\phi^{\alpha}_{ij}(u)$ et les éléments matriciels $K_{ij}(u)$. Élargissons le système de fonctions $\hat\phi^{\alpha}_{ij}(u)$ avec des fonctions de $\mathfrak O$ jusqu'à un système complet $\phi^{\alpha}_{ij}(u)$, $i,j=1,\ldots,n$, $\alpha=1,\ldots,N_{ij}$. Puis, on construit une algèbre $\A$ engendrée par $L^{\alpha}_{ij}$, où $i,j=1,\ldots,n$, $\alpha=1,\ldots,N_{ij}$, avec les relations de commutation~\eqref{RLL}, où
\begin{align} \label{LK}
 L(u)=K(u)+\sum_{ij=1}^n E_{ij}\otimes\sum_{\alpha=1}^{N_{ij}}L^{\alpha}_{ij}\phi^{\alpha}_{ij}(u).
\end{align}
C'est une algèbre plus universelle que $\hat\A$ et l'opérateur de Lax~\eqref{LK__} définit aussi une représentation de $\A$ par $L(u)\mapsto\hat L(u)$. Cette algèbre est munie de la comultiplication $\Delta\colon\A\to\A\otimes\A$ définie par~\eqref{Delta_L}. Si $K(u)=1$ elle est munie encore de la counité
\begin{align}
 \varepsilon\big(L(u)\big)=1, \label{vareps_L}
\end{align}
et, par conséquent, $\A$ est une bigèbre. Si, en plus, l'opérateur de Lax $L(u)$ est anti-inversible, id est il existe une matrice $L(u)^{ai}$ avec des éléments matriciels $L_{ij}(u)^{ai}\in\A$ tel que
\begin{align}
 L_{kj}(u)^{ai}L_{ik}(u)=L_{kj}(u)L_{ik}(u)^{ai}=\delta_{ij},
\end{align}
alors $\A$ est une algèbre de Hopf avec l'antipode
\begin{align}
 S\big(L(u)\big)=L(u)^{ai}. \label{antip_L}
\end{align}

D'habitude, $N_{ij}=\infty$ pour système de fonctions complet et, par conséquent, la somme~\eqref{LK} doit se comprendre comme une limite et $\A$ doit être complétée par une topologie correspondante. Ainsi, l'opérateur de Lax~\eqref{LK} est un élément de quelque complétion de $\End\mathbb C^n\otimes{\mathfrak O}\otimes\A$. Quelquefois il est raisonnable considérer l'opérateur de Lax comme un élément d'une complétion de $\End(\mathbb C^n\otimes{\mathfrak O})\otimes\A$.

Nous avons obtenu une algèbre de Hopf $\A$ engendrée par $L^\alpha_{ij}$. Les produits ordonnés des générateurs $L^\alpha_{ij}$ donnent une base de PBW, si une $R$-matrice $R(u,v)$ satisfait l'équation de Yang-Baxter~\eqref{EYB}. L'algèbre $\A$ n'est pas quasi-triangulaire, mais elle peut être complétée jusqu'à une algèbre quasi-triangulaire en utilisant la construction de double de Drinfeld.


\section{Double de Drinfeld et représentation d'évaluation}
\label{sec19}

Soit $\A$ une bigèbre sur $\mathbb K$ avec une multiplication $\mu\colon\A\otimes\A\to\A$, une comultiplication $\Delta\colon\A\to\A\otimes\A$ et une counité $\varepsilon\colon\A\to\mathbb K$. On introduit une application $\eta\colon\mathbb K\to\A$ par la règle $\eta(\alpha)=\alpha\cdot1$, où $1$ est l'unité de $\A$. Par abus de langage, cette application $\eta$ est aussi appelée unité de $\A$. Considérons un espace dual $\A^*$ et munissons-le de la multiplication $\Delta^*\colon\A^*\otimes\A^*\to\A^*$, la comultiplication $\mu^*\colon\A^*\to\A^*\otimes\A^*$, l'unité $\varepsilon^*\colon\mathbb K\to\A^*$ et la counité $\mu^*\colon\A^*\to\mathbb K$, où
\begin{align}
 \la\Delta^*(x\otimes y),a\ra_{\A^*\times\A}&=\la x\otimes y,\Delta(a)\ra_{\A^*\otimes\A^*\times\A\otimes\A}, \\
 \la\mu^*(x),a\otimes b\ra_{\A^*\otimes\A^*\times\A\otimes\A}&=\la x,\mu(a\otimes b)\ra_{\A^*\times\A}, \\
 \la\varepsilon^*(\alpha),a\ra_{\A^*\times\A}&=\alpha\cdot\varepsilon(a), \\
 \eta^*(x)&=\la x,\eta(1)\ra_{\A^*\times\A},
\end{align}
$\forall x,y\in\A^*$, $\forall a,b\in\A$, $\alpha\in\mathbb K$. La bigèbre obtenue est appelée une bigèbre duale à la bigèbre $\A$ et elle est désignée par $\A^*$. Si $\A$ est une algèbre de Hopf avec un antipode $S\colon\A\to\A$, alors $\A^*$ possède l'antipode $S^*\colon\A^*\to\A^*$, où
\begin{align}
 \la S^*(x),a\ra=\la x,S(a)\ra,
\end{align}
$\forall x\in\A^*$, $\forall a\in\A$.

Soit $\A^{cop}$ une bigèbre $\A$ avec la comultiplication $\Delta$ remplacée par $\Delta^{op}$. Si $\A$ est une algèbre de Hopf, alors $\A^{cop}$ est aussi une algèbre de Hopf avec antipode $S^{-1}$.

Considérons l'espace vectoriel $\HH=\A^*\otimes\A^{cop}$, où $\A$ est une algèbre de Hopf. On décrit une structure d'algèbre de Hopf sur $\HH$ :
\begin{align}
 \mu_{\HH}\big((x\otimes a)\otimes(y\otimes b)\big)&\equiv(x\otimes a)\cdot(y\otimes b)= \notag \\
  =x\cdot\big\la(\mu^*\otimes\id)&\circ\mu^*(y)\stackrel{\otimes}{,}(S^{-1}\otimes\id\otimes\id)\circ(\Delta^{op}\otimes\id)\circ\Delta^{op}(a)\big\ra_{1,3}\cdot b, \label{DdDmu} \\
 \Delta_{\HH}(x\otimes a)&=\mu^*(x)\otimes\Delta^{op}(a), \\
 \eta_{\HH}(\alpha)&=\varepsilon^*(\alpha)\otimes\eta(1), \\
 \varepsilon_{\HH}(x\otimes a)&=\eta^*(x)\varepsilon(a), \label{DdDve}
\end{align}
où $\la x\otimes y\otimes z\stackrel{\otimes}{,}a\otimes b\otimes c\ra_{1,3}\hm=\la x,a\ra_{\A^*\times\A}\cdot\la z,c\ra_{\A^*\times\A}\cdot yb$, $\forall x,y,z\in\A^*$, $\forall a,b,c\in\A$. La bigèbre $\HH$ avec la structure~\eqref{DdDmu} -- \eqref{DdDve} construite par la structure de $\A$ est une algèbre de Hopf. Elle est appelée {\itshape double de Drinfeld} d'algèbre de Hopf $\A^{cop}$ et désignée par ${\cal D}(\A^{cop})$. Les algèbres $\A^*$ et $\A^{cop}$ sont plongées dans ${\cal D}(\A^{cop})$ par $x\mapsto x\otimes1$ et $a\mapsto1\otimes a$ respectivement. En particulier, $xa=x\otimes a$. Ainsi, on peut considérer $\A^*$ et $\A^{cop}$ comme des sous-algèbres de Hopf de Double ${\cal D}(\A^{cop})$.

\begin{prop}
 Le double de Drinfeld ${\cal D}(\A^{cop})$ est une algèbre de Hopf quasi-triangu\-laire avec la matrice $R$ universelle
\begin{align}
 \R=\sum_k e^k\otimes e_k, \label{Ru_ee}
\end{align}
où $\{e_k\}$ est une base de $\A$ et $\{e^k\}$ est une base duale de $\A^*$. En générale, la somme dans la formule~\eqref{Ru_ee} est comprise comme une série formelle.
\end{prop}

Soit ${\mathfrak K}_u$ une algèbre associative commutative réalisée comme une algèbre de fonctions de la variable $u$ ou sa complétion. Supposons que l'algèbre ${\mathfrak K}_u$ possède un produit scalaire non-dégénéré invariant $\la\cdot,\cdot\ra_u$. Soit $\A$ une algèbre de Hopf sur $\mathbb C$ engendrée par des éléments $x_i$ tels que les produits ordonnés de ces éléments donnent une base de PBW. Un homomorphisme $\Pi^+_u\colon\A\to\End(\mathbb C^n)\otimes{\mathfrak K}_u\subset\End(\mathbb C^n\otimes{\mathfrak K}_u)$ est appelé {\itshape représentation d'évaluation} si $\Pi^+_u(x_i)$ sont linéairement indépendants et $\Pi^+_u(x_ix_j)=0$. On fixe une base de PBW dans $\A$ et on désigne par $x^i\in\A^*$ des éléments duaux de $x_i$. On définit un homomorphisme $\Pi^-_u\colon\A^*\to\End(\mathbb C^n)\otimes{\mathfrak K}_u$ par la formule
\begin{align}
 \la\Pi^-_u(x^i),\Pi^+_u(x_j)\ra_{\End(\mathbb C^n)\otimes{\mathfrak K}_u}=\delta^i_j,
\end{align}
où $\la A\otimes s(u),B\otimes t(u)\ra_{\End(\mathbb C^n)\otimes{\mathfrak K}_u}=\tr(AB)\la s(u),t(u)\ra_u$. Nous supposerons que $\Pi^-_u$ est une représentation d'évaluation de l'algèbre $\A^*$. En particulier, $\Pi^-_u(x^i)$ sont linéairement indépendants et $\Pi^-_u(x^ix^j)=0$. 

Considérons le double de Drinfeld $\HH={\cal D}(\A^{cop})$ avec la matrice $R$ universelle~\eqref{Ru_ee} et introduisons les notations
\begin{align}
 L_+(u)&=(\Pi^-_u\otimes\id)\R, & L_-(u)&=(\Pi^+_u\otimes\id)(\R^{-1})^{(21)},\label{LpmXx} \\
 R_+(u,v)&=(\Pi^-_u\otimes\Pi^+_v)\R, & R_-(u,v)&=(\Pi^+_u\otimes\Pi^-_v)(\R^{-1})^{(21)}. \label{RpRm_def}
\end{align}
En désignant $X_i(u)=\Pi^+_u(x_i)$, $X^i(u)=\Pi^-_u(x^i)$ et en prenant en considération $\R^{-1}\hm=\sum\limits_k S^*(e^k)\otimes e_k\hm=\sum\limits_k e^k\otimes S(e_k)$, où $S$ est l'antipode de $\A$, nous obtenons
\begin{align}
  L_+(u)&=\sum_i X^i(u)\otimes x_i, & L_-(u)&=\sum_{i,j}S^i_j X_i(u)\otimes x^j, \label{LXx} \\
  R_+(u,v)&=\sum_i X^i(u)\otimes X_i(v), & R_-(u,v)&=\sum_{i,j}S^i_j X_i(u)\otimes X^j(v), \label{RXX} 
\end{align}
où les coefficients $S^i_j$ de l'antipode de $\A$ :
\begin{align}
 S^*(x^i)&=\sum\limits_j S^i_j x^j, & S(x_i)&=\sum\limits_j S^j_i x_j.
\end{align}
Les sommes par $i$ et $j$ peuvent être comprises comme des séries convergeant sous quelque topologie.

L'Équation de Yang-Baxter universelle~\eqref{EYBU} entraîne l'Équation de Yang-Baxter pour les matrices~\eqref{RpRm_def} et les relations $RLL$ suivantes
\begin{align}
 R_+(u,v)L^{(1)}_+(u)L^{(2)}_+(v)=L^{(2)}_+(v)L^{(1)}_+(u)R_+(u,v), \label{RpLpLp} \\
 R_+(u,v)L^{(1)}_-(u)L^{(2)}_-(v)=L^{(2)}_-(v)L^{(1)}_-(u)R_+(u,v), \label{RpLmLm} \\
 R_+(u,v)L^{(1)}_+(u)L^{(2)}_-(v)=L^{(2)}_-(v)L^{(1)}_+(u)R_+(u,v). \label{RpLpLm} 
\end{align}
En agissant par $\Pi_u\otimes\id\otimes\id$ et par $\id\otimes\id\otimes\Pi_u$ à \eqref{cmR_R13R12} et \eqref{cmR_R13R23} on obtient
\begin{align}
 \Delta_{\HH} L_+(u)&=L^{(02)}_+(u)L^{(01)}_+(u), &
 \Delta_{\HH} L_-(u)&=L^{(02)}_-(u)L^{(01)}_-(u), \label{DeltaHLpm}
\end{align}
où 
\begin{align}
 L^{(01)}_+(u)&=\sum_i X^i(u)\otimes x_i\otimes1, & 
 L^{(02)}_+(u)&=\sum_i X^i(u)\otimes1\otimes x_i, \\
 L^{(01)}_-(u)&=\sum_{i,j}S^i_j X_i(u)\otimes x^j\otimes1, &
 L^{(02)}_-(u)&=\sum_{i,j}S^i_j X_i(u)\otimes1\otimes x^j.
\end{align}
Les formules~\eqref{LpmXx} signifient que les matrices $L_\pm(u)$ suffisent pour décrire l'algèbre de Hopf $\HH$. En effet, $\HH$ peut être définie comme l'algèbre engendrée par $L^{\alpha}_{\pm;ij}$ avec les relations de commutation~\eqref{RpLpLp}--\eqref{RpLpLm}, la comultiplication, la counité et l'antipode définies sur les générateurs par les formules~\eqref{DeltaHLpm} et
\begin{align}
 \varepsilon_{\HH}\big(L_\pm(u)\big)&=1, & S_{\HH}\big(L_\pm(u)\big)&=L_\pm(u)^{ai},
\end{align}
où
\begin{align}
 L_\pm(u)=1+\sum_{i,j=1}^n E_{ij}\otimes\sum_{\alpha=1}^{N_{\pm;ij}}L^{\alpha}_{\pm;ij}\phi^{\alpha}_{\pm;ij}(u). \label{Lpm_dec}
\end{align}
La matrice $L_+(u)$ et $L_-(u)$ sont appelées des opérateurs de Lax positif et négatif. Elles décrivent les sous-algèbres de Hopf $\A^{cop}$ et $\A^*$ respectivement. 

\begin{rem}
 Si les représentations d'évaluations ne satisfaisaient pas aux conditions $\Pi^+_u(x_ix_j)=0$ et $\Pi^-_u(x^ix^j)=0$, les opérateurs $L_+(u)$ et $L_-(u)$ contiendraient les produits $x_ix_j$ et $x^ix^j$. Dans ce cas ne peut pas définie comme une algèbre engendrée par $L^{\alpha}_{\pm;ij}$, parce que $L^{\alpha}_{\pm;ij}$ ne sont pas algébriquement indépendant.
\end{rem}

\begin{rem}
Les relations~\eqref{RpLpLp}, \eqref{RpLmLm}, \eqref{RpLpLm} sont équivalentes aux relations
\begin{align}
 R_-(u,v)L^{(1)}_+(u)L^{(2)}_+(v)=L^{(2)}_+(v)L^{(1)}_+(u)R_-(u,v), \label{RmLpLp} \\
 R_-(u,v)L^{(1)}_-(u)L^{(2)}_-(v)=L^{(2)}_-(v)L^{(1)}_-(u)R_-(u,v), \label{RmLmLm} \\
 R_-(u,v)L^{(1)}_-(u)L^{(2)}_+(v)=L^{(2)}_+(v)L^{(1)}_-(u)R_-(u,v). \label{RmLpLm} 
\end{align}
\end{rem}

Soit $\A^{cop}$ une algèbre de Hopf considérée à la section~\ref{sec18} (et désignée par $\A$), id est une algèbre de Hopf générée par $L^{\alpha}_{ij}$ avec les relations de commutation~\eqref{RLL} et la structure d'algèbre de Hopf définie par~\eqref{Delta_L_el}, \eqref{vareps_L} et \eqref{antip_L}. Le double de Drinfeld de cette algèbre est une algèbre de Hopf $\HH={\cal D}(\A^{cop})$ qui peut être décrite par les opérateurs de Lax~\eqref{Lpm_dec}, où $L^{\alpha}_{+;ij}=L^{\alpha}_{ij}$, $N_{+;ij}=N_{ij}$, $\phi^{\alpha}_{+;ij}(u)=\phi^{\alpha}_{ij}(u)$, $L_+(u)=L(u)$. On peut choisir la $R$-matrice $R(u,v)$ pour l'algèbre $\A^{cop}$ telle que $R(u,v)=\rho_+(u,v)R_+(u,v)$ et/ou $R(u,v)=\rho_-(u,v)R_-(u,v)$, où $\rho_\pm(u,v)$ sont des fonctions complexes. Les matrices $R^+(u,v)$ et $R^-(u,v)$ sont des matrices $R$ avec des normalisations naturelles.

Ainsi, chaque algèbre de Hopf étant décrite par les relations $RLL$ quadratiques peut être plongée dans une algèbre de Hopf quasi-triangulaire -- le double de Drinfeld -- qui est décrite par une paire d'opérateurs de Lax. En prenant la matrice $R$ rationnelle, par exemple, nous obtenons l'Yangien~\label{Yangien} $\A=Y(\gln)$ connue de travails pionniers. Son double de Drinfeld ${\cal D}Y(\gln)$ est introduit et examiné dans les travails de Khoroshkin et Tolstoy~\cite{Kh}. La matrice~\eqref{R_trig} correspond à $U_q({\mathfrak b}_+)$ -- l'algèbre enveloppante universelle quantique de l'algèbre de Borel de l'algèbre Affine $\widehat\slt$. Son double de Drinfeld est aussi une algèbre connue $U_q(\widehat\slt)$.

\section{Relations $RLL$ dynamiques}
\label{sec110}

Comme nous l'avons mentionné, les modèles des faces sont liés à l'Équation de Yang-Baxter Dynamique et ses solutions -- des matrices $R$ dynamiques. Dans~\cite{Bax2} Baxter a décrit un modèle des faces associé avec le modèle 8-vertex. Nous l'appellerons modèle SOS (Solid-On-Solid) ou modèle SOS elliptique. C'est une autre généralisation du modèle 6-vertex. La limite trigonométrique du ce modèle nous donne un modèle des faces avec des poids de Boltzmann trigonométriques. Ce dernier est appelé un modèle SOS trigonométrique et sa limite par rapport au paramètre dynamique équivalent au modèle 6-vertex. La matrice des poids de Boltzmann pour ce modèle a été généralisée par Felder au cas $\gln$~\cite{F2} et cette matrice est appelée matrice $R$ de Felder.

L'Équation de Yang-Baxter Dynamique est apparue en même temps dans plusieurs domaines. Pour la première fois elle a été découverte par J.\,L.\,Gervais et A.\,Neveu en étudiant la théorie de Liouville~\cite{GN}. Après  G.\,Felder a écrit cette équation développant son approche à la quantification de l'équation de Knizhnik-Zamolodchikov-Bernard~\cite{F1}. Finalement, le rôle de cette équation pour la quantification des modèles de Calogero-Moser~\cite{ABB} a été décelé. L'adjectif "dynamique" a été proposé après l'article~\cite{BBB}, où le sens algébrique de l'Équation de Yang-Baxter Dynamique a été expliqué.

Soit ${\mathfrak h}$ une sous-algèbre de Cartan d'une algèbre de Lie semi-simple et soit $\{h_k\}$ sa base, où $k=1,\ldots,r$ et $r=\dim{\mathfrak h}$. Soient $H_k\in\End\mathbb C^n$ des opérateurs représentant des éléments de Cartan $h_k$ dans $\End\mathbb C^n$. 
Nous utiliserons la notation
\begin{align}
 F(\lambda_k+P_k)=\sum_{i_1,\ldots,i_r=0}^{\infty}\frac{1}{i_1!\cdots i_r!} \frac{\partial^{i_1+\ldots+i_r}F(\lambda_1,\ldots,\lambda_r)}{\partial\lambda_1^{i_1}\cdots\partial\lambda_r^{i_r}}P_1^{i_1}\cdots P_r^{i_r}
\end{align}
pour une fonction $F\colon\mathbb C^n\to{\mathfrak A}$ et $P_k\in{\mathfrak A}$, où $k=1,\ldots,r$ et $\mathfrak A$ est une algèbre associative unitaire. La matrice $R(u,v;\lambda_k)\in\End\mathbb C^n\otimes\End\mathbb C^n$ dépendant de paramètres spectraux $u$ et $v$ et de paramètres dynamiques $\lambda_k$ est appelée matrice $R$ dynamique si elle satisfait Équation de Yang-Baxter Dynamique
\begin{multline}
 R^{(12)}(u_1,u_2;\lambda_k)R^{(13)}(u_1,u_3;\lambda_k+\hbar H^{(2)}_k)R^{(23)}(u_2,u_3;\lambda_k)= \\
 =R^{(23)}(u_2,u_3;\lambda_k+\hbar H^{(1)}_k)R^{(13)}(u_1,u_3;\lambda_k)R^{(12)}(u_1,u_2;\lambda_k+\hbar H^{(3)}_k). \label{DYBE}
\end{multline}
Soit $L(u;\lambda_k)$ une matrice sur une algèbre non-commutative dépendant d'un paramètre spectral $u$ et des paramètres dynamiques $\lambda_k$. La matrice $L(u;\lambda_k)$ est appelée un opérateur de Lax dynamique si elle satisfait
\begin{align}
 R^{(12)}(u-v;\lambda_k)L^{(1)}(u;\lambda_k+\hbar H^{(2)}_k)L^{(2)}(v;\lambda_k)=& \notag \\
 =L^{(2)}(v;\lambda_k+\hbar H^{(1)}_k)&L^{(1)}(u;\lambda_k)R^{(12)}(u-v;\lambda_k+\hbar h_k). \label{DRLL} 
\end{align}
La relation~\eqref{DRLL} est appelée une relation $RLL$ dynamique.

\begin{prop} \label{lem_LL}
 Si $L_1(u;\lambda_k)\in\End(\mathbb C^n)\otimes\mathfrak R_1$ et $L_2(u;\lambda_k)\in\End(\mathbb C^n)\otimes\mathfrak R_2$ sont deux opérateurs de Lax dynamiques alors leur produit matriciel
\begin{align} \label{cmdL}
 L_2(u;\lambda_k)L_1(u;\lambda_k+\hbar h_{k;2})\in\End(\mathbb C^n)\otimes\mathfrak R_1\otimes\mathfrak R_2
\end{align}
est aussi un opérateur de Lax dynamique satisfaisant~\eqref{DRLL} avec la même matrice $R$ dynamique, un élément de Cartan $h_k=h_{k;1}+h_{k;2}$ et les mêmes opérateurs $H_k$, où $h_{k;i}\in\mathfrak R_i$ sont les éléments de Cartan correspondant aux opérateurs de Lax $L_i(u;\lambda_k)$. Ainsi, si $L_1(u;\lambda_k)$, \ldots, $L_m(u;\lambda_k)$ sont des opérateurs de Lax dynamiques avec les éléments de Cartan $h_{k;1}$, \ldots, $h_{k;2}$, alors la matrice
\begin{align}
 \mathop{\overleftarrow\prod}\limits_{m\ge j\ge1}L_j\big(u;\lambda_k+\hbar\sum_{l=j+1}^{m} h_{k;l}\big)
 \end{align}
est un opérateur de Lax avec les éléments de Cartan $h_k=\sum\limits_{i=1}^m h_{k;i}$.
\end{prop}

La question naturelle apparaît : peut-t-on construire une algèbre de Hopf ou une bigèbre par un opérateur de Lax dynamique et la comultiplication~\eqref{cmdL} par le schéma décrit dans les sections~\ref{sec18} et \ref{sec19}? La comultiplication $\bar\Delta L(u;\lambda_k)=$\eqref{cmdL} a l'air coassociative, mais ce n'est pas un homomorphisme, parce qu'elle dépend des paramètres dynamiques: $\bar\Delta=\bar\Delta(\lambda_k)$, et ainsi
\begin{align}
 \bar\Delta(\lambda_k)L^{(1)}(u;\lambda_k+\hbar H^{(2)}_k)\ne\bar\Delta(\lambda_k+\hbar H^{(2)}_k)L^{(2)}(u;\lambda_k+\hbar H^{(2)}_k),
\end{align}
et par conséquent l'application de $\bar\Delta$ à~\eqref{DRLL} ne donne pas une identité. En plus, on ne peut pas définir une algèbre par la relation~\eqref{DRLL} comme à la section~\ref{sec18}. En effet, l'opérateur de Lax dynamique est de la forme
\begin{align} \label{LKD}
 L(u;\lambda_k)=1+\sum_{ij=1}^n E_{ij}\otimes\sum_{\alpha=1}^{N_{ij}}L^{\alpha}_{ij}(\lambda_k)\phi^{\alpha}_{ij}(u;\lambda_k).
\end{align}
Pour construire une algèbre qui contient les éléments de cet opérateur il faut fixer les paramètres dynamiques $\lambda_k$, générer une algèbre $\A_{\lambda_k}$ par coefficient $L^\alpha_{ij}(\lambda_k)$ et imposer la relation $RLL$ dynamique~\eqref{DRLL}. Mais cette relation définissant une multiplication contient toutes les dérivées $\dfrac{\partial^{i_1+\ldots+i_r}L^\alpha_{ij}(\lambda_k)}{\partial\lambda_1^{i_1}\cdots\partial\lambda_r^{i_r}}$ outre les générateurs $L^\alpha_{ij}(\lambda_k)$. Pour rendre cette définition correcte il est nécessaire de compléter l'ensemble de générateurs par ces dérivées dans les valeurs de $\lambda_k$ fixées. Mais pourtant, il y a deux obstacles. Premièrement, l'opérateur de Lax~\eqref{LKD} ne contient pas les dérivées de $L^\alpha_{ij}(\lambda_k)$ et par conséquent il ne peut par décrire toute l'algèbre construite. Deuxièmement, ces dérivées ne peuvent pas être fonctionnellement indépendantes et les contraintes nécessaires dépendent d'une spécification de la dépendance des coefficients $L^\alpha_{ij}(\lambda_k)$ de $\lambda_k$.

Néanmoins, pour les matrices $R$ dynamiques connues on construit des algèbres $\HH$ qui sont décrites par des paires d'opérateurs de Lax dynamiques. Ces algèbres sont analogues au double de Drinfeld regardé à la section~\ref{sec19}. Elles sont appelées groupes quantiques dynamiques et elles possèdent des comultiplications $\Delta$ liées à la formule~\eqref{cmdL}. Cependant, ces comultiplications ne satisfont pas à la condition de coassociativité $(\Delta\otimes\id)\circ\Delta=(\id\otimes\Delta)\circ\Delta$ (le premier diagramme~\eqref{diag_coass}). Néanmoins, elle satisfait à l'égalité
\begin{align} \label{quasicoass}
 (\id\otimes\Delta)\circ\Delta(a)= \Phi\cdot(\Delta\otimes\id)\circ\Delta(a)\cdot\Phi^{-1},
\end{align}
$\forall a\in\HH$, pour quelque élément inversible $\Phi\in\HH\otimes\HH\otimes\HH$. Par conséquence, les opérateurs de Lax dynamiques décrivent une classe d'algèbres plus générale que l'algèbre de Hopf. La condition de coassociativité pour ces algèbres est remplacée par l'équation~\eqref{quasicoass} appelée une condition de quasi-coassociativité. Ces algèbres ont été introduites dans un travail de Drinfeld~\cite{D90}.

\begin{de}
 L'algèbre $\HH$ sur $\mathbb K$ munie des homomorphismes $\Delta\colon\HH\to\HH\otimes\HH$ et $\varepsilon\colon\HH\to\mathbb K$ et d'un élément inversible $\Phi\in\HH\otimes\HH\otimes\HH$ est appelée quasi-bigèbre s'ils satisfont la condition de quasi-coassociativité~\eqref{quasicoass} et les conditions
\begin{align}
 (\id\otimes\id\otimes\Delta)(\Phi)\cdot(\Delta\otimes\id\otimes\id)(\Phi)&= (1\otimes\Phi)\cdot(\id\otimes\Delta\otimes\id)(\Phi)\cdot(\Phi\otimes1), \\
 (\varepsilon\otimes\id)\circ\Delta&=(\id\otimes\varepsilon)\circ\Delta=\id, \\
 (\id\otimes\varepsilon\otimes\id)(\Phi)&=1.
\end{align}
L'élément $\Phi$ est appelé un coassociateur pour la comultiplication $\Delta$. Cette quasi-bigèbre est appelée quasi-algèbre de Hopf, s'il existe, en plus, une application linéaire $S\colon\HH\to\HH$ et des éléments $\alpha,\beta\in\HH$ satisfaisant les conditions
\begin{align}
 \mu\circ(S\otimes\hat\alpha)\circ\Delta(a)=&\varepsilon(a)\cdot\alpha, \\
 \mu\circ(\check\beta\otimes S)\circ\Delta(a)=&\varepsilon(a)\cdot\beta, \\
 (\mu\otimes\id)\circ\mu\circ(\check\beta\otimes S\otimes\hat\alpha)&(\Phi)=1, \\
 (\mu\otimes\id)\circ\mu\circ(S\otimes\hat\alpha\check\beta\otimes S)&(\Phi^{-1})=1,
\end{align}
$\forall a\in\HH$, où $\mu\colon\HH\otimes\HH\to\HH$ est une multiplication de $\HH$ et $\hat\alpha\colon\HH\to\HH$ et $\check\beta\colon\HH\to\HH$ sont des opérateurs de multiplication par $\alpha$ à gauche et par $\beta$ à droite. La quasi-algèbre de Hopf $\HH$ appelée quasi-triangulaire s'il existe un élément inversible $\R\in\HH\otimes\HH$ tel que
\begin{align}
 \Delta^{op}(a)&=\R\Delta(a)\R^{-1}, \label{cmop_RcmRPhi} \\
 (\id\otimes\Delta)\R&=(\Phi^{-1})^{(231)}\R^{(13)}\Phi^{(213)}\R^{(12)}(\Phi^{-1})^{(123)}, \label{cmR_R13R12Phi} \\
 (\Delta\otimes\id)\R&=\Phi^{(312)}\R^{(13)}(\Phi^{-1})^{(132)}\R^{(23)}\Phi^{(123)}. \label{cmR_R13R23Phi}
\end{align}
\end{de}

En substituant~\eqref{cmR_R13R23Phi} à l'égalité
\begin{align}
 \R^{(12)}(\Delta\otimes\id)\R=(\Delta^{op}\otimes\id)\R\cdot\R^{(12)}
\end{align}
découlant de~\eqref{cmop_RcmRPhi}, nous obtenons l'équation
\begin{align}
 \R^{(12)}\Phi^{(312)}\R^{(13)}(\Phi^{-1})^{(132)}\R^{(23)}\Phi^{(123)}= \Phi^{(321)}\R^{(23)}(\Phi^{-1})^{(231)}\R^{(13)}\Phi^{(213)}\R^{(12)},
\end{align}
qui généralise l'Équation de Yang-Baxter Universelle pour les quasi-algèbres de Hopf.

Si $\Phi=1$ et $\alpha=\beta=1$ alors l'algèbre $\HH$ est une algèbre de Hopf. Parmi les exemples plus simples non-triviaux de quasi-algèbres de Hopf (quasi-triangulaires) il faut mentionner les groupes quantiques dynamiques qui peuvent être obtenus par {\itshape twists dynamiques} d'algèbre de Hopf (quasi-triangulaire). Nous considérons d'abord le twist général suivant~\cite{D90}.

Soit $\HH$ une quasi-bigèbre avec une comultiplication $\Delta$, counité $\varepsilon$ et coassociateur $\Phi$ et soient
\begin{align}
 \tilde\Delta(a)&=\F\cdot\Delta(a)\cdot \F^{-1}, \\
 \tilde\Phi&=\F^{(23)}\cdot(\id\otimes\Delta)\F\cdot\Phi\cdot(\Delta\otimes\id)\F^{-1} \cdot(\F^{-1})^{(12)},
\end{align}
$\forall a\in\HH$, où $\F\in\A\otimes\A$ est un élément inversible tel que
\begin{align}
 (\varepsilon\otimes\id)\F=&1, & (\id\otimes\varepsilon)\F=&1. \label{cuF}
\end{align}
Alors $\HH$ munie par $\tilde\Delta$, $\varepsilon$ et $\tilde\Phi$ est aussi une quasi-bigèbre. L'élément $\F$ est appelé un twist. On dit que la quasi-bigèbre $\tilde\HH=(\HH,\tilde\Delta,\varepsilon,\tilde\Phi)$ est obtenue de la bigèbre $\HH=(\HH,\Delta,\varepsilon,\Phi)$ par twist $\F$. Si $\HH=(\HH,\Delta,\varepsilon,\Phi,S,\alpha,\beta)$ est une algèbre de Hopf alors la bigèbre $\tilde\HH$ avec le même antipode $\tilde S=S$ et les éléments
\begin{align}
 \tilde\alpha&=(S\otimes\hat\alpha)\F^{-1}, & \tilde\beta&=(\check\beta\otimes S)\F
\end{align}
est aussi une algèbre de Hopf. Si, en plus, $\HH$ est une quasi-bigèbre de Hopf quasi-triangulaire alors $\tilde\HH$ est aussi quasi-triangulaire avec une matrice universelle
\begin{align}
 \tilde\R=\F^{(21)}\R \F^{-1}. \label{twistdeR}
\end{align}

Soit $\HH$ une algèbre de Hopf quasi-triangulaire, id est $\Phi=1$ et $\alpha=\beta=1$, et soit $\F\in\HH\otimes\HH$ un élément inversible satisfaisant~\eqref{cuF}. Alors l'algèbre $\tilde\HH$ obtenue de $\HH$ par twist $\F$ est muni du coassociateur
\begin{align} \label{tPhiF}
 \tilde\Phi=\F^{(23)}\cdot(\id\otimes\Delta)\F\cdot(\Delta\otimes\id)\F^{-1}\cdot(\F^{-1})^{(12)}.
\end{align}
En particulier, si $\F$ satisfait
\begin{align}
 \F^{(12)}\cdot(\Delta\otimes\id)\F=\F^{(23)}\cdot(\id\otimes\Delta)\F, \label{cocyc}
\end{align}
alors l'algèbre $\tilde\HH$ est aussi une algèbre de Hopf. La condition~\eqref{cocyc} est appelée une équation d'un cocycle ou une équation de Drinfeld. Une généralisation de cette équation au cas dynamique a été considérée dans~\cite{BBB}. Elle permet d'obtenir une Équation de Yang-Baxter Dynamique en termes des quasi-algèbres de Hopf.

Soit $\F(\lambda_k)=\F(\lambda_1,\ldots,\lambda_r)$ un élément inversible de $\HH\otimes\HH$ dépendant de $r$ paramètres dynamiques et satisfaisant à la condition~\eqref{cuF}. Soient $h_k\in\HH$ des éléments commutant deux par deux -- les éléments de Cartan correspondants. Si $\F$ satisfait à la relation
\begin{align} \label{shcocyc}
 \F^{(12)}(\lambda_k+\hbar h^{(3)}_k)\cdot(\Delta\otimes\id)\F(\lambda_k)=
 \F^{(23)}(\lambda_k)\cdot(\id\otimes\Delta)\F(\lambda_k),
\end{align}
appelée une équation d'un cocycle déplacé, alors la comultiplication et le coassociateur correspondants prennent les formes
\begin{align}
 \Delta(\lambda_k)(a)&=\F(\lambda_k)\cdot\Delta(a)\cdot \F(\lambda_k)^{-1}, \label{twistcmD} \\
 \Phi(\lambda_k)&=\F^{(12)}(\lambda_k+\hbar h^{(3)}_k)\cdot\F^{(12)}(\lambda_k)^{-1}, \label{tPhiFshcocyc} \\
 \R(\lambda_k)&=\F^{(21)}(\lambda_k)\cdot\R\cdot\F(\lambda_k)^{-1}, \label{twistdeRD}
\end{align}
$\forall a\in\HH$. (Nous pouvons omettre un tilde en indiquant explicitement la dépendance des paramètres dynamiques $\lambda_k$, parce que la comultiplication $\Delta$, le coassociateur $\Phi=1$ et la matrice $R$ universelle $\R$ de l'algèbre de Hopf initiale $\HH$ ne dépendent pas de $\lambda_k$). En utilisant la formule
\begin{align}
 \R^{(12)}(\lambda_k+\hbar h^{(3)}_k)&= \Phi^{(213)}(\lambda_k)\cdot\R(\lambda_k)\cdot\Phi^{(123)}(\lambda_k)^{-1}, \label{RPhiD}
\end{align}
découlant de~\eqref{twistdeRD}, on obtient une Équation de Yang-Baxter Universelle Dynamique
\begin{multline}
 \R^{(12)}(\lambda_k)\R^{(13)}(\lambda_k+\hbar h^{(2)}_k)\R^{(23)}(\lambda_k)= \\
 =\R^{(23)}(\lambda_k+\hbar h^{(1)}_k)\R^{(13)}(\lambda_k)\R^{(12)}(\lambda_k+\hbar h^{(3)}_k). \label{DYBEU}
\end{multline}
Appliquant une représentation d'évaluation correspondante nous arrivons à la relation $RLL$ dynamique~\eqref{DRLL} et à l'Équation de Yang-Baxter Dynamique~\eqref{DYBE}.

Enfin, notons que les formules~\eqref{quasicoass}, \eqref{twistcmD} et \eqref{tPhiFshcocyc} entraîne la formule
\begin{align} \label{coassD}
 \big(\id\otimes\Delta(\lambda_k)\big)\circ\Delta(\lambda_k)= \big(\Delta(\lambda_k+\hbar h^{(3)})\otimes\id\big)\circ\Delta(\lambda_k),
\end{align}
qui fonde la règle de coassociativité dynamique pour la comultiplication
\begin{align} \label{cmdLv}
 \Delta(\lambda_k)L(u;\lambda_k)=\F(\lambda_k) \F(\lambda_k+\hbar H_k)^{-1} L^{(02)}(u;\lambda_k)L^{(01)}(u;\lambda_k+\hbar h^{(2)}_{k}).
\end{align}

\selectlanguage{english}

\chapter{Transition function for the Toda chain}
\label{sec2}

In this chapter we present the main results of the work~\cite{S1}(see Appendix~\ref{SA1}). This work was inspired by the article~\cite{Derkachov} devoted to the relation of Separation of Variables (SoV) for $XXX$-model and the Baxter $Q$-operator for this model. The main idea of SoV is to transform the wave functions of the model such that the eigenfunctions of quantum integrals of motion become the products of functions of one variable~\cite{Sklyanin}. It was shown in~\cite{Derkachov} that in the $XXX$-model case this transformation can be constructed as consecutive application of operators called $\Lambda$-operators and the kernels of these operators are degenerated kernels of the Baxter $Q$-operators. Moreover, those functions of one variables are eigenvalues of the Baxter operators.

The first ideas of SoV for the periodic Toda chain were proposed by Gutzwiller~\cite{Gutzwiller}. It turned out that the transition function for the $(N+1)$-particle periodic Toda chain is constructed as an eigenfunction of the $N$-particle open Toda chain with a factor depending on the coordinate of $(N+1)$-st particle. This idea was successfully applied by him to the few particle case. The transition function for the $N$-particle periodic Toda chain was first obtained by Sklyanin using $R$-matrix formalism~\cite{Sklyanin}. The most resent description of the SoV method for the periodic Toda chain in terms of the Lax operators and $R$-matrices was done in~\cite{Kharchev_P}. The separated variables of the periodic Toda chain parametrize the eigenvalues of the integrals of motion of the open Toda chain. Sklyanin proposed to search the eigenfunctions of these integrals of motion as integrals transformation of eigenfunctions of the smaller chain over these variables, what was realized most completely in~\cite{Kharchev_O} resulting the producing an integral representation for the finite open Toda chain eigenfunctions, which was called a Mellin-Barns representation.

We apply methods of the paper~\cite{Derkachov} to obtain the eigenfunctions of the open Toda chain as a product of $\Lambda$-operators. We developed the method of a triangulation of the Lax matrices appeared in~\cite{Pasquier} and used in~\cite{Derkachov} for $XXX$-model. The triangulation is implemented by a gauge transformation parametrized by variables $y_0,\ldots,y_N$. In the periodic case one has to impose the condition $y_0=y_N$~\cite{Pasquier} and the method produces Baxter $Q$-operators for the periodic Toda chain model. Following~\cite{Derkachov} we impose an open boundary conditions: $y_0\to -\infty$, $y_N\to +\infty$ to construct the $\Lambda$-operator. More exactly, the kernel of the $\Lambda$-operator corresponding to the $N$-particle open Toda chain (and, consequently, to the $(N+1)$-particle periodic Toda chain) can be obtained by the limit $y_0\to -\infty$, $y_N\to +\infty$ of the kernel of the Baxter $Q$-operator corresponding to the $(N+1)$-particle periodic Toda chain. Thus the $\Lambda$-operator and the Baxter $Q$-operator for the periodic Toda chain correspond to the different choices of the boundary conditions in the method of triangulation of the Lax matrices. In particular, this explain the likeness of the properties of these operators.

The difference in the boundary conditions used by triangulation method leads in turn to the fact that the $\Lambda$-operator transforms a function of $N-1$ variables to a function of $N$ variables, as the Baxter $Q$-operator transoms a function of $N$ variables to a function of the same number of variables. The properties of the product of these operators imply that the action of a $\Lambda$-operator on an $(N-1)$-particle eigenfunction of open chain gives an $N$-particle eigenfunction. The product of $\Lambda$-operators acting on a constant (a function of zero number of variables) are the eigenfunctions for the open Toda chain with eigenvalues parametrized by the spectral parameters of these $\Lambda$-operators. This also gives the transition function for the $(N+1)$-particle periodic Toda chain. This form of eigenfunctions of the open Toda chain leads to its integral representation that appeared first in~\cite{Givental} employing a different approach. Recently this representation was interpreted from a group-theoretical point of view using the Gauss decomposition of $GL(N,\mathbb R)$~\cite{Kharchev_GG}, where this integral representation was called Gauss-Givental representation.

Availability of two kind of the integral representation for the open Toda chain eigenfunctions is explained in~\cite{Babelon} as follows. As we shall see the eigenvalues are parametrized by $N$ variables $\gamma_j$. The corresponding eigenfunctions $\psi_{\gamma_1,\ldots,\gamma_N}(x_1,\ldots,x_N)$ can be regarded as well as a function of $x_k$ satisfying the differential equations and in other hand as a function of $\gamma_j$ satisfying difference equations in $\gamma_j$ , i.e. as a wave function of some dual model. The duality of the same kind appears in the Representation Theory. The infinite-dimensional Gelfand-Zetlin representation of Lie algebra $\mathfrak{gl}(N)$ by shift operators in $\gamma_j$ allows to obtain the Mellin-Barns integral representation~\cite{Kharchev_GZ}, while the Gauss representation of the same Lie algebra by differential operators in $x_n$ leads to the Gauss-Givental representation~\cite{Kharchev_GG}.


\section{The separation of variables}
\label{sec21}

First of all we introduce a general notion of the SoV method for the quantum integrable systems following~\cite{Skl95}. Let $\{Q_1,\ldots,Q_N\}\subset\A$ be a quantum integrable systems. Suppose that the algebra $\A$ can be represented as an algebra of operators acting on a space $W$ consisting of functions of variables $y_1$, \ldots, $y_N$. So that the operators $Q_k$ transform a function of $y_1$, \ldots, $y_N$ to another function of these variables. The variables $y_1$, \ldots, $y_N$ are called separated if there exist $N$ relations of the form
\begin{align}
 F_j(y_j,Y_j,Q_1,\ldots,Q_N)&=0, & j&= 1,\ldots,N, \label{Phij} 
\end{align}
where $y_j$ are operators of the multiplication by the corresponding variables, $Y_j=-i\hbar\dfrac{\partial}{\partial y_j}$ are canonically dual operators and $F_j$ are functions of $N+2$ variables such that the determinant 
\begin{align}
\det_{j=1,\ldots,N \atop k=3,\ldots,N+2}\Big(\frac{\partial F_j(z_1,\ldots,z_{N+2})}{\partial z_k}\Big)
\end{align}
is not identically zero. We assume that the operators in~\eqref{Phij} are ordered exactly as they are enlisted.

Let $\Phi_E(y)$, where $y=(y_1,\ldots,y_N)$ and $E=(E_1,\ldots,E_N)$, be a common eigenfunction of the integrals of motion $Q_k$ with the eigenvalues $E_k$:
\begin{align}
 Q_k \Phi_E(y) = E_k \Phi_E(y). \label{QjPsi} 
\end{align}
The condition~\eqref{Phij} implies that the system of equations~\eqref{QjPsi} are equivalent to the system
\begin{align}
 F_j(y_j,Y_j,E_1,\ldots,E_N)\Phi_E(y)=0, \label{PhijE} 
\end{align}
which leads to the factorization of the eigenfunction $\Psi(y)$ into the product of functions of one variable:
\begin{align}
\Phi_E(y)=\prod_{j=1}^N \phi_j(y_j;E), \label{Psipsi} 
\end{align}
where each function $\phi_j(y_j;E)$ satisfies the equation
\begin{align}
 \Phi_j(y_j,Y_j,E_1,\ldots,E_N)\phi_j(y_j;E)=0. \label{PhijEpsi} 
\end{align}
Thus the spectral problem~\eqref{QjPsi} can reduced be reduced to the easier problem~\eqref{PhijEpsi}. 

Suppose that an integrable system is initially defined by operators acting on a space $V$ consisting of functions of variables $x_1$, \ldots, $x_N$ and these variables are not separated. Then in order to solve the spectral problem by SoV method we need to define the separated variables via the integral transformation
\begin{equation} \label{phiTV}
 \phi(y)=\int\mu_V(x)dx\, T(x;y) \psi(x),
\end{equation}
where $\mu_V(x)$ is a measure which define the scalar product of the state space $V$ and $T(x;y)=T(x_1,\ldots,x_N;y_1,\ldots,y_N)$ is a function (distribution, in general) defining the inclusion $\hat T\colon V\to W$ such that the reduced map $\hat T^{*}\colon V\to\hat T(V)$ is a unitary operator.
This function describes a transition from the initial set of variables $x=(x_1,\ldots,x_N)$ in which the model is given to the new set variables $y$ that factorizes the eigenfunctions. Its complex conjugation $U_y(x)=\overline{T(x;y)}$ is called {\itshape transition function}. Having found the solutions of the spectral problem~\eqref{QjPsi} (belonging to $\hat T(V)$) we need to present it in terms of the variables $x_1$,\ldots,$x_N$ applying the inverse integral transformation
\begin{equation}
 \psi(x)=\int\mu_W(y)dy\, U_y(x) \phi(y),
\end{equation}
where $\mu_W(y)$ is a measure defining the scalar product in $W$.

To write the relations~\eqref{Phij} we need to find an operators $Q_k$ in terms of the separated variables. The linear operator $\tilde A\colon W\to W$ correspond to a linear operator $\hat A\colon V\to V$ in separated variables if the diagram
\begin{align}
 \xymatrix{
  V\ar@{->}[r]^{\hat A}\ar@{->}[d]_{\hat T} & V\ar@{->}[d]^{\hat T} \\
  W\ar@{->}[r]^{\tilde A} & W
 }
 \end{align}
is commutative. That is it acts on the functions $\phi(y)\in\hat T(V)\subset W$ (the functions presenting as~\eqref{phiTV}) by the formula
\begin{align} \label{tildeAphi}
 (\tilde A\phi)(y)=\int\mu_V(x)dx\int\mu_W(y')dy'\,
  \overline{U_y(x)}(\hat A U_{y'})(x)\phi(y').
\end{align}
Note that if $\hat V\ne W$ the operator $\tilde A$ is not unique because one can extend the formula~\eqref{tildeAphi} to all the space $W$ in different ways.

\section{Transition to the separated variables of the periodic Toda chain}
\label{sec22}

In the section~\eqref{sec14} we define the periodic and open Toda chains. Consider the $(N+1)$-particle periodic case and define a transition to new variables: $\varepsilon, \gamma_1, \ldots, \gamma_{N}$.  This transition realizes with help of a transition function $U_{\varepsilon,\gamma}(\vec x)$, where $\gamma=(\gamma_1,\ldots, \gamma_N)$, $\vec x=(x_1,\ldots,x_{N+1})$ as the integral transformation
\begin{equation} \label{phiUpsi}
 \phi(\varepsilon,\gamma)=\int\limits_{\mathbb R^{N+1}} d\vec x\, \overline{U_{\varepsilon,\gamma}(\vec x)} \psi(\vec x),
\end{equation}
where $d\vec x=dx_1\cdots dx_{N+1}$ is a standard Lebesgue measure. As we shall see, a special choice of the function $U_{\varepsilon,\gamma}(\vec x)$ leads to the Separation of Variables.

Let us define the transition function $U_{\varepsilon,\gamma}(\vec x)$ to be a solution of the equations~\cite{Sklyanin}
\begin{align}
 C_{N+1}(u)\, U_{\varepsilon,\gamma}(\vec x)&=-e^{x_{N+1}}\prod_{k=1}^{N}(u-\gamma_k)\, U_{\varepsilon,\gamma}(\vec x), \label{CU}\\
 \sum_{k=1}^{N+1}p_k\; U_{\varepsilon,\gamma}(\vec x)&=\varepsilon U_{\varepsilon,\gamma}(\vec x), \label{H1U}
\end{align}
where $p_k=-i\hbar\dfrac{\partial}{\partial x_k}$. The relation $T_{[N+1]}(u)=L_{N+1}(u)T_{[N]}(u)$, which follows directly from~\eqref{TNLL}, implies the following recurrent relations
\begin{align}
 A_{N+1}(u)&=(u-p_{N+1})A_N(u)+e^{-x_{N+1}}C_N(u), \label{Arec} \\
 C_{N+1}(u)&=-e^{x_{N+1}}A_N(u), \label{Crec} \\
 A_{N+1}(u)&=(u-p_{N+1})A_N(u)-e^{x_N-x_{N+1}}A_{N-1}(u).  \label{AArec}
\end{align}
Taking into account~\eqref{Crec} one can see that the function $U_{\varepsilon,\gamma}(\vec x)$ is a common eigenfunction of the operators $A_{N}(u)$ and the operator of total momentum of $(N+1)$-particle chain $Q_1=\sum\limits_{k=1}^{N+1}p_k$. (This operators commute with each other for all values of spectral parameter $u$.)

We also require an additional condition
\begin{align} \label{addcU}
 U_{\varepsilon,\sigma(\gamma)}(\vec x)&=U_{\varepsilon,\gamma}(\vec x), & &\text{for all } \sigma\in S_{N},
\end{align}
where $\sigma(\gamma) =(\gamma_{\sigma(1)},\ldots, \gamma_{\sigma(N)})$ and $S_{N}$ is a permutation group. The sense of this condition is following. Since the equations~\eqref{CU}, \eqref{H1U} is invariant under permutations of the set of variables $\gamma$, their solution is not unique. The condition~\eqref{addcU} fix this solution and leads to the symmetry of the functions of the form~\eqref{phiUpsi}. It means in turn that the space $W$, on which the operators $\gamma_k$ act is a little more than the image of the state space $V$ under the transformation~\eqref{phiUpsi}. Indeed, this image consists of the functions~\eqref{phiUpsi}, which are symmetric under permutations of $\gamma$, while the operators of the multiplication by $\gamma_k$ gives a non-symmetric function, in general.

The conditions defining the transition function imply that $U_{\varepsilon,\gamma}(\vec x)$
satisfies the orthogonality condition
\begin{equation} \label{orthog}
 \int\limits_{\mathbb R^{N+1}} d\vec x \overline{U_{\varepsilon,\gamma}(\vec x)}U_{\varepsilon',\gamma'}(\vec x)=\mu^{-1}(\gamma)\delta(\varepsilon-\varepsilon')\delta_{SYM}(\gamma,\gamma')
\end{equation}
and the completeness condition
\begin{equation} \label{compl}
 \int\limits_{\mathbb R} d\varepsilon \int\limits_{\mathbb R^N} d\gamma\,\mu(\gamma)\overline{U_{\varepsilon,\gamma}(\vec x)}U_{\varepsilon,\gamma}(\vec x')=\delta(\vec x-\vec x'),
\end{equation}
where $\delta_{SYM}(\gamma,\gamma')=\dfrac1{N!}\sum\limits_{\sigma\in S_N}\delta\big(\gamma-\sigma(\gamma')\big)$
is a symmetrized delta-function and $\mu(\gamma)$ is some integration measure. In~\cite{S1} we check the condition~\eqref{orthog} proving that $\mu(\gamma)$ is a Sklyanin measure~\cite{Sklyanin}. The condition~\eqref{compl} follows from the general theory of rigged Hilbert space~\cite{Gelfand4} and it is equivalent to the fact that each function $\psi(x)$ from the state space $V$ can be presented as the integral
\begin{equation} \label{psiUphi}
 \psi(\vec x)=\int\limits_{\mathbb R} d\varepsilon \int\limits_{\mathbb R^N} d\gamma\,\mu(\gamma) U_{\varepsilon,\gamma}(\vec x)\phi(\varepsilon,\gamma),
\end{equation}
for some function $\phi(\varepsilon,\gamma)$ from $W$. The formula~\eqref{psiUphi} define an inverse transformation to the transformation~\eqref{phiUpsi}.

Consider the spectral problem for the $(N+1)$-particle periodic Toda chain
\begin{align} \label{spprpTc}
 Q_j \Psi_E(\vec x)=E_j\Psi_E(\vec x),
\end{align}
where $E=(E_1,\ldots,E_{N+1})$: by virtue of~\eqref{commQQ} the integrals of motion $Q_k$ have common eigenfunctions corresponding to the eigenvalues $E=(E_1,\ldots,E_{N+1})$ if each $E_k$ belongs to the spectrum of $Q_k$. In terms of the generating function~\eqref{hattuQ} they can be defined by the equation
\begin{equation} \label{Psi_def}
 \hat t_{N+1}(u)\Psi_E(\vec x)= t_{N+1}(u;E)\Psi_E(\vec x),
\end{equation}
where
\begin{equation}
 t_{N+1}(u;E)=\sum_{k=0}^{N+1} (-1)^k u^{N+1-k} E_k,
\end{equation}
and $E_0=1$. The eigenfunctions in the new variables are
\begin{equation} \label{Phi_def}
 \Phi_E(\varepsilon,\gamma)=\int\limits_{\mathbb R^{N+1}} d\vec x \overline{U_{\varepsilon,\gamma}(\vec x)} \Psi_E(\vec x).
\end{equation}

The choice of the transition function as a solution of the equations~\eqref{CU}, \eqref{H1U} leads to a system of equations of the form~\eqref{PhijE} for the eigenfunctions in the new variables $\Phi_E(\gamma)$, and consequently to the separation of variables~\eqref{Psipsi}. Now we shall briefly depict the derivation of this system of equations.

Rewriting by entries the $RLL$-relation~\eqref{RLL} with the matrix of monodromy~\eqref{TNLL} as a Lax operator and rational $R$-matrix~\eqref{R_rat} one obtains (in particular)
\begin{align}
 A_N(u)A_N(v)&=A_N(v)A_N(u), \label{RTTAA}\\
 (u-v)C_N(u)A_N(v)+i\hbar A_N(u)C_N(v)&=(u-v+i\hbar)A_N(v)C_N(u), \label{RTTAC}\\
 (u-v+i\hbar)D_N(u)C_N(v)&=(u-v)C_N(v)D_N(u)+i\hbar D_N(v)C_N(u). \label{RTTDC}
\end{align}
This relations can be used to derive the formulae for the actions of the operators $A_{N+1}(u)$ and $D_{N+1}(u)$ on the transition function as follows.  

Applying the relation~\eqref{RTTAA} for $N+1$ to $U_{\varepsilon,\gamma}(\vec x)$ and comparing the coefficients at $v^N$ in both sides one derives that the function $A_{N+1}(u)U_{\varepsilon,\gamma}(\vec x)$ satisfies the same equation~\eqref{H1U} with the same eigenvalue $\varepsilon$. Then the recurrent relations~\eqref{Arec}, \eqref{Crec} imply the relation
\begin{align}
 [A_{N+1}(u),e^{x_{N+1}}]=-[p_{N+1},e^{x_{N+1}}]A_N(u)=i\hbar e^{x_{N+1}}A_N(u)=-i\hbar C_{N+1}(u),
\end{align}
which leads in turn to
\begin{align}
 [A_{N+1}(\gamma_k),e^{x_{N+1}}]U_{\varepsilon,\gamma}(\vec x)=i\hbar e^{x_{N+1}}C_{N+1}(\gamma_k)U_{\varepsilon,\gamma}(\vec x)=0. \label{Aex0}
\end{align}
Substituting $v=\gamma_k$ to the relation~\eqref{RTTAC} for $N+1$ applied to $U_{\varepsilon,\gamma}(\vec x)$ and taking into account~\eqref{Aex0} the relation we conclude that the function $A_{N+1}(\gamma_k)U_{\varepsilon,\gamma}(\vec x)$ satisfies the equation~\eqref{CU} with the parameters $\gamma_1,\ldots,\gamma_k-i\hbar,\ldots,\gamma_N$. By virtue of the uniqueness of the solutions up to the factor it means that this function is proportional to the transition function
with these parameters:
\begin{equation} \label{AU}
 A_{N+1}(\gamma_k)U_{\varepsilon,\gamma}(\vec x)=a_k(\varepsilon,\gamma)U_{\varepsilon,\gamma-i\hbar e_k}(\vec x),
\end{equation}
where $\gamma-i\hbar e_k=(\gamma_1,\ldots,\gamma_k-i\hbar,\ldots,\gamma_N)$. Analogously, the relation
\begin{equation}
 [\hat t_{N+1}(u),D_{N+1}(v)]+[D_{N+1}(u),A_{N+1}(v)]=0,
\end{equation}
which follows from the relations~\eqref{ttcomm} and \eqref{AAcomm}, the relation~\eqref{RTTDC} for $N+1$ and the relation $[D_{N+1}(u),e^{x_{N+1}}]\hm=0$ imply
\begin{equation} \label{DU}
 D_{N+1}(\gamma_k)U_{\varepsilon,\gamma}(\vec x)=d_k(\varepsilon,\gamma)U_{\varepsilon,\gamma+i\hbar e_k}(\vec x).
\end{equation}

Taking into account the fact that the action of the operator $\hat t_N(u)$ on the transition function has the form
\begin{equation}
 \hat t_{N+1}(u)U_{\varepsilon,\gamma}(\vec x)=u^{N+1} U_{\varepsilon,\gamma}(\vec x)-u^N \varepsilon\, U_{\varepsilon,\gamma}(\vec x)+\ldots
\end{equation}
we can interpolate the formulae~\eqref{AU} and~\eqref{DU}:
\begin{multline} \label{tU}
 \hat t_{N+1}(u)U_{\varepsilon,\gamma}(\vec x)=\big(u-\varepsilon+\sum_{k=1}^N\gamma_k\big) \prod_{j=1}^N(u-\gamma_j)U_{\varepsilon,\gamma}(\vec x)+\\
  +\sum_{k=1}^N\prod_{j\ne k}\frac{u-\gamma_j}{\gamma_k-\gamma_j}\,\big(a_k(\varepsilon,\gamma)U_{\varepsilon,\gamma-i\hbar e_k}(\vec x)+d_k(\varepsilon,\gamma)U_{\varepsilon,\gamma+i\hbar e_k}(\vec x)\big).
\end{multline}

Let us suppose that the following functions
\begin{equation} \label{Delt}
  \begin{split}
    \Delta^+_k(\varepsilon,\gamma)&=a_k(\varepsilon,\gamma+i\hbar e_k)\prod_{j\ne k}\frac{\gamma_k-\gamma_j}{\gamma_k-\gamma_j+i\hbar}\frac{\mu(\gamma+i\hbar e_k)}{\mu(\gamma)},\\
    \Delta^-_k(\varepsilon,\gamma)&=d_k(\varepsilon,\gamma-i\hbar e_k)\prod_{j\ne k}\frac{\gamma_k-\gamma_j}{\gamma_k-\gamma_j-i\hbar}\frac{\mu(\gamma-i\hbar e_k)}{\mu(\gamma)}.
 \end{split}
\end{equation}
depend only on the variable $\gamma_k$:
\begin{align} \label{Delta_indep}
 \Delta^{\pm}_k(\varepsilon,\gamma)\hm=\Delta^{\pm}(\gamma_k),
\end{align}
(the functions $\Delta^\pm(u)$ do not depend on $k$). In the section~\ref{sec24} we shall see that the functions~\eqref{Delta_indep} are constants.

The formula~\eqref{tU} allows us to represent the generating functions $\hat t_{N+1}(u)$ in the separated variables. The operator $\tilde t_{N+1}(u)$ corresponding to $\hat t_{N+1}(u)$ in the separated variables acts on the functions $\phi(\varepsilon,\gamma)$ of the form~\eqref{phiUpsi} as follows (see the formula~\eqref{tildeAphi}):
\begin{multline} \label{ttildephi}
  (\tilde t_{N+1}(u)\phi)(\varepsilon',\gamma')=\int\limits_{\mathbb R^{N+1}} d\vec x \int\limits_{\mathbb R} d\varepsilon \int\limits_{\mathbb R^N} d\gamma\,\mu(\gamma) \overline{U_{\varepsilon',\gamma'}(\vec x)}\big(\hat t_{N+1}(u)U_{\varepsilon,\gamma}\big)(\vec x)\phi(\varepsilon,\gamma).
\end{multline}
Substituting the formula~\eqref{tU} in~\eqref{ttildephi}, shifting the variable $\gamma_k$ by $\pm i\hbar$ in the terms containing $a_k$ and $d_k$ and integrating over $\vec x$ one yields
\begin{multline} \label{ttildephi1}
  (\tilde t_{N+1}(u)\phi)(\varepsilon',\gamma')=\int\limits_{\mathbb R} d\varepsilon \int\limits_{\mathbb R^N} d\gamma\,\delta(\varepsilon-\varepsilon')\delta_{SYM}(\gamma,\gamma')\Big[\big(u-\varepsilon+\sum_{k=1}^N\gamma_k\big)\prod_{j=1}^N(u-\gamma_j)\phi(\varepsilon,\gamma)+ \\
   +\sum_{k=1}^N\prod_{j\ne k}\frac{u-\gamma_j}{\gamma_k-\gamma_j}\,\big(\Delta^+(\gamma_k)\phi(\varepsilon,\gamma+i\hbar e_k)+\Delta^-(\gamma_k)\phi(\varepsilon,\gamma-i\hbar e_k)\big)\Big].
\end{multline}
Note that the function $\phi(\varepsilon,\gamma)$ and, consequentely, the expression in the square brackets are symmetric with respect to $\gamma$. So that one obtain
\begin{multline} \label{ttildephi2}
  (\tilde t_{N+1}(u)\phi)(\varepsilon,\gamma)=\big(u-\varepsilon+\sum_{k=1}^N\gamma_k\big)\prod_{j=1}^N(u-\gamma_j)\phi(\varepsilon,\gamma)+ \\
   +\sum_{k=1}^N\prod_{j\ne k}\frac{u-\gamma_j}{\gamma_k-\gamma_j}\,\big(\Delta^+(\gamma_k)\phi(\varepsilon,\gamma+i\hbar e_k)+\Delta^-(\gamma_k)\phi(\varepsilon,\gamma-i\hbar e_k)\big).
\end{multline}
The operator $\tilde t_{N+1}(u)$ can be extended to all the space $W$ by the formula~\eqref{ttildephi2} for $\phi(\varepsilon,\gamma)\in W$.

The operators $\tilde Q_k\colon W\to W$ corresponding to $Q_k\colon V\to V$ can be defined as coefficients in powers of $u$:
\begin{align}
 \tilde t_{N+1}(u)=\sum_{k=0}^{N+1} (-1)^k u^{N+1-k} \tilde Q_k.
\end{align}
In particular,
\begin{align} \label{Q1eps}
 \tilde Q_1=\varepsilon.
\end{align}
Then the spectral problem~\eqref{spprpTc} in the separated variables takes the form
\begin{align} \label{spprpTcsv}
 \tilde Q_j \Phi_E(\varepsilon,\gamma)=E_j\Phi_E(\varepsilon,\gamma),
\end{align}
where the functions $\Phi_E(\vec x)$ are defined by the formula~\eqref{Phi_def}, or in terms of generating function:
\begin{equation} \label{ttildePhi}
 \tilde t_{N+1}(u)\Phi_E(\varepsilon,\gamma)= t_{N+1}(u;E)\Phi_E(\varepsilon,\gamma).
\end{equation}
Let us denote by $\tilde t_{N+1}(\gamma_k)$ the right substitution
\begin{align}
 \tilde t_{N+1}(\gamma_k)=\sum_{j=0}^{N+1} (-1)^j \gamma_k^{N+1-j} \tilde Q_j,
\end{align}
(the operators $y_k$ do not commute with $\tilde Q_j$). By virtue of the formula~\eqref{ttildephi2} we have
\begin{align} \label{ttildephigk}
  \tilde t_{N+1}(\gamma_k)=\Delta^+(\gamma_k) e^{i\hbar\frac{\partial}{\partial\gamma_k}}
	+\Delta^-(\gamma_k)e^{-i\hbar\frac{\partial}{\partial\gamma_k}}.
\end{align}
The equations~\eqref{Q1eps} and~\eqref{ttildephigk} play the roles of the relation~\eqref{Phij}. So that the problem~\eqref{spprpTcsv} is equivalent to the system of equations
\begin{align} \label{Baxter_md}
  \varepsilon\,\Phi_E(\varepsilon,\gamma)&=E_1\Phi_E(\varepsilon,\gamma), \\
	t_{N+1}(\gamma_k;E)\Phi_E(\varepsilon,\gamma)&=\Delta^+_k(\gamma)\Phi_E(\varepsilon,\gamma+i\hbar e_k)+\Delta^-_k(\gamma)\Phi_E(\varepsilon,\gamma-i\hbar e_k).
\end{align}
It means that the solution of~\eqref{Baxter_md} can be represented as follows
\begin{equation} \label{Phi_SoV}
 \Phi_E(\varepsilon,\gamma)=\delta(\varepsilon-E_1)\prod_{k=1}^N c(\gamma_k;E),
\end{equation}
where the function $c(u;E)$ satisfies the {\itshape Baxter equation}
\begin{equation} \label{Baxter_od}
  t_{N+1}(u;E)c(u;E)=\Delta^+(u)c(u+i\hbar;E)+\Delta^-(u)c(u-i\hbar;E).
\end{equation}

Thus, in order to find the eigenfunctions of the periodic Toda chain we need to find the transition function, the measure $\mu(\gamma)$, the coefficients $\Delta^\pm(u)$ and to solve the Baxter equation~\eqref{Baxter_od}. If we have a transition function $U_{E_1,\gamma}(\vec x)$, a solution~$c(u;E)$ and a measure~$\mu(\gamma)$, then we can present explicitly the eigenfunctions of the periodic Toda chain by the formula
\begin{equation} \label{Psi_Phi}
 \Psi_E(\vec x)=\int\limits_{\mathbb R^N} d\gamma\,\mu(\gamma) U_{E_1,\gamma}(\vec x) \prod_{k=1}^N c(\gamma_k;E).
\end{equation}
In the section~\ref{sec23} we consider a method to find the transition function. Substitutiong it to the formula~\eqref{orthog} one can find the measure $\mu(\gamma)$~\cite{S1}. In the section~\ref{sec24} we use the expression for the measure to obtain the coefficients $\Delta^\pm(\gamma_k)$. We do not concern the problem to solve the Baxter equetion~\eqref{Baxter_od}. It described in~\cite{Kharchev_P} in details.

\section{Transition function and $\Lambda$-operators}
\label{sec23}

Now we explain main ideas of~\cite{S1} how to apply the $\Lambda$-operators method of~\cite{Derkachov} to the derivation of the transition function defined in the previous section by the equations~\eqref{CU}, \eqref{H1U} and the condition~\eqref{addcU}. The equation~\eqref{CU} is equivalent to the equation
\begin{align}
 A_N(u)\, U_{\varepsilon,\gamma}(\vec x)&=\prod_{k=1}^{N}(u-\gamma_k)\, U_{\varepsilon,\gamma}(\vec x). \label{ANU}
\end{align}
The space of the solutions of this equations is invariant under multiplying by a function of $x_N$. The equation~\eqref{H1U} fix the $x_N$-depending factor up to a constant. Let us consider the equation
\begin{align}
 A_N(u)\, \psi_{\gamma}(x)&=\prod_{k=1}^{N}(u-\gamma_k)\, \psi_{\gamma}(x), \label{ANpsi}
\end{align}
for a function $\psi_\gamma(x)$ of $N$ variables, where $x=(x_1,\ldots,x_N)$. If $\psi_\gamma(x)$ is a solution of the equation~\eqref{ANpsi} satisfying the condition
\begin{align} \label{addcpsi}
 \psi_{\sigma(\gamma)}(x)&=\psi_\gamma(x), & &\text{for all } \sigma\in S_{N},
\end{align}
then the function
\begin{equation} \label{UxNpsi}
 U_{\varepsilon,\gamma}(x_1,\ldots,x_{N+1})=
  e^{\frac{i}{\hbar}\big(\varepsilon-\sum\limits_{k=1}^N\gamma_k\big)x_{N+1}}
   \psi_\gamma(x_1,\ldots,x_N)
\end{equation}
solves the equations~\eqref{CU}, \eqref{H1U} and satisfies the condition~\eqref{addcU}. Thus the derivation of the transition function for the periodic Toda chain is reduced to the problem~\eqref{ANpsi}, \eqref{addcpsi}, which is, in fact, the spectral problem for the open Toda chain
\begin{align}
 H_k\psi_\gamma(x)=E_k\psi_\gamma(x),
\end{align}
where $E_k=\sum\limits_{j_1<\ldots<j_k}\gamma_{j_1}\cdots\gamma_{j_k}$. The eigenfunction of the open Toda chain $\psi_\gamma(x)$ should satisfy the condition~\eqref{addcpsi} because the eigenvalues of its integral of motion $H_k$ depend symmetrically on the parameters $\gamma_1,\ldots,\gamma_N$.

The main idea to solve the equation~\eqref{ANpsi} is following. Let us first note that the equation~\eqref{ANpsi} is equivalent to the system of $N$ equations
\begin{align} \label{AU0}
 A_N(\gamma_k)\psi_\gamma(x)&=0, & &k=1,\ldots,N.
\end{align}
Consider the first equation
\begin{equation} \label{Aq1U0}
 A_N(\gamma_1)\psi_\gamma(x)=0
\end{equation}
and notice that any solution to this equation satisfying the condition~\eqref{addcpsi} is a solution of the whole system~\eqref{AU0}.

We apply the Lax operator triangulation method proposed in~\cite{Pasquier} to solve the equation~\eqref{Aq1U0}. Consider the following gauge transformation of the Lax operators
\begin{align}
 L_k(u)\to\widetilde L_k(u)&=M_k L_k(u)M_{k-1}^{-1}, & &k=1,\ldots,N,
\end{align}
by the invertible matrices
\begin{align}
 M_k&=\begin{pmatrix}
   1 & 0 \\
   ie^{y_k} & 1
  \end{pmatrix} & &k=0,\ldots,N.
\end{align}
This transformation leads to the corresponding transformation of the monodromy matrix
\begin{equation} \label{T_til}
 T_{[N]}(u)\to\widetilde T_{[N]}(u)\equiv
  \begin{pmatrix}
   \widetilde A_N(u) & \widetilde B_N(u) \\
   \widetilde C_N(u) & \widetilde D_N(u)
  \end{pmatrix}
  =\widetilde L_N(u)\cdots\widetilde L_1(u)=M_N T_{[N]}(u) M_0^{-1}.
\end{equation}
Let us consider the equation
\begin{equation} \label{C_tilW}
 \widetilde C_N(u) W_u(x;y_0,\ldots,y_N)=0.
\end{equation}
Its solution can be present as the product
\begin{equation}
 W_u(x;y_0,\ldots,y_N)=\prod\limits_{k=1}^N w_k(u;x_k;y_{k-1},y_k),
\end{equation}
where each function $w_k(u;x_k;y_{k-1},y_k)$ is a solution of the equation
\begin{equation} \label{Ltilde21wu0}
 \widetilde L_k(u)_{21}w_k(u;x_k;y_{k-1},y_k)=0,
\end{equation}
where $\widetilde L_k(u)_{21}$ is the corresponding entry of the matrix $\widetilde L_k(u)$. The equation~\eqref{Ltilde21wu0} is ODE of the first order with the unique solution (up to a factor) 
\begin{equation}
 w_k(u;x_k;y_{k-1},y_k)=\exp\Big(\frac{i}{\hbar}u(x_k-y_{k-1})-\frac{1}{\hbar}e^{y_{k-1}-x_k}-\frac{1}{\hbar}e^{x_k-y_k}\Big).
\end{equation}
This means that the equation~\eqref{C_tilW} has a solution
\begin{equation} \label{W_u_xy}
 W_u(x;y_0,\ldots,y_N)=\exp\sum\limits_{k=1}^N\Big(\frac{i}{\hbar}u(x_k-y_{k-1})-\frac{1}{\hbar}e^{y_{k-1}-x_k}-\frac{1}{\hbar}e^{x_k-y_k}\Big).
\end{equation}

In the limit $y_0\to -\infty$, $y_N\to +\infty$ the formula~\eqref{T_til} gives us the equality
\begin{equation} \label{ANlim}
A_N(u)\hm=-i\lim\limits_{\substack{y_0\to -\infty \\ y_N\to +\infty}} e^{-y_N}\widetilde C_N(u).
\end{equation}
Therefore, multiplying the equation~\eqref{C_tilW} by
$-ie^{-y_N}e^{\frac{i}{\hbar}u(y_0+y_N)}$, taking the same limit
as in~\eqref{ANlim} and setting $u\hm=\gamma_1$ we arrive to the equation~\eqref{Aq1U0} with the solution
$\psi_{\gamma}(x)\hm=\Lambda_{\gamma_1}(x;y)$, where $y=(y_1,\ldots,y_{N-1})$ and
\begin{equation} \label{Lamb}
 \begin{split}
 \Lambda_u(x_1,\ldots,x_N;y_1,\ldots,y_{N-1})
     &=\lim\limits_{\substack{y_0\to -\infty \\ y_N\to +\infty}}
        e^{\frac{i}{\hbar}u(y_0+y_N)}W_{u}(x;y_0,\ldots,y_N)=\\
     =\exp\Big(\frac{i}{\hbar}u(\sum\limits_{k=1}^{N}x_k-&\sum\limits_{k=1}^{N-1}y_k)
        -\frac{1}{\hbar}\sum\limits_{k=1}^{N-1}(e^{y_k-x_{k+1}}+e^{x_k-y_k})\Big).
 \end{split}
\end{equation}
Let $\Lambda_N(u)$ be an operator with the kernel
$\Lambda_u(x_1,\ldots,x_N;y_1,\ldots,y_{N-1})$, i.e.
\begin{equation} \label{Lambf}
 (\Lambda_N(u)\cdot s)(x)
   =\int\limits_{\mathbb R^{N-1}}dy\,\Lambda_u(x_1,\ldots,x_N;y_1,\ldots,y_{N-1})s(y).
\end{equation}
This operator acts from a space of functions of $N-1$ variables to a space of functions of $N$ variables and called $\Lambda$-operator. This notation is related with the graphical representation of this operator (see~\cite{Derkachov,S1}). The expression~\eqref{Lambf} for an arbitrary function of $N-1$ variables $s(y)$ and $u=\gamma_1$ is a more general solution of the equation~\eqref{Aq1U0}. So, we have to select the function $s(y)$ parametrised by $\gamma$ such that the function $(\Lambda_N(\lambda_1)\cdot s)(x)$ satisfy the condition~\eqref{addcpsi}.

In~\cite{S1} we prove the following property of $\Lambda$-operators:
\begin{align} \label{Lamoppr}
 \Lambda_N(u)\Lambda_{N-1}(v)=\Lambda_N(v)\Lambda_{N-1}(u).
\end{align}
This property is suggested in~\cite{Derkachov} as a key point of the $\Lambda$-operator method for the $XXX$-model. It leads to 

\begin{theor} \label{Th_Uq}
The following solution of the equation~\eqref{Aq1U0}
\begin{equation} \label{Uq}
 \psi_\gamma(x)=(\Lambda_N(\gamma_1)\cdots\Lambda_2(\gamma_{N-1})\Lambda_1(\gamma_N)\cdot1)(x_1,\ldots x_N)
\end{equation}
is invariant under the permutation of the variables $\gamma_1,\ldots,\gamma_N$, i.e. it satisfies the condition~\eqref{addcpsi}. (Here $1$ is a function of zero number of variables.) So that the function~\eqref{Uq} is a solution to the equation~\eqref{ANpsi} satisfying~\eqref{addcpsi}.
\end{theor}

Substituting the explicit expression~\eqref{Lamb} for the kernel of the first operator $\Lambda_N(\gamma_N)$ (we can exchange $\gamma_1$ and $\gamma_N$) one obtains the recurrent formula for the function~\eqref{Uq}
\begin{equation} \label{psi_rec}
 \begin{split}
 \psi_{\gamma_1,\ldots,\gamma_N}(x_1,\ldots,x_N)=\int\limits_{\mathbb R^{N-1}}dy_1\ldots dy_{N-1}\,
  \psi_{\gamma_1,\ldots,\gamma_{N-1}}(y_1,\ldots,y_{N-1}) \times \\
 \times\exp\Big(\frac{i}{\hbar}\gamma_N(\sum\limits_{k=1}^N x_k-\sum\limits_{k=1}^{N-1}y_k)
        -\frac{1}{\hbar}\sum\limits_{k=1}^{N-1}(e^{y_k-x_{k+1}}+e^{x_k-y_k})\Big).
 \end{split}
\end{equation}
The consecutive applications of this formula allow to derive the following integral representation for the eigenfunctions of open Toda chain
\begin{multline} \label{psi_GG}
 \psi_{\gamma}(z_{N1},\ldots,z_{NN})=\int\limits_{\mathbb R^{\frac{N(N-1)}2}}\prod_{k=1}^{N-1}\prod_{j=1}^{k}dz_{kj}\,
  \exp\bigg(\frac{i}{\hbar}\Big(\gamma_N\sum\limits_{j=1}^N z_{Nj}
    +\sum_{k=1}^{N-1}(\gamma_k-\gamma_{k+1})\sum_{j=1}^k z_{kj}\Big)-\\
    -\frac{1}{\hbar}\sum_{k=1}^N\sum_{j=1}^{k-1}\big(e^{z_{kj}-z_{k-1,j}}
    +e^{z_{k-1,j}-z_{k,j+1}}\big)\bigg).
\end{multline}
The corresponding representation for the Toda chain transition function is obtained via the multiplication by the factor
\begin{align}
 e^{\frac{i}{\hbar}\big(E_1-\sum\limits_{k=1}^N\gamma_k\big)x_{N+1}}.
\end{align}
This representation is called Gauss-Givental representation~\cite{Kharchev_GG}, \cite{Givental}.

\section{Properties of the Baxter $Q$-operators and $\Lambda$-op\-era\-tors}
\label{sec24}

The Baxter $Q$-operators, for which the Lax operator triangulation method was first introduced, defined also by the kernel $W_u(x;y_0,\ldots,y_N)$ but with the periodic boundary conditions $y_0=y_N$. The $Q$-operator for the $N$-particle periodic Toda chain acts from a space of functions of $N$ variables to the same space by the formula 
\begin{equation} \label{Lambf_Q}
 (\hat Q_N(u)\cdot s)(x)
   =\int\limits_{\mathbb R^N}dy_1\cdots dy_N\,W_u(x_1,\ldots,x_N;y_N,y_1,\ldots,y_N)s(y_1,\ldots,y_N).
\end{equation}
These operators satisfy the following relations
\begin{align}
 \hat Q_N(u)\hat Q_N(v)&=\hat Q_N(v)\hat Q_N(u), \label{Qoppra} \\
 \hat Q_N(u)\hat t_N(v)&=\hat t_N(v)\hat Q_N(u), \label{Qopprb} \\
\hat t_N(u)\hat Q_N(u)&=i^N \hat Q_N(u+i\hbar)(x)+i^{-N} \hat Q_N(u-i\hbar). \label{Qopprc}
\end{align}

The property~\eqref{Lamoppr} of the $\Lambda$-operators is analogous to the commutativity of the Baxter $Q$-operators~\eqref{Qoppra}. The formula~\eqref{Qopprb} is associated to the fact that the functions~\eqref{Lambf} satisfy the equation~\eqref{Aq1U0}. The first and second terms of the right hand side of the Baxter equation~\eqref{Qopprc} are equal to $\tilde D_N(u)\hat Q_N(u)$ and $\tilde A_N(u)\hat Q_N(u)$ respectively, where $y_0=y_N$. The corresponding equations for the $\Lambda$-operators have the form~\cite{S1}
\begin{align}
C_N(u)\Lambda_N(u)&=i^{-N-1}\Lambda_N(u-i\hbar), &
B_N(u)\Lambda_N(u)&=i^{N-1}\Lambda_N(u+i\hbar). \label{LpropCB}
\end{align}

The properties~\eqref{LpropCB} allow us to determine the coefficients~\eqref{Delt}. Indeed, taking into account that the function~\eqref{Uq} satisfies the condition~\eqref{addcpsi} we obtain
\begin{align}
C_N(\gamma_k)\psi_{\gamma}(x)&=i^{-N-1}\psi_{\gamma-i\hbar e_k}(x), &
B_N(\gamma_k)\psi_{\gamma}(x)&=i^{N-1}\psi_{\gamma+i\hbar e_k}(x). \label{psipropCB}
\end{align}
Substituting~\eqref{UxNpsi} to the left hand sides of the formulae~\eqref{AU}, \eqref{DU} and taking into account the relations~\eqref{Arec}, \eqref{AU0} and $D_{N+1}(u)=-e^{x_N}B_N(u)$ we obtain the formulae~\eqref{AU}, \eqref{DU} with the coefficients
\begin{align}
 a_k(\varepsilon,\gamma)&=i^{-N-1}, &  d_k(\varepsilon,\gamma)&=i^{N+1}. \label{ak_dk}
\end{align}
Then, substituting the formulae~\eqref{ak_dk} and the Sklyanin measure
\begin{align}
  \mu(\gamma)=\frac{(2\pi\hbar)^{-N}}{N!}\prod_{k<m}
  \left(\Gamma\Bigl(\frac{\gamma_m-\gamma_k}{i\hbar}\Bigr)
        \Gamma\Bigl(\frac{\gamma_k-\gamma_m}{i\hbar}\Bigr)\right)^{-1}
\end{align}
to the definitions~\eqref{Delt} we obtain
\begin{align}
 \Delta^+_k(\varepsilon,\gamma)&=\Delta^+(\gamma_k)=i^{N+1}, &  \Delta^-_k(\varepsilon,\gamma)&=\Delta^-(\gamma_k)=i^{-N-1}. \label{Delta_ak_dk}
\end{align}

The properties~\eqref{Qopprb} implies that the functions
\begin{equation} \label{LambPhi_Q}
 (\hat Q_{N+1}(u)\cdot \Psi_E)(\vec x)
   =\int\limits_{\mathbb R^{N+1}}d\vec y\,W_u(x_1,\ldots,x_{N+1};y_{N+1},y_1,\ldots,y_{N+1})\Psi_E(\vec y),
\end{equation}
(where $\vec y=(y_1,\ldots,y_{N+1})$ and $d\vec y=\prod\limits_{k=1}^{N+1} dy_k$), also satisfy the equation~\eqref{Psi_def}. Let us represent these functions in terms of the Theorem~\ref{Th_Uq}. Substituting~\eqref{Psi_Phi} and taking into account that the integration kernel~\eqref{W_u_xy} with $y_0=y_{N+1}$ has the form
\begin{multline}
 W_u(x_1,\ldots,x_{N+1};y_{N+1},y_1,\ldots,y_{N+1})= \\
 =\exp\Big(-\frac{i}\hbar u y_{N+1}-\frac1\hbar e^{y_{N+1}-x_1}-\frac1\hbar e^{x_{N+1}-y_{N+1}}\Big)\Lambda_u(x_1,\ldots,x_{N+1};y_1,\ldots,y_N)
\end{multline}
we derive
\begin{equation} \label{LambPhi_Q}
 (\hat Q_{N+1}(u)\cdot \Psi_E)(\vec x)
   = \int\limits_{\mathbb R^N} d\gamma\,\mu(\gamma)\prod_{k=1}^N c(\gamma_k;E) \psi_{(u,\gamma_1,\ldots,\gamma_N)}(\vec x)\psi_{(0,E_1-\sum_{k=1}^N\gamma_k-u)}(x_{N+1},x_1).
\end{equation}
Thus using Theorem~\ref{Th_Uq} we obtain another expression of a solution of the spectral problem~\eqref{Psi_def}.

\selectlanguage{french}

\chapter[Algèbres des courants associées à la courbe elliptique]{Algèbres des courants et groupes quantiques dynamique associées à la courbe elliptique}
\label{sec3}

Ici nous présentons les résultats des articles~\cite{S21,S22} (voir Appendices~\ref{SA21}, \ref{SA22}) sur les groupes quantiques dynamiques elliptiques décrits par les courants. Tous d'abord nous considérons les courants et leurs propriétés en détail.

Pour explicitement décrire des groupes quantiques plusieurs langages sont utilisées. Dans la section~\ref{sec18} nous avons considéré le langage des opérateurs de Lax. Puis, dans la section~\ref{sec19}, nous avons montré que les paires d'opérateurs de Lax décrivent les algèbres de Hopf quasi-triangulaires. Les exemples les plus simples d'algèbre de Hopf quasi-triangulaire sont des groupes quantiques sans paramètre spectral engendrées par un nombre fini de générateurs. Ils peuvent être décrits par des paires d'opérateurs de Lax~\eqref{Lpm_dec} qui ne dépendent pas de $u$, c'est le cas de $N_{\pm;ij}=1$, $\phi^{\alpha}_{\pm;ij}(u)=1$. On peut aussi considérer ces algèbres comme des quantifications des algèbres enveloppantes universelles des algèbres de Lie semi-simples et il est très utile d'écrire ce type d'algèbres par une base de Cartan-Weyl quantique. La quantification de la base de Cartan-Weyl est facilement généralisée au cas d'algèbres de Lie affines. Pour ces dernières et pour les Yangiens mentionnés à la fin de la section~\ref{sec19} (page~\pageref{Yangien}) des réalisations d'un type nouveau ont été construites~\cite{D88}. Dans ce travail Drinfeld a introduit une notion de courants jouant maintenant un rôle important dans la théorie des groupes quantiques et surtout pour des groupes quantiques (dynamiques) associées à la courbe elliptique.

\section{Courants en termes des distributions}
\label{sec31}

Les courants apparus dans~\cite{D88} pour la description des algèbres affines quantiques et Yangiens ont été définis comme éléments de l'espace $\HH[[z,z^{-1}]]$, où $\HH$ est une algèbre correspondante. Dans un cas plus général ils peuvent être compris comme des distributions avec valeurs dans $\HH$. Ce concept est convenable même pour des courants qui sont décomposés aux intégrales au lieu des séries formelles~\cite{KLPST,KLP98}.

On commence par quelques définitions. Soit $\mathfrak{K}$ une algèbre de fonctions sur une variété complexe $\Sigma$ (pas toutes les fonction sur $\Sigma$) avec la multiplication par points. Supposons l'algèbre $\mathfrak{K}$ munie d'une topologie. Soit $\la\cdot,\cdot\ra\colon\mathfrak{K}\times\mathfrak{K}\to\mathbb C$ une forme bilinéaire non-dégénérée, invariante et continue. Nous appellerons une algèbre $\mathfrak{K}$ munie de la forme $\la\cdot,\cdot\ra$ {\em une algèbre de fonction test}. La non-dégénération de cette forme permet de continûment plonger l'algèbre $\mathfrak{K}$ dans l'espace $\mathfrak{K}'$ -- l'espace des fonctionnelles linéaires continues sur $\mathfrak{K}$. Nous appellerons les éléments de $\mathfrak{K}$ des fonctions test, les éléments de $\mathfrak{K}'$ des distributions. Nous désignons les fonctions test avec leur arguments: $s(u)\in\mathfrak{K}$ ou sans leur arguments $s\in\mathfrak{K}$. Si $a\in\mathfrak{K}'$ est une fonctionnelle nous le désignons par $a(u)$. L'action de $a\in\mathfrak{K}'$ sur $s\in\mathfrak{K}$ est désignée par $\la a,s\ra$, ou par $\la a(u),s(u)\ra_u$. Ces notations sont aussi valable dans le cas $a\in\mathfrak{K}$. Nous munissons l'espace $\mathfrak{K}'$ de la topologie (la convergence) faible: on a $a(u)=\lim\limits_{i\to\infty}a_i(u)$ pour $a,a_i\in\mathfrak{K}'$ si et seulement si $\la a(u),s(u)\ra_u=\lim\limits_{i\to\infty}\la a_i(u),s(u)\ra_u$ pour tous $s(u)\in\mathfrak{K}$. La multiplication $\cdot\colon\mathfrak{K}\otimes\mathfrak{K}'\to\mathfrak{K}'$ est uniquement définie comme une extension de la multiplication $\cdot\colon\mathfrak{K}\otimes\mathfrak{K}\to\mathfrak{K}$. On peut aussi uniquement (mais pas toujours) définir un produit de deux distributions étendant un produit d'une fonction test avec une distribution: si $a(u)\in\mathfrak{K}'$ est représentée comme $a(u)=\lim\limits_{i\to\infty} a_i(u)$, où $a_i(u)\in\mathfrak{K}$ (nous supposons que toutes les distributions peuvent être représentées comme ça), et $b(u)\in\mathfrak{K}'$ alors
\begin{align}
 a(u)b(u)=b(u)a(u)=\lim\limits_{i\to\infty} a_i(u)b(u). \label{ab_lim}
\end{align}
On dit que le produit de $a(u)$ et $b(u)$ existe si la limite~\eqref{ab_lim} existe.
Dans le cas où la fonction $1$, égale à l'unité, appartient à $\mathfrak{K}$ nous utiliserons aussi la notation $\la a(u)\ra_u=\la a(u),1\ra_u$. Autrement, cette expression est définie comme $\la a(u)\ra_u=\lim\limits_{i\to\infty}\la a(u),s_i(u)\ra_u$, où $\lim\limits_{i\to\infty}s_i(u)=1$, si cette limite existe. Puisque la forme $\la\cdot,\cdot\ra$ est invariante nous avons $\la a(u),s(u)\ra_u=\la s(u)\cdot a(u)\ra_u$. 

Un produit tensoriel $\mathfrak{K}\otimes\mathfrak{K}$ peut être regardé comme un espace de fonction test sur $\Sigma\times\Sigma$
\begin{align}
 (s\otimes t)(u,v)&=s(u)t(v), \\
 \la s_1(u)t_1(v),s_2(u)t_2(v)\ra_{u,v}&=\la s_1(u),s_2(u)\ra_u\,\la t_1(v),t_2(v)\ra_v.
\end{align}
Un élément $a\in(\mathfrak{K}\otimes\mathfrak{K})'$ est appelé distribution de deux variables et désigné par $a(u,v)$. Pour ces distributions nous pouvons définir des actions partielles par une de ses variables: $\la a(u,v),s(u,v)\ra_u$ et $\la a(u,v),s(u,v)\ra_v$. L'action partielle d'une distribution $a(u,v)\in(\mathfrak{K}\otimes\mathfrak{K})'$ sur une fonction test $s(u)t(v)$, où $s,t\in\mathfrak{K}$, par $u$ est une distribution de variable $v$ désignée par $\la a(u,v),s(u,v)\ra_u$ et agissant sur une fonction test $\tilde t(v)\in\mathfrak{K}$ par la formule
\begin{align}
 \big\la\,\la a(u,v),s(u)t(v)\ra_u\, ,\tilde t(v)\big\ra_v=\la a(u,v),s(u)t(v)\tilde t(v)\ra_{u,v}.
\end{align}
L'action partielle sur un quelconque $s(u,v)\in\mathfrak{K}\otimes\mathfrak{K}$ est définie par linéarité. On peut définir aussi l'action partielle par $v$ à la même façon. En particulier, nous avons
\begin{align}
 \la a(u)b(v),s(u)t(v)\ra_u&=\la a(u),s(u)\ra_u\,t(v)b(v), \\
 \la a(u)b(v),s(u)t(v)\ra_v&=s(u)a(u)\,\la b(v),t(v)\ra_v,
\end{align}
où $a,b\in\mathfrak{K}'$ et $s,t\in\mathfrak{K}$.
Ces notations sont utilisées aussi pour les distributions de plusieurs variables.

Avant de donner une définition du courant, nous considérerons un espace important. C'est un sous-espace de $(\mathfrak{K}\otimes\mathfrak{K})'$ consistant des distributions $a(u,v)\in(\mathfrak{K}\otimes\mathfrak{K})'$ telles que
\begin{align}
 \la a(u,v),s(u)\ra_u&\in\mathfrak{K}, & \la a(u,v),s(v)\ra_v&\in\mathfrak{K},
\end{align}
$\forall s\in\mathfrak{K}$. Désignons cet espace par $\mathfrak{K}_2$. Il possède une propriété importante: pour tous les éléments $a\in\mathfrak{K}_2$, $b\in\mathfrak{K}'$ et $c\in(\mathfrak{K}\otimes\mathfrak{K})'$ les produits $a(u,v)b(u)$, $a(u,v)b(v)$, $a(u,v)c(v,w)$ et $a(v,w)c(u,v)$ existent toujours. L'espace $\mathfrak{K}_2$ joue un rôle important pour la construction des courants. Notons que ses éléments peuvent être identifiés avec des éléments de $\End\mathfrak{K}$ par l'action partielle: une distribution $a(u,v)\in\mathfrak{K}_2$ correspond à l'opérateur agissant sur $s\in\mathfrak{K}$ comme $\la a(u,v),s(u)\ra_u$. La composition $\hat b\circ\hat a$ des opérateurs $\hat b,\hat a\in\End\mathfrak{K}$ représentés par $b(u,v),a(u,v)\in\mathfrak{K}_2$ correspond à la "convolution" $\la a(u,v)b(v,w)\ra_v\in\mathfrak{K}_2$ définie comme
\begin{align}
 \big\la\la a(u,v)b(v,w)\ra_v,s(u)t(w)\big\ra_{u,w}=\big\la b(v,w),\la a(u,v),s(u)\ra_u\,t(w)\big\ra_{v,w}.
\end{align}
Par analogie nous définissons les distributions $\la a(u,v)b(u)\ra_u$ et $\la a(u,v)b(v)\ra_v$ agissant sur $s\in\mathfrak{K}$ par
\begin{align}
 \big\la \la a(u,v)b(u)\ra_u,s(v)\big\ra_v=\big\la b(u),\la a(u,v),s(v)\ra_v\big\ra_u, \\
 \big\la \la a(u,v)b(v)\ra_v,s(u)\big\ra_u=\big\la b(v),\la a(u,v),s(u)\ra_u\big\ra_v,
\end{align}
où $a\in\mathfrak{K}_2$, $b\in\mathfrak{K}'$.

L'exemple principal d'un élément de $\mathfrak{K}_2$ est une distribution $\delta(u,v)$ appelée la fonction delta. Elle agit comme
\begin{align}
 \la\delta(u,v),s(u,v)\ra_{u,v}=\la s(u,u)\ra_u
\end{align}
et peut être identifiée avec $\id_{\mathfrak{K}}$, puisque ses actions partielles sont égales à
\begin{align}
 \la\delta(u,v),s(u)\ra_u&=s(v), & \la\delta(u,v),s(v)\ra_v&=s(u).
\end{align}
Soient $\{\epsilon^i(u)\}$ et $\{\epsilon_i(u)\}$ des bases duales de l'espace $\mathfrak{K}$: $\la \epsilon^i(u),\epsilon_j(u)\ra=\delta^i_j$. Alors
\begin{align}
 \delta(u,z)=\sum_i\epsilon^i(u)\epsilon_i(z),
\end{align}
où la somme par $i$ converge par topologie faible de l'espace des distributions $(\mathfrak{K}\otimes\mathfrak{K})'$. La fonction delta possède la propriété suivante:
\begin{align}
 a(u)\delta(u,v)=a(v)\delta(u,v),
\end{align}
où les deux produits existent pour toutes distributions $a\in\mathfrak{K}'$.

Soit $\mathfrak{g}$ un espace vectoriel sur $\mathbb C$ de dimension infinie et soit $\hat{x}\colon\mathfrak{K}\to\mathfrak{g}$ un opérateur linéaire continu. Un courant est présenté d'habitude par l'expression
\begin{align}\label{def_xu}
 x(u)=\sum_i \epsilon^i(u)\hat{x}[\epsilon_i],
\end{align}
où $\hat{x}[\epsilon_i]$ signifie une action de $\hat{x}$ sur $\epsilon_i$. Puisque il ne dépend pas de choix des bases duales de $\mathfrak{K}$, un unique courant correspond à l'opérateur $\hat{x}$. Nous comprenons le courant~\eqref{def_xu} comme une distribution avec valeurs dans l'espace $\mathfrak{g}$, id est $x(u)\in\mathfrak{K}'\otimes\mathfrak{g}\cong\Hom(\mathfrak{K},\mathfrak{g})$. Il agit sur une fonction test $s(u)$ par
\begin{align}\label{x_act}
 \la x(u),s(u)\ra_u=\hat{x}[s].
\end{align}
Nous considérerons cette formule comme une définition universelle du courant $x(u)$ correspondant à l'opérateur $\hat x$.

Considérons un nombre fini d'opérateurs $\hat x_k\colon\mathfrak{K}\to\mathfrak{g}$, où $\mathfrak{g}$ est un espace vectoriel de dimension infini. Supposons que $\hat x_k[\epsilon_i]$ est une base de l'espace $\mathfrak{g}$ et soient $x_k(u)$ des courants correspondant à ces opérateurs. Imposons aux courants $x_k(u)$ les relations de commutation suivantes
\begin{align} \label{cu-com-rel}
 [x_k(u),x_l(v)]=\sum_m C^m_{kl}x_m(u)\delta(u,v)\,
\end{align}
où $C^m_{kl}$ sont les constantes de structure d'une algèbre de Lie de dimension finie $\mathfrak a$. Les relations~\eqref{cu-com-rel} munissent l'espace $\mathfrak{g}$ d'une structure d'algèbre de Lie. L'algèbre de Lie $\mathfrak{g}$ peut être regardée comme une algèbre de Lie $\mathfrak{a}\otimes\mathfrak{K}$ avec les crochets
\begin{align}
 [x\otimes s(u),y\otimes t(u)]=[x,y]_{\mathfrak a}\otimes s(u)t(u),
\end{align}
où $x,y\in\mathfrak a$ et $s,t\in\mathfrak{K}$. Dans le cas $\mathfrak{K}=\mathbb C[u^{-1},u]$ cette algèbre de Lie est appelée {\it une algèbre des lacets}. A chaque élément $x\in\mathfrak{a}$ on peut associer un courant
\begin{align}
 x(u)=\sum_i \epsilon^i(u)\, x\otimes\epsilon_i(v)=x\otimes\delta(u-v)
\end{align}
correspondant à l'opérateur
\begin{align}
 \hat x[s]=x\otimes s(v).
\end{align}
Si l'algèbre de Lie $\mathfrak{a}$ possède une forme bilinéaire invariante $(\cdot,\cdot)$, alors l'algèbre de Lie $\mathfrak{g}$ possède aussi une forme bilinéaire invariante 
\begin{align} \label{scpr_xsyt}
 \la x\otimes s,y\otimes t\ra_{\mathfrak{g}}=(x,y)\la s(u),t(u)\ra_u.
\end{align}
En termes de courants elle peut être définie par la formule
\begin{align} \label{scpr_xuyv}
 \la x(u),y(v)\ra_{\mathfrak{g}}=(x,y)\delta(u,v),
\end{align}
où $x(u)$ et $y(v)$ sont des courants associés aux éléments $x\in\mathfrak{a}$ et $y\in\mathfrak{a}$.

Par analogie, on peut considérer des courants pour une algèbre associative $\A$. Ce sont des courants $x_k(u)\in\mathfrak{K}'\otimes\A$ correspondant aux opérateurs $\hat x_k\colon\mathfrak{K}\to\A$ tels que les éléments $\hat x_k[\epsilon_i]\hm=\la x_k(u),\epsilon_i(u)\ra_u$ engendrent toute l'algèbre $\A$. L'algèbre $\A$ peut être présentée comme une quotient de l'algèbre engendrée librement par ces éléments par les relation de commutation entre eux. Le but principale d'utilisation des courants consiste à la représentation de ces relations de commutation sous forme compacte. Par exemple, les relations de commutation quadratiques (qui d'habitude décrivent les groupes quantiques) en termes des courants sont de la forme
\begin{align}
 \sum_{k,l}a_{k,l}(u,v)x_k(u)x_l(v)+\sum_k b_k(u,v)x_k(u)+\sum_k c_k(u,v)x_k(v)+d(u,v)=0, \label{rdcqetdc}
\end{align}
où $a_{k,l},b_k,c_k,d\in\mathfrak{K}_2$. Les algèbres décrites par des courants sont appelées algèbres des courants. En particulier, des courants pour une algèbre de Lie $x_k(u)\in\mathfrak{K}'\otimes\mathfrak{g}$ peuvent être considérée comme des courants pour son algèbre enveloppante universelle $U(\mathfrak{g})$ -- une algèbre associative. De telles algèbres sont appelées {\it algèbres de courants classiques}.

\section{Demi-courants et quantification des algèbres de courants}
\label{sec32}

Les courants jouent un rôle important pour la quantification des algèbres de courants. Considérons une algèbre de courants $\A_0=U(\mathfrak{g})$, où $\mathfrak{g}$ est une algèbre de Lie. Les quantification de cette algèbre considérées dans la théorie des groupes quantiques peuvent être associées avec des séparations des courants en deux ensembles $\{x^+_k(u)\}$ et $\{x^-_k(u)\}$. Des séparations différentes correspondent à des quantifications différentes. Elles sont liées à des structures différentes de bigèbre pour l'algèbre de Lie donnée. 
En particulier, si une structure de bigèbre est définie par une triplet de Manin $(\mathfrak{g},\mathfrak{g}^+,\mathfrak{g}^-)$, où $\mathfrak{g}^+$ et $\mathfrak{g}^-$ sont des sous-algèbres de $\mathfrak{g}$ telles que $\mathfrak{g}=\mathfrak{g}^+\oplus\mathfrak{g}^-$, alors $U(\mathfrak{g}^+)$ et $U(\mathfrak{g}^-)$ sont les algèbres de courants décrites par les courants $\{x^+_k(u)\}$ et $\{x^-_k(u)\}$ respectivement.

En pratique on doit souvent diviser un courant $x(u)$ en deux parties:
\begin{align} \label{xpm}
 x(u)=x^+(u)-x^-(u),
\end{align}
où $x^+(u),x^-(u)\in\mathfrak{K}'\otimes\A$ sont des distributions appelées demi-courants du courant $x(u)$. Le demi-courant $x^+(u)$ est appelé positif et $x^-(u)$ -- négatif. Pour contraste le courant initial $x(u)$ est appelé courant total. Les demi-courants sont aussi des courants pour la même algèbre $\A$. On demande que les demi-courants auraient des bonnes propriétés analytiques, qui sont spécifiques pour chaque cas. La division d'un courant total en demi-courants est réalisée par des distributions $G^+(u,v),G^-(u,v)\in\mathfrak{K}_2$ satisfaisant à l'équation
\begin{align}
 \delta(u,v)=G^+(u,v)-G^-(u,v). \label{deltaGG}
\end{align}
Elle sont appelées {\itshape distributions de Green}. Ces distribution sont aussi supposées avoir des propriétés analytiques données et l'équation~\eqref{deltaGG} est appelée un problème de Riemann pour la fonction delta. Soient des distributions de Green décomposées comme
\begin{align} \label{Gdecalbet}
 G^+(u,v)&=\sum_i\alpha^+_i(u)\beta^+_i(v), & G^-(u,v)=\sum_i\alpha^-_i(u)\beta^-_i(v),
\end{align}
où $\alpha^\pm_i,\beta^\pm_i\in\mathfrak{K}$. Alors les demi-courants correspondants sont de forme
\begin{align}
 x^+(u)&=\sum_i \alpha^+_i(u)\hat{x}[\beta^+_i], & x^-(u)&=\sum_i \alpha^-_i(u)\hat{x}[\beta^-_i].
\end{align}
Cette définitions de demi-courants ne dépend pas du choix des décompositions~\eqref{Gdecalbet} et ils sont exprimées à partir du courant total par
\begin{align}\label{hc_tc}
 x^+(u)&=\la G^+(u,v)x(v)\ra_v, & x^-(u)&=\la G^-(u,v)x(v)\ra_v.
\end{align}
Les demi-courants sont les courants correspondant aux opérateurs
\begin{align}
 \hat{x}^+&=\,\hat{x}\circ P^+, & \hat{x}^-&=-\,\hat{x}\circ P^-,
\end{align}
où les opérateurs $P^+$ et $-P^-$ correspondent aux distributions de Green : 
\begin{align}
 P^+[s](v)&=\la G^+(u,v),s(u)\ra_u, & P^-[s](v)&=-\la G^-(u,v),s(u)\ra_u,
\end{align}
$\forall s\in\mathfrak{K}$.

Considérons l'exemple le plus connu. Soit $\mathfrak{K}=\mathbb C[[u^{-1},u]]$ un espace de fonction test sur $\Sigma=\mathbb C\backslash\{0\}$ avec la forme
\begin{align}
 \la s(u),t(u)\ra_u=\oint\limits_{|u|=1}\frac{du}{(2\pi i)u} s(u)t(u),
\end{align}
pour $s,t\in\mathfrak{K}$. Soit $\mathfrak{g}=\slt\otimes\mathfrak{K}$ une algèbre des lacets. Soit $\{h,e,f\}\subset\slt$ une base de $\slt$ avec les relations de commutation
\begin{align}
 [e,f]&=h, & [h,e]&=2e, & [h,f]&=-2f.
\end{align}
Les courants totaux $h(u)$, $e(u)$ et $f(u)$ sont définis par
\begin{align} \label{tcUqslt}
 x(u)=\sum_{i\in\mathbb Z}\epsilon^i(u)x_i=\sum_{i\in\mathbb Z}u^{-i} x_i,
\end{align}
où $x\in\{h,e,f\}$, $x_i=x\otimes\epsilon_i$, $\epsilon^i(u)=u^{-i}$ et $\epsilon_i(u)=u^i$. Ils satisfont aux relations de commutation
\begin{align}
 [e(u),f(v)]&=h(u)\delta(u/v),\\
 [h(u),e(v)]&=2e(u)\delta(u/v), \\
 [h(u),f(v)]&=-2f(u)\delta(u/v), 
\end{align}
où $\delta(u/v)$ est une fonction delta pour $\mathfrak{K}=\mathbb C[[u^{-1},u]]$:
\begin{align}
 \delta(u/v)=\delta(u,v)=\sum_{i\in\mathbb Z} u^{-i}v^i
\end{align}
Le courant $e(u)$ est divisé en demi-courants par les distributions de Green
\begin{align}
 G_{(e)}^+(u,v)=G^+(u/v)&=\sum_{i\ge0} u^{-i}v^i, &  G_{(e)}^-(u,v)=G^-(u/v)&=-\sum_{i<0} u^{-i}v^i \label{GDtrig}
\end{align}
agissant comme
\begin{align}
 \la G^+(u/v),s(u)\ra_u&=\oint\limits_{|u|>|v|}\frac{du}{2\pi i}\frac{s(u)}{u-v} , &  \la G^-(u/v),s(u)\ra_u&=\oint\limits_{|u|<|v|}\frac{du}{2\pi i}\frac{s(u)}{u-v},
\end{align}
sur $s\in\mathfrak{K}$. Pour le courant de Cartan $h(u)$ on introduit les distributions de Green 
\begin{align}
 G_{(h)}^+(u,v)&=G^+(u/v)-\frac12, &  G_{(h)}^-(u,v)&=G^-(u/v)-\frac12
\end{align}
et pour le courant $f(u)$ --
\begin{align}
 G_{(f)}^+(u,v)&=G^+(u/v)-1, &  G_{(f)}^-(u,v)&=G^-(u/v)-1.
\end{align}
Cela nous donne les demi-courants suivants 
\begin{align}
 h^+(u)&=\frac12h_0+\sum_{i>0}u^{-i}h_i, &
 h^-(u)&=-\frac12h_0-\sum_{i<0}u^{-i}h_i, \label{dchUqslt} \\
 e^+(u)&=\sum_{i\ge0}u^{-i}e_i, &
 e^-(u)&=-\sum_{i<0}u^{-i}e_i, \label{dceUqslt} \\
 f^+(u)&=\sum_{i>0}u^{-i}f_i, &
 f^-(u)&=-\sum_{i\le0}u^{-i}f_i, \label{dcfUqslt}
\end{align}
où $h_i=h\otimes\epsilon_i$, $e_i=e\otimes\epsilon_i$, $f_i=f\otimes\epsilon_i$. Puisque les distributions de Green satisfont une condition
\begin{align}
 G_{(x)}^+(u,v)-G_{(x)}^-(u,v)=G^+(u/v)-G^-(u/v)=\delta(u/v),
\end{align}
où $x\in\{h,e,f\}$, nous avons
\begin{align}
 h(u)&=h^+(u)-h^-(u), & e(u)&=e^+(u)-e^-(u), & f(u)&=f^+(u)-f^-(u).
\end{align}
L'extension centrale et cocentrale d'algèbre de Lie $\slt\otimes\mathfrak{K}$ est l'algèbre affine $\widehat\slt$. Pour simplicité nous considérons toutes les formules au niveau zéro. Elles correspondent au cas sans charges centrale et cocentrale.

L'algèbre $U_q(\widehat\slt)$ mentionnée à la fin de la section~\eqref{sec19} est une quantification de l'algèbre $U(\widehat\slt)$, où $q=e^{\hbar/2}$ est un paramètre multiplicatif de quantification. Pour la première fois elle a été construite comme une quantification de $U(\widehat\slt)$ en termes de la base de Cartan-Weyl en utilisant l'analogie avec l'algèbre de Lie semi-simple. Puis, Drinfeld a décrit ces algèbres quantifiées en termes des courants pour toutes les algèbres de Lie semi-simples $\mathfrak{a}$~\cite{D88}. A niveau zéro c'est une algèbre engendrée par $h_i$, $e_i$ et $f_i$, $i\in\mathbb Z$. Les courants totaux et les demi-courants pour cette algèbre sont définis par~\eqref{tcUqslt} et \eqref{dchUqslt} -- \eqref{dcfUqslt} et satisfont les relations de commutation
\begin{gather}
[\psi^\pm(u),\psi^\pm(v)]=0, \qquad [\psi^+(u),\psi^-(v)]=0, \\
(q^{-1}u-qv)\psi^\pm(u)e(v)=(qu-q^{-1}v)e(v)\psi^\pm(u), \label{cr_UqKe} \\
(qu-q^{-1}v)\psi^\pm(u)f(v)=(q^{-1}u-qv)f(v)\psi^\pm(u), \label{cr_UqKf}\\
(q^{-1}u-qv)e(u)e(v)=(qu-q^{-1}v)e(v)e(u), \label{cr_Uqee} \\
(qu-q^{-1}v)f(u)f(v)=(q^{-1}u-qv)f(v)f(u), \label{cr_Uqff} \\
[e(u),f(v)]=\frac1{q-q^{-1}}\delta(u/v)\Big(\psi^+(u)-\psi^-(v)\Big), \label{cr_Uqef}
\end{gather}
où $\psi^+(u)=q^{2h^+(u)}$ et $\psi^-(u)=q^{2h^-(u)}$ sont les courants de Cartan multiplicatifs, qui sont utilisés dans le cas quantique. La comultiplication $\Delta$ de $U(\widehat\slt)$, qui a été introduite en termes de la base de Cartan-Weyl, ne peut pas représentée en termes des courants $\psi^+(u)$, $\psi^-(u)$, $e(u)$ et $f(u)$, mais Drinfeld a proposé la comultiplication simple pour ces courants:
\begin{align}
 &\Delta^{(D)} \psi^+(u)=\psi^+(u)\otimes \psi^+(u), &
 &\Delta^{(D)} \psi^-(u)=\psi^-(u)\otimes \psi^-(u), \label{cpKdef} \\
 &\Delta^{(D)} f(u)=f(u)\otimes \psi^+(u)+1\otimes f(u), & &\Delta^{(D)} e(u)=e(u)\otimes1+\psi^-(u)\otimes e(u). \label{cpefdef}
\end{align}
Nous désignons par $U^{(D)}_q(\widehat\slt)$ l'algèbre $U_q(\widehat\slt)$ munie de cette comultiplication. C'est une algèbre de Hopf avec une counité $\varepsilon$ et un antipode $S$ donnés par les formules
\begin{align}
 &\varepsilon(\psi^+(u))=1, & &\varepsilon(\psi^-(u))=1, \label{cuUqDK} \\ 
 &\varepsilon(f(u))=0, & &\varepsilon(e(u))=0, \label{cuUqDef} \\
 &S^{(D)}(\psi^+(u))=\psi^+(u)^{-1}, &
 &S^{(D)}(\psi^-(u))=\psi^-(u)^{-1}, \label{SUqDK} \\
 &S^{(D)}(f(u))=-f(u)\psi^+(u)^{-1}, & &S^{(D)}(e(u))=-\psi^-(u)^{-1}e(u). \label{SUqDef}
\end{align}
L'algèbre de Hopf $U^{(D)}_q(\widehat\slt)$ est une autre quantification de $U(\widehat\slt)$. Elle correspond à la séparation des courants en deux ensembles :
\begin{align}
 &\{h^+(u),f(u)\}, & &\{h^-(u),e(u)\}. \label{KpfKme}
\end{align}
En effet, les formules~\eqref{cpKdef}, \eqref{cpefdef}, \eqref{SUqDK} et \eqref{SUqDef} sont fermées par rapport aux ensembles~\eqref{KpfKme}. La quantification standard $U_q(\widehat\slt)$ correspond à la séparation suivant
\begin{align}
 &\{h^+(u),e^+(u),f^+(u)\}, & &\{h^-(u),e^-(u),f^-(u)\}, \label{dcpdcm}
\end{align}
id est en ensembles des demi-courants positifs et négatifs. La comultiplication $\Delta$ d'un courant positif (négatif) peut être exprimée par les courants positifs (négatifs), quoique les formules correspondantes soient plus compliquées que~\eqref{cpKdef}, \eqref{cpefdef}.

Les deux algèbres de Hopf $U^{(D)}_q(\widehat\slt)$ et $U_q(\widehat\slt)$ peuvent être représentées comme des produits de leurs sous-algèbre de Hopf. Soient $U_F$ et $U_E$ les sous-algèbres de $U^{(D)}_q(\widehat\slt)$ engendrées par $h_i$, $i\ge0$, $f_i$, $i\in\mathbb Z$, et $h_i$, $i\le0$, $e_i$, $i\in\mathbb Z$. Alors nous avons la dualité de Hopf $U_E=(U_F^{cop})^*$ par rapport au couple~\eqref{scpr_xuyv}. L'algèbre de Hopf $U^{(D)}_q(\widehat\slt)$ est naturellement isomorphe au double de Drinfeld ${\cal D}(U_F)$ quotienté par l'égalité $h_0\otimes1=1\otimes h_0$. Par analogie, nous avons la dualité $U^-=\big((U^+)^{cop})^*$ par rapport au même couple, où $U^-$ et $U^+$ sont les sous-algèbres de Hopf de l'algèbre $U_q(\widehat\slt)$ engendrées par $h_i$, $i\ge0$, $e_i$, $i\ge0$,$f_i$, $i>0$, et  $h_i$, $i\le0$, $e_i$, $i<0$,$f_i$, $i\ge0$. L'algèbre de Hopf $U_q(\widehat\slt)$ est naturellement isomorphe au double de Drinfeld ${\cal D}(U^+)$ quotienté par $h_0\otimes1=1\otimes h_0$. Ainsi les séparations des courants différentes pour l'algèbre $U(\widehat\slt)$ correspondent aux factorisations différentes de la même algèbre associative $U=U_q(\widehat\slt)$ quantifiant cette algèbre. Les ensembles~\eqref{KpfKme} correspondent à la factorisation $U=U_E\cdot U_F$  et les ensembles \eqref{dcpdcm} -- à $U=U^-\cdot U^+$.

Pour finir la section nous expliquons la relation des courants avec les opérateurs de Lax. Soit $U(R)$ une algèbre de Hopf décrite par une paire des opérateurs de Lax $L_+(u)$, $L_-(u)$ et une matrice $R$ trigonométrique~\eqref{R_trig} (voir la section~\ref{sec19}). Les opérateurs de Lax ont des décompositions de Gauss uniques 
\begin{align}
 L_\pm(u)=\begin{pmatrix}
					 1 & F^\pm(u) \\
					 0 & 1
          \end{pmatrix}
					\begin{pmatrix}
					 k^\pm(q^{-1}u) & 0 \\
					 0 & k^\pm(qu)^{-1}
					\end{pmatrix}
					\begin{pmatrix}
					 1 & 0 \\
					 E^\pm(u) & 1
          \end{pmatrix}.
\end{align}
Il est démontré dans~\cite{DF} (pour le cas $\gln$) que l'algèbre de Hopf $U(R)$ est isomorphe à l'algèbre de Hopf $U_q(\widehat\slt)$ et cet isomorphisme est donné par les formules
\begin{gather}
 F^\pm(u)=(q-q^{-1})f^\pm(u), \qquad\qquad\qquad  E^\pm(u)=(q-q^{-1})e^\pm(u), \\
 k^\pm(u)=q^{\pm h_0/2}\exp\Big(\pm2\log q\sum\limits_{i>0}\frac{h_{\pm i}u^{\mp i}}{q^i+q^{-i}}\Big)=\exp\Big(\frac{2\log q}{q^{u\partial_u}+q^{-u\partial_u}}h^\pm(u)\Big). \label{kpmUqdef}
\end{gather}
Les courants multiplicatifs de Cartan $\psi^\pm(u)$ sont liés avec~\eqref{kpmUqdef} comme
\begin{align}
 \psi^\pm(u)=k^\pm(q^{-1}u)k^\pm(qu).
\end{align}
Puis, Ding et Khoroshkin~\cite{DKh} ont montré que l'algèbre $U^{(D)}_q(\widehat\slt)$ peut être représenté en termes des opérateurs de Lax
\begin{align}
 L_F(u)=\begin{pmatrix}
					 1 & (q-q^{-1})f(u) \\
					 0 & 1
          \end{pmatrix}
					\begin{pmatrix}
					 k^\pm(q^{-1}u) & 0 \\
					 0 & k^\pm(qu)^{-1}
					\end{pmatrix}, \\
 L_E(u)=\begin{pmatrix}
					 k^\pm(q^{-1}u) & 0 \\
					 0 & k^\pm(qu)^{-1}
					\end{pmatrix}
					\begin{pmatrix}
					 1 & 0 \\
					 -(q-q^{-1})e(u) & 1
          \end{pmatrix},
\end{align}
où les courants $k^\pm(u)$ sont aussi définis par les formules~\eqref{kpmUqdef}. En particulier, les formules~\eqref{Delta_L_el}, \eqref{vareps_L}, \eqref{antip_L} nous donnent la comultiplication~\eqref{cpKdef}, \eqref{cpefdef}, la co\-uni\-té~\eqref{cuUqDK}, \eqref{cuUqDef} et l'antipode~\eqref{SUqDK}, \eqref{SUqDef}.

\section{Algèbres associées aux surfaces de Riemann}
\label{sec33}

Les quantifications d'une algèbre des lacets regardée à la section~\ref{sec32} correspond au cas trigonométrique, puisque le noyau des distributions de Green~\eqref{GDtrig} $\dfrac{u}{u-v}$ est une fonction trigonométrique des paramètres spectraux additives $u_{\text{add}}=\frac{1}{2\pi i}\log u$ et $v_{\text{add}}=\frac{1}{2\pi i}\log v$. Ces distributions de Green sont associées à la courbe trigonométrique -- le cylindre infini. En choisissant autres distributions de Green on peut obtenir une algèbre associée à une autre surface de Riemann. Soit $\mathfrak{K}$ une algèbre $\mathbb C[u^{-1},u]$ avec la forme
\begin{align} \label{spYang}
 \la s(u),t(u)\ra_u=\oint\limits_{|u|=1}\frac{du}{2\pi i} s(u)t(u),
\end{align}
où $s,t\in\mathfrak{K}$. Alors $u$ est une variable additive. L'algèbre quantique des courants définie par les distributions de Green avec le noyau $\frac{1}{u-v}$ est une quantification rationnelle de l'algèbre des lacets $\mathfrak{g}=\slt\otimes\mathfrak{K}$. Elle est isomorphe au double de Yangien ${\cal D}Y(\slt)$.


Les algèbres de Hopf peuvent être associées seulement aux courbes rationnelle, trigonométrique et elliptique. Elles correspondent aux solutions de l'Équation de Yang-Baxter Classique classifiées par Belavin et Drinfeld~\cite{BD}. En introduisant les quasi-algèbres de Hopf Drinfeld a proposé une problème de quantification des algèbres des courants classiques associées aux surfaces de Riemann du genre arbitraire au sens des quasi-algèbres de Hopf~\cite{D90}. Ces algèbres classiques peuvent être définies comme suit. Soit $X$ une surface de Riemann, soit $\omega$ une forme différentielle méromorphe sur $X$ et soit $\{p_1,\ldots,p_n\}$ un ensemble fini de points de $X$ contenant tous les pôles de la forme $\omega$. Soit $\lfK_p$ un corps locale dans le point $p$. C'est une complétion de l'espace de fonctions localement définies, qui sont holomorphe dans $U\backslash\{p\}$ pour un domaine $U\subset X$ quelconque contenant $p$: $p\in U$. La séquence $s_j\in\lfK_p$ converge au zéro si les ordres des pôles de $s_j$ dans $p$ sont bornés et tous les coefficients $s_{j,k}$ de la décomposition de Laurent
\begin{align}
 s_j(u_p)=\sum_{k=-k_0}^{\infty}s_{j,k}u_p^k
\end{align}
dans le point $p$ converge au zéro: $\lim\limits_{j\to\infty}s_{j,k}=0$, où $u_p$ est une variable complexe locale sur $X$ disparaît au point $p$. Ainsi $\lfK_p$ est une algèbre commutative sur $\mathbb C$, qui peut être représentée par $\mathbb C[u_p^{-1}][[u_p]]$. Soit $\mathfrak{K}=\lfK=\bigoplus\limits_{i=1}^n\lfK_{p_i}$ une algèbre commutative sur $\mathbb C$ munie par la forme
\begin{align} \label{spRes}
 \la (s_1,\ldots,s_n),(t_1,\ldots,s_n)\ra=\sum_{i=1}^n\res_{p_i}(\omega s_i t_i ),
\end{align}
où $s_i,t_i\in\lfK_{p_i}$. Définissons une algèbre de Lie correspondante comme $\mathfrak{g}=\slt\otimes\lfK$. Soit $\mathfrak{O}$ une algèbre de fonctions méromorphes sur $X$ et holomorphe sur $X\backslash\{p_1,\ldots,p_n\}$. Elle est une sous-algèbre de $\lfK$ avec la plongée $\mathfrak{O}\to\lfK$ défini par $s\to(s,\ldots,s)$. Considérons un sous-algèbre de Lie $\mathfrak{p}=(h\otimes\mathfrak{O})\oplus(f\otimes\lfK)\subset\mathfrak{g}$ et le problème suivant: construire une quantification $\HH$ de l'algèbre $U(\mathfrak{g})$ (en sens des quasi-algèbres de Hopf) telle que l'espace $\A\subset\HH$ défini comme image de $U(\mathfrak{p})$ par l'isomorphisme $U(\mathfrak{g})\cong\HH$ est une sous-quasi-algèbre de Hopf de la quasi-algèbre de Hopf $\HH$. Si cette quantification existe, on dit que l'algèbre $\HH$ est associée à les données $(\slt,X,\omega,\{p_1,\ldots,p_n\})$. Cette problème a été résolue pour le cas général par Enriquez et Roubtsov en termes des courants~\cite{ER1}. En particulier, une série formelle appelée noyau de Green a été introduite. La solution de la même problème pour $\mathfrak{p}=\slt\otimes\mathfrak{O}$ est présentée dans les travails postérieurs~\cite{ER2,ER3}. La notion de projection, à qui le chapitre prochain, a été introduit dans ces articles.

Dans les travails~\cite{ER2,ER3} des exemple des algèbre associées aux surface de Riemann du genre supérieur et, en particulier, une algèbre nouvelle associée à la courbe elliptique a été construite. Dans le travail~\cite{EF} Enriquez et Felder ont montré que après application d'un twist cette algèbre peut être représentée en termes des opérateur de Lax considérés par Felder dans~\cite{F2}. Ces opérateurs de Lax satisfont les relation $RLL$ avec la matrice $R$ de Felder. Au même temps une autre algèbre décrite en termes des opérateur de Lax avec la même matrice $R$ de Felder a apparue~\cite{ABRR97,JKOS1}. Dans~\cite{K98} il a été désignée par $U_{q,p}(\widehat\slt)$, puisque sa limite trigonométrique est isomorphe à l'algèbre de Hopf $U_q(\widehat\slt)$. Pour l'algèbre de Hopf regardée dans~\cite{EF} nous utilisons la notation $E_{\tau,\hbar}(\slt)$.

Au niveau du zéro les relations de commutation pour les deux algèbres coïncident en termes des courant et ainsi que en terme des opérateurs de Lax, mais les algèbres ne sont pas isomorphes. En particulier, les courants et les opérateurs de Lax possèdent des propriétés analytiques différentes. Les opérateurs de Lax présentés dans~\cite{F2} conviennent également aux deux algèbres. L'objet introduit là n'est pas complètement défini (comme une algèbre) et l'auteur l'appelle "algebra" (étant guillemets). La même situation est discutée dans~\cite{KLP99}. Deux algèbres rationnelles: le double du Yangien ${\cal D}Y(\slt)$ et l'algèbre $A_\hbar(\widehat\slt)$ (obtenue comme une limite rationnelle d'une autre algèbre elliptique $A_{q,p}(\widehat{\slt})$ correspondant à la matrice $R$ de Baxter-Belavin) possèdent les mêmes relations de commutation en termes des courants et en termes des opérateurs de Lax pour tous les niveaux de la charge centrale. Mais leurs courants et leurs opérateurs de Lax possèdent aussi des propriétés analytiques différentes et par conséquent ils sont différement décomposés par modes $\hat x[\epsilon_i]$ (au sens de~\eqref{def_xu}). Les courants de ${\cal D}Y(\slt)$ sont décomposés en séries formelles lorsque les courants de $A_\hbar(\widehat\slt)$ --  en intégrales formelles. Les relations de commutation en termes des modes sont différentes.

\section{Algèbres $E_{\tau,\hbar}(\slt)$ et $U_{q,p}(\widehat\slt)$}
\label{sec34}

Pour le cas elliptique la langue des modes est trop compliquée et elle n'est pas utilisée. C'est pourquoi les algèbres $E_{\tau,\hbar}(\slt)$ et $U_{q,p}(\widehat\slt)$ peuvent être confondues. Notre but est de décrire la différence entre ces algèbres en termes des courants. La première remarque sur cette différence a été écrite dans~\cite{JKOS2}. Les auteurs ont noté que les relations de commutation pour les extensions centrales des algèbres $E_{\tau,\hbar}(\slt)$ et $U_{q,p}(\widehat\slt)$ sont différentes. Ils ont expliqué ce fait par des contours d'intégration différents entrant dans l'expression des demi-courants via les courants totaux. Le contour d'intégration pour l'algèbre $E_{\tau,\hbar}(\slt)$ est un petit cercle autour du point $u=0$, lorsque le contour d'intégration pour l'algèbre $U_{q,p}(\widehat\slt)$ est un cercle autour du point $z=0$, où $z=q^{2u}$ est une variable multiplicative. Au niveau classique ces contours définissent le 2-cocycle qui définit l'extension centrale. Ensuite, dans l'article~\cite{EPR} on développe les idées de~\cite{KLP99} pour montrer la différence entre ces algèbres même à niveau zéro, quand les relations de commutation (et ainsi les formules pour comultiplication) coïncident. L'argument principal proposé là vient des décompositions différentes des demi-courants en séries formelles liées aux décompositions différentes des noyaux de Green en séries formelles.

Soit $\theta(u)$ une fonction thêta définie par les équations~\footnote{La fonction $\theta(u)$ est lié avec la fonction
\begin{align}
 &\vartheta_1(u)=\sum_{k\in\mathbb Z}e^{\pi i\tau(k+\frac12)^2+2\pi i(k+\frac12)(u+\frac12)}
& &\text{par la relation}
 &\theta(u)=\frac{\vartheta_1(u)}{\vartheta'_1(0)}.
\end{align}
}
\begin{align}
 \theta(u+1)&=-\theta(u), &  \theta(u+\tau)&=-e^{-2\pi iu-\pi i\tau}\theta(u), &   \theta'(0)&=1.
 \label{D_theta}
\end{align}
Les noyaux de Green de l'algèbre $E_{\tau,\hbar}(\slt)$ sont des fonctions méromorphes quasi-doublement périodiques décomposées en séries de Taylor autour de zéro:
\begin{align}
 \frac{\theta'(u-v)}{\theta(u-v)}&=\sum_{k=0}^\infty\frac{(-1)^k}{k!} \Big(\frac{\theta'(u)}{\theta(u)}\Big)^{(k)}v^k, \label{GkEk} \\
 \frac{\theta(u-v+\lambda)}{\theta(u-v)\theta(\lambda)}&=\sum_{k=0}^\infty\frac{(-1)^k}{k!} \Big(\frac{\theta(u+\lambda)}{\theta(u)\theta(\lambda)}\Big)^{(k)}v^k,\label{GkEef}
\end{align}
Les noyaux de Green de l'algèbre $U_{q,p}(\widehat\slt)$ sont les mêmes fonctions décomposées aux séries de Fourier:
\begin{align}
 \frac{\theta'(u-v)}{\theta(u-v)}&=\pi i+2\pi i\sum_{n\ne0}\frac{e^{-2\pi in(u-v)}}{1-e^{2\pi in\tau}}, \label{GkUk} \\
 \frac{\theta(u-v+\lambda)}{\theta(u-v)\theta(\lambda)}&=2\pi i\sum_{n\in\mathbb Z}\frac{e^{-2\pi in(u-v)}}{1-e^{2\pi i(n\tau-\lambda)}}. \label{GkUef} 
\end{align}

Dans~\cite{S21}, \cite{S22} (voir Appendices~\ref{SA21}, \ref{SA22}) nous combinons les idées des articles~\cite{JKOS2} et \cite{EPR} pour analyser la différence des algèbres $E_{\tau,\hbar}(\slt)$ et $U_{q,p}(\widehat\slt)$ en détail. Nous regardons les séries formelles~\eqref{GkEk} -- \eqref{GkUef} comme des distributions agissant sur les espaces correspondants de fonctions test. C'est pourquoi nous les appelons {\it distributions} de Green. L'espace de fonction test $\mathfrak{K}$ pour l'algèbre $E_{\tau,\hbar}(\slt)$ est un corps local à l'origine $\lfK_0\hm=\mathbb C[u^{-1}][[u]]$ muni de la forme
\begin{align} \label{spE}
 \la s(u),t(u)\ra_u=\oint\limits_{C_0}\frac{du}{2\pi i} s(u)t(u),
\end{align}
où $s,t\in\lfK_0$ et $C_0$ est un petit contour autour de l'origine. C'est un cas particulier de la forme~\eqref{spRes} pour la courbe elliptique  $X={\cal E}=\mathbb C/\Gamma$, où $(\Gamma=\mathbb Z+\tau\mathbb Z)$, avec un module $\tau$, $\Im\tau>0$, l'ensemble de points $\{x_1,\ldots,x_n\}=\{0\}\subset{\cal E}$ et la forme différentielle $\omega=du$. Les distributions de Green correspondantes sont
\begin{align}
 G^+_{(h)}(u,v)&=G(u,v),  & G^+_{(e)}(u,v)&=G^+_\lambda(u,v), & G^+_{(f)}(u,v)&=G^+_{-\lambda}(u,v), \label{DdGEp}\\
 G^-_{(h)}(u,v)&=-G(v,u), & G^-_{(e)}(u,v)&=G^-_\lambda(u,v), & G^-_{(f)}(u,v)&=G^-_{-\lambda}(u,v), \label{DdGEm}
\end{align}
où $G^+_\lambda(u,v)$, $G^-_\lambda(u,v)$ et $G(u,v)$ sont des éléments de $(\lfK_0\otimes\lfK_0)'$ avec les actions partielles par $u$ définies comme
\begin{align}
\la G_\lambda^+(u,v),s(u)\ra_u&=
  \oint\limits_{|u|>|v|}\frac{du}{2\pi i}\frac{\theta(u-v+\lambda)}{\theta(u-v)
  \theta(\lambda)}s(u)\,,
  \label{Glpmact} \\
 \la G^-_\lambda(u,v),s(u)\ra_u&=
  \oint\limits_{|u|<|v|}\frac{du}{2\pi i}\frac{\theta(u-v+\lambda)}{\theta(u-v)\theta(\lambda)}s(u)\,,
  \label{Glmact} \\
 \la G(u,v),s(u)\ra_u&=
  \oint\limits_{|u|>|v|}\frac{du}{2\pi i}\frac{\theta'(u-v)}{\theta(u-v)}s(u). \label{Gact}
\end{align}
Ces formules définissent aussi les actions par $v$:
\begin{align}
\la G_\lambda^+(u,v),s(v)\ra_v&=
  \oint\limits_{|u|>|v|}\frac{dv}{2\pi i}\frac{\theta(u-v+\lambda)}{\theta(u-v)
  \theta(\lambda)}s(v)\,,
  \label{Glpmactv} \\
 \la G^-_\lambda(u,v),s(v)\ra_v&=
  \oint\limits_{|u|<|v|}\frac{dv}{2\pi i}\frac{\theta(u-v+\lambda)}{\theta(u-v)\theta(\lambda)}s(v)\,,
  \label{Glmactv} \\
 \la G(u,v),s(v)\ra_v&=
  \oint\limits_{|u|>|v|}\frac{dv}{2\pi i}\frac{\theta'(u-v)}{\theta(u-v)}s(v). \label{Gactv}
\end{align}
Les parties gauches de~\eqref{Glpmact} -- \eqref{Gact} et \eqref{Glpmactv} -- \eqref{Gactv} sont des fonctions test de $v$ et de $u$ respectivement. Grâce à l'identité (en particulier)
\begin{align}
 G^+_\lambda(u,v)=-G^-_{-\lambda}(v,u)
\end{align}
les distributions de Green~\eqref{DdGEp}, \eqref{DdGEm} appartiennent à l'espace $\mathfrak{K}_2=(\lfK_0)_2$. Les distributions ${\cal G}(u-v)$, et ${\cal G}^+_\lambda(u-v)$ peuvent être représentées comme des séries~\eqref{GkEk}, \eqref{GkEef}.

L'algèbre de fonctions pour $U_{q,p}(\widehat\slt)$ est une algèbre $K$. Elle consiste en fonctions périodiques $s(u)$ sur $\mathbb C$: $s(u)=s(u+1)$, telles que $|s(u)|\le C e^{p|\Im u|}$, où $C,p>0$ dépendent de fonctions $s(u)$. Id est c'est un espace des séries de Fourier finies: $K=\mathbb[e^{-2\pi i u},e^{2\pi i u}]$. Grâce à périodicité, c'est un espace de fonctions sur le cylindre $\Cyl=\mathbb C/\mathbb Z$: $K=K(\Cyl)$. L'algèbre $K$ est munie de la forme
\begin{align}
 \la s(u),t(u)\ra=\int\limits_{-\frac12+\alpha}^{\frac12+\alpha}\frac{du}{2\pi i}s(u)t(u), \label{spJ}
\end{align}
où $s,t\in K$, qui ne dépend pas de $\alpha\in\mathbb C$ et est représentée comme une intégrale autour du point $z=0$, où $z=e^{2\pi i u}$ est une variable multiplicative. Le rôle du contour d'intégration est joué par le cycle du cylindre $\Cyl$. La convergence dans $K$ est donnée comme suit: la séquence d'éléments $s_i\in K$ tend vers zéro si on peut trouver des constantes $C,p>0$ telles que $|s_i(u)|\le C e^{p|\Im u|}$ et pour toutes les valeurs $u\in\mathbb C$ la séquence $s_i(u)\to0$. En particulier, si $s_n\to0$ alors les fonctions $s_i(u)$ tendent uniformément vers zéro. Par conséquent, la forme~\eqref{spJ} est continue par rapport aux deux arguments et définit le plongement continu $K\to K'$. Notons que un opérateur ${\cal T}_t\colon s(u)\mapsto s(u+t)$ est aussi continu. Par conséquent, pour chaque distribution $a(u)\in K'$ on peut définir la distribution $a(u-t)$ agissant comme
\begin{align}
 \la a(u-t),s(u)\ra_u=\la a(u),s(u+t)\ra_u.
\end{align}
Puisque $\la a(u),s(u+v)\ra_u\in K$ pour tous $a(u)\in K'$ et $s(u)\in K$ nous pouvons considérer $a(u-v)$ comme un élément de $K_2$. Les distributions de Green pour ce cas sont
\begin{align}
 {\cal G}^+_{(h)}(u,v)&={\cal G}(u-v),  & {\cal G}^+_{(e)}(u,v)&={\cal G}^+_\lambda(u-v), & {\cal G}^+_{(f)}(u,v)&={\cal G}^+_{-\lambda}(u-v), \label{DdGEpJ}\\
 {\cal G}^-_{(h)}(u,v)&=-{\cal G}(v,u), & {\cal G}^-_{(e)}(u,v)&={\cal G}^-_\lambda(u-v), & {\cal G}^-_{(f)}(u,v)&={\cal G}^-_{-\lambda}(u-v), \label{DdGEmJ}
\end{align}
où ${\cal G}^+_\lambda(u)$, ${\cal G}^-_\lambda(u)$ et ${\cal G}(u)$ sont des éléments de $K'$ agissant sur $s(u)\in K$ par les formules
\begin{align}
 \la{\cal G}_\lambda^+(u),s(u)\ra_u&=\int\limits_{-\Im\tau<\Im u<0}
 \frac{du}{2\pi i} \frac{\theta(u+\lambda)}{\theta(u)\theta(\lambda)}s(u), \label{GlpmdefJ} \\
 \la{\cal G}_\lambda^-(u),s(u)\ra_u&=\int\limits_{0<\Im(u)<\Im\tau}\frac{du}{2\pi i}
 \frac{\theta(u+\lambda)}{\theta(u)\theta(\lambda)}s(u), \label{GlmdefJ} \\
 \la{\cal G}(u),s(u)\ra_u
   &=\int\limits_{-\Im\tau<\Im(u)<0}\frac{du}{2\pi i}\frac{\theta'(u)}{\theta(u)}s(u), \label{GdefJ}
\end{align}
où les contours d'intégration sont des segments horizontaux $[\alpha-\frac12,\alpha+\frac12]$ dont les points satisfont les conditions correspondantes. Les décompositions des distributions de Green ${\cal G}(u-v)$ et ${\cal G}^+_\lambda(u-v)$ par des bases duales sont les parties droites de~\eqref{GkUk} et \eqref{GkUef}.

Puisque les demi-courants ont la structure~\eqref{hc_tc} avec les mêmes distributions de Green dans les cas classique et quantique, nous considérons toutes les formules au niveau classique: $\hbar\to0$ ($q\to1$). Désignons les algèbres de Lie correspondant par $\mathfrak{e}_\tau(\slt)$ et $\mathfrak{u}_\tau(\widehat\slt)$. Considérons d'abord les relations de commutation entre demi-courants $e^+_\lambda(u)$ et $f^+_\lambda(u)$, par exemple. Pour l'algèbre de Lie $\mathfrak{e}_\tau(\slt)$ elles sont écrites comme:
\begin{align}
 [e_\lambda^+(u),f_\lambda^+(v)]&=G_\lambda^+(u,v)\big(-h^+(u)+h^+(v)\big)+h\frac{\partial}{\partial\lambda}G_\lambda^+(u,v), \label{epmfpm}
\end{align}
lorsque pour $\mathfrak{u}_\tau(\widehat\slt)$ --
\begin{align}
 [e_\lambda^+(u),f_\lambda^+(v)]&={\cal G}_\lambda^+(u-v)\big(-h^+(u)+h^+(v)\big)+h\frac{\partial}{\partial\lambda}{\cal G}_\lambda^+(u-v), \label{epmfpmJ}
\end{align}
où $h=\la h(u),1\ra_u$. Les relations coïncident quand on écrit les noyaux correspondants au lieu des distributions. Mais en termes des distributions elles sont différentes.

Les demi-courants de l'algèbre de Lie $\mathfrak{u}_\tau(\widehat\slt)$ possèdent les propriétés importantes
\begin{align}
 h^+(u-\tau)&=2\pi i h+h^-(u), & e^+_\lambda(u-\tau)&=e^{2\pi i\lambda}e^-_\lambda(u), & f^+_\lambda(u-\tau)&=e^{-2\pi i\lambda}f^-_\lambda(u)
\end{align}
suivantes de les relations correspondantes pour les distributions de Green
\begin{align}
{\cal G}_\lambda^+(u-v-\tau)&=e^{2\pi i\lambda}{\cal G}_\lambda^-(u-v),  &{\cal G}(u-v-\tau)&=2\pi i-{\cal G}(v-u). \label{shift_tau}
\end{align}
Cela nous permet d'exprimer les demi-courants négatifs par les demi-courants positifs et vice versa, ce qui signifie qu'on peut effectivement décrire l'algèbre de Lie $\mathfrak{u}_\tau(\widehat\slt)$ (et algèbre $U_{q,p}(\widehat\slt)$ dans le cas quantique) par les demi-courants positifs seuls ou par les demi-courants négatifs seuls.

En même temps on ne peut pas exprimer les demi-courants positifs (négatifs) de l'algèbre $\mathfrak{e}_\tau(\slt)$ par d'autres. Ce fait est lié d'abord au fait que les contours d'intégration pour les distributions de Green positives et négatives ne sont pas liés par une translation parallèle. Puis les expressions de type $h^+(u-\tau)$ en ce cas ne sont pas définies, parce que l'opérateur $T_t\colon s(u)\mapsto s(u+t)$ n'est pas continu. En effet, considérons, par exemple, les sommes $s_N(u)=\sum_{n=0}^N(\frac{u}{\alpha})^n$. Pour chaque $v$ on peut trouver un nombre $\alpha$ tel que la séquence $s_N(u+v)$ diverge, lorsque $N\to\infty$.

Mais la cause principale de l'impossibilité de relation entre les demi-courants positifs et négatifs est liée à la particularité suivante de l'algèbre de Lie $\mathfrak{e}_\tau(\slt)$. Considérons les opérateurs  $P^+$, $P^-$ et $P_\lambda^+$, $P_\lambda^-$ correspondant aux distributions de Green~\eqref{DdGEp},~\eqref{DdGEm}. Ils agissent sur $s\in\lfK_0$ comme
\begin{align}
 P_\lambda^+[s](v)&=\la G_\lambda^+(u,v),s(u)\ra_u, & P_\lambda^-[s](v)&=-\la G_\lambda^-(u,v),s(u)\ra_u, \\
 P^+[s](v)&=\la G(u,v),s(u)\ra_u, & P^-[s](v)&=\la G(v,u),s(u)\ra_u.
\end{align}
Les opérateurs $P^+$, $P^-$ et $P_\lambda^+$, $P_\lambda^-$ sont idempotents orthogonaux:
\begin{align}
 P^++P^-&=\id, & P^\pm\circ P^\pm&=P^\pm, & P^-\circ P^+&=P^+\circ P^-=0, \\
 P_\lambda^++P_\lambda^-&=\id, & P_\lambda^\pm\circ P_\lambda^\pm&=P_\lambda^\pm, & P_\lambda^-\circ P_\lambda^+&=P_\lambda^+\circ P_\lambda^-=0.
\end{align}
Ils définissent une décomposition de $\lfK_0$ en deux sous-espaces: $\lfK_0={\cal O}\oplus\Lambda_0$ et $\lfK_0={\cal O}\oplus\Lambda_\lambda$, où ${\cal O}=\mathbb C[[u]]$, $\Lambda_0$ est un sous-espace de $\lfK_0$ engendré par les fonctions $\epsilon^{k;0}(u)=\frac{1}{k!}\big(\frac{\theta'(u)}{\theta(u)}\big)^{(k)}$, $k\ge0$, et $\Lambda_\lambda$ pour $\lambda\not\in\Gamma$ est un sous-espace de $\lfK_0$ engendré par les fonctions $\epsilon^{k;\lambda}(u)=\frac{1}{k!}\Big(\frac{\theta(u+\lambda)}{\theta(u)\theta(\lambda)}\Big)^{(k)}$, $k\ge0$.
Cela signifie que la décomposition des courants totaux en demi-courants entraîne la représentation de l'algèbre de Lie $\mathfrak{g}=\mathfrak{e}_\tau(\slt)$ en une somme directe. Ainsi les ensembles~\eqref{KpfKme} correspond à la décomposition
\begin{align}
 \mathfrak{g}=\mathfrak{g}_e\oplus\mathfrak{g}_f,
\end{align}
où $\mathfrak{g}_e=(h\otimes\Lambda_0)\oplus(e\otimes\lfK_0)$ et $\mathfrak{g}_f=(h\otimes{\cal O})\oplus(f\otimes\lfK_0)$ sont les sous-algèbres de Lie engendrées par les courants du premier et du second ensemble respectivement.
Cela nous donne la factorisation
\begin{align}
 E^{(D)}_{\tau,\hbar}(\slt)=\A_E\cdot\A_F,
\end{align}
où $E^{(D)}_{\tau,\hbar}(\slt)$ est une algèbre $E_{\tau,\hbar}(\slt)$ muni d'une comultiplication analogue à la comultiplication de Drinfeld~\eqref{cpKdef}, \eqref{cpefdef} de l'algèbre $U_{q}(\widehat\slt)$, et $\A_E$ et $\A_F$ sont des sous-algèbres de Hopf engendrées par $\la h(u),\epsilon^{k;0}(u)\ra_u$, $k\ge0$, $\la e(u),s(u)\ra_u$, et $\la h(u),\epsilon_{k;0}(u)\ra_u$, $k\ge0$, $\la f(u),s(u)\ra_u$, respectivement, où $s\in\lfK_0$ et $\epsilon_{k;0}=(-u)^k$ pour $k\ge0$.
Les ensembles~\eqref{dcpdcm} correspond à la décomposition
\begin{align}
 \mathfrak{g}=\mathfrak{g}^-_\lambda\oplus\mathfrak{g}^+,
\end{align}
où $\mathfrak{g}^-_\lambda=(h\otimes\Lambda_0)\oplus(e\otimes\Lambda_\lambda) \oplus(f\otimes\Lambda_{-\lambda})$ et $\mathfrak{g}^+=\slt\otimes{\cal O}$ sont sous-espaces de $\mathfrak{g}$ engendrés par les courants correspondants ($\mathfrak{g}^+$ est une sous-algèbre de Lie). Nous avons aussi la factorisation
\begin{align}
 E^{(std)}_{\tau,\hbar}(\slt)=\A^-\cdot\A^+,
\end{align}
où $E^{(sdt)}_{\tau,\hbar}(\slt)$ est une algèbre $E_{\tau,\hbar}(\slt)$ muni par une comultiplication analogique à la comultiplication standard d'algèbre $U_{q}(\widehat\slt)$, $\A^+$ est une sous-algèbre de Hopf de $E_{\tau,\hbar}(\slt)$ engendrée par $\la x(u),\epsilon_{k;0}(u)\ra_u$, où $x\in\{h,e,f\}$ et $\A^-$ est un sous-espace complété algébrique\-ment $\A^+$, qui est défini au Chapitre~\ref{sec4}. 

Ainsi nous voyons que les demi-courants~\eqref{KpfKme} et \eqref{dcpdcm} sont associés aux factorisations d'algèbre de Hopf $E_{\tau,\hbar}(\slt)$ à des sous-espaces décrits par ces courants. Cela est aussi valable pour toutes les algèbres de courants, qui sont obtenues par la construction de Enriquez et Roubtsov~\cite{ER1,ER2,ER3}.

Dans le cas de l'algèbre $U_{q,p}(\widehat\slt)$ la situation est essentiellement différente. Les ensembles des courants~\eqref{KpfKme} et \eqref{dcpdcm} ne sont pas associés avec des sous-espaces de l'algèbre de Lie $\mathfrak{u}_\tau(\widehat\slt)$ et de l'algèbre $U_{q,p}(\widehat\slt)$. Au contraire, on peut décrire toute l'algèbre $U^{(std)}_{q,p}(\widehat\slt)$ par chaque ensemble $\{h^+(u),e^+_\lambda(u),f^+_\lambda(u)\}$ ou $\{h^-(u),e^-_\lambda(u),f^-_\lambda(u)\}$.
L'algèbre $U^{(D)}_{q,p}(\widehat\slt)$ peut être complètement décrite par l'ensemble $\{h^+(u),e(u),f(u)\}$ ou $\{h^-(u),e(u),f(u)\}$. C'est aussi lié à des propriétés algébriques des distributions de Green. En effet, soient ${\cal P}^+$, ${\cal P}^-$ et ${\cal P}_\lambda^+$, ${\cal P}_\lambda^-$ les opérateurs correspondant aux distributions de Green~\eqref{DdGEpJ},~\eqref{DdGEmJ}. Ils agissent sur $s\in K$ comme
\begin{align}
 {\cal P}_\lambda^+[s](v)&=\la {\cal G}_\lambda^+(u,v),s(u)\ra_u, & {\cal P}_\lambda^-[s](v)&=-\la {\cal G}_\lambda^-(u,v),s(u)\ra_u, \\
 {\cal P}^+[s](v)&=\la {\cal G}(u,v),s(u)\ra_u, & {\cal P}^-[s](v)&=\la {\cal G}(v,u),s(u)\ra_u.
\end{align}
Les images de tous les opérateurs ${\cal P}^+$, ${\cal P}^-$ et ${\cal P}_\lambda^+$, ${\cal P}_\lambda^-$ coïncident avec toute l'algèbre de fonction test $K$. Cela signifie, par exemple, que l'algèbre $U_{q,p}(\widehat\slt)$ est engendrée par $\la h^+(u),s(u)\ra$, $\la e^+_\lambda(u),s(u)\ra$ et $\la f^+_\lambda(u),s(u)\ra$, $s\in K$. Ces opérateurs satisfont ${\cal P}_\lambda^++{\cal P}_\lambda^-=\id$, mais ils ne sont pas idempotents orthogonaux. En particulier, on a d'obstacle suivant
\begin{align}
 {\cal P}_\lambda^-\circ {\cal P}_\lambda^+&={\cal P}_\lambda^+\circ {\cal P}_\lambda^-=\frac{1}{2\pi i}\frac{\partial}{\partial \lambda}{\cal P}_\lambda^+=-\frac{1}{2\pi i}\frac{\partial}{\partial \lambda}{\cal P}_\lambda^-.
\end{align}

\section{Dégénérescences d'algèbres elliptiques}
\label{sec35}

Considérons des dégénérescences rationnelles et trigonométriques des algèbres $E_{\tau,\hbar}(\slt)$ et $U_{q,p}(\widehat\slt)$. Les différences entre ces algèbres entraînent à des différences entre leurs dégénérescences. Ces dégénérescences sont aussi des algèbres des courants qui peuvent être définies par une algèbre de fonctions test et les distributions de Green correspondantes. C'est pourquoi on peut les rechercher au niveau classique~\cite{S22}. Nous résumons ici les résultats.

Une dégénérescence rationnelle est une limite $\omega\to+\infty$, $\omega'\to+i\infty$ d'une algèbre associée à la courbe elliptique $\mathbb C/(\omega\mathbb Z+\omega'\mathbb Z)$, où $\dfrac{\omega'}{\omega}=\tau$. Pour la réaliser il faut d'abord replacer $u\to\dfrac u\omega$, $v\to\dfrac v\omega$ et $\lambda\to\mu+\dfrac\lambda\omega$, où $\mu$ est un nombre complexe arbitraire. Nous désignons cette limite par \textbf{(a)} et nous montrons qu'elle ne dépend pas de $\mu$. On peut considérer aussi deux types de limites trigonométriques. La première limite, désignée par \textbf{(b)}, est définie comme $\tau\to+i\infty$ (id est $\omega$ fixé et $\omega'\to+i\infty$). Elle n'exige pas des renormalisations des paramètres spectral et dynamique. La second limite $\omega\to+\infty$ (lorsque $\omega'$ est fixé) exige les renormalisations des paramètres spectral et dynamique et elle est désignée par \textbf{(c)}. Dans ce cas la limite dépend du paramètre $\mu$.

Pour les deux algèbres les noyaux d'intégration des distributions de Green ont les mêmes dégénérescences~\footnote{La renormalisation de paramètre spectral $u\to\frac{u}{\omega}$ entraîne la substitution $du\to\frac{du}{\omega}$ dans les formes~\eqref{spE} et \eqref{spJ} ce qui entraîne dans son lieu à la renormalisation de noyaux de distributions}:
\begin{align}
 \lim_{\omega\to+\infty \atop \omega'\to+i\infty} \frac1{\omega}\frac{\theta'(\frac{u-v}{\omega})}{\theta(\frac{u-v}{\omega})} &=\frac{1}{u-v}, &
  &\lim_{\omega\to+\infty \atop \omega'\to+i\infty} \frac1{\omega}\frac{\theta(\frac{u-v+\lambda}{\omega})}{\theta(\frac{u-v}{\omega})\theta(\frac{\lambda}{\omega})}=\frac{1}{u-v}+\frac1{\lambda}, \label{degkera} \\
 \lim_{\tau\to+i\infty} \frac{\theta'(u-v)}{\theta(u-v)} &=\pi\ctg\pi(u-v), &
 &\lim_{\tau\to+i\infty} \frac{\theta'(u-v+\lambda)}{\theta(u-v)\theta(\lambda)} =\pi\ctg\pi(u-v)+\pi\ctg\pi\lambda, \label{degkerb} \\
 \lim_{\omega\to+\infty \atop \omega'=const} \frac1{\omega}\frac{\theta'(\frac{u-v}{\omega})}{\theta(\frac{u-v}{\omega})}&=\pi\eta\cth\pi\eta(u-v), &
  &\lim_{\omega\to+\infty \atop \omega'=const} \frac1{\omega}\frac{\theta(\mu+\frac{u-v+\lambda}{\omega})}{\theta(\frac{u-v}{\omega})\theta(\mu+\frac{\lambda}{\omega})}=2\pi\eta\frac{e^{-2\pi\eta\mu(u-v)}}{1-e^{-2\pi\eta(u-v)}}, \label{degkerc} \\
  &\lim_{\omega\to+\infty \atop \omega'=const} \frac1{\omega}\frac{\theta(\frac{u-v+\lambda}{\omega})}{\theta(\frac{u-v}{\omega})\theta(\frac{\lambda}{\omega})} &&=\pi\eta\cth\pi\eta(u-v)+\pi\eta\cth\pi\eta\lambda, \label{degkerc0} 
\end{align}
où $\eta=\dfrac{i}{\omega'}$, $\Re\eta>0$ et $\mu$ est un nombre arbitraire satisfaisant à l'inégalité
\begin{align}
 \frac{\Im\eta\Im\mu}{\Re\eta}<\Re\mu<\frac{\Im\eta\Im\mu}{\Re\eta}+1. \label{zone_mu}
\end{align}
Puisque les algèbres $E_{\tau,\hbar}(\slt)$ et $U_{q,p}(\widehat\slt)$ sont dynamiques avec paramètre dynamique $\lambda$, leurs dégénérescences peuvent donner aussi la dynamique par $\lambda$. Dans les cas \textbf{(a)} et \textbf{(b)} c'est la dynamique "triviale", qui est définie par des termes dynamiques séparés: $\frac1\lambda$ et $\pi\ctg\pi\lambda$ respectivement. Cette dynamique disparaît à la limite "dédynamisant" correspondante: $\lambda\to\infty$ et $\lambda\to-i\infty$. Cette dynamique "triviale" mérite l'attention séparée, mais pour nos buts il suffis de considérer les dégénérescences aux limites "dédynamisantes" correspondantes. Dans le cas \textbf{(c)} la limite $\lambda\to+\infty$ dans la formule~\eqref{degkerc0} correspond à la continuation analytique de la formule~\eqref{degkerc} au point $\mu=0$ (voir~\eqref{degkerc0}). Pour une valeur d'un paramètre $\mu$ satisfaisant~\eqref{zone_mu} la dynamique par $\lambda$ disparaît à la limite \textbf{(c)}. Nous verrons sur des formules explicites que le paramètre $\mu$ n'est pas un paramètre dynamique.

L'algèbre de fonctions test pour les dégénérescences de l'algèbre $E_{\tau,\hbar}(\slt)$ est la même algèbre $\lfK_0$ (avec la même forme). La dégénérescence rationnelle (limite~\textbf{(a)}) est défini par les distribution de Green avec les noyaux~\eqref{degkera} et les contours d'intégrations $|u|>|v|$ et $|u|<|v|$. Cela nous donne le double de Yangien ${\cal D}Y(\slt)$. Les distributions de Green de la dégénérescence trigonométrique (limite~\textbf{(b)}) correspondent aux noyaux~\eqref{degkerb} avec les mêmes contours. Cette algèbre n'est pas apparue dans la littérature, et quelquefois on l'appelle $q$-Yangien ou double de $q$-Yangien ${\cal D}Y_q(\slt)$. D'habitude son analogue plus connu $U_q(\widehat\slt)$ est utilisée. La limite~\textbf{(c)} de l'algèbre $E_{\tau,\hbar}(\slt)$ est équivalente à la limite~\textbf{(b)} parce que on peut échanger les rôles de périodes elliptiques $\omega$ et $\omega'$.

Les rôles des algèbres de fonctions test pour les dégénérescences de l'algèbre $U_{q,p}(\widehat\slt)$ sont joués par l'algèbre $K$ et $Z$. Cette dernière est une algèbre de fonctions entières $s(u)$ satisfaisant aux inégalités $|u^n s(u)|<C_n e^{p|\Im u|}$, $n\in\mathbb Z_+$, avec des nombres $C_n,p>0$ dépendent de $s(u)$~\cite{GSh1}. Elle est muni de la forme 
\begin{align}
 \la s(u),t(u)\ra_u=\int\limits_{-\infty}^{+\infty}\frac{du}{2\pi i}s(u)t(u).
\end{align}
L'espace des distributions $K'$ peut être regardé comme le sous-espace de l'espace $Z'$ consistant en des distributions périodiques. La dégénérescence rationnelle de $U_{q,p}(\widehat\slt)$ (limite~\textbf{(a)}) est une algèbre $A_\hbar(\slt)$. Les distributions de Green correspondantes agissent sur $Z$. Elles correspondent aux mêmes noyaux~\eqref{degkera}, et aux autres contours d'intégration. Ce sont les contours horizontaux au-dessous et au-dessus de l'axe réel pour les distributions de Green positives et négatives respectivement. Ainsi la différence entre les algèbres rationnelles ${\cal D}Y(\slt)$ et $A_\hbar(\slt)$ décrites en~\cite{KLP99} en détail est héritée de la différence entre les algèbres elliptiques $E_{\tau,\hbar}(\slt)$ et $U_{q,p}(\widehat\slt)$. Comme c'est déjà écrit, la dégénérescence rationnelle de $U_{q,p}(\widehat\slt)$ (limite~\textbf{(b)}) est une algèbre $U_q(\widehat\slt)$. Les distributions de Green agissent sur $K$, ont les noyaux~\eqref{degkerb} et les mêmes contours d'intégration comme dans le cas d'algèbre $U_{q,p}(\widehat\slt)$.

L'algèbre $U_{q,p}(\widehat\slt)$ a une autre dégénérescence trigonométrique -- la limite~\textbf{(c)}. Elle n'est pas équivalente à la limite~\textbf{(b)} parce que les contours d'intégration dans ce cas ne sont pas symétriques par rapport à l'échange $\omega\leftrightarrow\omega'$. Elle entraîne une série d'algèbre $A_{\hbar,\eta}(\slt;\mu)$ paramétrisée par $\eta$ et $\mu$. Le paramètre $\eta$ est hérité du module elliptique $\tau$ lorsque $\mu$ est un paramètre de dégénérescence apparu comme la partie finie d'un paramètre dynamique $\lambda$. Néanmoins, il vaut souligner que $\mu$ n'est pas un paramètre dynamique et l'algèbre $A_{\hbar,\eta}(\slt;\mu)$ n'est pas une algèbre dynamique au niveau du zéro. Comme une partie fini du paramètre $\lambda$ le paramètre $\mu$ permet la transformation $\mu\to\mu+1$ et, par conséquent, on peut se restreindre à la zone~\eqref{zone_mu}. En même temps la formule~\eqref{degkerc} n'a pas de période 1 par rapport de $\mu$. C'est une spécificité de cette dégénérescence. En fait, la formule~\eqref{degkerc} doit être considérée comme sa continuation analytique.

\selectlanguage{english}

\chapter[Elliptic projections and SOS model with DWBC]{Elliptic projections and SOS model with Domain Wall Boundary Conditions}
\label{sec4}

In this chapter we consider projections for the elliptic quasi-Hopf algebra $E_{\tau,\hbar}(\slt)$ introduced in~\cite{EF} and regarded in Chapter~\ref{sec3}. We apply these projections to the SOS model. These results published in~\cite{S3} (see Appendix~\ref{SA3}). We also present some algebraic properties of the elliptic projections. 

The projections of currents first appeared in the works of Enriquez and Rubtsov~\cite{ER2}, \cite{ER3} in a purely algebraic framework. This was a method to construct current algebras for higher genus corresponding to the sets of currents of type~\eqref{dcpdcm}. Further, Khoroshkin and Pakuliak applied this method for the quantum affine algebras to factorize the Universal $R$-matrix~\cite{DKhP} and to obtain the universal Bethe ansatz~\cite{KhP}, \cite{EKhP}. It was observed that the projections for the algebra $U_q(\widehat{\slt})$ can be presented as an integral transform and the kernel of this transform is proportional to the partition function of the finite 6-vertex model with Domain Wall Boundary Conditions (DWBC)~\cite{KhP}. We conjecture that the elliptic projections described in~\cite{EF} can help to derive the partition function for some elliptic model and the first candidate for the role of such model is an elliptic SOS model mentioned in the section~\ref{sec110} with the corresponding boundary conditions.

In~\cite{Kor} Korepin derived recurrent relations for the partition function of the finite 6-vertex model with DWBC. Further Izergin used these relations to find the expression for the partition function in a determinant form~\cite{I87}. The integral kernel of projections calculated by Khoroshkin and Pakuliak satisfies the same recursive relations and it gives another formula for the partition function.

Unfortunately, the Izergin's determinant formula can not be generalized to the elliptic case, but it is natural to expect that the theory of projections gives an expression for the partition function for the SOS model. In one hand we generalize Korepin's recurrent relations and in other hand we generalize the method of calculation of projections proposed by Khoroshkin and Pakuliak to the elliptic case. By analogy with the trigonometric (6-vertex model) case we present the result of the calculation as an integral transform. We check that the kernel extracted from this transform and multiplied by a corresponding factor satisfies the obtained recurrent relations, which uniquely define the partition function for SOS model with DWBC. 

Recently H.\,Rosengren~\cite{Ros08} has independently shown that this partition function can be written as a sum of $2^n$ determinants (where $n$ is a range of the lattice) which generalizes the Izergin's determinant formula. His approach relates to a dynamical generalization of Alternating-Sign Matrices.

\section{Projections for the Hopf algebras}
\label{sec41}

Here we introduce the notion of the projections for the Hopf algebras and discuss some their properties. Let us consider a Hopf algebra $\A$ and its two subalgebras~\footnote{In general, these are not Hopf subalgebras.} $\A^-$ and $\A^+$ satisfying the following conditions:
\begin{itemize}
\item[(i)] The algebra $\A$ admits a factorisation $\A=\A^-\cdot\A^+$, such that the corresponding restriction of the multiplication map
\begin{align} \label{hsac1}
\mu\colon\A^-\otimes \A^+\to \A 
\end{align}
is an isomorphism of linear spaces;
\item[(ii)] the subalgebra $\A^-$ is a left coideal:
\begin{align} \label{hsac2}
\Delta(\A^-)\subset \A\otimes \A^-; 
\end{align}
\item[(iii)] the subalgebra $\A^+$ is a right coideal:
\begin{align} \label{hsac3}
 \Delta(\A^+)\subset \A^+\otimes \A. 
\end{align}
\end{itemize}

Let us define the linear operators $P^-\colon\A\to\A$ and $P^+\colon\A\to\A$ by the formulae
\begin{align} \label{defprojHA}
 P^-(ab)&=a\varepsilon(b), &
 P^+(ab)&=\varepsilon(a)b, & a&\in \A^-,\; b\in \A^+.
\end{align}
Due to the condition~(i) these formulae define these operators on the whole algebra $\A$. The operators $P^-$ and $P^+$ are called {\itshape projections}: they are idempotents:
\begin{align}
 P^-\circ P^-&=P^-, & P^+\circ P^+&=P^+.
\end{align}

One also say that $P^-$ and $P^+$ are projections onto the subalgebras $\A^-$ and $\A^+$ parallel to the $\A^+$ and $\A^-$, what is presented in the equalities
\begin{align}
 P^-(a^-)&=a^-, & P^+(a^+)&=a^+, \label{projpropPaa} \\
 P^-(a^+)&=\varepsilon(a^+), & P^+(a^-)&=\varepsilon(a^-),
\end{align}
where $a^-\in\A^-$ and $a^+\in\A^+$. The definition~\eqref{defprojHA} implies more general properties:
\begin{align}
 P^-(a^-b)&=a^-P^-(b), & P^+(ba^+)&=P^+(b)a^+, \label{projpopPP} \\
 P^-(ba^+)&=P^-(b)\varepsilon(a^+), & P^+(a^-b)&=\varepsilon(a^-)P^+(b), \label{projpopPPe}
\end{align}
for all $a^-\in\A^-$, $a^+\in\A^+$ and $b\in\A$. They can be used to prove the following proposition.

\begin{prop} \label{prop_proj}
The projections $P^-$ and $P^+$ onto subalgebras $\A^-$ and $\A^+$ subjected to the conditions~\eqref{hsac1}, \eqref{hsac2}, \eqref{hsac3} satisfy the operator equality
\begin{align}
\mu\circ(P^-\otimes P^+)\circ\Delta=\id_{\A}.
\label{addition}
\end{align}
\end{prop}

\noindent{\bf Proof.} Due to the condition~\eqref{hsac1} it is sufficient to check the equality on the products $ab$ of the elements $a\in\A^-$ and $b\in\A^+$. Let the coproducts of these elements look as
\begin{align}
 \Delta(a)&=\sum_i a'_i\otimes a''_i, &
 \Delta(b)&=\sum_i b'_i\otimes b''_i.
\end{align}
Then due to the conditions~\eqref{hsac2} and \eqref{hsac3} we have $a''_i\in\A^-$, $b'_i\in\A^+$. Further, taking into account the properties~\eqref{projpopPP}, \eqref{projpopPPe} and the definition of the counity we obtain
\begin{multline}
 \mu\circ(P^-\otimes P^+)\circ\Delta(ab)=\sum_{i,j}P^-(a'_ib'_j)P^+(a''_ib''_j)= \\ =\sum_{i,j}P^-(a'_i)\varepsilon(b'_j)\varepsilon(a''_i)P^+(b''_j)=P^-(a)P^+(b).
\end{multline}
The properties~\eqref{projpropPaa} means in turn that the last expression is equal to $ab$. \qed

The condition~\eqref{hsac1} means that each element $a\in\A$ can be decomposed as
\begin{align} \label{decproja}
 &a=\sum_i a^-_i a^+_i, & a^-_i\in\A^-,\; a^+_i\in\A^+.
\end{align}
The proposition~\ref{prop_proj} allows us to write the explicit formula for this decomposition for each element $a\in\A$:
\begin{align}
 a&=\sum_i P^-(a'_i)P^+(a''_i), &
 &\text{where} & &\Delta(a)=\sum_i a'_i\otimes a''_i,
\end{align}
because of $P^\pm(\A)=\A^\pm$. Inversely, as soon as we have the decomposition~\eqref{decproja} for some element $a\in\A$ we can calculate its projections substituting this decomposition to the definition~\eqref{defprojHA}. The main idea for calculation of the projections proposed by Khoroshkin and Pakuliak~\cite{Kh} is to find the terms of the decomposition~\eqref{decproja} not annihilated by the corresponding projection. This idea was realized to the full in calculations of the projections for the quantum affine algebras~\cite{KhPT}, \cite{KP-GLN}.

In the classical case the projections is defined by the decomposition of a Lie algebra as a vector space: $\mathfrak{g}=\mathfrak{g}^-\oplus\mathfrak{g}^+$, where $\mathfrak{g}^-$ and $\mathfrak{g}^+$ are Lie subalgebras. This gives the factorisation $U(\mathfrak{g})=U(\mathfrak{g}^-)\cdot U(\mathfrak{g}^+)$, where $U(\mathfrak{g}^-)$ and $U(\mathfrak{g}^+)$ are (cocommutative) Hopf subalgebras of $U(\mathfrak{g})$ (each Hopf subalgebra is a left and right coideal) satisfying~\eqref{hsac1}. It allows us to introduce the corresponding projections. So, the notion of the projections for the Hopf algebras generalizes the notion of the orthogonal projections for the vector spaces applied to the Lie algebras $\mathfrak{g}$ decomposed as $\mathfrak{g}=\mathfrak{g}^-\oplus\mathfrak{g}^+$.

Let us discuss a quantum example. Consider the Hopf algebra $U_q(\widehat\slt)$ -- the quantum universal enveloping algebra of the affine Lie algebra $\widehat\slt$. Let the Hopf subalgebra $U_F^{cop}$ defined in the section~\ref{sec32} play the role of the Hopf algebra $\A$. This is the algebra $U_F$ described by the currents $\psi^+(u)$ and $f(u)$ and endowed with the comultiplication $\Delta^{(D),op}$, which is opposite to the Drinfeld comultiplication:
\begin{align} \label{op1co}
 &\Delta^{(D),op}\; \psi^+(u)=\psi^+(u)\otimes\psi^+(u), \\
 &\Delta^{(D),op}\; f(u)=f(u)\otimes 1+ \psi^+(u)\otimes f(u).
\end{align}
Consider the subalgebras $U^-_f=U_F\cap U^-$ and $U^+_F=U_F\cap U^+$ of the algebra $U_F$, where $U^-$ and $U^+$ are Hopf algebras of $U_q(\widehat\slt)$ also defined in the section~\eqref{sec32}. The algebra $U^-_f$ described by the current $f^-(u)$ while $U^+_F$ -- by the currents $\psi^+(u)$ and $f^+(u)$. The construction of the Drinfeld double implies that the multiplication gives the isomorphism between the spaces $U=U_q(\widehat\slt)$ and $U^-\otimes U^+$, what in turn proves the condition~\eqref{hsac1} for the algebra $\A=U_F$ and the subalgebras $\A^-=U^-_f$ and $\A^+=U^+_F$. The last two subalgebras are not Hopf subalgebras, but one can check that they are left coideal and right coideal respectively. That is the conditions~\eqref{hsac2} and \eqref{hsac3} for this subalgebras are also fulfilled.

Let $P^-$ and $P^+$ be projections onto $U^-_f$ and $U^+_F$ parallel to $U^+_F$ and $U^-_f$ defined by the formulae~\eqref{defprojHA}. They act on the currents as
\begin{align}
 &P^-\big(\psi^+(u)\big)=1, & &P^+\big(\psi^+(u)\big)=\psi^+(u), \\
 &P^-\big(f(u)\big)=-f^-(u), & &P^+\big(f(u)\big)=f^+(u), \\
 &P^-\big(f^+(u)\big)=0,     & &P^+\big(f^+(u)\big)=f^+(u), \\
 &P^-\big(f^-(u)\big)=f^-(u), & &P^+\big(f^-(u)\big)=0.
\end{align}
In order to find the projections of an arbitrary element of $U_F$ it is sufficient to obtain the projections of products of currents
\begin{align}
 &P^+\big(f(u_1)\cdots f(u_n)\big), & &P^-\big(f(u_1)\cdots f(u_n)\big) \label{PPCs}
\end{align}
for arbitrary $n\in\mathbb Z_{\ge0}$, since each element of $U_F$ can be represent as a sum of elements of the form
\begin{align}
 \big\la f(u_1)\cdots f(u_n),s_1(u_1)\cdots s_n(u_n)\big\ra_{u_1,\ldots,u_n}\cdot t^+
\end{align}
where $n\in\mathbb Z_{\ge0}$, $s_1(u),\ldots,s_n(u)\in\mathbb C[u^{-1},u]$, $t^+\in H^+$ and $n\in\mathbb Z_{\ge0}$. Here $H^+$ is the subalgebra generated by $\hat h^+[s]$, $s(u)\in\mathbb C[u^{-1},u]$. It turns out that the objects~\eqref{PPCs} play important role in the theory of the integrable systems as well as in the Conformal Field Theory (CFT). In the $U_q(\hat{\mathfrak a})$ case it has the form
\begin{align}
 &P^\pm\big(f_1(u_{11})\cdots f_1(u_{1n_1})\cdots\cdots f_r(u_{r1})\cdots f_r(u_{rn_r})\big), &  \label{PPCsr}
\end{align}
where $f_i(u)$ is a current corresponding to the $i$-th simple root of a semi-simple Lie algebra $\mathfrak a$. In the theory of the integrable systems the object~\eqref{PPCsr} is called universal algebraic (off-shell) Bethe ansatz, while from the point of view of the CFT it is called universal weight function. The importance of this object is explained by the `coproduct property' (see~\cite{EKhP}).

In~\cite{KP-GLN} the calculation of the projections~\eqref{PPCsr} is reduced to the calculation of the projection of the product of the currents corresponding to the same root, id est to the calculations of the projections~\eqref{PPCs}. The last ones are calculated in~\cite{KhP} in terms of the corresponding half-currents $f^+(u)$ and $f^-(u)$ using interpolation formulae and, then, represented as integrals of the products of the total currents using the formulae~\eqref{hc_tc}. For the positive projection this looks as the integral transformation
\begin{align} \label{PpK}
P^+\big(f(u_n)\cdots f(u_1)\big)=
\oint\limits_{|u_i|>|v_j|}\frac{dv_1\cdots dv_n}{(2\pi iv_1)\cdots(2\pi iv_n)}K(u;v)
f(v_n)\cdots f(v_1).
\end{align}
with the integral kernel
\begin{align} \label{PpK_ker}
 K(u;v)=\prod_{m=1}^n v_m\;\prod_{k>j}\frac{u_k-u_j}{qu_k-q^{-1}u_j}\prod_{k=1}^n \frac{\displaystyle \prod_{i=k+1}^n(qu_i-q^{-1}v_k)}
{\displaystyle \prod_{i=k}^n(u_i-v_k)},
\end{align}
$u=(u_1,\ldots,u_n)$ and $v=(v_1,\ldots,v_n)$. By virtue of the commutation relation~\eqref{cr_Uqff} the integral kernel $K(u;v)$ can be replaced by the integral kernel
\begin{align}
 \frac{qv_{i+1}-q^{-1}v_i}{q^{-1}v_{i+1}-qv_i}K(u;v_1,\ldots,v_{i+1},v_i,\ldots,v_n).
\end{align}
So, the integral kernel in the formula~\eqref{PpK} can chosen in different ways.

One can check that the formula
\begin{align} \label{piqudef}
 \pi^q_u(\sigma_{i,i+1})g(u)=\frac{q^{-1}u_{i+1}-qu_i}{qu_{i+1}-q^{-1}u_i}g(u_1,\ldots,u_{i+1},u_i,\ldots,u_n)
\end{align}
define an action of the permutation group $S_n$ on the functions of $u_1,\ldots,u_n$, where $g(u)\in V[[u^{-1},u]]$, $V$ is a vector space over $\mathbb C$ and $\sigma_{i,i+1}$ is an elementary transposition. Therefore the integral kernel $K(u;v)$ in the formula~\eqref{PpK} can be also replaced by
\begin{align}
 \pi^{q^{-1}}_v(\sigma)K(u;v)=\prod_{i<j \atop \sigma(i)>\sigma(j)}\frac{qv_{\sigma(i)}-q^{-1}v_{\sigma(j)}}{q^{-1}v_{\sigma(i)}-qv_{\sigma(j)}}K(u;v^\sigma),
\end{align}
where $v^\sigma=(v_{\sigma(1)},\ldots,v_{\sigma(n)})$, or by $\frac1{n!}{\cal K}(u;v)$, where
\begin{align} \label{Kqsym}
 {\cal K}(u,v)=\sum_{\sigma\in S_n}\pi^{q^{-1}}_v(\sigma)K(u;v)=\sum_{\sigma\in S_n}\prod_{i<j \atop \sigma(i)>\sigma(j)}\frac{qv_{\sigma(i)}-q^{-1}v_{\sigma(j)}}{q^{-1}v_{\sigma(i)}-qv_{\sigma(j)}}K(u;v^\sigma).
\end{align}

We shall say that a function $g(u)$ is {\it $q$-symmetric with respect to the variables $u$} if it is invariant under the action~\eqref{piqudef}:
\begin{align}
 \pi^q_u(\sigma)g(u)&=g(u), & &\forall\sigma\in S_n.
\end{align}
The function~\eqref{Kqsym} is $q$-symmetric with respect to $u$ and $q^{-1}$-symmetric with respect to the variables $v$. Let us also note that if the function $g(u)$ is $q$-symmetric then the function
\begin{align}
 \prod_{k>j}\frac{qu_k-q^{-1}u_j}{u_k-u_j}\;g(u)
\end{align}
is symmetric.

\section[Partition function of the finite 6-vertex model with DWBC]{Partition function of the finite 6-vertex model with Domain Wall Boundary Conditions}
\label{sec42}

\begin{figure}[h]
\begin{center}
\begin{picture}(150,90)
\put(50,70){\line(1,0){80}}
\put(35,70){$n$}
\put(50,50){\line(1,0){80}}
\put(35,50){$\ldots$}
\put(50,30){\line(1,0){80}}
\put(35,30){$2$}
\put(50,10){\line(1,0){80}}
\put(35,10){$1$}
\put(60,0){\line(0,1){80}}
\put(58,-10){$n$}
\put(80,0){\line(0,1){80}}
\put(78,-10){$\ldots$}
\put(100,0){\line(0,1){80}}
\put(98,-10){$2$}
\put(120,0){\line(0,1){80}}
\put(118,-10){$1$}
\end{picture}%
\end{center}
\caption{\footnotesize  The square $n\times n$ lattice.}
\label{fig3}
\end{figure}
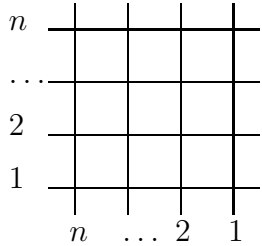

Let us consider a 6-vertex model defined in the section~\eqref{sec16} on the square lattice of the size $n\times n$, where the columns and rows are enumerated from $1$ to $n$ from the right to the left and upward respectively (see~Fig.~\ref{fig3}). Let us briefly recall it. Six possible configurations are shown in the Fig.~\ref{fig1} and the weights of vertex is defined by~\eqref{Bw}, \eqref{Bwc}. The sign $+$ is associated to the upward arrows and to the arrows directed to the left, while the sign $-$ -- to the downward arrows, to the arrows directed to the right like it shown in the fig.~\ref{fig1}. The Boltzmann weights~\eqref{Bw}, \eqref{Bwc} are gathered to the matrix
\begin{equation} \label{Bw1}
R(z,w)=\left(\begin{array}{cccc}a(z,w)&0&0&0\\ 0&b(z,w)& \bar c(z,w)&0\\
0&c(z,w)&b(z,w)&0\\0&0&0&a(z,w)
\end{array}\right).
\end{equation}
acting in the space $\mathbb C^2\otimes\mathbb C^2$ with the basis $e_{++}$, $e_{+-}$, $e_{-+}$, $e_{--}$. This is the trigonometric $R$-matrix~\eqref{R_trig}. The entry $R(z,w)^{\alpha\beta}_{\gamma\delta}$, $\alpha,\beta,\gamma,\delta=\pm$, coincides with the Boltzmann weight corresponding to the fig.~\ref{fig2} (the Boltzmann weight of the unavailable configuration vanishes). We consider an inhomogeneous model: the weight matrix for the $(i,j)$-th vertex is $R(u_i,v_j)$, where the variables $u_i$ associated to the $i$-th column and the variable $v_j$ to the $j$-th row (see fig.~\ref{fig4}).

\begin{figure}[b]
\begin{center}
\begin{picture}(160,95)
\put(50,70){\line(1,0){80}}
\put(50,70){\vector(1,0){3}}
\put(130,70){\vector(-1,0){3}}
\put(15,70){$v_4$} \put(35,70){$+$} \put(135,70){$-$}
\put(50,50){\line(1,0){80}}
\put(50,50){\vector(1,0){3}}
\put(130,50){\vector(-1,0){3}}
\put(15,50){$v_3$} \put(35,50){$+$} \put(135,50){$-$}
\put(50,30){\line(1,0){80}}
\put(50,30){\vector(1,0){3}}
\put(130,30){\vector(-1,0){3}}
\put(15,30){$v_2$} \put(35,30){$+$} \put(135,30){$-$}
\put(50,10){\line(1,0){80}}
\put(50,10){\vector(1,0){3}}
\put(130,10){\vector(-1,0){3}}
\put(15,10){$v_1$} \put(35,10){$+$} \put(135,10){$-$}

\put(60,0){\vector(0,1){80}}
\put(60,0){\vector(0,-1){3}}
\put(58,-10){$-$} \put(58,-20){$u_4$} \put(58,83){$+$}
\put(80,0){\vector(0,1){80}}
\put(80,0){\vector(0,-1){3}}
\put(78,-10){$-$} \put(78,-20){$u_3$} \put(78,83){$+$}
\put(100,0){\vector(0,1){80}}
\put(100,0){\vector(0,-1){3}}
\put(98,-10){$-$} \put(98,-20){$u_2$} \put(98,83){$+$}
\put(120,0){\vector(0,1){80}}
\put(120,0){\vector(0,-1){3}}
\put(118,-10){$-$} \put(118,-20){$u_1$} \put(118,83){$+$}
\end{picture}
\end{center}
\caption{\footnotesize Inhomogeneous lattice with Domain Wall Boundary Conditions.}
\label{fig4}
\end{figure}
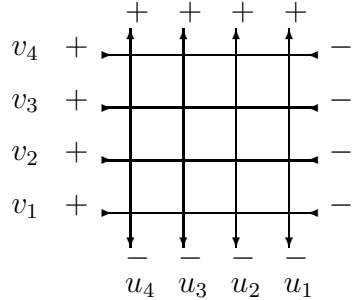

We choose so-called Domain Wall Boundary Conditions (DWBC) that fix the boundary arrows (signs) like it shown in the fig.~\ref{fig3}: the arrows are entering on the left and right boundaries and leaving on the lower and upper ones. In terms of the section~\ref{sec16} it means that the partition function with DWBC is defined by the formula~\eqref{ZsumRabgd} with $\alpha_1=\ldots=\alpha_n=+$, $\beta_1=\ldots=\beta_n=-$, $\gamma_1=\ldots=\gamma_n=-$ and $\delta_1=\ldots=\delta_n=+$. That is this partition function, denoted by $Z^{(n)}(u;v)$, has the form
\begin{align} \label{ZRRpmmp}
 Z^{(n)}(u;v)=\mathbb R(u;v)^{+,\ldots,+;-,\ldots,-}_{-,\ldots,-;+,\ldots,+},
\end{align}
where $\mathbb R(u;v)$ is the matrix acting in $V_1\otimes\ldots\otimes V_n\otimes V_{1'}\otimes\ldots\otimes V_{n'}$, $V_i\cong V_{j'}\cong\mathbb C^2$, which defined as
\begin{align} \label{Z6vDWBCdef}
 \mathbb R(u;v)=\mathop{\overrightarrow\prod}\limits_{1\le i\le n}\; \mathop{\overleftarrow\prod}\limits_{n\ge j\ge 1} R^{(i,j')}(u_i,v_j).
\end{align}

V.\,Korepin~\cite{Kor} prove the following properties of the function~\eqref{Z6vDWBCdef}:
\begin{itemize}
\item[({\it i})] $Z^{(n)}(u;v)$ is a polynomial of order $n-1$ in each variable $u_i$;
\item[({\it ii})] $Z^{(n)}(u;v)$ is symmetric with respect to the set of the variables $u$: $Z^{(n)}(u;v)=Z^{(n)}(u^\sigma;v)$, $\forall\sigma\in S_n$;
\item[({\it iii})] $Z^{(n)}(u;v)$ is a polynomial of order $n$ in each variable $v_j$;
\item[({\it iv})] $Z^{(n)}(u;v)$ is symmetric with respect to the set of the variables $v$: $Z^{(n)}(u;v)=Z^{(n)}(u;v^\sigma)$, $\forall\sigma\in S_n$;
\item[({\it v})] $Z^{(n)}(u;v)$ satisfies the recurrent relation
\begin{multline} \label{condvrecrel}
 Z^{(n)}(u_1,\ldots,u_{n-1},u_n=q^{-2}v_n;v_1,\ldots,v_n)= \\
=(1-q^{-2})v_n\,\prod_{k=1}^{n-1} (u_k-v_n)(q^{-1}v_n-qv_k)\; Z^{(n-1)}(u_1,\ldots,u_{n-1};v_1,\ldots,v_{n-1}),
\end{multline}
where $Z^{(n-1)}(u_1,\ldots,u_{n-1};v_1,\ldots,v_{n-1})$ is a partition function for the $(n-1)\times(n-1)$ lattice with DWBC;
\item[({\it vi})] the partition function for $1\times1$ lattice with DWBC is
\begin{align}
 Z^{(1)}(u_1;v_1)=(q-q^{-1})v_1.
\end{align}
\end{itemize}

Note that a polynomial of degree $n-1$ can be given by its values in $n$ different points. For a fixed polymonial $Z^{(n-1)}(u_1,\ldots,u_{n-1};v_1,\ldots,v_{n-1})$ the conditions~({\it iv}) and ({\it v}) fix the values of $Z^{(n)}(u_1,\ldots,u_n;v_1,\ldots,v_n)$ considered as a polynomial of $u_n$ in the points $u_n=q^{-2}v_j$, $j=1,\ldots,n$. Taking into account the `initial' condition~({\it vi}) we conclude that there are unique functions $Z^{(n)}(u_1,\ldots,u_n;v_1,\ldots,v_n)$ satisfying the conditions~({\it i}), ({\it iv}) -- ({\it vi}). It means that if some functions satisfy these conditions then they are partition functions for the 6-vertex model with DWBC. In the paper~\cite{I87}, A.\,G.\,Izergin use this idea obtaining the determinant representation of the partitions function with DWBC 
\begin{multline} \label{stat-s}
Z^{(n)}(u;v)= (q-q^{-1})^n \prod_{m=1}^n v_m \times \\
  \times\frac{\displaystyle\prod_{i,j=1}^n(u_i-v_j)(qu_i-q^{-1}v_j)}
{\displaystyle\prod_{n\ge i>j\ge1}(u_i-u_j)(v_j-v_i)}
\det\left|\left|\frac{1}{(u_i-v_j)(qu_i-q^{-1}v_j)}\right|\right|_{i,j=1,\ldots,n}\, .
\end{multline}
One can check that the functions~\eqref{stat-s} satisfy all the conditions~({\it i}) -- ({\it vi}).

In the other hand it was observed that the integral kernel~\eqref{Kqsym} is proportional to the partition function $Z^{(n)}(u;v)$~\cite{KhP}. Indeed, multiplying ${\cal K}(u;v)$ by $\prod_{k>j}\frac{qu_k-q^{-1}u_j}{u_k-u_j}$ and $\prod_{k>j}\frac{q^{-1}v_k-qv_j}{v_k-v_j}$ we obtain the function symmetric function with respect to the both sets of variables: $u$ and $v$. Multiplying result by $\prod_{i,j=1}^n(u_i-v_i)$ we annihilate the poles of this functions. The precise factor can be deduced by the `initial' condition~({\it iv}). One can check that the function
\begin{multline} \label{stat-s_pr}
 Z^{(n)}(u;v)=(q-q^{-1})^n \prod_{i,j=1}^n(u_i-v_i)\prod_{n\ge k>j\ge1}\frac{(qu_k-q^{-1}u_j)(q^{-1}v_k-qv_j)}{(u_k-u_j)(v_k-v_j)}\;{\cal K}(u,v)=\\
 =(q-q^{-1})^n \prod_{m=1}^n v_m \prod_{n\ge k>j\ge1}\frac{q^{-1}v_k-qv_j}{v_k-v_j}\times \\
 \times \sum_{\sigma\in S_n}\prod_{1\le i<j\le n \atop \sigma(i)>\sigma(j)}
  \frac{qv_{\sigma(i)}-q^{-1}v_{\sigma(j)}}{q^{-1}v_{\sigma(i)}-qv_{\sigma(j)}}
  \prod_{n\ge i>k\ge1}(qu_i-q^{-1}v_{\sigma_k})\prod_{1\le i<k\le n} (u_i-v_{\sigma_k}) 
\end{multline}
satisfy the conditions~({\it i}), ({\it iv}) -- ({\it vi}). Thus the theory of the projections for the Hopf algebra  $U_q(\widehat\slt)$ gives another expression for the partition function. Further we explain the fact that the purely mathematical object turned out to be a partition function for 6-vertex model with DWBC giving an alternative proof of this fact using only algebraic approach~\footnote{This proof has not been published before}.

The main relation of the 6-vertex model with the algebra $U_q(\widehat\slt)$ is the $R$-matrix~\eqref{Bw1}: in the section~\eqref{sec32} we mentioned that this algebra can be described by $RLL$ relations with this $R$-matrix. Let us also recall that the subalgebra $U_F$ generated by the currents $f(z)$ and $K^+(z)$ with the coproduct $\Delta^{op}$ is dual to the subalgebra $U_E$ generated by the currents $e(z)$ and $K^-(z)$ with the coproduct $\Delta$ with respect to the pairing $\la\cdot,\cdot\ra\colon U_F\otimes U_E\to\mathbb C$ defined by
\begin{align}
 \La f(u),e(v)\Ra&=(q-q^{-1})^{-1}\delta(u/v), &
 \La K^+(u),K^-(v)\Ra&=\frac{q^{-1}u-qv}{qu-q^{-1}v}. \label{HPdef_FE}
\end{align}
Since the whole algebra $U=U_q(\widehat\slt)$ is constructed as a Drinfeld double of $U_F^{cop}$ this pairing can be thought as a map $\la\cdot\ra\colon U\to\mathbb C$:
\begin{align}
 \la a\cdot x\ra&=\la a,x\ra, & &\forall a\in U_F,\; x\in U_E.
\end{align}

Consider the Hopf subalgebra $(U^+)^{cop}$ described by the Lax operator $L_+(u)$ with the opposite comultiplication
\begin{align}
 \Delta^{op} L_+(u)^i_k=\sum_{j\in\{+,-\}} L_+(u)^i_j\otimes L_+(u)^j_k, \label{Delta_RLLp}
\end{align}
where $\Delta$ is a standard comultiplication, and the Hopf subalgebra $U^-$ described by the Lax operator $L_-(u)$ with the standard comultiplication
\begin{align}
 \Delta L_-(u)^i_k=\sum_{j\in\{+,-\}} L_-(u)^j_k\otimes L_-(u)^i_j. \label{Delta_RLLm}
\end{align}
The subalgebras $(U^+)^{cop}$ and $U^-$ are dual to each other with respect to the same paring
\begin{align}
 \la a,b\ra&=\la a\cdot b\ra, & &\text{where} & &a\in U^+, b\in U^-.
\end{align}
The formulae~\eqref{LXx}, \eqref{RXX}, the duality of the generators $\la x_i\cdot x^j\ra=\delta^j_i$, and the relation $R^{(21)}_-(v,u)=R_+(u,v)^{-1}$ imply the equality
\begin{align}
 \La L_+^{(1)}(u)\cdot L_-^{(2)}(v)\Ra=R_+(u,v)^{-1}. \label{HPLLRpi}
\end{align}
Taking into account the duality of the multiplication~\eqref{Delta_RLLp} with the multiplication in $U^-$ one yields
\begin{align}
 \La L_+^{(1)}(u)\cdot L_-^{(2)}(v)^{-1}\Ra=R_+(u,v). \label{HPLLRp}
\end{align}
The matrix $R_+(u,v)$ is proportional to the matrix of Boltzmann weights $R(u,v)$. The straightforward calculations lead to the following formula for the pairing in terms of this matrix:
\begin{align}
 \La L_+^{(1)}(u)\cdot L_-^{(2)}(v)^{-1}\Ra=\frac{\la k^+_2(u),k^-_1(v)^{-1}\ra}{u-v} R(u,v), \label{HPLL}
\end{align}
where $k^\pm_1(u)=k^\pm(q^{-1}u)$, $k^\pm_2(u)=k^\pm(qu)^{-1}$ and $k^\pm(u)$ is defined by~\eqref{kpmUqdef}.
Using the duality between $(U^+)^{cop}$ and $U^-$ and the formulae for the coproducts~\eqref{Delta_RLLp} and \eqref{Delta_RLLm} one can generalize the formula~\eqref{HPLL} up to~%
\footnote{There is an alternative derivation of this formula. The formula~\eqref{Ru_ee} for $\A=(U^+)^{cop}$, $\A^*=U^-$ and $\la e_i\cdot e^j\ra=\delta_i^j$ imply
\begin{align}
 \la \R^{(13)}\R^{(32)}\ra_3=\R^{(12)}, \label{R133212}
\end{align}
where the paring $\la\cdot\ra_3$ is taken over the $3^{\text{d}}$ tensor space. Let us consider the following representations (instead of the evaluation representations):
\begin{align}
\Pi^-_{u_1,\ldots,u_n} &=(\Pi^-_{u_1}\otimes\ldots\otimes\Pi^-_{u_n})\circ\Delta^n, &
\Pi^+_{v_1,\ldots,v_n} &=(\Pi^+_{v_1}\otimes\ldots\otimes\Pi^+_{v_n})
\circ\Delta^n,
\end{align}
where $\Delta^n$ is defined by the formulae
\begin{align}
\Delta^n&=\big(\Delta\otimes(\id)^{\otimes(n-1)}\big)\circ\Delta^{n-1}, & \Delta^1&=\Delta.
\end{align}
Applying the map $(\Pi^-_{u_1,\ldots,u_n}\otimes\Pi^+_{v_1,\ldots,v_n})$ to the formula~\eqref{R133212} we obtain
\begin{align}
 \La L_+^{(1)}(u_1)\cdots L_+^{(n)}(u_n),L_-^{(n')}(v_n)^{-1}\cdots L_-^{(1')}(v_1)^{-1}\Ra=\mathop{\overrightarrow\prod}\limits_{1\le i\le n}\; \mathop{\overleftarrow\prod}\limits_{n\ge j\ge 1} R^{(i,j')}_+(u_i,v_j). \label{HPLLLp}
\end{align}
Substituting the expression~\eqref{HPLL} for $R_+(u,v)$ one yields~\eqref{HPLLL}.
}
\begin{multline}
 \La L_+^{(1)}(u_1)\cdots L_+^{(n)}(u_n),L_-^{(n')}(v_n)^{-1}\cdots L_-^{(1')}(v_1)^{-1}\Ra= \\
    =\prod_{i,j=1}^n\frac{\la k^+_2(u_i),k^-_1(v_j)^{-1}\ra}{u_i-v_j}\mathop{\overrightarrow\prod}\limits_{1\le i\le n}\; \mathop{\overleftarrow\prod}\limits_{n\ge j\ge 1} R^{(i,j')}(u_i,v_j). \label{HPLLL}
\end{multline}
This is a crusial formula explaining the relation between the statistical models on the lattices and the quantum groups. 

Let $\bar P^-\colon U_E\to U_E$ is a linear operator defined by the formula
\begin{align}
 \bar P^-(ab)&=\varepsilon(a)b, & &\forall a\in U^+_e,\; b\in U^-_E,
\end{align}
where $U^+_e$ is subalgebra generated by $\hat e^+[s]$, $s(u)\in\mathbb C[u^{-1},u]$,and $U^-_E$ is a subalgebra generated by $\hat e^-[s]$ ans $\hat h^-[s]$, $s(u)\in\mathbb C[u^{-1},u]$. The operator $\bar P^-$ is a projection dual to the projection $P^+\colon U_F\to U_F$ with respect to the pairing~\eqref{HPdef_FE}:
\begin{align} \label{projduality}
 \la a,\bar P^-(x)\ra&=\la P^+(a),x\ra, & &\forall a\in U_F,\; x\in U_E.
\end{align}

Let us introduce the following notations for the following entries of the matrices $L_+(u)$ and $L_-(v)^{-1}$:
\begin{align}
B^+(u)&=L_+(u)^+_-=(q-q^{-1})f^+(u)k^+_2(u),  \\
\tilde C^-(v)&=\big(L_-(v)^{-1}\big)^-_+=-(q-q^{-1})e^-(v)k^-_1(v)^{-1}.
\end{align}
Taking the matrix entry $\Big(\cdot\Big)^{+\ldots+;-\ldots-}_{-\ldots-;+\ldots+}$ in the both hand sides of~\eqref{HPLLL} and taking into account the formulae~\eqref{ZRRpmmp} and \eqref{Z6vDWBCdef} we obtain
\begin{multline}
 \La B^+(u_1)\cdots B^+(u_n),\tilde C^-(v_n)\cdots \tilde C^-(v_1)\Ra= \\
    =\prod_{i,j=1}^n\frac{\la k^+_2(u_i),k^-_1(v_j)^{-1}\ra}{u_i-v_j}Z^{(n)}(u;v), \label{HPBC}
\end{multline}
The left hand side of~\eqref{HPBC} can be expressed in terms of the projection of total currents:
\begin{align}
 B^+(u_1)\cdots B^+(u_n)&= B^+(u_n)\cdots B^+(u_1)= \label{BB_CC} \\
  &=(q-q^{-1})^n\prod_{n\ge i>j\ge 1}\frac{qu_i-q^{-1}u_j}{u_i-u_j}P^+\big(f(u_n)\cdots f(u_1)\big)\prod_{i=1}^n k^+_2(u_i), \notag \\
 \tilde C^-(v_n)\cdots \tilde C^-(v_1)
  &=(q-q^{-1})^n\prod_{n\ge i>j\ge 1}\frac{q^{-1}v_i-qv_j}{v_i-v_j}\bar P^-\big(e(v_n)\cdots e(v_1)\big)\prod_{j=1}^n k^-_2(v_j)^{-1}, \notag
\end{align}
Substituting the formulae~\eqref{BB_CC} to~\eqref{HPBC} and taking into account the duality~\eqref{projduality} and the formulae
\begin{align*}
 \La f(u_n)\cdots f(u_1) a^+_K,e(v_n)\cdots e(v_1)a^-_K\Ra&=
 \La f(u_n)\cdots f(u_1),e(v_n)\cdots e(v_1)\Ra \la a^+_K,a^-_K\ra, \\
 \La\prod_{i=1}^n k^+_2(u_i),\prod_{j=1}^n k^-_1(v_j)\Ra&=\prod_{i,j=1}^n\La k^+_2(u_i),k^-_1(v_j)\Ra,
\end{align*}
(where $a^\pm_K\in H^\pm$ and $H^-$ is the subalgebra generated by $\hat h^-[s]$, $s(u)\in\mathbb C[u^{-1},u]$), we derive the following formula for the partition function
\begin{multline}
 Z^{(n)}(u;v)=(q-q^{-1})^{2n}\prod_{i,j=1}^n(u_i-v_j) \times \\
  \times\prod_{n\ge i>j\ge 1}\frac{(qu_i-q^{-1}u_j)(q^{-1}v_i-qv_j)}{(u_i-u_j)(v_i-v_j)}
   \La P^+\big(f(u_n)\cdots f(u_1)\big),e(v_n)\cdots e(v_1)\Ra. \label{Z6ffee}
\end{multline}
Substituting the integral formula for the projection~\eqref{PpK} to~\eqref{Z6ffee} and calculating the Hopf pairing by the formula
\begin{multline}
 \La f(t_n)\cdots f(t_1),e(v_n)\cdots e(v_1)\Ra=(q-q^{-1})^{-n}\sum_{\sigma\in S_n}\pi^{q^{-1}}_v(\sigma)\prod_{m=1}^n\delta(t_m/v_m)= \\
  =(q-q^{-1})^{-n}\sum_{\sigma\in S_n}\prod_{\substack{l<l' \\ \sigma(l)>\sigma(l')}}
   \frac{qv_{\sigma(l)}-q^{-1}v_{\sigma(l')}}{q^{-1}v_{\sigma(l)}-qv_{\sigma(l')}}\prod_{m=1}^n\delta(t_m/v_{\sigma(m)}),
    \label{la_ff_ee_ra_FE_}
\end{multline}
taking into account the formulae~\eqref{Kqsym} and \eqref{PpK_ker}
we arrive to the formula~\eqref{stat-s_pr}.

\section[Algebra $E_{\tau,\hbar}(\slt)$ with `Drinfeld' comultiplication]{Algebra $E_{\tau,\hbar}(\slt)$ with `Drinfeld' comultiplication, Hopf paring of currents}
\label{sec43}

In this section we consider the currents of the algebra $E_{\tau,\hbar}(\slt)$ and calculate the Hopf paring between its two dual subalgebras in terms of currents. For these purposes we introduce an analogue of the $q$-symmetrization for the elliptic case and calculate the coproduct of the product of currents.

Let $\A$ be an algebra generated by $\hat h^+[s]$, $\hat h^-[s]$, $\hat e[s]$, $\hat f[s]$, $s\in\lfK_0$, where
\begin{align}
 \hat h^+,\hat h^-,\hat e,\hat f\colon\lfK_0\to\A
\end{align}
are the operators corresponding to the currents $h^+(u)$, $h^-(u)$, $e(u)$, $f(u)$ (defined in the sections~\ref{sec32}, \ref{sec34}) subjected to the commutation relations~\cite{EF}:
\begin{gather}
[K^\pm(u),K^\pm(v)]=0, \qquad [K^+(u),K^-(v)]=0, \\
K^\pm(u)e(v)K^\pm(u)^{-1}=\frac{\theta(u-v+\hbar)}{\theta(u-v-\hbar)}e(v), \label{cr_EFR0Ke} \\
K^\pm(u)f(v)K^\pm(u)^{-1}=\frac{\theta(u-v-\hbar)}{\theta(u-v+\hbar)}f(v), \label{cr_EFR0Kf}\\
\theta(u-v-\hbar)e(u)e(v)=\theta(u-v+\hbar)e(v)e(u), \label{cr_EFR0ee} \\
\theta(u-v+\hbar)f(u)f(v)=\theta(u-v-\hbar)f(v)f(u), \label{cr_EFR0ff} \\
[e(u),f(v)]=\frac1\hbar\delta(u-v)\Big(K^+(u)-K^-(v)\Big), \label{cr_EFR0ef}
\end{gather}
where
$K^+(u)=\exp\Big(\dfrac{e^{\hbar\partial_u}-e^{-\hbar\partial_u}}{2\partial_u}
h^+(u)\Big)$, $K^-(u)=\exp\big(\hbar h^-(u)\big)$. Here and further the meromorphic functions of $u$ and $v$ is understood as decomposed into formal series by the corresponding way. The algebra $\A$ is the algebra~$E_{\tau,\hbar}(\slt)$ on the zero level of central charge, which was mentioned in the section~\ref{sec33}. This algebra is equipped with the following comultiplication and counity:
\begin{align}
 \Delta K^\pm(u)&=K^\pm(u)\otimes K^\pm(u), \label{cpKdef_} \\
 \Delta e(u)&=e(u)\otimes1+K^-(u)\otimes e(u), \label{cpedef} \\
 \Delta f(u)&=f(u)\otimes K^+(u)+1\otimes f(u), \label{cpfdef} \\
 \varepsilon(K^\pm(u))&=1, \qquad \varepsilon(e(u))=0, \qquad \varepsilon(f(u))=0.
\end{align}
This comultiplication is an analogue of the Drinfeld comultiplication~\eqref{cpKdef}, \eqref{cpefdef} of the algebra $U_q(\widehat\slt)$.

Let $\A_F$ and $\A_E$ be subalgebras of $\A$ generated by the generators $\hat h^+[s]$, $\hat f[s]$, and $\hat h^-[s]$, $\hat e[s]$, $s\in\lfK_0$, respectively. The subalgebra $\A_F$ is described by the currents $K^+(u)$, $f(u)$, and the subalgebra $\A_E$ --  by $K^-(u)$, $e(u)$. We also introduce the notations $H^+$ and $H^-$ for the subalgebras of $\A$ generated by $\hat h^+[s]$ and $\hat h^-[s]$ respectively. They are described by the currents $K^+(u)$ and $K^-(u)$. As was stated in~\cite{EF} the bialgebras $(\A_F)^{cop}$ and $\A_E$ are dual to each other with respect to the pairing
\begin{align}
 \La f(u),e(v)\Ra_{FE}&=\hbar^{-1}\delta(u,v), &
 \La K^+(u),K^-(v)\Ra_{FE}&=\frac{\theta(u-v-\hbar)}{\theta(u-v+\hbar)}.
\end{align}
The pairing between the other elements is defined by the formulae
\begin{align}
 \La f(u),K^-(v)\Ra_{FE}&=\La K^+(u),e(v)\Ra_{FE}=0 \label{HPfKmeKp}
\end{align}
and the property of duality.

The algebras $\A_F$ and $\A_E$ is $\mathbb Z$-graduated with respect to the number of generators $\hat f[s]$ and $\hat e[s]$ respectively:
\begin{align}
 \A_F&=\bigoplus\limits_{n\ge0}\A^n_F, & \A_E&=\bigoplus\limits_{n\ge0}\A^n_E.
\end{align}
The spaces $\A^n_F$ and $\A^n_E$ are spanned by the elements $\hat f[s_n]\cdots \hat f[s_1]t^+$ and $\hat e[s_n]\cdots \hat e[s_1]t^-$ respectively, where $s_i\in\lfK_0$, $t^+\in H^+$, $t^-\in H^-$. These gradings are dual with respect to the paring $\La\cdot,\cdot\Ra_{FE}$ in the following sense:

\begin{prop} \label{propHP}
If $k\ne m$ then $\La X_k,Y_m\Ra_{FE}=0$ for all $X_k\in\A^k_F$, $Y_m\in\A^m_E$. The paring between the elements of $\A^n_F$ and $\A^n_E$ is then non-degenerated and it can be written in terms of the currents as follows:
\begin{gather}
 \La f(u_n)\cdots f(u_1)t^+,e(v_n)\cdots e(v_1)t^-\Ra_{FE}= \notag \\
  =\hbar^{-n}\sum_{\sigma\in S_n}\prod_{\substack{l<l' \\ \sigma(l)>\sigma(l')}}
   \frac{\theta(v_{\sigma(l)}-v_{\sigma(l')}+\hbar)}{\theta(v_{\sigma(l)}-v_{\sigma(l')}-\hbar)}\prod_{m=1}^n\delta(u_m,v_{\sigma(m)})
   \La t^+,t^-\Ra_{FE}. \label{la_ff_ee_ra}
\end{gather}
\end{prop}

Before proving this proposition we rewrite the formula~\eqref{la_ff_ee_ra} in terms of an elliptic $\hbar$-symmetrization analogous to the $q$-symmetrization introduced in the section~\ref{sec41} and also calculate the coproduct of the product of currents in these terms.

Define first the elliptic $\hbar$-action of the permutation group $S_n$:
\begin{align}
 \pi^{\hbar}_u(\sigma)g(u)=\prod_{\substack{l<l' \\ \sigma(l)>\sigma(l')}}
   \frac{\theta(u_{\sigma(l)}-u_{\sigma(l')}-\hbar)}{\theta(u_{\sigma(l)}-u_{\sigma(l')}+\hbar)}g(u^\sigma), \label{Elh_def}
\end{align}
The distribution $g(u)$ is called $\hbar$-symmetric if it is invariant under the elliptic $\hbar$-action:
\begin{align}
 \pi^{\hbar}_u(\sigma)g(u)&=g(u),  & &\forall\sigma\in S_n. \label{a_is_h-sym}
\end{align}
Let us remark that it is sufficient to check the condition~\eqref{a_is_h-sym} for the generators
$\sigma_{i,i+1}$, $i=1,\ldots,n-1$. The main example of $\hbar$-symmetric distribution is a product of total currents $f(u_n)\cdots f(u_1)$. Define the elliptic $\hbar$-symmetrizator as
\begin{align}
 \Sym^\hbar_u g(u)=\sum_{\sigma\in S_n}\pi^{\hbar}_u(\sigma)g(u)
 =\sum_{\sigma\in S_n}\prod_{\substack{l<l' \\ \sigma(l)>\sigma(l')}}
   \frac{\theta(u_{\sigma(l)}-u_{\sigma(l')}-\hbar)}{\theta(u_{\sigma(l)}-u_{\sigma(l')}+\hbar)}g(u^\sigma). \label{ElSymh_def}
\end{align}
The $\hbar$-symmetrized distribution $\Sym^\hbar_u a(u)$ is $\hbar$-symmetric:
\begin{align}
 \pi^\hbar_u(\sigma)\Sym^\hbar_{u_n,\ldots,u_0} a(u)=\Sym^\hbar_{u_n,\ldots,u_0} a(z).
\end{align}
If distribution $g(u)$ is $\hbar$-symmetric then $\Sym^\hbar_u g(u)=n! g(u)$. In particular, $\Sym^\hbar_u\Sym^\hbar_u=n!\Sym^\hbar_u$.
In these terms the formula~\eqref{la_ff_ee_ra} reads
\begin{multline} \label{la_ff_ee_ra_}
 \La f(u_n)\cdots f(u_1)t^+,e(v_n)\cdots e(v_1)t^+\Ra_{FE} =\hbar^{-n}\Sym^\hbar_u(\sigma)\prod_{m=1}^n\delta(u_m,v_m)\La t^+,t^-\Ra_{FE}= \\
 =\hbar^{-n}\Sym^{-\hbar}_v(\sigma)\prod_{m=1}^n\delta(u_m,v_m)\La t^+,t^-\Ra_{FE}.
\end{multline}

\begin{lem} Let $S_{s,n}$ and $S_{s}$ are subgroups of $S_n$ consisting of permutations of elements $\{s+1,\ldots,n\}$ and $\{1,\ldots,s\}$ respectively.
Consider one more subset
\begin{align}
S^{(s)}_n=\{\sigma\in S_n\mid \sigma(s+1)<\ldots<\sigma(n),\sigma(1)<\ldots<\sigma(s)\}.
\end{align}
Then $S_n=S^{(s)}_nS_{s,n}S_{s}=S^{(s)}_nS_{s}S_{s,n}$, moreover, each element $\sigma\in S_n$ has unique decomposition
\begin{align}
&\sigma=\sigma_3\sigma_2\sigma_1=\sigma_3\sigma_1\sigma_2, & &\sigma_3\in S^{(s)}_n, & &\sigma_2\in S_{s,n}, & &\sigma_1\in S_{s}. \label{sigma_dec}
\end{align}
\end{lem}

\noindent{\bfseries Proof.} Let us consider the sequence of numbers
$\sigma^{-1}(1),\ldots,\sigma^{-1}(n)$. Write out consecutively its elements belonging to the segment $[s+1;n]$ to the first line and its elements belonging to the segment $[1;s]$ to the second line:
\begin{align*}
 &s\le\sigma^{-1}(i_{s+1}),\ldots,\sigma^{-1}(i_n)\le n; & &i_{s+1}<\ldots<i_n; \\
 &1\le\sigma^{-1}(i_1),\ldots,\sigma^{-1}(i_{s})\le s; & &i_1<\ldots<i_{s}.
\end{align*}
Setting $\sigma^{-1}_1(k)=\sigma^{-1}(i_k)$ for $k=s+1,\ldots,n$, $\sigma^{-1}_2(k)=\sigma^{-1}(i_k)$ for $k=1,\ldots,s$ (that is $\sigma^{-1}_2\sigma^{-1}_1(k)=\sigma^{-1}(i_k)$ for $k=1,\ldots,n$) and $\sigma_3(k)=i_k$ for $k=1,\ldots,n$ one obtains the decomposition~\eqref{sigma_dec}.

Let $\sigma\in S_n$ have two decompositions $\sigma=\sigma_3\sigma_2\sigma_1=\bar\sigma_3\bar\sigma_2\bar\sigma_1$ of type~\eqref{sigma_dec}, then
$\bar\sigma_3=\sigma_3\sigma_4$, where $\sigma_4=(\sigma_2\bar\sigma^{-1}_2)(\sigma_1\bar\sigma^{-1}_1)$. Consider the sequence of numbers 
\begin{align}
 \sigma_3(\sigma_4(s+1)),\ldots,\sigma_3(\sigma_4(n)). \label{seqSigma34}
\end{align}
Since $\sigma_3\sigma_4=\bar\sigma_3\in S^{(s)}_n$ this sequence is increasing. Since the numbers $\sigma_4(k)=\sigma_2\bar\sigma^{-1}_2(k)$ belong to $[s+1;n]$ for $k=s+1,\ldots,n$, the sequence~\eqref{seqSigma34} can be obtained by a permutation of the sequence $\sigma_3(s+1),\ldots,\sigma_3(n)$, but the last one is also increasing because $\sigma_3\in S^{(s)}_n$ and, consequently it coincides with~\eqref{seqSigma34}. This implies $\sigma_2=\bar\sigma_2$. Analogously, considering the sequences $\sigma_3(\sigma_4(1)),\ldots,\sigma_3(\sigma_4(s))$ and $\sigma_3(1),\ldots,\sigma_3(s)$ one concludes $\sigma_1=\bar\sigma_1$, and, therefore $\sigma_3=\bar\sigma_3$. \qed

The meaning of this lemma in terms of the elliptic $\hbar$-action is reduced to the formula
\begin{gather}
 \Sym^\hbar_u=\Sym^\hbar_{(u_1,\ldots,u_n)}=\sum_{\sigma\in S^{(s)}_n}\pi^\hbar_u(\sigma)
  \Sym^\hbar_{(u_{s+1},\ldots,u_n)}\Sym^\hbar_{(u_1,\ldots,u_{s})}. \label{mlad}
\end{gather}
It can be used to calculate the coproduct of the product of the total currents.

Applying the formulae~\eqref{cpfdef}, $\Delta^{op}(xy)=\Delta^{op}(x)\Delta^{op}(y)$ and moving all $K^+(u)$ to the right using the relation~\eqref{cr_EFR0Kf} one yields
\begin{multline}
 \Delta^{op}\big(f(u_n)\cdots f(u_1)\big)
  =(K^+(u_n)\otimes f(u_n)+f(u_n)\otimes1)\cdots(K^+(u_1)\otimes f(u_1)+f(u_1)\otimes1)= \\
  =\sum_{s=0}^n\sum_{\sigma\in S^{(s)}_n}
   \prod_{\substack{s+1\le k'\le n \\ 1\le k\le s \\ \sigma(k)>\sigma(k')}}
      \frac{\theta(u_{\sigma(k)}-u_{\sigma(k')}-\hbar)}
           {\theta^+(u_{\sigma(k)}-u_{\sigma(k')}+\hbar)} \times \\
  \times f(u_{\sigma(n)})\cdots f(u_{\sigma(s)}) K^+(u_{\sigma(s)})\cdots K^+(u_{\sigma(1)})
    \otimes f(u_{\sigma(s)})\cdots f(u_{\sigma(1)})= \\
  =\sum_{s=0}^n\sum_{\sigma\in S^{(s)}_n}\pi^{\hbar}_u(\sigma)
   f(u_n)\cdots f(u_{s+1}) K^+(u_{s})\cdots K^+(u_1)
    \otimes f(u_{s})\cdots f(u_1). \label{copr_ff}
\end{multline}
Taking into account the $\hbar$-symmetry of the product of the total currents with respect to the subgroups $S_{s,n}$ and $S_{s}$ and using~\eqref{mlad} one can rewrite it in the form
\begin{multline} \label{Deltaopff}
 \Delta^{op}\big(f(u_n)\cdots f(u_1)\big)= \\
 =\Sym^\hbar_u\sum_{s=0}^{n}\frac1{s!(n-s)!}
   f(u_n)\cdots f(u_{s+1}) K^+(u_{s})\cdots K^+(u_1)\otimes f(u_{s})\cdots f(u_1).
\end{multline}

\noindent{\bf Proof of Proposition~\ref{propHP}.}
For $k>1$ and $t^+\in H^+$ one has the formula
\begin{align}
 \la f(u_{k-1})\cdots f(u_1)t^+,1\ra_{FE}=\varepsilon\big(f(u_{k-1})\cdots f(u_1)t^+\big)=0.
\end{align}
Using it, the duality $\la xy,a\ra_{FE}=\la x\otimes y,\Delta(a)\ra_{FE}$ and the formulae~\eqref{cpedef}, \eqref{HPfKmeKp} one yields
\begin{multline}
\La f(u_{k})f(u_{k-1})\cdots f(u_1)t^+,e(v)\Ra_{FE}= \\
 =\La f(u_{k})\otimes f(u_{k-1})\cdots f(u_1)t^+,e(v)\otimes1+K^-(v_1)\otimes e(v)\Ra_{FE}=0. \label{eq4328}
\end{multline}
Thus we have $\la\A^k_F,e(v)\ra_{FE}=0$, for $k>1$.

Let us prove the proposition~\ref{propHP} by induction. The formulae~\eqref{HPfKmeKp} implies that $\la\A^1_F,\A^0_E\ra_{FE}=\la\A^0_F,\A^1_E\ra_{FE}=0$. One can also check the formula~\eqref{la_ff_ee_ra} for $n=1$: using the formulae of the duality $\la x,ab\ra_{FE}=\la \Delta^{op}(x),a\otimes b\ra_{FE}$, $\la xy,a\ra_{FE}=\la x\otimes y,\Delta(a)\ra_{FE}$, $\la x,1\ra_{FE}=\epsilon(x)$ and the formula $(\id\otimes\varepsilon)\circ\Delta=\id$ one obtains
\begin{multline}
 \la f(u)t^+,e(v)t^-\ra_{FE}=\sum_i\La (K^+(u)t''_i\otimes f(u)t'_i+f(u)t''_i\otimes t'_i),e(v)\otimes t^-\Ra_{FE}= \\
=\sum_i\la f(u)t''_i,e(v)\ra_{FE}\la t'_i,t^-\ra_{FE}=\sum_i\La f(u)\otimes t''_i,\big(e(v)\otimes1+K^-(v)\otimes e(v)\big)\Ra_{FE}\la t'_i,t^-\ra_{FE}= \\
=\la f(u),e(v)\ra_{FE}\sum_i\varepsilon(t''_i)\la t'_i,t^-\ra_{FE}
=\hbar^{-1}\delta(u,v)\la t^+,t^-\ra_{FE},
\end{multline}
where $\Delta t^+=\sum\limits_i t'_i\otimes t''_i$. Suppose that if $k,m<n$ and $k\ne m$ then $\la\A^k_F,\A^m_E\ra_{EF}=0$ and that the paring between $\A^{n-1}_F$ and $\A^{n-1}_E$ is determined by formula
\begin{multline}
 \La f(u_{n-1})\cdots f(u_1)t^+,e(v_{n-1})\cdots e(v_1)t^-\Ra_{FE} = \\
 =\hbar^{-n+1}\Sym^\hbar_{(u_1,\ldots,u_{n-1})}\prod_{m=1}^{n-1}\delta(u_m-v_m)
 \la t^+,t^-\ra_{FE}, \label{eq4329}
\end{multline}
where $t^+\in H^+$, $t^-\in H^-$. Let $k,m\le n$. Using the duality $\la x,ab\ra_{FE}=\la \Delta^{op}(x),a\otimes b\ra_{FE}$ and applying the formula~\eqref{Deltaopff} in the following paring we see taking into account~\eqref{eq4328} that the only term $s=k-1$ does not vanish:
\begin{multline}
 \La f(u_k)\cdots f(u_1)t^+,e(v_m)\cdots e(v_1)t^-\Ra_{FE}= \\
 =\La\Delta^{op}\big(f(u_k)\cdots f(u_1)t^+\big),e(v_m)\otimes e(v_{m-1})\cdots e(v_1)t^-\Ra_{FE}= \label{eq4330} \\
 =\Sym^\hbar_{(u_1,\ldots,u_k)}\frac1{(k-1)!}\sum_i
 \La f(u_k)K^+(u_{k-1})\cdots K^+(u_1)t''_i,e(v_m)\Ra_{EF}\times \\
 \times \La f(u_{k-1})\cdots f(u_1)t'_i,e(v_{m-1})\cdots e(v_1)t^-\Ra_{FE},
\end{multline}
By induction we have $\la\A^{k-1}_F,\A^{m-1}_E\ra_{FE}=0$ for $k\ne m$ and hence the equation~\eqref{eq4330} implies $\la\A^{k}_F,\A^{m}_E\ra_{FE}=0$ for $m,k\le n$ such that $k\ne m$. Substituting $k=m=n$ to the formula~\eqref{eq4330}, using the formula $$\Sym^\hbar_{(u_1,\ldots,u_n)}\Sym^\hbar_{(u_1,\ldots,u_{n-1})}=(n-1)!\Sym^\hbar_{(u_1,\ldots,u_n)}$$ and~\eqref{eq4329} one derives
\begin{multline}
 \La f(u_k)\cdots f(u_1)t^+,e(v_m)\cdots e(v_1)t^-\Ra_{FE}= \\
 =\hbar^{-n}\Sym^\hbar_{(u_1,\ldots,u_n)}\frac1{(n-1)!}\delta(u_n-v_n) \Sym^\hbar_{(u_1,\ldots,u_{n-1})}\prod_{m=1}^{n-1}\delta(u_m-v_m)\sum_i\varepsilon(t''_i)\la t'_i,t^-\ra= \\
 =\hbar^{-n}\Sym^\hbar_u\prod_{m=1}^n\delta(u_m-v_m)\la t^+,t^-\ra.
\end{multline}
Thus we have proved the formula~\eqref{la_ff_ee_ra}, which is used in the section~\ref{sec45} to obtain the exprassion for the partition function of the SOS model with DWBC. \qed

\section{Elliptic projections}
\label{sec44}

Before introducing the projection for $E_{\tau,\hbar}(\slt)$ we define the analogues of the left and right coideals $\A^-$ and $\A^+$. We generalize the property~\eqref{addition} for this dynamical case proving the analogue of the coideal conditions~\eqref{hsac2}, \eqref{hsac3}.

Let $\hat e^+$, $\hat e^-$, $\hat f^+$, $\hat f^-$ be the operators corresponding to the currents
$e^+(u)$, $e^-(u)$, $f^+(u)$, $f^-(u)$ introducing in the section~\ref{sec34} by the formulae~\eqref{hc_tc} with the Green distributions~\eqref{DdGEp}, \eqref{DdGEm}:
\begin{align}
 \hat f^-_\lambda[s]&=\la f^-_\lambda(u),s(u)\ra_u, & \hat f^+_\lambda[s]&=\la f^+_\lambda(u),s(u)\ra_u, \\
  \hat e^-_\lambda[s]&=\la e^-_\lambda(u),s(u)\ra_u, & \hat e^+_\lambda[s]&=\la e^+_\lambda(u),s(u)\ra_u,
\end{align}
where $s\in\lfK_0$. Let $\A^{-,(n)}_{f,\lambda}$ be a subspace of $\A^{(n)}_F$ spanned by the elements of the form $\hat f^-_\lambda[s_1]\cdots\hat f^-_{\lambda-2(n-1)\hbar}[s_n]$, where $s_k\in\lfK_0$. Similarly we denote by $\A^{+,(n)}_{F,\lambda}$ the subspace of $\A^{(n)}_F$ spanned by the elements of the form $\hat f^+_{\lambda+2(n-1)\hbar}[s_1]\cdots\hat f^+_\lambda[s_n]t^+$, where $s_k\in\lfK_0$, $t^+\in H^+$; in fact, the subspace $\A^{+,(n)}_{F,\lambda}$ is spanned by $\hat f[s_1]\cdots\hat f[s_n]t^+$, where $s_k\in{\cal O}=\mathbb C[[z]]$, $t^+\in H^+$; so it does not depend on $\lambda$: $\A^{+,(n)}_{F,\lambda}=\A^{+,(n)}_F$.
The intersection of each two these subspaces is $\{0\}$. Considering the direct sums of these subspaces over all integer $n\ge0$ we obtain the following graduated spaces:
\begin{align}
 \A^{-}_{f,\lambda}&=\bigoplus_{n\in\mathbb Z_+}\A^{-,(n)}_{f,\lambda}
  =\Span\{\hat f^-_\lambda[s_1]\cdots\hat f^-_{\lambda-2(n-1)\hbar}[s_n]\mid n\in\mathbb Z_+, s_k\in\lfK_0\}, \\
 \A^{+}_F&=\bigoplus_{n\in\mathbb Z_+}\A^{+,(n)}_F
  =\Span\{\hat f[s_1]\cdots\hat f[s_n]t^+\mid n\in\mathbb Z_+, s_k\in{\cal O}, t^+\in H^+\}. 
\end{align}
Actually, $\A^{+}_F$ is a graduated subalgebra generated by $\hat f[s]$, $\hat h[s]$, $s\in{\cal O}$.
The subspases $\A^{+,(n)}_F$, $\A^{-,(n)}_{f,\lambda}$ and $\A^{(n)}_F$ we have introduced can be decomposed as follows
\begin{align}
 \A^{-,(n)}_{f,\lambda}&=\A^{-,(m)}_{f,\lambda}\A^{-,(n-m)}_{f,\lambda-2m\hbar}, &
 \A^{+,(n)}_F&=\A^{+,(m)}_F\A^{+,(n-m)}_F, &
 \A^{(n)}_F&=\A^{(m)}_F\A^{(n-m)}_F, \label{An_Am_Anm}
\end{align}
where we use the notation: for $X$ and $Y$ are subspaces of $\A$ we set $XY=\Span\{xy\mid x\in X,y\in Y\}$.

The elements of the spaces $\A^{(n)}_F$ and $\A^{(n)}_E$ can be also decomposed into the products of negative and positive parts. The existence of this decomposition are provided by the following statement~\cite{EF}:
\begin{prop} \label{prop_Apm}
 The multiplication map $\mu\colon\A_F\otimes\A_F\to\A_F$ establishes an isomorphism of linear spaces $\bigoplus\limits_{m=0}^n\A^{-,(m)}_{f,\lambda}\otimes\A^{+,(n-m)}_F\cong\A^{(n)}_F$ and therefore an isomorphism $\A^{-}_{f,\lambda}\otimes\A^{+}_F\cong\A_F$.
\end{prop}

Define the negative projections as a linear map $P^-_\lambda\colon\A_F\to\A_F$ acting on the elements $a\in\A_F$ which has the decomposition $a=a^-_\lambda a^+$ with $a^-_\lambda\in\A^-_{f,\lambda}$, $a^+\in\A^+_F$ by the rules
\begin{align}
 P^-_\lambda(a^-a^+)=a^-\varepsilon(a^+). \label{Pm_def}
\end{align}
The positive projections is defined as a linear map $P^+_\lambda\colon\A_F\to\A_F$ acting on the elements
$a^{(n)}\in\A^{(n)}_F$ which has decomposition $a^{(n)}=a^{-,(m)}_{\lambda+2(n-1)\hbar}a^{+,(n-m)}$ with $a^{-,(m)}_{\lambda+2(n-1)\hbar}\in\A^{-,(m)}_{f,\lambda+2(n-1)\hbar}$, $a^{+,(n-m)}\in\A^{+,(n-m)}_F$ by the rules
\begin{align}
 P^+_\lambda(a^{-,(m)}_{\lambda+2(n-1)\hbar}a^{+,(n-m)})=\varepsilon(a^{-,(m)}_{\lambda+2(n-1)\hbar})a^{+,(n-m)}. \label{Pp_def}
\end{align}
By virtue of the proposition~\ref{prop_Apm} these rules define the projections on whole $\A_F$. These projections are idempotents. Being restricted on the space of generators $\A^{(1)}_F$ they define a decomposition of this space into direct sum $\A^{(1)}_F=\A^{-,(1)}_{f,\lambda}\oplus\A^{+,(1)}_F$, which can be present in terms of currents as follows
\begin{gather}
 P^\pm_\lambda(f(z))=\pm f^\pm_\lambda(z)=\pm\la G^\pm_{-\lambda}(z-w)f(w)\ra_w 
. \label{Ppmf}
\end{gather}
Thus the projections relate total current with the half-currents. Actually the projections on all the homogeneous components can be calculated using formulae~\eqref{Ppmf} and
 
\begin{align}
 P^-_\lambda(at^+)&=P^-_\lambda(a)\varepsilon(t^+), & P^+_\lambda(at^+)&=P^+_\lambda(a)t^+\\
 P^-_\lambda(a b^+)&=P^-_\lambda(a)\epsilon(b^+), &
 P^+_\lambda(c^{-,(k)}_{\lambda+2(n-1)\hbar} a^{(n-k)}) &=\epsilon(c^{-,(k)}_{\lambda+2(n-1)\hbar})P^+_\lambda(a^{(n-k)}), \label{Ppm_epsilon} \\
 P^-_\lambda(b^{-,(k)}_\lambda a)&=b^{-,(k)}_\lambda P^-_{\lambda-2k\hbar}(a), &
 P^+_\lambda(a c^{+,(k)})&=P^+_{\lambda+2k\hbar}(a) c^{+,(k)}, \label{Ppm_bc}
\end{align}
where $a\in\A_F$, $t^+\in H^+$, $b^+\in\A^+_F$, $b^{-,(k)}_\lambda\in\A^{-,(k)}_{f,\lambda}$, $c^{-,(k)}_{\lambda+2(n-1)\hbar}\in\A^{-,(k)}_{f,\lambda+2(n-1)\hbar}$, $a^{(n-k)}\in\A^{(n-k)}_F$, $c^{+,(k)}\in\A^{+,(k)}_F$. These equalities are direct consequences of the definitions~\eqref{Pm_def}, \eqref{Pp_def}, proposition~\ref{prop_Apm} and formulae~\eqref{An_Am_Anm}.

\begin{prop} \label{prop_copr_f}
For all $a\in\A^{(n+1)}_F$ the following formula is valid
\begin{align}
 \mu\circ(P^-_{\lambda+n\hbar}\otimes P^+_{\lambda-n\hbar})\circ\Delta^{op}(a)=a. \label{copr_f}
\end{align}
\end{prop}

This proposition generalizes Proposition~\ref{prop_proj} for the elliptic projections~\footnote{The Proposition~\ref{prop_copr_f} has not been published before.}. To prove it we need the following lemma, which plays the role of the conditions~\eqref{hsac2} and \eqref{hsac3} for the dynamical case.

\begin{lem} \label{lem_coideals}
 The coproducts of any elements $a^{-,(n)}_\lambda\in\A^{-,(n)}_{f,\lambda}$ and $a^{+,(n)}\in\A^{+,(n)}_F$ have the forms
\begin{align}
 \Delta^{op}(a^{-,(n)}_\lambda)&=\sum_i b^{(m_i)}_i\otimes c^{-,(n-m_i)}_i(\lambda+\hbar h^{(1)})
  \equiv\sum_i\sum_{l\ge0} \frac{(\hbar h)^l}{l!}b^{(m_i)}_i\otimes\frac{\partial^l}{\partial\lambda^l}c^{-,(n-m_i)}_i(\lambda),  \label{Am_coideal} \\
 \Delta^{op}(a^{+,(n)})&=\sum_j b^{+,(k_j)}_j\otimes c^{(n-k_j)}_j, \label{Ap_coideal}
\end{align}
where $b^{(m_i)}_i\in\A^{(m_i)}_F$, $c^{-,(n-m_i)}_i(\lambda)\in\A^{-,(n-m_i)}_{f,\lambda}$, $b^{+,(k_j)}_j\in\A^+_F$, $c^{(n-k_j)}_j\in\A^{(n-k_j)}_F$.
In other words we have the relations
\begin{align}
 \Delta^{op}\A^{-,(n)}_{f,\lambda}&\subset\sum_{m=0}^n\A^{(m)}_F\otimes\A^{-,(n-m)}_{f,\lambda+\hbar h^{(1)}}, \label{Delta_Amn} \\
 \Delta^{op}\A^{+,(n)}_{f,\lambda}&\subset\sum_{k=0}^n\A^{+,(k)}_F\otimes\A^{(n-k)}_F. \label{Delta_Apn}
\end{align}
Therefore the subalgebra $\A^+_F$ is right coideal and subspace $\A^-_{f,\lambda}$ could be called a dynamical left coideal:
\begin{align}
 &\Delta^{op}\A^-_{f,\lambda}\subset\A_F\otimes\A^-_{f,\lambda+\hbar h^{(1)}}, &
 &\Delta^{op}\A^+_F\subset\A^+_F\otimes\A_F.
\end{align}
\end{lem}

\noindent{\bfseries Proof of the lemma~\ref{lem_coideals}.} (This proof is based on the proof of {\itshape properties of the coproducts} described in~\cite{EF}.) First we prove the formulae~\eqref{Am_coideal} on the half current $f^-_\lambda(u)$ thereby proving the case $n=1$. Substituting~\eqref{hc_tc} to the following coproduct, using the formula~\eqref{cpfdef} and substituting the decomposition
$f(u)=f^+_{\lambda+\hbar h^{(1)}}(u)-f^-_{\lambda+\hbar h^{(1)}}(u)$ we obtain
\begin{multline}
 \Delta^{op}f^-_\lambda(z)=f^-_\lambda(z)\otimes1
  -\La G^-_{-\lambda}(z,u) K^+(u)\otimes f^-_{\lambda+\hbar h^{(1)}}(u)\Ra_u+ \\
  +\La G^-_{-\lambda}(z,u) K^+(u)\otimes f^+_{\lambda+\hbar h^{(1)}}(u)\Ra_u. \label{cpfdefm_op}
\end{multline}
The first and second terms in the right hand side of~\eqref{cpfdefm_op} (their actions on $\lfK_0$) belong to $\A_F\otimes\A^-_{f,\lambda+\hbar h^{(1)}}$. Let us show that the third term vanishes by virtue of the following formula~\cite{EF}
\begin{align}
 K^+(u)=\Big(\frac{\theta(u+\hbar)}{\theta(u-\hbar)}\Big)^{h/2}\Big(t^+_0+\sum_{i\ge1}t^+_i\epsilon^{i;0}(u)\Big),
\end{align}
where $t^+_i\in H^+$ for $i\ge0$ and $\epsilon^{i;0}(u)=\frac1{i!}\big(\frac{\theta'(u)}{\theta(u)}\big)^{(i)}$ -- were defined in the section~\eqref{sec34}. The function $\big(\frac{\theta(u+\hbar)}{\theta(u-\hbar)}\big)^{h/2}$ belongs to $L_{\hbar h}$ and therefore it can be represent in the form $\la G^+_{\hbar h}(u,w),s(w)\ra_w$, for some $s\in\lfK_0$. Taking into account this fact and substituting $f^+_{\lambda+\hbar h^{(1)}}(u)\hm=\la G^+_{-\lambda-\hbar h^{(1)}}(u,v)f(v)\ra_v$ we can represent the third term in the right hand side of~\eqref{cpfdefm_op} in the form
\begin{align}
  \La G^+_{\hbar h^{(1)}}(u,w)G^+_{-\lambda-\hbar h^{(1)}}(u,v)G^-_{-\lambda}(z,u),
    s(w)\Big(t^+_0+\sum_{i\ge1}t^+_i\epsilon^{i;0}(u)\Big)\otimes f(v)\Ra_{u,v,w}. \label{cpfdefm_op3}
\end{align}
The integral kernel of the expression~\eqref{cpfdefm_op3} can be rewritten as follows
\begin{multline}
  G^+_{\hbar h^{(1)}}(u,w)  G^+_{-\lambda-\hbar h^{(1)}}(u,v)G^-_{-\lambda}(z,u)=
  G^-_{-\hbar h^{(1)}}(w,u) G^-_{\lambda+\hbar h^{(1)}}(v,u)G^-_{-\lambda}(z,u)= \\
  =G^-_{-\hbar h^{(1)}}(w,v) G^-_{\lambda}(v,u) G^-_{-\lambda}(z,u)
  +G^+_{\lambda+\hbar h^{(1)}}(v,w)G^-_\lambda(w,u) G^-_{-\lambda}(z,u)= \\
  =-G^-_{-\hbar h^{(1)}}(w,v) \Big(G^-_{-\lambda}(z,v)G(u,v)
  +G^+_{\lambda}(v,z)G(u,z)
  +\frac{\partial}{\partial\lambda}G^-_{-\lambda}(z,v)\Big) + \\
  -G^+_{\lambda+\hbar h^{(1)}}(v,w)\Big(G^-_{-\lambda}(z,w)G(u,w)
  +G^+_{\lambda}(w,z)G(u,z)
  +\frac{\partial}{\partial\lambda}G^-_{-\lambda}(z,w)\Big), \label{lem_pr_ml}
\end{multline}
where we used the identity
\begin{align}
  G^-_{\lambda}(v,u)G^-_{-\lambda}(z,u)&=-G^-_{-\lambda}(z,v)G(u,v)
  -G^+_{\lambda}(v,z)G(u,z)
  -\frac{\partial}{\partial\lambda}G^-_{-\lambda}(z,v). \label{GG_deg}
\end{align}
The vanishing of the term containing the sum over $i\ge1$ in the expression~\eqref{cpfdefm_op3} follows from the fact that $\la G(u,v),\epsilon^{i;0}(u)\ra_u=0$ and $\la1,\epsilon^{i;0}(u)\ra_u=0$ for $i\ge1$. Considering the second line of~\eqref{lem_pr_ml} and taking into account $\la G^-_{\lambda}(v,u) G^-_{-\lambda}(z,u)\ra_u=0$, we conclude that the term containing $t^+_0$ in the expression~\eqref{cpfdefm_op3} also vanishes.

Thus we have relation $\Delta^{op}\A^{-,(1)}_{f,\lambda}\subset\A^{(1)}_F\otimes1+H^+\otimes\A^{-,(1)}_{f,\lambda+\hbar h^{(1)}}$. Using it one can prove the relation~\eqref{Delta_Amn} by induction. Indeed, let~\eqref{Delta_Amn} is true for some $n\ge0$, then using~\eqref{An_Am_Anm} and the commutation relation $(h+2m\hbar)a^{(m)}=a^{(m)}h$, where $a^{(m)}\in\A^{(m)}_F$, one concludes that it is true for $n+1$:
\begin{multline*}
 \Delta^{op}\A^{-,(n+1)}_{f,\lambda}\subset\sum_{m=0}^n\A^{(m)}_F\otimes\A^{-,(n-m)}_{f,\lambda+\hbar h^{(1)}}
 \big(\A^{(1)}_F\otimes1+H^+\otimes\A^{-,(1)}_{f,\lambda+\hbar h^{(1)}-2n\hbar}\big)\subset \\
 \subset\sum_{m=0}^n\A^{(m)}_F\A^{(1)}_F\otimes\A^{-,(n-m)}_{f,\lambda+\hbar h^{(1)}}
+\sum_{m=0}^n\A^{(m)}_F\otimes\A^{-,(n-m)}_{f,\lambda+\hbar h^{(1)}}
\A^{-,(1)}_{f,\lambda+\hbar h^{(1)}-2n\hbar+2m\hbar}\subset \\
 \subset\sum_{m=0}^n\A^{(m+1)}_F\otimes\A^{-,(n-m)}_{f,\lambda+\hbar h^{(1)}}
+\sum_{m=0}^n\A^{(m)}_F\otimes\A^{-,(n+1-m)}_{f,\lambda+\hbar h^{(1)}}
 \subset\sum_{m=0}^{n+1}\A^{(m)}_F\otimes\A^{-,(n+1-m)}_{f,\lambda+\hbar h^{(1)}}.
\end{multline*}

By the same reason the formula~\eqref{Delta_Apn} for general $n$ follows from~\eqref{Delta_Apn} for $n=1$. Since the coproduct $\Delta^{op}$ is a homomorphism it is sufficient to prove the formula~\eqref{Ap_coideal} on the half currents $K^+(u)$, $f^+_\lambda(u)$. For $K^+(u)$ it immediately follows from~\eqref{cpKdef_}. In the second case we have the relation $\Delta^{op}f^+_\lambda(u)\in\A^+_F\otimes1+H^+\otimes\A_F$ following from the formula
\begin{align}
 \Delta^{op}f^+_\lambda(u)&=f^+_\lambda(u)\otimes1+\La G^+_{-\lambda}(u,v) K^+(v)\otimes f(v)\Ra_v, \label{cpfdefp_op}
\end{align}
which is in turn obtained in the same way as the formula~\eqref{cpfdefm_op}. \qed

\noindent{\bfseries Proof of the proposition~\ref{prop_copr_f}.} It is sufficient to prove the formula~\eqref{copr_f} on the elements of the form $a^{-,(m)}_{\lambda+n\hbar}a^{+,(k)}$, with $a^{-,(m)}_{\lambda+n\hbar}\in\A^{-,(m)}_{f,\lambda}$, $a^{+,(k)}\in\A^{+,(k)}_F$, where $m+k=n+1$. Due to the lemma~\ref{lem_coideals} the coproduct of elements $a^{-,(m)}_{\lambda+n\hbar}$ and $a^{+,(k)}$ can be represented as~\eqref{Am_coideal} and \eqref{Ap_coideal}. Then, the coproduct of the considering element $a^{-,(m)}_{\lambda+n\hbar}a^{+,(k)}$ is equal to
\begin{align}
 \Delta^{op}(a^{-,(m)}_{\lambda+n\hbar}a^{+,(k)})
=\sum_{i,j}\sum_{l\ge0}
  \frac{(\hbar h)^l}{l!}b^{(m_i)}_i b^{+,(k_j)}_j \otimes\frac{\partial^l}{\partial\lambda^l}c^{-,(m-m_i)}_i(\lambda+n\hbar)c^{(k-k_j)}_j= \\
=\sum_{i,j}\sum_{l\ge0}
  b^{(m_i)}_i b^{+,(k_j)}_j \frac{(\hbar h-2m_i\hbar-2k_j\hbar)^l}{l!} \otimes\frac{\partial^l}{\partial\lambda^l}c^{-,(m-m_i)}_i(\lambda+n\hbar)c^{(k-k_j)}_j= \\
=\sum_{i,j}\sum_{l\ge0}
  b^{(m_i)}_i b^{+,(k_j)}_j \frac{(\hbar h)^l}{l!} \otimes\frac{\partial^l}{\partial\lambda^l}c^{-,(m-m_i)}_i(\lambda+n\hbar-2m_i\hbar-2k_j\hbar)c^{(k-k_j)}_j,
\end{align}
where $b^{(m_i)}_i\in\A^{(m_i)}_F$, $c^{-,(n-m_i)}_i(\lambda)\in\A^{-,(n-m_i)}_{f,\lambda}$, $b^{+,(k_j)}_j\in\A^{+,(k_j)}_F$, $c^{(n-k_j)}_j\in\A^{(n-k_j)}_F$. Using the formulae~\eqref{Ppm_epsilon} one yields
\begin{multline*}
 \mu\circ(P^-_{\lambda+n\hbar}\otimes P^+_{\lambda-n\hbar})\Delta^{op}(a^{-,(m)}_{\lambda+n\hbar}a^{+,(k)})= \\
=\sum_{i,j}\sum_{l\ge0}
  P^-_{\lambda+n\hbar}\Big(b^{(m_i)}_i b^{+,(k_j)}_j \frac{(\hbar h)^l}{l!}\Big) P^+_{\lambda-n\hbar}\Big(\frac{\partial^l}{\partial\lambda^l}c^{-,(m-m_i)}_i(\lambda+n\hbar-2m_i\hbar-2k_j\hbar)
  c^{(k-k_j)}_j\Big)= \\
=\sum_{i,j}P^-_{\lambda+n\hbar}\big(b^{(m_i)}_i\big) \varepsilon(b^{+,(k_j)}_j)
           P^+_{\lambda-n\hbar}\big(c^{-,(m-m_i)}_i(\lambda+n\hbar-2m_i\hbar-2k_j\hbar)c^{(k-k_j)}_j\big)= \\
=\sum_{i,j}P^-_{\lambda+n\hbar}\big(b^{(m_i)}_i\big) \varepsilon(b^{+,(k_j)}_j)
 \varepsilon(c^{-,(m-m_i)}_i(\lambda+n\hbar-2m_i\hbar-2k_j\hbar)) P^+_{\lambda-n\hbar}\big(c^{(k-k_j)}_j\big).
\end{multline*}
Since $\varepsilon(b^{+,(k_j)}_j)=0$ while $k_j\ne0$ we have only contribution of terms with $k_j=0$. Substituting $k_j=0$ to the argument of the function $c^{-,(m-m_i)}_i$ and taking into account the properties of the counity $(\varepsilon\otimes\id_{\A_F})\Delta^{op}=(\id_{\A_F}\otimes\varepsilon)\Delta^{op}=\id_{\A_F}$ we derive
\begin{multline*}
 \mu\circ(P^-_{\lambda+n\hbar}\otimes P^+_{\lambda-n\hbar})\Delta^{op}(a^{-,(m)}_{\lambda+n\hbar}a^{+,(k)})
=\sum_iP^-_{\lambda+n\hbar}\big(b^{(m_i)}_i\big)
 \varepsilon(c^{-,(m-m_i)}_i(\lambda+n\hbar-2m_i\hbar))\times \\
 \times P^+_{\lambda-n\hbar}\big(a^{+,(k)}\big)
=\sum_i\sum_{l\ge0}P^-_{\lambda+n\hbar}\Big(\frac{(\hbar h+2m_i\hbar)^l}{l!}b^{(m_i)}_i\Big)
 \varepsilon\Big(\frac{\partial^l}{\partial\lambda^l}c^{-,(m-m_i)}_i(\lambda+n\hbar-2m_i\hbar)\Big)\times \\
 \times P^+_{\lambda-n\hbar}\big(a^{+,(k)}\big)
 =P^-_{\lambda+n\hbar}(a^{-,(m)}_{\lambda+n\hbar})P^+_{\lambda-n\hbar}(a^{+,(k)})
 =a^{-,(m)}_{\lambda+n\hbar}a^{+,(k)}.
\end{multline*}
\qed

Consider the expression of the form
\begin{align}
 P^+_{\lambda-(n-1)\hbar}\big(f(u_n)f(u_{n-1})\cdots f(u_2)f(u_1)\big). \label{Pfn0}
\end{align}
This is an elliptic version of the function~\eqref{PPCs}. It is called {\itshape elliptic universal weight function}. The parameter $\lambda-(n-1)\hbar$ is chosen for the symmetry reason: as we shall see the distribution~\eqref{Pfn0} is represented as linear combinations of the terms
\begin{align}
 f^+_{\lambda-(n-1)\hbar}(u_{i_n})f^+_{\lambda-(n-3)\hbar}(u_{i_{n-1}})\cdots f^+_{\lambda+(n-3)\hbar}(u_{i_2})f^+_{\lambda+(n-1)\hbar}(u_{i_1}).
\end{align}

The projections~\eqref{Pfn0} can be calculated by the generalization of the method proposed  in~\cite{KhP} (see~\cite{S3}):
\begin{gather}
   P^+_{\lambda-(n-1)\hbar}\big(f(u_n)\cdots f(u_2)f(u_1)\big)
   =\prod_{n\ge m\ge1}^{\longleftarrow}f^+_{\lambda-(n-2m+1)\hbar}(u_m;u_n,\ldots,u_{m+1}). \label{Pfn0_prod_f}
\end{gather}
where
\begin{gather}
 f^+_{\lambda-(n-2m+1)\hbar}(u_m;u_n,\ldots,u_{m+1})= f^+_{\lambda-(n-2m+1)\hbar}(u_m)
 -\sum_{i=m+1}^n\frac{\theta(u_i-u_m+\lambda+(m-1)\hbar)}{\theta(\lambda+(m-1)\hbar)}\times \notag \\
 \times\prod_{k=m+1}^n \frac{\theta(u_k-u_i+\hbar)}{\theta(u_k-u_m+\hbar)}
  \prod\limits_{\substack{k=m+1 \\ k\ne i}}^n\frac{\theta(u_k-u_m)}{\theta(u_k-u_i)}
   f^+_{\lambda-(n-2m+1)\hbar}(u_i). \label{fpm_nm}
\end{gather}
Representing each current~\eqref{fpm_nm} in~\eqref{Pfn0_prod_f} as integral transforms of the total currents one yields
\begin{gather}
 P^+_{\lambda-(n-1)\hbar}\big(f(u_n)\cdots f(u_2)f(u_1)\big)
 =\prod_{n\ge k>m\ge1}\frac{\theta(u_k-u_m)}{\theta(u_k-u_m+\hbar)}
   \oint\limits_{|u_i|>|v_j|}\frac{dv_n\cdots dv_1}{(2\pi i)^n} \notag \\
     \prod_{n\ge k>m\ge1}\frac{\theta(u_k-v_m+\hbar)}{\theta(u_k-v_m)}
    \prod_{m=1}^n \frac{\theta(u_m-v_m-\lambda-(m-1)\hbar)}{\theta(u_m-v_m)\theta(-\lambda-(m-1)\hbar)}
           f(v_n)\cdots f(v_1). \label{Pfn0_prod_f_int}
\end{gather}

\section{Partition function for the SOS model}
\label{sec45}

We have seen in the section~\ref{sec42} that the connecting-link between the projection method and the statistical models is the $R$-matrix. As it was shown in~\cite{EF}, the algebra $\A$ is described by the dynamical $RLL$-relations with the Felder $R$-matrix, as we mentioned, the statistical model with the Felder $R$-matrix as a Boltzmann weight matrix is the SOS model. Therefore it is natural to expect that the integral kernel of the projections for the algebra $\A$ plays the same role for the SOS model as the kernel~\eqref{PpK_ker} plays for the 6-vertex model.

\begin{figure}[h]
\begin{center}
\begin{picture}(150,90)
\put(50,70){\line(1,0){80}}
\put(35,76){$n$}
\put(50,50){\line(1,0){80}}
\put(35,60){$\ldots$}
\put(50,30){\line(1,0){80}}
\put(5,38){$j=$}
\put(35,38){$2$}
\put(50,10){\line(1,0){80}}
\put(35,19){$1$}
\put(35,00){$0$}
\put(60,0){\line(0,1){80}}
\put(48,-15){$n$}
\put(80,0){\line(0,1){80}}
\put(68,-15){$\ldots$}
\put(100,0){\line(0,1){80}}
\put(88,-15){$2$}
\put(120,0){\line(0,1){80}}
\put(108,-15){$1$}
\put(123,-15){$0$}
\put(20,-15){$i=$}
\end{picture}
\end{center}
\caption{\footnotesize The numeration of faces.}
\label{fig5}
\end{figure}
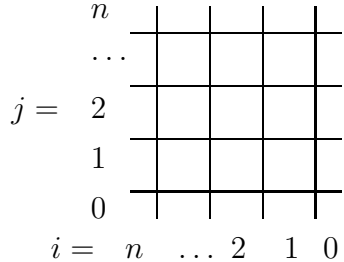

The SOS model is a face model mentioned in the section~\eqref{sec16}. In order to consider this model as a generalized 6-vertex model we present it in terms of $R$-matrix formalism as in~\cite{FS}. 
Consider a square $n\times n$ lattice with the vertices enumerated by index $i=1,\ldots,n$ like in Fig.~\ref{fig3}. It has $(n+1)\times(n+1)$ faces enumerated by pairs $(i,j)$, $i,j=0,\ldots,n$ (see Fig.~\ref{fig5}). The heights $d_{ij}$ putted to each face satisfy the conditions: $|d_{ij}-d_{i-1,j}|=1$ for $i=1,\ldots,n$, $j=0,\ldots,n$, and $|d_{ij}-d_{ij-1}|=1$ for $i=0,\ldots,n$, $j=1,\ldots,n$, where $d_{ij}$ is a height corresponding to the face $(i,j)$, which is placed to the up-left from the $(i,j)$-th vertex.

The Boltzmann weight of the $(i,j)$-th vertex $W_{ij}(d_{i,j-1},d_{i-1,j-1},d_{i-1,j},d_{ij})$ depends on the configuration via connected heights as follows~\cite{J}
\begin{align}
 W_{ij}(d+1,d+2,d+1,d)&=a(u_i-v_j)=\theta(u_i-v_j+\hbar), \label{BWa} \\
 W_{ij}(d-1,d-2,d-1,d)&=a(u_i-v_j)=\theta(u_i-v_j+\hbar), \\
 W_{ij}(d-1,d,d+1,d)&=b(u_i-v_j;\hbar d)=\frac{\theta(u_i-v_j)\theta(\hbar d+\hbar)}{\theta(\hbar d)}, \\
 W_{ij}(d+1,d,d-1,d)&=\bar b(u_i-v_j;\hbar d)=\frac{\theta(u_i-v_j)\theta(\hbar d-\hbar)}{\theta(\hbar d)}, \\
 W_{ij}(d-1,d,d-1,d)&=c(u_i-v_j;\hbar d)=\frac{\theta(u_i-v_j+\hbar d)\theta(\hbar)}{\theta(\hbar d)}, \\
 W_{ij}(d+1,d,d+1,d)&=\bar c(u_i-v_j;\hbar d)=\frac{\theta(u_i-v_j-\hbar d)\theta(\hbar)}{\theta(-\hbar d)},
  \label{BWc}
\end{align}
where $u_i$, $v_j$ are additive variables attached to the $i$-th vertical and $j$-th horizontal lines respectively, and $\hbar$ is a non-zero additive anisotropy parameter.
Each distribution of heights $d_{ij}$ ($i,j=0,\ldots,n$) subjected to boundary conditions defines a configuration of the model. The partition function of this model is the sum over these configurations:
\begin{align}
 Z=\sum\prod_{i,j=1}^n W_{ij}(d_{i,j-1},d_{i-1,j-1},d_{i-1,j},d_{ij}), \label{pf_def}
\end{align}

\begin{figure}
\begin{center}
 \begin{picture}(260,170)
\put(00,30){\line(1,0){60}}\put(30,00){\line(0,1){60}}
\put(-10,30){\bfseries $-$}\put(60,30){$-$}\put(30,-10){\bfseries $-$}\put(30,60){\bfseries $-$}
\put(0,10){\small $d-1$}\put(10,40){\small $d$}\put(35,40){\small $d-1$}\put(35,10){\small $d-2$}
\put(10,-20){\small  $a(u_i-v_j)$}
\put(00,130){\line(1,0){60}}\put(30,100){\line(0,1){60}}
\put(-10,130){\bfseries $+$}\put(60,130){\bfseries $+$}\put(30,90){\bfseries $+$}\put(30,160){\bfseries $+$}
\put(0,110){\small $d+1$}\put(10,140){\small $d$}\put(35,140){\small $d+1$}\put(35,110){\small $d+2$}
\put(10,77){\small  $a(u_i-v_j)$}
\put(100,30){\line(1,0){60}}\put(130,00){\line(0,1){60}}
\put(90,30){\bfseries $+$}\put(160,30){$+$}\put(130,-10){\bfseries $-$}\put(130,60){\bfseries $-$}
\put(100,10){\small $d+1$}\put(110,40){\small $d$}\put(135,40){\small $d-1$}\put(145,10){\small $d$}
\put(105,-20){\small  $\bar b(u_i-v_j;d)$}
\put(100,130){\line(1,0){60}}\put(130,100){\line(0,1){60}}
\put(90,130){\bfseries $-$}\put(160,130){\bfseries $-$}\put(130,90){\bfseries $+$}\put(130,160){\bfseries $+$}
\put(100,110){\small $d-1$}\put(110,140){\small $d$}\put(135,140){\small $d+1$}\put(145,110){\small $d$}
\put(105,77){\small  $b(u_i-v_j;d)$}
\put(200,30){\line(1,0){60}}\put(230,00){\line(0,1){60}}
\put(190,30){\bfseries $+$}\put(260,30){$-$}\put(230,-10){\bfseries $-$}\put(230,60){\bfseries $+$}
\put(200,10){\small $d+1$}\put(210,40){\small $d$}\put(235,40){\small $d+1$}\put(245,10){\small $d$}
\put(205,-20){\small  $\bar c(u_i-v_j;d)$}
\put(200,130){\line(1,0){60}}\put(230,100){\line(0,1){60}}
\put(190,130){\bfseries $-$}\put(260,130){\bfseries $+$}\put(230,90){\bfseries $+$}\put(230,160){\bfseries $-$}
\put(200,110){\small $d-1$}\put(210,140){\small $d$}\put(235,140){\small $d-1$}\put(245,110){\small $d$}
\put(205,77){\small  $c(u_i-v_j;d)$}
  \end{picture}
 \end{center}
\vspace{3mm}
\caption{\footnotesize The Boltzmann weights for the SOS model.}
\label{fig6}
\end{figure}

The Boltzmann weights~\eqref{BWa} -- \eqref{BWc} can be presented as the entries of the Felder $R$-matrix~\cite{F1,FS}:
\begin{align}
 &W_{ij}(d_{i,j-1},d_{i-1,j-1},d_{i-1,j},d_{ij})
   =R(u_i-v_j;\hbar d_{ij})^{\alpha_{ij}\beta_{ij}}_{\gamma_{ij}\delta_{ij}}, \\
 &R(u;\lambda)=\begin{pmatrix}
   a(u) & 0              & 0              & 0              \\
  0              & b(u;\lambda) & \bar c(u;\lambda)      & 0              \\
  0              & c(u;\lambda) & \bar b(u;\lambda) & 0              \\
  0              & 0              & 0              & a(u)
 \end{pmatrix}. \label{Rz}
\end{align}
where the matrix indexes take the values $+$ and $-$ being defined through the heights as follows
\begin{align}
 \alpha_{ij}&=d_{i-1,j}-d_{ij}, & \beta_{ij}&=d_{i-1,j-1}-d_{i-1,j},
 & \gamma_{ij}&=d_{i-1,j-1}-d_{i,j-1}, & \delta_{ij}&=d_{i,j-1}-d_{ij}. \label{albegade}
\end{align}
We attach the differences~\eqref{albegade} to the corresponding edges as in Fig.~\ref{fig2}: the signs $\gamma_{i,j+1}=\alpha_{ij}$ are attached to the vertical edges and the signs $\beta_{i,j+1}=\delta_{ij}$ -- to the horizontal edges.
The DWBC are imposed in terms of the signs on the external edges: they should coincide with the signs in the Fig.~\ref{fig4}. The configuration of the model can be considered as a distribution of these signs on the internal edges and a value of one of the boundary heights, for example, $d_{nn}$. The partition function in terms of the Felder $R$-matrix can be written as an entry
\begin{align} \label{ZSOSRRpmmp}
 Z^{(n)}(u;v;\lambda)=\mathbb R(u;v;\lambda)^{+,\ldots,+;-,\ldots,-}_{-,\ldots,-;+,\ldots,+}
\end{align}
of the matrix~\footnote{We use the notations of the section~\ref{sec110}: we understand an expression of the form $A(\lambda+\hbar B)$, where $A(\lambda)$ and $B$ are matrices, as the series
\begin{align}
 A(\lambda+\hbar B)=\sum_{k=0}^\infty\frac{\hbar^k}{k!}\frac{\partial^k}{\partial\lambda^k}A(\lambda) B^k.
\end{align}
}
\begin{align} \label{ZSOSDWBCdef}
 \mathbb R(u;v)=\mathop{\overrightarrow\prod}\limits_{1\le i\le n}\; \mathop{\overleftarrow\prod}\limits_{n\ge j\ge 1} R^{(i,j')}(u_i,v_j;\Lambda_{ij}),
\end{align}
where $\lambda=\hbar d_{nn}$ and $\Lambda_{ij}=\lambda+\hbar(n-i)+\hbar\sum\limits_{l=j+1}^n H^{(l)}$.
This description of the SOS model generalizes the description of the 6-vertex model.

Consider a group homomorphism $\chi\colon\Gamma\to\mathbb C^\times$, where $\Gamma=\mathbb Z+\tau\mathbb Z$ and $\mathbb C^\times$ is a group of non-zero complex numbers with respect to the multiplication. It can be set fixing two values $\chi(1)$ and $\chi(\tau)$. The holomorphic functions on $\mathbb C$ with the translation properties
\begin{align}
 \phi(u+1)&=\chi(1)\phi(u), \\
 \phi(u+\tau)&=\chi(\tau)e^{-2\pi inu-\pi in\tau}\phi(u)
\end{align}
are called elliptic polynomials (or theta-functions) of degree $n$ with the {\itshape character} $\chi$. Let $\Theta_n(\chi)$ be a space of these functions. If $n>0$ then $\dim\Theta_n(\chi)=n$ (and $\dim\Theta_n(\chi)=0$ if $n<0$).

In~\cite{S3} we derive the analytical properties of the function~\eqref{ZSOSDWBCdef} analogous to the properties found by Korepin for the 6-vertex model:
\begin{itemize}
\item[({\it I})] $Z^{(n)}(u;v;\lambda)$ is an elliptic polynomial of order $n$ with the character $\chi$ in each variable $u_i$, where the character $\chi$ are defined by the values
\begin{align}
 \chi(1)&=(-1)^n, &
 \chi(\tau)&=(-1)^ne^{2\pi i(\lambda+\sum\limits_{j=1}^n v_j)}. \label{char_chi}
\end{align}
\item[({\it II})] $Z^{(n)}(u;v;\lambda)$ is symmetric with respect to the variables $u$: $Z^{(n)}(u;v;\lambda)=Z^{(n)}(u^\sigma;v;\lambda)$, $\forall\sigma\in S_n$;
\item[({\it III})] $Z^{(n)}(u;v;\lambda)$ is an elliptic polynomial of order $n$  with the character $\tilde\chi$ in each variable $v_i$, where
\begin{align}
\tilde\chi(1)&=(-1)^n, &
\tilde\chi(\tau)&=(-1)^n e^{2\pi i(-\lambda+\sum\limits_{i=1}^n u_i)}
\end{align}
\item[({\it IV})] $Z^{(n)}(u;v;\lambda)$ is symmetric with respect to the variables $v$: $Z^{(n)}(u;v;\lambda)=Z^{(n)}(u;v^\sigma;\lambda)$, $\forall\sigma\in S_n$;
\item[({\it V})] $Z^{(n)}(u;v;\lambda)$ satisfies the recurrent relation
\begin{align}
 Z^{(n)}(u_1,&\ldots,u_{n-1},u_n=v_n-\hbar;v_1,\ldots,v_n;\lambda)=\frac{\theta(\lambda+n\hbar)\theta(\hbar)}{\theta(\lambda+(n-1)\hbar)}\times \label{Zrec} \\
  &\times\prod_{m=1}^{n-1}\Big(\theta(v_n-v_m-\hbar)\theta(u_m-v_n)\Big)\,
 Z^{(n-1)}(u_1,\ldots,u_{n-1};v_{n-1},\ldots,v_1;\lambda); \notag
\end{align}
\item[({\it VI})] the partition function for $1\times1$ lattice with DWBC is
\begin{gather}
 Z^{(1)}(u_1;v_1;\lambda)=\bar c(u_1-v_1)=\frac{\theta(u_1-v_1-\lambda)\theta(\hbar)}{\theta(-\lambda)}. \label{Z0c}
\end{gather}
\end{itemize}

Since the space of elliptic polimonials of order $n$ is $n$-dimentional there exists a unique function having the given values in $n$ gereric ponts (see Lemma~\ref{D_ep_coin} from Appendix~\ref{D_Appendix_ep} of~\cite{S3}. This implies that there is a unique functions $Z^{(n)}(u_1,\ldots,u_n;v_1,\ldots,v_n)$ satisfying the conditions~({\it I}), ({\it IV}) -- ({\it VI}). These conditions are similar to conditions~({\it i}), ({\it iv}) -- ({\it vi}) for the 6-vertex case, but in the elliptic case we have an additional proviso -- the character. This was an obstacle to find a determinant formula generalizing the Izergin's formula~\eqref{stat-s}. However we can generalize the formula~\eqref{stat-s_pr} using the projections for the algebra $\A_F$.

\begin{theor}
The functions
\begin{gather}
 Z^{(n)}(u;v;\lambda)=\prod_{i,j=1}^n\theta(u_i-v_j)
  \prod_{n\ge k>m\ge1}\frac{\theta(u_k-u_m+\hbar)\theta(v_k-v_m-\hbar)}{\theta(u_k-u_m)
  \theta(v_k-v_m)}\times \notag \\
\times\big(\hbar\theta(\hbar)\big)^n\La P^+_{\lambda-(n-1)\hbar}\big(f(u_n)\cdots f(u_1)\big),
e(v_n)\cdots e(v_1)\Ra= \notag \\
 =\prod_{n\ge k>m\ge1}\frac{\theta(v_k-v_m-\hbar)}{\theta(v_k-v_m)}
  \sum_{\sigma\in S_n}\prod_{\substack{l<l' \\ \sigma(l)>\sigma(l')}}
   \frac{\theta(v_{\sigma(l)}-v_{\sigma(l')}+\hbar)}{\theta(v_{\sigma(l)}-v_{\sigma(l')}-\hbar)}
	 \prod_{1\le k<m\le n}\theta(u_k-v_{\sigma(m)}) \times  \notag \\
  \times\prod_{n\ge k>m\ge1}\theta(u_k-v_{\sigma(m)}+\hbar)\prod_{m=1}^n\frac{\theta(u_m-v_{\sigma(m)}-\lambda-(m-1)\hbar)\theta(\hbar)}
                      {\theta(-\lambda-(m-1)\hbar)}. \label{Z_Z_}
\end{gather}
satisfy the conditions~({\it I}), ({\it IV}) -- ({\it VI}) and, consequently, the formula~\eqref{Z_Z_} gives the partition functions for the SOS model with DWBC.
\end{theor}

The second equality in the formula~\eqref{Z_Z_} is derived using the integral representation of the projection~\eqref{Pfn0_prod_f_int} and the formula for the paring of the total currents~\eqref{la_ff_ee_ra} (for the case $t^+=t^-=1$).

Let us remark that the factors in the formula~\eqref{Z_Z_} before the paring are not purely fitting (maybe excepting $\theta(\hbar)^n$). The factor $\prod_{n\ge k>m\ge1}\frac{\theta(u_k-u_m+\hbar)\theta(v_k-v_m-\hbar)}{\theta(u_k-u_m)\theta(v_k-v_m)}$ is necessary for the symmetry over $\{u\}$ and $\{v\}$. The factor $\prod_{i,j=1}^n\theta(u_i-v_j)$ is needed to annihilate the poles. The numerical factor $\theta(\hbar)^n$ can be changed by a renormalization of the currents and of the pairing. 

The formula~\eqref{Z_Z_} in the trigonometric limit $\tau\to i\infty$ gives the partition function for the trigonometric SOS model:
\begin{gather}
Z^{(n)}_{\text{trig}}(z_1,\ldots,z_n;w_1\ldots,w_n;\mu)
 =\prod_{k,j=1}^n\big(2\pi i e^{\pi i(u_k+v_j)}\big)\lim_{\tau\to i\infty}
 Z^{(n)}(u;v;\lambda)= \notag \\
 =\prod_{n\ge k>m\ge1}\frac{q^{-1}w_k-qw_m}{w_k-w_m}
  \sum_{\sigma\in S_n}\prod_{\substack{l<l' \\ \sigma(l)>\sigma(l')}}
   \frac{qw_{\sigma(l)}-q^{-1}w_{\sigma(l')}}{q^{-1}w_{\sigma(l)}-qw_{\sigma(l')}}\times \label{Z_Z_t} \\
  \times\prod_{n\ge k>m\ge1}\big(qz_k-q^{-1}w_{\sigma(m)}\big)\prod_{1\le k<m\le n}\big(z_k-w_{\sigma(m)}\big)
   \prod_{m=1}^n\frac{\big(z_m-q^{2(m-1)}w_{\sigma(m)}\mu\big)(q-q^{-1})}
                      {\big(1-q^{2(m-1)}\mu\big)}, \notag
\end{gather}
where $z_i=e^{2\pi iu_i}$ and $w_j=e^{2\pi iv_j}$ are multiplicative variables, $q=e^{\pi i\hbar}$ is a multiplicative anisotropy parameter and $\mu=e^{2\pi i\lambda}$ is a multiplicative dynamical parameter. The trigonometric SOS model is defined by the matrix of Boltzmann weights
\begin{align}
 &R(z,w;\mu)=2\pi i e^{\pi i(u+v)}\lim_{\tau\to i\infty}R(u-v;\lambda)= \\
 &\begin{pmatrix}
  zq-wq^{-1}  & 0              & 0              & 0              \\
  0 & \frac{(z-w)(\mu q-q^{-1})}{(\mu-1)} &\frac{(z-w\mu)(q-q^{-1})}{(1-\mu)}  & 0              \\
  0 & \frac{(z\mu-w)(q-q^{-1})}{(\mu-1)} &\frac{(z-w)(\mu q^{-1}-q)}{(\mu-1)}  & 0   \\
  0              & 0              & 0              & zq-wq^{-1}
 \end{pmatrix}.  \label{Rzmu}
\end{align}
The limit $\lambda\to-i\infty$ implying $\mu\to\infty$ (or $\lambda\to i\infty$ implying $\mu\to0$) of the formula~\eqref{Z_Z_t} coincide with the formula~\eqref{stat-s_pr} for the 6-vertex partition function corresponding to DWBC.

%
%

\renewcommand{\thesection}{\arabic{section}}

\appendix

\chapter{Transition function for the Toda chain model}
\label{SA1}

\thispagestyle{empty} \pagestyle{myheadings}
\markboth{}{A.\,Silantyev\hfil{\itshape Transition function for the Toda chain model}}

\author{A.\,Silantyev \footnote{E-mail: silant@thsun1.jinr.ru, silant@tonton.univ-angers.fr} \bigskip\\
{\normalsize \itshape Bogoliubov Laboratory of Theoretical Physics, JINR } \\
{\normalsize \itshape 141980 Dubna, Moscow region, Russia} \\ [5pt]
{\normalsize \itshape D\'epartement de Math\'ematiques, Universit\'e d'Angers, 49045 Angers, France}}
\title{Transition function for the Toda chain model}
\date{}

\maketitle

\begin{abstract}
The method of $\Lambda$-operators developed by S.\,Derkachov, G.\,Korchemsky, A.\,Ma\-na\-shov is
applied to a derivation of eigenfunctions for the open Toda chain.
The Sklyanin measure is reproduced using diagram techniques developed
for these $\Lambda$-operators. The properties of the
$\Lambda$-operators are studied. This approach to the open Toda
chain eigenfunctions reproduces Gauss-Givental representation for
these eigenfunctions.
\end{abstract}

\section{Introduction}

This work was inspired by the article~\cite{Derkachov} devoted to
the Separation of Variables (SoV) method for $XXX$-model. The main
idea of this method is to find an integral transformation such
that eigenfunctions of quantum integrals of motion in new
variables becomes the product of functions of one
variable~\cite{Sklyanin}. If everyone of these functions satisfies
the Baxter equation, then the initial multivariable function
becomes an eigenfunction. The kernel of this transform is called a
transition function and can be constructed as consecutive
application of operators $\Lambda_k(u)$:
$\Lambda_N(\gamma_1)\Lambda_{N-1}(\gamma_2)\cdots\Lambda_1(\gamma_N)$.
Every operator $\Lambda_k(u)$ is an integral transformation, which
maps a function of $k-1$ variables onto function of $k$ variables.
The properties of the transition function can be translated to the
properties of these operators (see
section~\ref{A_Pr_Lambda}). \\

The transition function for the $N$-particle periodic Toda chain
was obtained in the works~\cite{Gutzwiller}, \cite{Sklyanin},
\cite{Kharchev_P}. In this case, the transition function is
proportional to the eigenfunction of the open Toda chain, with a
factor depending on the coordinate of $N$-th particle. We apply
methods of the paper~\cite{Derkachov} to obtain
these eigenfunctions as a product of $\Lambda$-operators. \\

This form of eigenfunctions of the open Toda chain leads to an
integral representation that appeared first in~\cite{Givental}
employing a different approach. Recently it was interpreted from a
group-theoretical point of view using the Gauss decomposition of
$GL(N,\mathbb R)$~\cite{Kharchev_GG}. Therefore, this integral
representation of the eigenfunctions for the open Toda chain is
called a Gauss-Givental representation. \\

The method of a triangulation of the Lax matrix described
in~\cite{Pasquier} was used in~\cite{Derkachov}. We also use a
triangulation, which is implemented by a gauge transformation
parametrized by variables $y_0,\ldots,y_N$. In the periodic case
one has to impose the condition $y_0=y_N$ and the method described
in~\cite{Pasquier} produces Baxter's $Q$-operators for the
periodic Toda chain model. Following~\cite{Derkachov} we impose a
different boundary condition: $y_0\to -\infty$, $y_N\to +\infty$
to construct $\Lambda$-operator. Thus $\Lambda$-operator and
Baxter's $Q$-operator for the periodic Toda chain correspond to the
different choice of the boundary conditions in the method of
triangulation of the Lax matrices. \\

To describe the construction of eigenfunctions for open Toda
chain we develop a kind of the Feynman diagram technique similar to
one exploited in~\cite{Derkachov}. It allows to reduce
calculations with kernels of $\Lambda$-operators to simple
manipulations with diagrams. \\

The article is organized as follows. In section~\ref{A_OTCh} we
recall a definition of the open Toda chain model following~\cite{Kharchev_O, Kharchev_OP}.
Section~\ref{A_Eig_OTCh} is devoted to a description of
eigenfunctions in terms of the product of $\Lambda$-operators
and formulation of a diagram technique developed in~\cite{Derkachov}. In
section~\ref{A_Int_meas} we use this technique in order to prove
that eigenfunctions satisfy an orthogonality condition. As a by-product
of this calculation we obtain a Sklyanin measure, which is necessary
to prove a completeness condition.
Section~\ref{A_Pr_Lambda} is devoted to algebraic properties of
$\Lambda$-operators.

\section{Open Toda chain model}
\label{A_OTCh}

 The quantum $N$-particle open Toda chain is a one-dimensional model with
the exponential interaction between the nearest particles.
The hamiltonian of the system is equal to
\begin{equation}
 H=\frac12\sum_{n=1}^N p_n^2+\sum_{n=1}^{N-1} e^{x_n-x_{n+1}},
\end{equation}
where $p_n\hm=-i\hbar\dfrac{\partial}{\partial x_n}$ is an operator of momentum
for the $n$-th particle. Due to the translational invariance, the total momentum
\begin{equation}
 P=\sum_{n=1}^N p_n
\end{equation}
commutes with the hamiltonian, i.e. it is also an integral of motion.
There are $N$ functionally independent integrals of motion for
this system. It is relevant to use the $R$-matrix formalism to find them.
First of all, introduce the Lax operator of the Toda chain
\begin{equation} \label{A_Ln}
 L_n(u)=
  \begin{pmatrix}
   u-p_n & e^{-x_n} \\
   -e^{x_n} & 0
  \end{pmatrix}, \qquad n=1,\ldots,N,
\end{equation}
and monodromy matrix for the $N$-particle Toda chain
\begin{equation} \label{A_TN}
 T_N(u)=L_N(u)\cdots L_1(u)=
  \begin{pmatrix}
   A_N(u) & B_N(u) \\
   C_N(u) & D_N(u)
  \end{pmatrix},
\end{equation}
where $u$ is a spectral parameter. \\

The following recurrent relations, which are direct consequence of this definition,
will be useful below:
\begin{gather}
 A_N(u)=(u-p_N)A_{N-1}(u)+e^{-x_N}C_{N-1}(u), \label{A_Arec} \\
 C_N(u)=-e^{x_N}A_{N-1}(u), \label{A_Crec} \\
 A_N(u)=(u-p_N)A_{N-1}(u)-e^{x_{N-1}-x_N}A_{N-2}(u).  \label{A_AArec}
\end{gather}
These relations show that $A_N(u)$ and $C_N(u)$ are polynomials in
$u$ of degree $N$ and $N-1$ respectively. Analogously, $B_N(u)$
and $D_N(u)$ have degree $N-1$ and $N-2$. \\

The monodromy matrix satisfies to the quantum $RTT$-relation
\begin{equation} \label{A_RTT}
  R(u-v)\,(T_N(u)\otimes I)(I\otimes T_N(v))=
  (I\otimes T_N(v))(T_N(u)\otimes I)\,R(u-v)
\end{equation}
with the rational $R$-matrix
\begin{equation} \label{A_RP}
 R(u)=I\otimes I+\frac{i\hbar}{u}\,\mathcal P,
\end{equation}
where $\mathcal P$ is a permutation matrix: $\mathcal P_{ij,kl}\hm=\delta_{il}\delta_{jk}$.

Rewriting~\eqref{A_RTT} by entries one obtains, in particular, the
relation
\begin{equation} \label{A_RTTAA}
 A_N(u)A_N(v)=A_N(v)A_N(u).
\end{equation}

This means that $A_N(u)$ is a generation function of the integrals
of motion of the integrable system with $N$ degrees of freedom.
Explicit calculations of two first integrals show that these are
integrals for open Toda chain model:
\begin{gather}
 A_N(u)=\sum_{k=0}^N (-1)^k u^{N-k} H_k, \\
 H_0=1, \quad H_1=P, \quad H_2=\frac12 P^2-H, \\
 [H_k,H_j]=0. \label{A_commHH}
\end{gather}

By virtue of~\eqref{A_commHH} there exist common eigenfunctions of the integrals $H_k$
corresponding to the eigenvalues $E_k$. They are defining by the equation
\begin{equation} \label{A_psi_def_E}
 A_N(u)\psi_E(x)= a_N(u;E)\psi_E(x),
\end{equation}
where
\begin{equation*}
 a_N(u;E)=\sum_{k=0}^N (-1)^k u^{N-k} E_k,
\end{equation*}
$E_0=1$, $E=(E_1,\ldots,E_N)$, $x=(x_1,\ldots,x_N)$. Representing
eigenvalues $E_k$ as symmetric combinations
\begin{equation} \label{A_Ek}
 E_k=\sum\limits_{j_1<\ldots<j_k}\gamma_{j_1}\ldots\gamma_{j_k}
\end{equation}
of real variables $\gamma=(\gamma_1,\ldots,\gamma_N)$, one can rewrite equation~\eqref{A_psi_def_E} as follows
\begin{equation} \label{A_psi_def}
 A_N(u)\psi_{\gamma}(x)= \prod_{j=1}^N(u-\gamma_j)\psi_{\gamma}(x).
\end{equation}

\section{Eigenfunctions of the open Toda chain}
\label{A_Eig_OTCh}

In this section we shall find eigenfunctions of the open Toda
chain defined in the previous section by the
equation~\eqref{A_psi_def}. This equation is equivalent to the
system of $N$ equations
\begin{equation} \label{A_Apsi0}
 A_N(\gamma_j)\psi_{\gamma}(x)=0, \qquad j=1,\ldots,N.
\end{equation}
The eigenvalues~\eqref{A_Ek} are invariant under the permutations of $\gamma_1,\ldots,\gamma_N$.
Therefore, it is reasonably to require the invariance of eigenfunction under these permutations, which we shall call the Weyl invariance:
\begin{equation} \label{A_Winv}
 \psi_{\sigma(\gamma)}(x)=\psi_{\gamma}(x), \qquad \text{for all } \sigma\in S_N,
\end{equation}
where $S_N$ is a permutation group and $\sigma(\gamma)
=(\gamma_{\sigma(1)},\ldots, \gamma_{\sigma(N)})$. \\

It is sufficiently to find a Weyl invariant solution of the unique
equation
\begin{equation} \label{A_Agamma1psi0}
 A_N(\gamma_1)\psi_{\gamma}(x)=0,
\end{equation}
which will be a solution for the whole system~\eqref{A_Apsi0} due to
its Weyl invariance. \\

To solve the last equation we shall consider a gauge transformation of the Lax operators
\begin{equation} \label{A_MLM}
 \widetilde L_n(u)=M_{n}L_n(u)M_{n-1}^{-1}, \qquad n=1,\ldots,N
\end{equation}
by the matrices
\begin{equation}
 M_n=\begin{pmatrix}
   1 & 0 \\
   ie^{y_n} & 1
  \end{pmatrix}, \qquad n=0,\ldots,N.
\end{equation}
The deformed $N$-particle monodromy matrix is
\begin{equation}
 \widetilde T_N(u)\equiv
  \begin{pmatrix}
   \widetilde A_N(u) & \widetilde B_N(u) \\
   \widetilde C_N(u) & \widetilde D_N(u)
  \end{pmatrix}
  =\widetilde L_N(u)\cdots\widetilde L_1(u)=M_N T_N(u) M_0^{-1}.
\end{equation}
In particular, we have
\begin{gather}
 \widetilde L_n(u)_{21}=ie^{y_n}(u-p_n-ie^{y_{n-1}-x_n}+ie^{x_n-y_n}), \label{A_L21_til}\\
 \widetilde C_N(u)=ie^{y_N} A_N(u)+e^{y_N+y_0}B_N(u)+C_N(u)-ie^{y_0}D_N(u). \label{A_C_til}
\end{gather}
Here $\widetilde L_n(u)_{21}$ is a lower off-diagonal entry of the matrix $\widetilde L_n(u)$. \\

Let us consider the auxiliary equation
\begin{equation}
 \widetilde L_n(u)_{21}w_n(u)=0,
\end{equation}
which has the following solution
\begin{equation}
 w_n(u)=\exp\Bigl\{\frac{i}{\hbar}u(x_n-y_{n-1})-\frac{1}{\hbar}e^{y_{n-1}-x_n}-\frac{1}{\hbar}e^{x_n-y_n}\Bigr\}.
\end{equation}
It is clear that the function
\begin{equation}
 W_u(x;y)\hm=\prod\limits_{n=1}^N w_n(u)
 =\exp\sum\limits_{n=1}^N\Bigl\{\frac{i}{\hbar}u(x_n-y_{n-1})-\frac{1}{\hbar}e^{y_{n-1}-x_n}-\frac{1}{\hbar}e^{x_n-y_n}\Bigr\}
\end{equation}
is a solution to the equation
\begin{equation} \label{A_C_tilW}
 \widetilde C_N(u) W_u(x;y)=0.
\end{equation}
In the limit $y_0\to -\infty$, $y_N\to +\infty$ the formula~\eqref{A_C_til} gives us
the equality
\begin{equation} \label{A_ANlim}
A_N(u)\hm=-i\lim\limits_{\substack{y_0\to -\infty \\ y_N\to +\infty}} e^{-y_N}\widetilde C_N(u).
\end{equation}
Therefore, multiplying the equation~\eqref{A_C_tilW} by
$-ie^{-y_N}e^{\frac{i}{\hbar}u(y_0+y_N)}$, taking the same limit
as in~\eqref{A_ANlim} and setting $u\hm=\gamma_1$ we arrive to the
equation~\eqref{A_Agamma1psi0} with the solution
$\psi_{\gamma}(x)\hm=\Lambda_{\gamma_1}(x;y)$, where
\begin{equation} \label{A_Lamb}
 \begin{split}
 \Lambda_u(x_1,\ldots,x_N;y_1,\ldots,y_{N-1})
     &=\lim\limits_{\substack{y_0\to -\infty \\ y_N\to +\infty}}
        e^{\frac{i}{\hbar}u(y_0+y_N)}W_{u}(x;y)=\\
     =\exp\Bigl\{\frac{i}{\hbar}u(\sum\limits_{n=1}^{N}x_n-&\sum\limits_{n=1}^{N-1}y_n)
        -\frac{1}{\hbar}\sum\limits_{n=1}^{N-1}(e^{y_n-x_{n+1}}+e^{x_n-y_n})\Bigr\}.
 \end{split}
\end{equation}

Let $\Lambda_N(u)$ be an operator with the kernel
$\Lambda_u(x_1,\ldots,x_N;y_1,\ldots,y_{N-1})$, i.e.
\begin{equation} \label{A_Lambf}
 (\Lambda_N(u)\cdot f)(x)
   =\int\limits_{\mathbb R^{N-1}}dy\,\Lambda_u(x_1,\ldots,x_N;y_1,\ldots,y_{N-1})f(y).
\end{equation}
This operator acts from the space of functions of $N-1$ variables
to the space of functions of $N$ variables. Setting $u=\gamma_1$
in~\eqref{A_Lambf} we obtain a solution to~\eqref{A_Agamma1psi0} for
arbitrary function of $N-1$ variables $f(y)$.

\begin{theor} \label{A_Th_psi}
 The following solution to the equation~\eqref{A_Agamma1psi0}
\begin{equation} \label{A_psi}
 \psi_{\gamma}(x)=(\Lambda_N(\gamma_1)\cdots\Lambda_2(\gamma_{N-1})\Lambda_1(\gamma_N)\cdot 1)(x_1,\ldots,x_N),
\end{equation}
where $(\Lambda_1(\gamma_N)\cdot
1)(x_1)=e^{\frac{i}{\hbar}\gamma_N x_1}$, is Weyl invariant, i.e.
satisfies to the condition~\eqref{A_Winv}, and, therefore, is a
solution to the equation~\eqref{A_psi_def}.
\end{theor}

\noindent{\bfseries Proof.} It is sufficient to establish the invariance under the elementary
permutations, i.e. we need to check the equality
\begin{equation} \label{A_LgammaLgamma}
 \begin{split}
  &\int\limits_{\mathbb R^{N-n+1}}dy\,\Lambda_{\gamma_{n-1}}(x_1,\ldots,x_{N-n+2};y_1,\ldots,y_{N-n+1})
                            \Lambda_{\gamma_{n}}(y_1,\ldots,y_{N-n+1};z_1,\ldots,z_{N-n})= \\
  &=\int\limits_{\mathbb R^{N-n+1}}dy\,\Lambda_{\gamma_{n}}(x_1,\ldots,x_{N-n+2};y_1,\ldots,y_{N-n+1})
                            \Lambda_{\gamma_{n-1}}(y_1,\ldots,y_{N-n+1};z_1,\ldots,z_{N-n}),
 \end{split}
\end{equation}
for $n=2,\ldots,N$. \\

\begin{figure}[h]
\hspace{\fill} \setlength{\unitlength}{0.00087489in}
\begingroup\makeatletter\ifx\SetFigFont\undefined%
\gdef\SetFigFont#1#2#3#4#5{%
  \reset@font\fontsize{#1}{#2pt}%
  \fontfamily{#3}\fontseries{#4}\fontshape{#5}%
  \selectfont}%
\fi\endgroup%
{\renewcommand{\dashlinestretch}{30}
\begin{picture}(781,1336)(0,-10)
\path(60,343)(330,883)
\path(303.167,762.252)(330.000,883.000)(249.502,789.085)
\path(60,343)(510,1243)
\put(150,298){\makebox(0,0)[lb]{\smash{{\SetFigFont{12}{14.4}{\rmdefault}{\mddefault}{\updefault}$x_k$}}}}
\put(600,1198){\makebox(0,0)[lb]{\smash{{\SetFigFont{12}{14.4}{\rmdefault}{\mddefault}{\updefault}$y_k$}}}}
\put(5,-23){\makebox(0,0)[lb]{\smash{{\SetFigFont{14}{14.4}{\rmdefault}{\mddefault}{\updefault}fig.
1a}}}}
\end{picture}
} \hspace{\fill} \setlength{\unitlength}{0.00087489in}
\begingroup\makeatletter\ifx\SetFigFont\undefined%
\gdef\SetFigFont#1#2#3#4#5{%
  \reset@font\fontsize{#1}{#2pt}%
  \fontfamily{#3}\fontseries{#4}\fontshape{#5}%
  \selectfont}%
\fi\endgroup%
{\renewcommand{\dashlinestretch}{30}
\begin{picture}(736,1318)(0,-10)
\path(465,343)(15,1243) \path(465,343)(15,1243)
\path(465,343)(195,883) \path(465,343)(195,883)
\whiten\path(275.498,789.085)(195.000,883.000)(221.833,762.252)(264.765,743.469)(275.498,789.085)
\put(105,1198){\makebox(0,0)[lb]{\smash{{\SetFigFont{12}{14.4}{\rmdefault}{\mddefault}{\updefault}$y_k$}}}}
\put(555,298){\makebox(0,0)[lb]{\smash{{\SetFigFont{12}{14.4}{\rmdefault}{\mddefault}{\updefault}$x_{k+1}$}}}}
\put(15,793){\makebox(0,0)[lb]{\smash{{\SetFigFont{10}{14.4}{\rmdefault}{\mddefault}{\updefault}$u$}}}}
\put(50,-23){\makebox(0,0)[lb]{\smash{{\SetFigFont{14}{14.4}{\rmdefault}{\mddefault}{\updefault}fig.
1b}}}}
\end{picture}
} \hspace{\fill} \setlength{\unitlength}{0.00087489in}
\begingroup\makeatletter\ifx\SetFigFont\undefined%
\gdef\SetFigFont#1#2#3#4#5{%
  \reset@font\fontsize{#1}{#2pt}%
  \fontfamily{#3}\fontseries{#4}\fontshape{#5}%
  \selectfont}%
\fi\endgroup%
{\renewcommand{\dashlinestretch}{30}
\begin{picture}(736,913)(0,-10)
\put(555,298){\makebox(0,0)[lb]{\smash{{\SetFigFont{12}{14.4}{\rmdefault}{\mddefault}{\updefault}$x_k$}}}}
\path(465,343)(195,883) \path(465,343)(195,883)
\whiten\path(275.498,789.085)(195.000,883.000)(221.833,762.252)(275.498,789.085)
\put(15,793){\makebox(0,0)[lb]{\smash{{\SetFigFont{10}{14.4}{\rmdefault}{\mddefault}{\updefault}$u$}}}}
\put(160,-23){\makebox(0,0)[lb]{\smash{{\SetFigFont{14}{14.4}{\rmdefault}{\mddefault}{\updefault}fig.
1c}}}}
\end{picture}
} \hspace{\fill} \setlength{\unitlength}{0.00087489in}
\begingroup\makeatletter\ifx\SetFigFont\undefined%
\gdef\SetFigFont#1#2#3#4#5{%
  \reset@font\fontsize{#1}{#2pt}%
  \fontfamily{#3}\fontseries{#4}\fontshape{#5}%
  \selectfont}%
\fi\endgroup%
{\renewcommand{\dashlinestretch}{30}
\begin{picture}(522,913)(0,-10)
\path(105,838)(330,388) \path(105,838)(330,388)
\path(330,388)(285,478) \path(330,388)(285,478)
\whiten\path(365.498,384.085)(285.000,478.000)(311.833,357.252)(365.498,384.085)
\put(195,793){\makebox(0,0)[lb]{\smash{{\SetFigFont{12}{14.4}{\rmdefault}{\mddefault}{\updefault}$x_k$}}}}
\put(105,388){\makebox(0,0)[lb]{\smash{{\SetFigFont{10}{14.4}{\rmdefault}{\mddefault}{\updefault}$u$}}}}
\put(05,-23){\makebox(0,0)[lb]{\smash{{\SetFigFont{14}{14.4}{\rmdefault}{\mddefault}{\updefault}fig.
1d}}}}
\end{picture}
} \hspace{\fill} \setlength{\unitlength}{0.00087489in}
\begingroup\makeatletter\ifx\SetFigFont\undefined%
\gdef\SetFigFont#1#2#3#4#5{%
  \reset@font\fontsize{#1}{#2pt}%
  \fontfamily{#3}\fontseries{#4}\fontshape{#5}%
  \selectfont}%
\fi\endgroup%
{\renewcommand{\dashlinestretch}{30}
\begin{picture}(589,1093)(0,-10)
\path(195,343)(195,1018) \path(195,343)(195,793)
\whiten\path(225.000,673.000)(195.000,793.000)(165.000,673.000)(195.000,709.000)(225.000,673.000)
\put(285,298){\makebox(0,0)[lb]{\smash{{\SetFigFont{12}{14.4}{\rmdefault}{\mddefault}{\updefault}$x_{k+1}$}}}}
\put(285,973){\makebox(0,0)[lb]{\smash{{\SetFigFont{12}{14.4}{\rmdefault}{\mddefault}{\updefault}$z_k$}}}}
\put(05,-23){\makebox(0,0)[lb]{\smash{{\SetFigFont{14}{14.4}{\rmdefault}{\mddefault}{\updefault}fig.
1e}}}}
\put(15,658){\makebox(0,0)[lb]{\smash{{\SetFigFont{10}{12.0}{\rmdefault}{\mddefault}{\updefault}$u$}}}}
\end{picture}
} \hspace{\fill} \\
\end{figure}

To do it we shall use a diagram technique introduced
in~\cite{Derkachov}. Let us denote the function
$I(x_k,y_k)\hm=\exp\Bigl\{-\dfrac{1}{\hbar}e^{x_k-y_k}\Bigr\}$ by
a line pictured in the fig.~1a, the function
$J_u(x_{k+1},y_k)\hm=\exp\Bigl\{\dfrac{i}{\hbar}u(x_{k+1}-y_k)\hm-\dfrac{1}{\hbar}e^{y_k-x_{k+1}}\Bigr\}$
by a fig.~1b, the function
$Z_u(x_k)\hm=\exp\Bigl\{\dfrac{i}{\hbar}u x_k\Bigr\}$ by a
fig.~1c, the function
$Z_u^{-1}(x_k)\hm=\exp\Bigl\{-\dfrac{i}{\hbar}u x_k\Bigr\}$ by a
fig.~1d and the function
$Y_u(x_{k+1},z_k)\hm=\Bigl(1+e^{x_{k+1}-z_k}\Bigr)^{\tfrac{i}{\hbar}u}$
by a fig.~1e. \\

\begin{figure}[h]
\hspace{\fill} \setlength{\unitlength}{0.00087489in}
\begingroup\makeatletter\ifx\SetFigFont\undefined%
\gdef\SetFigFont#1#2#3#4#5{%
  \reset@font\fontsize{#1}{#2pt}%
  \fontfamily{#3}\fontseries{#4}\fontshape{#5}%
  \selectfont}%
\fi\endgroup%
{\renewcommand{\dashlinestretch}{30}
\begin{picture}(4291,4293)(0,-10)
\path(465,105)(735,645)
\path(708.167,524.252)(735.000,645.000)(654.502,551.085)
\path(465,105)(915,1005) \path(1365,105)(915,1005)
\path(1365,105)(915,1005) \path(1365,105)(1095,645)
\path(1365,105)(1095,645)
\whiten\path(1175.498,551.085)(1095.000,645.000)(1121.833,524.252)(1164.765,505.469)(1175.498,551.085)
\path(2265,105)(2535,645)
\path(2508.167,524.252)(2535.000,645.000)(2454.502,551.085)
\path(2265,105)(2715,1005) \path(3165,105)(2715,1005)
\path(3165,105)(2715,1005) \path(3165,105)(2895,645)
\path(3165,105)(2895,645)
\whiten\path(2975.498,551.085)(2895.000,645.000)(2921.833,524.252)(2964.765,505.469)(2975.498,551.085)
\path(2715,1005)(2265,1905) \path(2715,1005)(2265,1905)
\path(2715,1005)(2445,1545) \path(2715,1005)(2445,1545)
\whiten\path(2525.498,1451.085)(2445.000,1545.000)(2471.833,1424.252)(2514.765,1405.469)(2525.498,1451.085)
\put(915,555){\makebox(0,0)[lb]{\smash{{\SetFigFont{10}{14.4}{\rmdefault}{\mddefault}{\updefault}$\gamma_1$}}}}
\put(2715,555){\makebox(0,0)[lb]{\smash{{\SetFigFont{10}{14.4}{\rmdefault}{\mddefault}{\updefault}$\gamma_1$}}}}
\put(3615,555){\makebox(0,0)[lb]{\smash{{\SetFigFont{10}{14.4}{\rmdefault}{\mddefault}{\updefault}$\gamma_1$}}}}
\put(3615,555){\makebox(0,0)[lb]{\smash{{\SetFigFont{10}{14.4}{\rmdefault}{\mddefault}{\updefault}$\gamma_1$}}}}
\put(3165,1455){\makebox(0,0)[lb]{\smash{{\SetFigFont{10}{14.4}{\rmdefault}{\mddefault}{\updefault}$\gamma_2$}}}}
\put(1815,2355){\makebox(0,0)[lb]{\smash{{\SetFigFont{10}{14.4}{\rmdefault}{\mddefault}{\updefault}$\gamma_3$}}}}
\put(2715,2355){\makebox(0,0)[lb]{\smash{{\SetFigFont{10}{14.4}{\rmdefault}{\mddefault}{\updefault}$\gamma_3$}}}}
\path(915,1005)(1185,1545)
\path(1158.167,1424.252)(1185.000,1545.000)(1104.502,1451.085)
\path(915,1005)(1365,1905) \path(1365,1905)(1635,2445)
\path(1608.167,2324.252)(1635.000,2445.000)(1554.502,2351.085)
\path(1365,1905)(1815,2805) \path(2265,1905)(1815,2805)
\path(2265,1905)(1815,2805) \path(2265,1905)(1995,2445)
\path(2265,1905)(1995,2445)
\whiten\path(2075.498,2351.085)(1995.000,2445.000)(2021.833,2324.252)(2064.765,2305.469)(2075.498,2351.085)
\path(1815,2805)(2085,3345)
\path(2058.167,3224.252)(2085.000,3345.000)(2004.502,3251.085)
\path(1815,2805)(2265,3705) \path(2715,2805)(2265,3705)
\path(2715,2805)(2265,3705) \path(2715,2805)(2445,3345)
\path(2715,2805)(2445,3345)
\whiten\path(2525.498,3251.085)(2445.000,3345.000)(2471.833,3224.252)(2514.765,3205.469)(2525.498,3251.085)
\path(2265,1905)(2535,2445)
\path(2508.167,2324.252)(2535.000,2445.000)(2454.502,2351.085)
\path(2265,1905)(2715,2805) \path(3165,1905)(2715,2805)
\path(3165,1905)(2715,2805) \path(3165,1905)(2895,2445)
\path(3165,1905)(2895,2445)
\whiten\path(2975.498,2351.085)(2895.000,2445.000)(2921.833,2324.252)(2964.765,2305.469)(2975.498,2351.085)
\path(2715,1005)(2985,1545)
\path(2958.167,1424.252)(2985.000,1545.000)(2904.502,1451.085)
\path(2715,1005)(3165,1905) \path(3615,1005)(3165,1905)
\path(3615,1005)(3165,1905) \path(3615,1005)(3345,1545)
\path(3615,1005)(3345,1545)
\whiten\path(3425.498,1451.085)(3345.000,1545.000)(3371.833,1424.252)(3414.765,1405.469)(3425.498,1451.085)
\path(3165,105)(3435,645)
\path(3408.167,524.252)(3435.000,645.000)(3354.502,551.085)
\path(3165,105)(3615,1005) \path(4065,105)(3615,1005)
\path(4065,105)(3615,1005) \path(4065,105)(3795,645)
\path(4065,105)(3795,645)
\whiten\path(3875.498,551.085)(3795.000,645.000)(3821.833,524.252)(3864.765,505.469)(3875.498,551.085)
\put(15,555){\makebox(0,0)[lb]{\smash{{\SetFigFont{10}{14.4}{\rmdefault}{\mddefault}{\updefault}$\gamma_1$}}}}
\put(555,60){\makebox(0,0)[lb]{\smash{{\SetFigFont{12}{14.4}{\rmdefault}{\mddefault}{\updefault}$x_1$}}}}
\path(465,105)(195,645) \path(465,105)(195,645)
\whiten\path(275.498,551.085)(195.000,645.000)(221.833,524.252)(275.498,551.085)
\path(1365,105)(1635,645)
\path(1608.167,524.252)(1635.000,645.000)(1554.502,551.085)
\path(1365,105)(1815,1005) \path(2265,105)(1815,1005)
\path(2265,105)(1815,1005) \path(2265,105)(1995,645)
\path(2265,105)(1995,645)
\whiten\path(2075.498,551.085)(1995.000,645.000)(2021.833,524.252)(2064.765,505.469)(2075.498,551.085)
\put(1815,1005){\blacken\ellipse{46}{46}}
\put(1815,1005){\ellipse{46}{46}} \path(1815,1005)(2085,1545)
\path(2058.167,1424.252)(2085.000,1545.000)(2004.502,1451.085)
\path(1815,1005)(2265,1905)
\put(1815,555){\makebox(0,0)[lb]{\smash{{\SetFigFont{10}{14.4}{\rmdefault}{\mddefault}{\updefault}$\gamma_1$}}}}
\put(1365,1455){\makebox(0,0)[lb]{\smash{{\SetFigFont{10}{14.4}{\rmdefault}{\mddefault}{\updefault}$\gamma_2$}}}}
\put(2265,1455){\makebox(0,0)[lb]{\smash{{\SetFigFont{10}{14.4}{\rmdefault}{\mddefault}{\updefault}$\gamma_2$}}}}
\path(1815,1005)(1365,1905) \path(1815,1005)(1365,1905)
\path(1815,1005)(1545,1545) \path(1815,1005)(1545,1545)
\whiten\path(1625.498,1451.085)(1545.000,1545.000)(1571.833,1424.252)(1614.765,1405.469)(1625.498,1451.085)
\put(915,1005){\blacken\ellipse{46}{46}}
\put(915,1005){\ellipse{46}{46}}
\put(2715,1005){\blacken\ellipse{46}{46}}
\put(2715,1005){\ellipse{46}{46}}
\put(1815,2805){\blacken\ellipse{46}{46}}
\put(1815,2805){\ellipse{46}{46}}
\put(2265,3705){\blacken\ellipse{46}{46}}
\put(2265,3705){\ellipse{46}{46}}
\put(2715,2805){\blacken\ellipse{46}{46}}
\put(2715,2805){\ellipse{46}{46}}
\put(3165,1905){\blacken\ellipse{46}{46}}
\put(3165,1905){\ellipse{46}{46}}
\put(3615,1005){\blacken\ellipse{46}{46}}
\put(3615,1005){\ellipse{46}{46}}
\put(2265,1905){\blacken\ellipse{46}{46}}
\put(2265,1905){\ellipse{46}{46}}
\put(1365,1905){\blacken\ellipse{46}{46}}
\put(1365,1905){\ellipse{46}{46}} \path(915,1005)(645,1545)
\path(915,1005)(645,1545)
\whiten\path(725.498,1451.085)(645.000,1545.000)(671.833,1424.252)(725.498,1451.085)
\path(1365,1905)(1095,2445) \path(1365,1905)(1095,2445)
\whiten\path(1175.498,2351.085)(1095.000,2445.000)(1121.833,2324.252)(1175.498,2351.085)
\path(1815,2805)(1545,3345) \path(1815,2805)(1545,3345)
\whiten\path(1625.498,3251.085)(1545.000,3345.000)(1571.833,3224.252)(1625.498,3251.085)
\path(2265,3705)(1995,4245) \path(2265,3705)(1995,4245)
\whiten\path(2075.498,4151.085)(1995.000,4245.000)(2021.833,4124.252)(2075.498,4151.085)
\put(4110,60){\makebox(0,0)[lb]{\smash{{\SetFigFont{12}{14.4}{\rmdefault}{\mddefault}{\updefault}$x_5$}}}}
\put(1905,2760){\makebox(0,0)[lb]{\smash{{\SetFigFont{12}{14.4}{\rmdefault}{\mddefault}{\updefault}$f_1$}}}}
\put(2805,2760){\makebox(0,0)[lb]{\smash{{\SetFigFont{12}{14.4}{\rmdefault}{\mddefault}{\updefault}$f_2$}}}}
\put(3255,60){\makebox(0,0)[lb]{\smash{{\SetFigFont{12}{14.4}{\rmdefault}{\mddefault}{\updefault}$x_4$}}}}
\put(1005,960){\makebox(0,0)[lb]{\smash{{\SetFigFont{12}{14.4}{\rmdefault}{\mddefault}{\updefault}$y_1$}}}}
\put(2805,960){\makebox(0,0)[lb]{\smash{{\SetFigFont{12}{14.4}{\rmdefault}{\mddefault}{\updefault}$y_3$}}}}
\put(3705,960){\makebox(0,0)[lb]{\smash{{\SetFigFont{12}{14.4}{\rmdefault}{\mddefault}{\updefault}$y_4$}}}}
\put(3255,1860){\makebox(0,0)[lb]{\smash{{\SetFigFont{12}{14.4}{\rmdefault}{\mddefault}{\updefault}$z_3$}}}}
\put(2355,3660){\makebox(0,0)[lb]{\smash{{\SetFigFont{12}{14.4}{\rmdefault}{\mddefault}{\updefault}$g_1$}}}}
\put(2265,3255){\makebox(0,0)[lb]{\smash{{\SetFigFont{10}{14.4}{\rmdefault}{\mddefault}{\updefault}$\gamma_4$}}}}
\put(465,1455){\makebox(0,0)[lb]{\smash{{\SetFigFont{10}{14.4}{\rmdefault}{\mddefault}{\updefault}$\gamma_2$}}}}
\put(915,2355){\makebox(0,0)[lb]{\smash{{\SetFigFont{10}{14.4}{\rmdefault}{\mddefault}{\updefault}$\gamma_3$}}}}
\put(1365,3255){\makebox(0,0)[lb]{\smash{{\SetFigFont{10}{14.4}{\rmdefault}{\mddefault}{\updefault}$\gamma_4$}}}}
\put(1815,4155){\makebox(0,0)[lb]{\smash{{\SetFigFont{10}{14.4}{\rmdefault}{\mddefault}{\updefault}$\gamma_5$}}}}
\put(2355,1860){\makebox(0,0)[lb]{\smash{{\SetFigFont{12}{14.4}{\rmdefault}{\mddefault}{\updefault}$z_2$}}}}
\put(1455,60){\makebox(0,0)[lb]{\smash{{\SetFigFont{12}{14.4}{\rmdefault}{\mddefault}{\updefault}$x_2$}}}}
\put(2355,60){\makebox(0,0)[lb]{\smash{{\SetFigFont{12}{14.4}{\rmdefault}{\mddefault}{\updefault}$x_3$}}}}
\put(1905,960){\makebox(0,0)[lb]{\smash{{\SetFigFont{12}{14.4}{\rmdefault}{\mddefault}{\updefault}$y_2$}}}}
\put(1455,1860){\makebox(0,0)[lb]{\smash{{\SetFigFont{12}{14.4}{\rmdefault}{\mddefault}{\updefault}$z_1$}}}}
\end{picture}
}
 \hspace{\fill} \\
\vspace{-2pt} \hspace{77mm} {\large fig. 2}
\end{figure}

The function~\eqref{A_psi} can be represented in these graphical
notations. For the case $N=5$ it is pictured in fig.~2, where bold
bullets $\bullet$ signify that we integrate over corresponding
variables.

\begin{lem}
 The equalities represented in the figures~3a, 3b, 3c are valid.
\end{lem}
\vspace{7pt}

\begin{figure}
\hspace{\fill} \setlength{\unitlength}{0.00087489in}
\begingroup\makeatletter\ifx\SetFigFont\undefined%
\gdef\SetFigFont#1#2#3#4#5{%
  \reset@font\fontsize{#1}{#2pt}%
  \fontfamily{#3}\fontseries{#4}\fontshape{#5}%
  \selectfont}%
\fi\endgroup%
{\renewcommand{\dashlinestretch}{30}
\begin{picture}(4186,1980)(0,-10)

\put(555,1005){\makebox(0,0)[rb]{\smash{{\SetFigFont{10}{12.0}{\rmdefault}{\mddefault}{\updefault}\makebox[0.8\width][s]{$(\gamma_{n-1}-\gamma_n)$}}}}}
\path(690,105)(960,645)
\path(933.167,524.252)(960.000,645.000)(879.502,551.085)
\path(690,105)(1140,1005) \path(1590,105)(1140,1005)
\path(1590,105)(1140,1005) \path(1590,105)(1320,645)
\path(1590,105)(1320,645)
\whiten\path(1400.498,551.085)(1320.000,645.000)(1346.833,524.252)(1389.765,505.469)(1400.498,551.085)
\put(1140,1005){\blacken\ellipse{46}{46}}
\put(1140,1005){\ellipse{46}{46}}
\put(780,1860){\makebox(0,0)[lb]{\smash{{\SetFigFont{12}{14.4}{\rmdefault}{\mddefault}{\updefault}$z_{k-1}$}}}}
\put(1680,1860){\makebox(0,0)[lb]{\smash{{\SetFigFont{12}{14.4}{\rmdefault}{\mddefault}{\updefault}$z_k$}}}}
\put(1230,960){\makebox(0,0)[lb]{\smash{{\SetFigFont{12}{14.4}{\rmdefault}{\mddefault}{\updefault}$y_k$}}}}
\put(780,60){\makebox(0,0)[lb]{\smash{{\SetFigFont{12}{14.4}{\rmdefault}{\mddefault}{\updefault}$x_k$}}}}
\put(1680,60){\makebox(0,0)[lb]{\smash{{\SetFigFont{12}{14.4}{\rmdefault}{\mddefault}{\updefault}$x_{k+1}$}}}}
\path(1140,1005)(1410,1545)
\path(1383.167,1424.252)(1410.000,1545.000)(1329.502,1451.085)
\path(1140,1005)(1590,1905) \path(1140,1005)(690,1905)
\path(1140,1005)(690,1905) \path(1140,1005)(870,1545)
\path(1140,1005)(870,1545)
\whiten\path(950.498,1451.085)(870.000,1545.000)(896.833,1424.252)(939.765,1405.469)(950.498,1451.085)
\path(3030,1005)(3300,1545)
\path(3273.167,1424.252)(3300.000,1545.000)(3219.502,1451.085)
\path(3030,1005)(3480,1905) \path(2580,105)(2850,645)
\path(2823.167,524.252)(2850.000,645.000)(2769.502,551.085)
\path(2580,105)(3030,1005) \path(3480,105)(3030,1005)
\path(3480,105)(3030,1005) \path(3480,105)(3210,645)
\path(3480,105)(3210,645)
\whiten\path(3290.498,551.085)(3210.000,645.000)(3236.833,524.252)(3279.765,505.469)(3290.498,551.085)
\put(3030,1005){\blacken\ellipse{46}{46}}
\put(3030,1005){\ellipse{46}{46}}
\put(2670,1860){\makebox(0,0)[lb]{\smash{{\SetFigFont{12}{14.4}{\rmdefault}{\mddefault}{\updefault}$z_{k-1}$}}}}
\put(3570,1860){\makebox(0,0)[lb]{\smash{{\SetFigFont{12}{14.4}{\rmdefault}{\mddefault}{\updefault}$z_k$}}}}
\put(3120,960){\makebox(0,0)[lb]{\smash{{\SetFigFont{12}{14.4}{\rmdefault}{\mddefault}{\updefault}$y_k$}}}}
\put(2670,60){\makebox(0,0)[lb]{\smash{{\SetFigFont{12}{14.4}{\rmdefault}{\mddefault}{\updefault}$x_k$}}}}
\put(3570,60){\makebox(0,0)[lb]{\smash{{\SetFigFont{12}{14.4}{\rmdefault}{\mddefault}{\updefault}$x_{k+1}$}}}}
\path(3480,105)(3480,1905) \path(3480,105)(3480,1905)
\path(3480,105)(3480,1140) \path(3480,105)(3480,1140)
\whiten\path(3510.000,1020.000)(3480.000,1140.000)(3450.000,1020.000)(3480.000,1056.000)(3510.000,1020.000)
\put(3570,1005){\makebox(0,0)[lb]{\smash{{\SetFigFont{10}{12.0}{\rmdefault}{\mddefault}{\updefault}\makebox[0.8\width][s]{$(\gamma_{n-1}-\gamma_n)$}}}}}
\put(2455,1410){\makebox(0,0)[lb]{\smash{{\SetFigFont{10}{12.0}{\rmdefault}{\mddefault}{\updefault}\makebox[0.8\width][s]{$\gamma_{n-1}$}}}}}
\path(3030,1005)(2580,1905) \path(3030,1005)(2580,1905)
\path(3030,1005)(2760,1545) \path(3030,1005)(2760,1545)
\whiten\path(2840.498,1451.085)(2760.000,1545.000)(2786.833,1424.252)(2829.765,1405.469)(2840.498,1451.085)
\put(735,1410){\makebox(0,0)[lb]{\smash{{\SetFigFont{10}{12.0}{\rmdefault}{\mddefault}{\updefault}$\gamma_n$}}}}
\put(1015,510){\makebox(0,0)[lb]{\smash{{\SetFigFont{10}{12.0}{\rmdefault}{\mddefault}{\updefault}\makebox[0.8\width][s]{$\gamma_{n-1}$}}}}}
\put(3030,510){\makebox(0,0)[lb]{\smash{{\SetFigFont{10}{12.0}{\rmdefault}{\mddefault}{\updefault}$\gamma_n$}}}}
\path(690,105)(690,1905) \path(690,105)(690,1905)
\path(690,105)(690,1140) \path(690,105)(690,1140)
\whiten\path(720.000,1020.000)(690.000,1140.000)(660.000,1020.000)(690.000,1056.000)(720.000,1020.000)
\put(2130,915){\makebox(0,0)[lb]{\smash{{\SetFigFont{14}{16.8}{\rmdefault}{\mddefault}{\updefault}=}}}}
\put(1970,-250){\makebox(0,0)[lb]{\smash{{\SetFigFont{14}{16.8}{\rmdefault}{\mddefault}{\updefault}fig.
3a}}}}
\end{picture}
} \hspace{\fill} \\ [17pt]
\end{figure}

\begin{figure}
\hspace{31mm} \setlength{\unitlength}{0.00087489in}
\begingroup\makeatletter\ifx\SetFigFont\undefined%
\gdef\SetFigFont#1#2#3#4#5{%
  \reset@font\fontsize{#1}{#2pt}%
  \fontfamily{#3}\fontseries{#4}\fontshape{#5}%
  \selectfont}%
\fi\endgroup%
{\renewcommand{\dashlinestretch}{30}
\begin{picture}(3905,1980)(0,-10)
\put(555,1005){\makebox(0,0)[rb]{\smash{{\SetFigFont{10}{12.0}{\rmdefault}{\mddefault}{\updefault}\makebox[0.8\width][s]{$(\gamma_{n-1}-\gamma_n)$}}}}}
\path(1140,1005)(690,1905) \path(1140,1005)(690,1905)
\path(1140,1005)(870,1545) \path(1140,1005)(870,1545)
\whiten\path(950.498,1451.085)(870.000,1545.000)(896.833,1424.252)(939.765,1405.469)(950.498,1451.085)
\path(2580,105)(2850,645)
\path(2823.167,524.252)(2850.000,645.000)(2769.502,551.085)
\path(2580,105)(3030,1005) \path(3480,105)(3030,1005)
\path(3480,105)(3030,1005) \path(3480,105)(3210,645)
\path(3480,105)(3210,645)
\whiten\path(3290.498,551.085)(3210.000,645.000)(3236.833,524.252)(3279.765,505.469)(3290.498,551.085)
\put(3030,1005){\blacken\ellipse{46}{46}}
\put(3030,1005){\ellipse{46}{46}}
\put(2455,1410){\makebox(0,0)[lb]{\smash{{\SetFigFont{10}{12.0}{\rmdefault}{\mddefault}{\updefault}$\gamma_{n-1}$}}}}
\path(3030,1005)(2580,1905) \path(3030,1005)(2580,1905)
\path(3030,1005)(2760,1545) \path(3030,1005)(2760,1545)
\whiten\path(2840.498,1451.085)(2760.000,1545.000)(2786.833,1424.252)(2829.765,1405.469)(2840.498,1451.085)
\put(735,1410){\makebox(0,0)[lb]{\smash{{\SetFigFont{10}{12.0}{\rmdefault}{\mddefault}{\updefault}$\gamma_n$}}}}
\put(1015,510){\makebox(0,0)[lb]{\smash{{\SetFigFont{10}{12.0}{\rmdefault}{\mddefault}{\updefault}$\gamma_{n-1}$}}}}
\put(3030,510){\makebox(0,0)[lb]{\smash{{\SetFigFont{10}{12.0}{\rmdefault}{\mddefault}{\updefault}$\gamma_n$}}}}
\path(690,105)(960,645)
\path(933.167,524.252)(960.000,645.000)(879.502,551.085)
\path(690,105)(1140,1005) \path(1590,105)(1140,1005)
\path(1590,105)(1140,1005) \path(1590,105)(1320,645)
\path(1590,105)(1320,645)
\whiten\path(1400.498,551.085)(1320.000,645.000)(1346.833,524.252)(1389.765,505.469)(1400.498,551.085)
\put(1140,1005){\blacken\ellipse{46}{46}}
\put(1140,1005){\ellipse{46}{46}}
\put(780,1860){\makebox(0,0)[lb]{\smash{{\SetFigFont{12}{14.4}{\rmdefault}{\mddefault}{\updefault}$z_{N-n}$}}}}
\put(2670,1860){\makebox(0,0)[lb]{\smash{{\SetFigFont{12}{14.4}{\rmdefault}{\mddefault}{\updefault}$z_{N-n}$}}}}
\put(1230,960){\makebox(0,0)[lb]{\smash{{\SetFigFont{12}{14.4}{\rmdefault}{\mddefault}{\updefault}$y_{N-n+1}$}}}}
\put(3120,960){\makebox(0,0)[lb]{\smash{{\SetFigFont{12}{14.4}{\rmdefault}{\mddefault}{\updefault}$y_{N-n+1}$}}}}
\put(2670,60){\makebox(0,0)[lb]{\smash{{\SetFigFont{12}{14.4}{\rmdefault}{\mddefault}{\updefault}$x_{N-n+1}$}}}}
\put(3570,60){\makebox(0,0)[lb]{\smash{{\SetFigFont{12}{14.4}{\rmdefault}{\mddefault}{\updefault}$x_{N-n+2}$}}}}
\put(1680,60){\makebox(0,0)[lb]{\smash{{\SetFigFont{12}{14.4}{\rmdefault}{\mddefault}{\updefault}$x_{N-n+2}$}}}}
\put(780,60){\makebox(0,0)[lb]{\smash{{\SetFigFont{12}{14.4}{\rmdefault}{\mddefault}{\updefault}$x_{N-n+1}$}}}}
\path(690,105)(690,1905) \path(690,105)(690,1905)
\path(690,105)(690,1140) \path(690,105)(690,1140)
\whiten\path(720.000,1020.000)(690.000,1140.000)(660.000,1020.000)(690.000,1056.000)(720.000,1020.000)
\put(2030,915){\makebox(0,0)[lb]{\smash{{\SetFigFont{14}{16.8}{\rmdefault}{\mddefault}{\updefault}=}}}}
\put(1870,-250){\makebox(0,0)[lb]{\smash{{\SetFigFont{14}{14.4}{\rmdefault}{\mddefault}{\updefault}fig.
3b}}}}
\end{picture}
} \\ [17pt]
\end{figure}

\begin{figure}
\hspace{35mm} \setlength{\unitlength}{0.00087489in}
\begingroup\makeatletter\ifx\SetFigFont\undefined%
\gdef\SetFigFont#1#2#3#4#5{%
  \reset@font\fontsize{#1}{#2pt}%
  \fontfamily{#3}\fontseries{#4}\fontshape{#5}%
  \selectfont}%
\fi\endgroup%
{\renewcommand{\dashlinestretch}{30}
\begin{picture}(4411,1980)(0,-10)
\path(510,105)(780,645)
\path(753.167,524.252)(780.000,645.000)(699.502,551.085)
\path(510,105)(960,1005) \path(1410,105)(960,1005)
\path(1410,105)(960,1005) \path(1410,105)(1140,645)
\path(1410,105)(1140,645)
\whiten\path(1220.498,551.085)(1140.000,645.000)(1166.833,524.252)(1209.765,505.469)(1220.498,551.085)
\put(960,1005){\blacken\ellipse{46}{46}}
\put(960,1005){\ellipse{46}{46}} \path(960,1005)(1230,1545)
\path(1203.167,1424.252)(1230.000,1545.000)(1149.502,1451.085)
\path(960,1005)(1410,1905)
\put(555,1410){\makebox(0,0)[lb]{\smash{{\SetFigFont{10}{12.0}{\rmdefault}{\mddefault}{\updefault}$\gamma_n$}}}}
\put(835,510){\makebox(0,0)[lb]{\smash{{\SetFigFont{10}{12.0}{\rmdefault}{\mddefault}{\updefault}$\gamma_{n-1}$}}}}
\path(960,1005)(690,1545) \path(960,1005)(690,1545)
\whiten\path(770.498,1451.085)(690.000,1545.000)(716.833,1424.252)(770.498,1451.085)
\path(510,105)(240,645) \path(510,105)(240,645)
\whiten\path(320.498,551.085)(240.000,645.000)(266.833,524.252)(320.498,551.085)
\put(-65,510){\makebox(0,0)[lb]{\smash{{\SetFigFont{10}{12.0}{\rmdefault}{\mddefault}{\updefault}$\gamma_{n-1}$}}}}
\path(3255,1005)(3525,1545)
\path(3498.167,1424.252)(3525.000,1545.000)(3444.502,1451.085)
\path(3255,1005)(3705,1905) \path(2805,105)(3075,645)
\path(3048.167,524.252)(3075.000,645.000)(2994.502,551.085)
\path(2805,105)(3255,1005) \path(3705,105)(3255,1005)
\path(3705,105)(3255,1005) \path(3705,105)(3435,645)
\path(3705,105)(3435,645)
\whiten\path(3515.498,551.085)(3435.000,645.000)(3461.833,524.252)(3504.765,505.469)(3515.498,551.085)
\put(3255,1005){\blacken\ellipse{46}{46}}
\put(3255,1005){\ellipse{46}{46}} \path(3705,105)(3705,1905)
\path(3705,105)(3705,1905) \path(3705,105)(3705,1140)
\path(3705,105)(3705,1140)
\whiten\path(3735.000,1020.000)(3705.000,1140.000)(3675.000,1020.000)(3705.000,1056.000)(3735.000,1020.000)
\put(3795,1005){\makebox(0,0)[lb]{\smash{{\SetFigFont{10}{12.0}{\rmdefault}{\mddefault}{\updefault}\makebox[0.8\width][s]{$(\gamma_{n-1}-\gamma_n)$}}}}}
\put(2680,1410){\makebox(0,0)[lb]{\smash{{\SetFigFont{10}{12.0}{\rmdefault}{\mddefault}{\updefault}$\gamma_{n-1}$}}}}
\put(3255,510){\makebox(0,0)[lb]{\smash{{\SetFigFont{10}{12.0}{\rmdefault}{\mddefault}{\updefault}$\gamma_n$}}}}
\path(3255,1005)(2985,1545) \path(3255,1005)(2985,1545)
\whiten\path(3065.498,1451.085)(2985.000,1545.000)(3011.833,1424.252)(3065.498,1451.085)
\path(2805,105)(2535,645) \path(2805,105)(2535,645)
\whiten\path(2615.498,551.085)(2535.000,645.000)(2561.833,524.252)(2615.498,551.085)
\put(2355,510){\makebox(0,0)[lb]{\smash{{\SetFigFont{10}{12.0}{\rmdefault}{\mddefault}{\updefault}$\gamma_n$}}}}
\put(600,60){\makebox(0,0)[lb]{\smash{{\SetFigFont{12}{14.4}{\rmdefault}{\mddefault}{\updefault}$x_1$}}}}
\put(1500,60){\makebox(0,0)[lb]{\smash{{\SetFigFont{12}{14.4}{\rmdefault}{\mddefault}{\updefault}$x_2$}}}}
\put(1050,960){\makebox(0,0)[lb]{\smash{{\SetFigFont{12}{14.4}{\rmdefault}{\mddefault}{\updefault}$y_1$}}}}
\put(1500,1860){\makebox(0,0)[lb]{\smash{{\SetFigFont{12}{14.4}{\rmdefault}{\mddefault}{\updefault}$z_1$}}}}
\put(1905,915){\makebox(0,0)[lb]{\smash{{\SetFigFont{14}{16.8}{\rmdefault}{\mddefault}{\updefault}=}}}}
\put(2895,60){\makebox(0,0)[lb]{\smash{{\SetFigFont{12}{14.4}{\rmdefault}{\mddefault}{\updefault}$x_1$}}}}
\put(3795,60){\makebox(0,0)[lb]{\smash{{\SetFigFont{12}{14.4}{\rmdefault}{\mddefault}{\updefault}$x_2$}}}}
\put(3345,960){\makebox(0,0)[lb]{\smash{{\SetFigFont{12}{14.4}{\rmdefault}{\mddefault}{\updefault}$y_1$}}}}
\put(3795,1860){\makebox(0,0)[lb]{\smash{{\SetFigFont{12}{14.4}{\rmdefault}{\mddefault}{\updefault}$z_1$}}}}
\put(1745,-250){\makebox(0,0)[lb]{\smash{{\SetFigFont{14}{16.8}{\rmdefault}{\mddefault}{\updefault}fig.
3c}}}}
\end{picture}
} \\ [17pt]
\end{figure}

\noindent{\bfseries Proof.} Integration over $y_k$ in the left
hand side of fig.~3a yields
\begin{equation}
\begin{split} \label{A_rhsf3a}
 \int\limits_{-\infty}^{+\infty} dy_k I(x_k,y_k) J_{\gamma_{n-1}}(x_{k+1},y_k) I(y_k,z_k) & J_{\gamma_n}(y_k,z_{k-1})
 =e^{\frac{i}{\hbar}(\gamma_{n-1}x_{k+1}-\gamma_nz_{k-1})}\times \\
  \times\int\limits_{-\infty}^{+\infty} dy_k
  \exp\Bigl\{\frac{i}{\hbar}(\gamma_n-\gamma_{n-1})y_k-\frac1{\hbar}(e^{-x_{k+1}}&+e^{-z_k})e^{y_k}
             -\frac1{\hbar}(e^{x_k}+e^{z_{k-1}})e^{-y_k}\Bigr\}= \\
 =2e^{\frac{i}{\hbar}(\gamma_{n-1}x_{k+1}-\gamma_nz_{k-1})} & \left( \frac{e^{x_k}+e^{z_{k-1}}}{e^{-x_{k+1}}+e^{-z_k}}
 \right)^{\frac{i(\gamma_n-\gamma_{n-1})}{2\hbar}} \times \\
 \times K_{\frac{i}{\hbar}(\gamma_n-\gamma_{n-1})} & \left(\frac{2}{\hbar}
                                  \sqrt{(e^{x_k}+e^{z_{k-1}})(e^{-x_{k+1}}+e^{-z_k})}\right),
\end{split}
\end{equation}
where $K_\nu (z)$ is a Macdonald function~\cite{BE2}.
Interchanging $\gamma_{n-1}$ and $\gamma_n$ in~\eqref{A_rhsf3a} we
obtain the expression for the integral in the right hand side of
fig.~3a
\begin{equation} \label{A_lhsf3a}
\begin{split}
 2e^{\frac{i}{\hbar}(\gamma_n x_{k+1}-\gamma_{n-1}z_{k-1})}&\left( \frac{e^{x_k}+e^{z_{k-1}}}{e^{-x_{k+1}}+e^{-z_k}}
 \right)^{-\frac{i(\gamma_n-\gamma_{n-1})}{2\hbar}} \times \\
 \times K_{\frac{i}{\hbar}(\gamma_n-\gamma_{n-1})} & \left(\frac{2}{\hbar}
                                  \sqrt{(e^{x_k}+e^{z_{k-1}})(e^{-x_{k+1}}+e^{-z_k})}\right).
\end{split}
\end{equation}
The ratio of~\eqref{A_rhsf3a} and \eqref{A_lhsf3a} is exactly equal to
$Y_{\gamma_{n-1}-\gamma_n}(x_{k+1},z_k)Y_{\gamma_{n-1}-\gamma_n}^{-1}(x_k,z_{k-1})$. \\

The equalities shown in the fig.~3b and 3c can be proved analogously. \qed \\

Let us continue the proof of the theorem~\ref{A_Th_psi}. It is shown
in fig.~4. The left hand side of~\eqref{A_LgammaLgamma} after
application of fig.~3c is reflected in this diagram. Then, using
the fig.~3a, one can move the vertical line picturing the function
$Y_{\gamma_{n-1}-\gamma_n}(x_j,z_{j-1})$ to the right, as shown in
the figure. When this line has arrived to the right one can apply
the fig. 3b. In each step the parameters $\gamma_{n-1}$ and
$\gamma_n$ are interchanged, and,
eventually, one has the right hand side of~\eqref{A_LgammaLgamma}. \qed \\

\begin{figure}[h] 
\hspace{\fill}\setlength{\unitlength}{0.00087489in}
\begingroup\makeatletter\ifx\SetFigFont\undefined%
\gdef\SetFigFont#1#2#3#4#5{%
  \reset@font\fontsize{#1}{#2pt}%
  \fontfamily{#3}\fontseries{#4}\fontshape{#5}%
  \selectfont}%
\fi\endgroup%
{\renewcommand{\dashlinestretch}{30}
\begin{picture}(5885,2497)(0,-10)
\path(1320,622)(870,1522) \path(1320,622)(870,1522)
\path(1320,622)(1050,1162) \path(1320,622)(1050,1162)
\whiten\path(1130.498,1068.085)(1050.000,1162.000)(1076.833,1041.252)(1119.765,1022.469)(1130.498,1068.085)
\path(420,622)(690,1162)
\path(663.167,1041.252)(690.000,1162.000)(609.502,1068.085)
\path(420,622)(870,1522)
\put(1655,1027){\makebox(0,0)[lb]{\smash{{\SetFigFont{10}{12.0}{\rmdefault}{\mddefault}{\updefault}$\gamma_{n-1}$}}}}
\put(05,1027){\makebox(0,0)[lb]{\smash{{\SetFigFont{10}{12.0}{\rmdefault}{\mddefault}{\updefault}$\gamma_n$}}}}
\put(905,1027){\makebox(0,0)[lb]{\smash{{\SetFigFont{10}{12.0}{\rmdefault}{\mddefault}{\updefault}$\gamma_n$}}}}
\put(1410,2377){\makebox(0,0)[lb]{\smash{{\SetFigFont{12}{14.4}{\rmdefault}{\mddefault}{\updefault}$z_1$}}}}
\put(1350,1972){\makebox(0,0)[lb]{\smash{{\SetFigFont{10}{12.0}{\rmdefault}{\mddefault}{\updefault}$\gamma_n$}}}}
\put(295,1972){\makebox(0,0)[lb]{\smash{{\SetFigFont{10}{12.0}{\rmdefault}{\mddefault}{\updefault}$\gamma_{n-1}$}}}}
\path(4155,1522)(3705,2422) \path(4155,1522)(3705,2422)
\path(4155,1522)(3885,2062) \path(4155,1522)(3885,2062)
\whiten\path(3965.498,1968.085)(3885.000,2062.000)(3911.833,1941.252)(3954.765,1922.469)(3965.498,1968.085)
\path(3705,622)(3975,1162)
\path(3948.167,1041.252)(3975.000,1162.000)(3894.502,1068.085)
\path(3705,622)(4155,1522)
\put(4040,1027){\makebox(0,0)[lb]{\smash{{\SetFigFont{10}{12.0}{\rmdefault}{\mddefault}{\updefault}$\gamma_{n-1}$}}}}
\put(4940,1027){\makebox(0,0)[lb]{\smash{{\SetFigFont{10}{12.0}{\rmdefault}{\mddefault}{\updefault}$\gamma_{n-1}$}}}}
\put(3705,1972){\makebox(0,0)[lb]{\smash{{\SetFigFont{10}{12.0}{\rmdefault}{\mddefault}{\updefault}$\gamma_n$}}}}
\put(4610,1972){\makebox(0,0)[lb]{\smash{{\SetFigFont{10}{12.0}{\rmdefault}{\mddefault}{\updefault}$\gamma_n$}}}}

\put(1020,1702){\makebox(0,0)[lb]{\smash{{\SetFigFont{10}{12.0}{\rmdefault}{\mddefault}{\updefault}\makebox[0.7\width][s]{$(\gamma_{n-1}-\gamma_n)$}}}}}
\put(870,1522){\blacken\ellipse{46}{46}}
\put(870,1522){\ellipse{46}{46}}
\put(1770,1522){\blacken\ellipse{46}{46}}
\put(1770,1522){\ellipse{46}{46}}
\put(4155,1522){\blacken\ellipse{46}{46}}
\put(4155,1522){\ellipse{46}{46}}
\put(5055,1522){\blacken\ellipse{46}{46}}
\put(5055,1522){\ellipse{46}{46}} \path(870,1522)(600,2062)
\path(870,1522)(600,2062)
\whiten\path(680.498,1968.085)(600.000,2062.000)(626.833,1941.252)(680.498,1968.085)
\path(420,622)(150,1162) \path(420,622)(150,1162)
\whiten\path(230.498,1068.085)(150.000,1162.000)(176.833,1041.252)(230.498,1068.085)
\path(1320,1747)(1320,2422) \path(1320,1747)(1320,2422)
\path(1320,622)(1320,1657) \path(1320,622)(1320,1657)
\whiten\path(1350.000,1537.000)(1320.000,1657.000)(1290.000,1537.000)(1320.000,1573.000)(1350.000,1537.000)
\path(870,1522)(1140,2062)
\path(1113.167,1941.252)(1140.000,2062.000)(1059.502,1968.085)
\path(870,1522)(1320,2422) \path(1635,1792)(1320,2422)
\path(1635,1792)(1320,2422) \path(1635,1792)(1500,2062)
\path(1635,1792)(1500,2062)
\whiten\path(1580.498,1968.085)(1500.000,2062.000)(1526.833,1941.252)(1569.765,1922.469)(1580.498,1968.085)
\path(2220,622)(1710,1642) \path(2220,622)(1710,1642)
\path(2220,622)(1950,1162) \path(2220,622)(1950,1162)
\whiten\path(2030.498,1068.085)(1950.000,1162.000)(1976.833,1041.252)(2019.765,1022.469)(2030.498,1068.085)
\path(1320,622)(1590,1162)
\path(1563.167,1041.252)(1590.000,1162.000)(1509.502,1068.085)
\path(1320,622)(1770,1522) \path(1770,1522)(2040,2062)
\path(2013.167,1941.252)(2040.000,2062.000)(1959.502,1968.085)
\path(1770,1522)(2220,2422) \path(1770,1522)(2220,2422)
\path(5505,622)(5055,1522) \path(5505,622)(5055,1522)
\path(5505,622)(5235,1162) \path(5505,622)(5235,1162)
\whiten\path(5315.498,1068.085)(5235.000,1162.000)(5261.833,1041.252)(5304.765,1022.469)(5315.498,1068.085)
\path(4605,622)(4875,1162)
\path(4848.167,1041.252)(4875.000,1162.000)(4794.502,1068.085)
\path(4605,622)(5055,1522) \path(5055,1522)(4605,2422)
\path(5055,1522)(4605,2422) \path(5055,1522)(4785,2062)
\path(5055,1522)(4785,2062)
\whiten\path(4865.498,1968.085)(4785.000,2062.000)(4811.833,1941.252)(4854.765,1922.469)(4865.498,1968.085)
\path(4605,622)(4155,1522) \path(4605,622)(4155,1522)
\path(4605,622)(4335,1162) \path(4605,622)(4335,1162)
\whiten\path(4415.498,1068.085)(4335.000,1162.000)(4361.833,1041.252)(4404.765,1022.469)(4415.498,1068.085)
\path(4155,1522)(4425,2062)
\path(4398.167,1941.252)(4425.000,2062.000)(4344.502,1968.085)
\path(4155,1522)(4605,2422)
\dashline{60.000}(1365,1297)(4515,1297)
\blacken\path(4395.000,1267.000)(4515.000,1297.000)(4395.000,1327.000)(4395.000,1267.000)
\put(960,1477){\makebox(0,0)[lb]{\smash{{\SetFigFont{12}{14.4}{\rmdefault}{\mddefault}{\updefault}$y_1$}}}}
\put(1860,1477){\makebox(0,0)[lb]{\smash{{\SetFigFont{12}{14.4}{\rmdefault}{\mddefault}{\updefault}$y_2$}}}}
\put(510,577){\makebox(0,0)[lb]{\smash{{\SetFigFont{12}{14.4}{\rmdefault}{\mddefault}{\updefault}$x_1$}}}}
\put(2310,577){\makebox(0,0)[lb]{\smash{{\SetFigFont{12}{14.4}{\rmdefault}{\mddefault}{\updefault}$x_3$}}}}
\put(1410,577){\makebox(0,0)[lb]{\smash{{\SetFigFont{12}{14.4}{\rmdefault}{\mddefault}{\updefault}$x_2$}}}}
\put(2310,2377){\makebox(0,0)[lb]{\smash{{\SetFigFont{12}{14.4}{\rmdefault}{\mddefault}{\updefault}$z_2$}}}}
\put(2895,892){\makebox(0,0)[lb]{\smash{{\SetFigFont{20}{24.0}{\rmdefault}{\mddefault}{\updefault}...}}}}
\put(2715,82){\makebox(0,0)[lb]{\smash{{\SetFigFont{14}{16.8}{\rmdefault}{\mddefault}{\updefault}fig.
4}}}}
\put(4695,577){\makebox(0,0)[lb]{\smash{{\SetFigFont{12}{14.4}{\rmdefault}{\mddefault}{\updefault}$x_{N-n+1}$}}}}
\put(4245,1477){\makebox(0,0)[lb]{\smash{{\SetFigFont{12}{14.4}{\rmdefault}{\mddefault}{\updefault}$y_{N-n}$}}}}
\put(5145,1477){\makebox(0,0)[lb]{\smash{{\SetFigFont{12}{14.4}{\rmdefault}{\mddefault}{\updefault}$y_{N-n+1}$}}}}
\put(4695,2377){\makebox(0,0)[lb]{\smash{{\SetFigFont{12}{14.4}{\rmdefault}{\mddefault}{\updefault}$z_{N-n}$}}}}
\put(3795,2377){\makebox(0,0)[lb]{\smash{{\SetFigFont{12}{14.4}{\rmdefault}{\mddefault}{\updefault}$z_{N-n-1}$}}}}
\put(3795,577){\makebox(0,0)[lb]{\smash{{\SetFigFont{12}{14.4}{\rmdefault}{\mddefault}{\updefault}$x_{N-n}$}}}}
\put(5550,577){\makebox(0,0)[lb]{\smash{{\SetFigFont{12}{14.4}{\rmdefault}{\mddefault}{\updefault}$x_{N-n+2}$}}}}
\end{picture}
}\hspace{\fill} \\
\end{figure}

Substituting the explicit expression~\eqref{A_Lamb} for the kernel of
the operator~\eqref{A_Lambf} one obtains the recurrent formula for the function~\eqref{A_psi}
\begin{equation} \label{A_psi_rec}
 \begin{split}
 \psi_{\gamma_1,\ldots,\gamma_N}(x_1,\ldots,x_N)=\int\limits_{\mathbb R^{N-1}}dy_1\ldots dy_{N-1}\,
  \psi_{\gamma_1,\ldots,\gamma_{N-1}}(y_1,\ldots,y_{N-1}) \times \\
 \times\exp\Bigl\{\frac{i}{\hbar}\gamma_N(\sum\limits_{n=1}^N x_n-\sum\limits_{n=1}^{N-1}y_n)
        -\frac{1}{\hbar}\sum\limits_{n=1}^{N-1}(e^{y_n-x_{n+1}}+e^{x_n-y_n})\Bigr\}.
 \end{split}
\end{equation}
The consecutive applications of this formula allow to derive the following integral representation for the
eigenfunctions of open Toda chain model
\begin{multline} \label{A_psi_GG}
 \psi_{\gamma}(z_{N1},\ldots,z_{NN})=\int\limits_{\mathbb R^{\frac{N(N-1)}2}}\prod_{n=1}^{N-1}\prod_{j=1}^{n}dz_{nj}\,
  \exp\biggl\{\frac{i}{\hbar}\Bigl[\gamma_N\sum\limits_{j=1}^N z_{Nj}
    +\sum_{n=1}^{N-1}(\gamma_n-\gamma_{n+1})\sum_{j=1}^n z_{nj}\Bigr]-\\
    -\frac{1}{\hbar}\sum_{n=1}^N\sum_{j=1}^{n-1}\bigl(e^{z_{nj}-z_{n-1,j}}
    +e^{z_{n-1,j}-z_{n,j+1}}\bigr)\biggr\}.
\end{multline}
This is a Gauss-Givental representation~\cite{Kharchev_GG},
\cite{Givental} of the Toda chain transition function.

\section{Integration measure}
\label{A_Int_meas}

In this section we shall check the orthogonality condition using
the diagram technique. The normalization function which appears in
this calculation coincides exactly with the Sklyanin integration
measure using in the SoV method for the periodic Toda
chain model~\cite{Sklyanin},
\cite{Kharchev_P},\cite{Kharchev_O}, \cite{Kharchev_OP}. \\

\begin{theor} \label{A_orthog_meas}
 The functions $\psi_{\gamma}(x)$ defined by the formula~\eqref{A_psi} satisfy to the
orthogonality condition
\begin{equation} \label{A_orthog}
 \int\limits_{\mathbb R^N}dx \overline{\psi_{\gamma}(x)}\psi_{\gamma'}(x)=\mu^{-1}(\gamma)\delta_{SYM}(\gamma,\gamma'),
\end{equation}
where $\delta_{SYM}(\gamma,\gamma')=
\dfrac1{N!}\sum\limits_{\sigma\in S_N}
\prod\limits_{i=1}^N\delta(\gamma_i-\gamma'_{\sigma(i)})$ is a
symmetrized delta function and $\mu(\gamma)$ is the Sklyanin
measure
 \begin{equation} \label{A_measure}
  \mu(\gamma)=\frac{(2\pi\hbar)^{-N}}{N!}\prod_{k<m}
  \left[\Gamma\Bigl(\frac{\gamma_m-\gamma_k}{i\hbar}\Bigr)
        \Gamma\Bigl(\frac{\gamma_k-\gamma_m}{i\hbar}\Bigr)\right]^{-1}.
 \end{equation}
\end{theor}

Since $H_k$ are Hermitian operators, the set of the functions
$\psi_{\gamma}(x)$ is complete \cite{Gelfand4}. This means that
any function $f(x)$ belonging to the Hilbert space $L^2(\mathbb
R^N)$ can be represented as an integral
\begin{equation} \label{A_f_psi}
 f(x)=\int\limits_{\mathbb R^N}\psi_{\gamma}(x) g(\gamma)\mu(\gamma)\,d\gamma
\end{equation}
for some function $g(\gamma)$ with some good properties.
As a consequence we have the
completeness condition
\begin{equation} \label{A_compl}
 \int\limits_{\mathbb R^N}d\gamma\,\mu(\gamma)\overline{\psi_{\gamma}(x)}\psi_{\gamma}(x')=\prod_{i=1}^N\delta(x_i-x'_i),
\end{equation}.

{\bfseries Proof of the theorem~\ref{A_orthog_meas}.} First of all,
we need to obtain a diagram representation of
$\overline{\psi_{\gamma}(x)}$ in order to calculate the integral
in the left hand side of~\eqref{A_orthog} using the diagram
technique. Let us to consider the diagram for $\psi_{\gamma}(x)$,
which are shown in fig.~2 for $N=5$ and in fig.~6a for $N=3$, and
to implement the following steps.

{\bfseries First step.} The imaginary unit is contained only in the
functions $J$ and $Z$. To reduce the conjugation of whole function
$\psi_{\gamma}(x)$ to the conjugation of the functions $Z$ we
decompose $J_u(x_{k+1},y_k)$ into the product
$Z_u^{-1}(y_k)I(y_k,x_{k+1})Z_u(x_{k+1})$ (fig.~5a). This
corresponds to the transition from the fig.~6a to the fig.~6b.

{\bfseries Second step.} We replace all $Z_u(x_k)$ by
$Z_{u}^{-1}(x_k)$ (fig.~5b) implementing the complex conjugation
and arrive to the fig.~6c, in which the function
$\overline{\psi_{\gamma}(x)}$ is pictured. \\ [8pt]

\begin{figure}[h]
\noindent\hspace{\fill} \setlength{\unitlength}{0.00087489in}
\begingroup\makeatletter\ifx\SetFigFont\undefined%
\gdef\SetFigFont#1#2#3#4#5{%
\reset@font\fontsize{#1}{#2pt}%
  \fontfamily{#3}\fontseries{#4}\fontshape{#5}%
  \selectfont}%
\fi\endgroup%
{\renewcommand{\dashlinestretch}{30}
\begin{picture}(1636,1426)(0,-10)
\put(555,388){\makebox(0,0)[lb]{\smash{{\SetFigFont{12}{14.4}{\rmdefault}{\mddefault}{\updefault}$x_{k+1}$}}}}
\path(915,1333)(1185,793)
\path(1104.502,886.915)(1185.000,793.000)(1158.167,913.748)
\path(915,1333)(1365,433)
\put(1455,388){\makebox(0,0)[lb]{\smash{{\SetFigFont{12}{14.4}{\rmdefault}{\mddefault}{\updefault}$x_{k+1}$}}}}
\path(1050,1063)(1005,1153)
\whiten\path(1085.498,1059.085)(1005.000,1153.000)(1031.833,1032.252)(1085.498,1059.085)
\path(1230,703)(1185,793)
\whiten\path(1265.498,699.085)(1185.000,793.000)(1211.833,672.252)(1265.498,699.085)
\put(870,1018){\makebox(0,0)[lb]{\smash{{\SetFigFont{10}{14.4}{\rmdefault}{\mddefault}{\updefault}$u$}}}}
\put(1005,658){\makebox(0,0)[lb]{\smash{{\SetFigFont{10}{14.4}{\rmdefault}{\mddefault}{\updefault}$u$}}}}
\put(1005,1288){\makebox(0,0)[lb]{\smash{{\SetFigFont{12}{14.4}{\rmdefault}{\mddefault}{\updefault}$y_k$}}}}
\path(465,433)(15,1333) \path(465,433)(15,1333)
\path(465,433)(195,973) \path(465,433)(195,973)
\whiten\path(275.498,879.085)(195.000,973.000)(221.833,852.252)(264.765,833.469)(275.498,879.085)
\path(510,883)(825,883)
\blacken\path(705.000,853.000)(825.000,883.000)(705.000,913.000)(741.000,883.000)(705.000,853.000)
\put(15,883){\makebox(0,0)[lb]{\smash{{\SetFigFont{10}{14.4}{\rmdefault}{\mddefault}{\updefault}$u$}}}}
\put(105,1288){\makebox(0,0)[lb]{\smash{{\SetFigFont{12}{14.4}{\rmdefault}{\mddefault}{\updefault}$y_k$}}}}
\put(455,73){\makebox(0,0)[lb]{\smash{{\SetFigFont{14}{14.4}{\rmdefault}{\mddefault}{\updefault}fig.
5a}}}}
\end{picture}
} \hspace{\fill} \setlength{\unitlength}{0.00087489in}
\begingroup\makeatletter\ifx\SetFigFont\undefined%
\gdef\SetFigFont#1#2#3#4#5{%
  \reset@font\fontsize{#1}{#2pt}%
  \fontfamily{#3}\fontseries{#4}\fontshape{#5}%
  \selectfont}%
\fi\endgroup%
{\renewcommand{\dashlinestretch}{30}
\begin{picture}(1197,1000)(0,-10)
\path(150,973)(330,613) \path(1185,613)(1005,973)
\put(870,748){\makebox(0,0)[lb]{\smash{{\SetFigFont{10}{14.4}{\rmdefault}{\mddefault}{\updefault}$u$}}}}
\path(420,838)(735,838)
\blacken\path(615.000,808.000)(735.000,838.000)(615.000,868.000)(651.000,838.000)(615.000,808.000)
\put(15,748){\makebox(0,0)[lb]{\smash{{\SetFigFont{10}{14.4}{\rmdefault}{\mddefault}{\updefault}$u$}}}}
\put(365,73){\makebox(0,0)[lb]{\smash{{\SetFigFont{14}{14.4}{\rmdefault}{\mddefault}{\updefault}fig.
5b}}}} \path(240,793)(195,883)
\whiten\path(275.498,789.085)(195.000,883.000)(221.833,762.252)(275.498,789.085)
\path(1095,793)(1140,703)
\whiten\path(1059.502,796.915)(1140.000,703.000)(1113.167,823.748)(1059.502,796.915)
\end{picture}
} \hspace{\fill} \\ [10pt]

{\bfseries Third step.} Since the functions $Z_{u}^{-1}(x_k)$ are
attached to only one point, namely to $x_k$, one can turn these
functions in the manner shown in fig.~5c and fig.~5d. It means
that we can represent $\overline{\psi_{\gamma}(x)}$ by the
fig.~6d.

{\bfseries Fourth step.} Now we replace the product
$Z_u^{-1}(x_k)I(x_k,y_k)Z_u(y_k)$ by $J_u(y_k,x_k)$ (fig.~5e) to
obtain fig.~6e. \\ [8pt]

\noindent\hspace{\fill} \setlength{\unitlength}{0.00087489in}
\begingroup\makeatletter\ifx\SetFigFont\undefined%
\gdef\SetFigFont#1#2#3#4#5{%
  \reset@font\fontsize{#1}{#2pt}%
  \fontfamily{#3}\fontseries{#4}\fontshape{#5}%
  \selectfont}%
\fi\endgroup%
{\renewcommand{\dashlinestretch}{30}
\begin{picture}(1519,1138)(0,-10)
\put(303,1018){\makebox(0,0)[rb]{\smash{{\SetFigFont{12}{14.4}{\rmdefault}{\mddefault}{\updefault}$x_k$}}}}
\put(1248,1041){\blacken\ellipse{46}{46}}
\put(1248,1041){\ellipse{46}{46}} \path(1023,613)(1248,1063)
\put(1338,1018){\makebox(0,0)[lb]{\smash{{\SetFigFont{12}{14.4}{\rmdefault}{\mddefault}{\updefault}$x_k$}}}}
\put(933,748){\makebox(0,0)[lb]{\smash{{\SetFigFont{10}{14.4}{\rmdefault}{\mddefault}{\updefault}$u$}}}}
\put(78,1041){\blacken\ellipse{46}{46}}
\put(78,1041){\ellipse{46}{46}} \path(303,613)(78,1063)
\path(438,838)(753,838)
\blacken\path(633.000,808.000)(753.000,838.000)(633.000,868.000)(669.000,838.000)(633.000,808.000)
\put(123,748){\makebox(0,0)[rb]{\smash{{\SetFigFont{10}{14.4}{\rmdefault}{\mddefault}{\updefault}$u$}}}}
\put(473,73){\makebox(0,0)[lb]{\smash{{\SetFigFont{14}{14.4}{\rmdefault}{\mddefault}{\updefault}fig.
5c}}}} \path(1113,793)(1068,703)
\whiten\path(1094.833,823.748)(1068.000,703.000)(1148.498,796.915)(1094.833,823.748)
\path(213,793)(258,703)
\whiten\path(177.502,796.915)(258.000,703.000)(231.167,823.748)(177.502,796.915)
\end{picture}
} \hspace{\fill} \setlength{\unitlength}{0.00087489in}
\begingroup\makeatletter\ifx\SetFigFont\undefined%
\gdef\SetFigFont#1#2#3#4#5{%
  \reset@font\fontsize{#1}{#2pt}%
  \fontfamily{#3}\fontseries{#4}\fontshape{#5}%
  \selectfont}%
\fi\endgroup%
{\renewcommand{\dashlinestretch}{30}
\begin{picture}(1339,1045)(0,-10)
\put(573,523){\makebox(0,0)[rb]{\smash{{\SetFigFont{12}{14.4}{\rmdefault}{\mddefault}{\updefault}$x_k$}}}}
\put(348,568){\blacken\ellipse{46}{46}}
\put(348,568){\ellipse{46}{46}}
\put(1068,568){\blacken\ellipse{46}{46}}
\put(1068,568){\ellipse{46}{46}} \path(528,793)(843,793)
\blacken\path(723.000,763.000)(843.000,793.000)(723.000,823.000)(759.000,793.000)(723.000,763.000)
\path(348,568)(123,1018) \path(1068,568)(1293,1018)
\put(428,73){\makebox(0,0)[lb]{\smash{{\SetFigFont{14}{14.4}{\rmdefault}{\mddefault}{\updefault}fig.
5d}}}}
\put(123,793){\makebox(0,0)[rb]{\smash{{\SetFigFont{10}{14.4}{\rmdefault}{\mddefault}{\updefault}$u$}}}}
\put(1023,793){\makebox(0,0)[lb]{\smash{{\SetFigFont{10}{14.4}{\rmdefault}{\mddefault}{\updefault}$u$}}}}
\put(1158,523){\makebox(0,0)[lb]{\smash{{\SetFigFont{12}{14.4}{\rmdefault}{\mddefault}{\updefault}$x_k$}}}}
\path(213,838)(258,748)
\whiten\path(177.502,841.915)(258.000,748.000)(231.167,868.748)(177.502,841.915)
\path(1203,838)(1158,748)
\whiten\path(1184.833,868.748)(1158.000,748.000)(1238.498,841.915)(1184.833,868.748)
\end{picture}
} \hspace{\fill} \setlength{\unitlength}{0.00087489in}
\begingroup\makeatletter\ifx\SetFigFont\undefined%
\gdef\SetFigFont#1#2#3#4#5{%
  \reset@font\fontsize{#1}{#2pt}%
  \fontfamily{#3}\fontseries{#4}\fontshape{#5}%
  \selectfont}%
\fi\endgroup%
{\renewcommand{\dashlinestretch}{30}
\begin{picture}(1726,1426)(0,-10)
\path(60,433)(330,973)
\path(303.167,852.252)(330.000,973.000)(249.502,879.085)
\path(60,433)(510,1333) \path(555,883)(870,883)
\blacken\path(750.000,853.000)(870.000,883.000)(750.000,913.000)(786.000,883.000)(750.000,853.000)
\path(195,703)(150,613)
\whiten\path(176.833,733.748)(150.000,613.000)(230.498,706.915)(176.833,733.748)
\path(375,1063)(330,973)
\whiten\path(356.833,1093.748)(330.000,973.000)(410.498,1066.915)(356.833,1093.748)
\path(1455,1333)(1005,433) \path(1455,1333)(1005,433)
\path(1455,1333)(1185,793) \path(1455,1333)(1185,793)
\whiten\path(1211.833,913.748)(1185.000,793.000)(1265.498,886.915)(1254.765,932.531)(1211.833,913.748)
\put(15,658){\makebox(0,0)[lb]{\smash{{\SetFigFont{10}{14.4}{\rmdefault}{\mddefault}{\updefault}$u$}}}}
\put(150,1018){\makebox(0,0)[lb]{\smash{{\SetFigFont{10}{14.4}{\rmdefault}{\mddefault}{\updefault}$u$}}}}
\put(1005,833){\makebox(0,0)[lb]{\smash{{\SetFigFont{10}{14.4}{\rmdefault}{\mddefault}{\updefault}$u$}}}}
\put(600,1288){\makebox(0,0)[lb]{\smash{{\SetFigFont{12}{14.4}{\rmdefault}{\mddefault}{\updefault}$y_k$}}}}
\put(1545,1288){\makebox(0,0)[lb]{\smash{{\SetFigFont{12}{14.4}{\rmdefault}{\mddefault}{\updefault}$y_k$}}}}
\put(150,388){\makebox(0,0)[lb]{\smash{{\SetFigFont{12}{14.4}{\rmdefault}{\mddefault}{\updefault}$x_k$}}}}
\put(1095,388){\makebox(0,0)[lb]{\smash{{\SetFigFont{12}{14.4}{\rmdefault}{\mddefault}{\updefault}$x_k$}}}}
\put(500,73){\makebox(0,0)[lb]{\smash{{\SetFigFont{14}{14.4}{\rmdefault}{\mddefault}{\updefault}fig.
5e}}}}
\end{picture}
} \hspace{\fill} 
\end{figure}

{\bfseries Fifth step.} Finally, reflecting the fig.~6e with respect to a horizontal line
we obtain the fig.~6f. \\ [8pt]

\begin{figure}[h]
\noindent \setlength{\unitlength}{0.00075in}
\begingroup\makeatletter\ifx\SetFigFont\undefined%
\gdef\SetFigFont#1#2#3#4#5{%
  \reset@font\fontsize{#1}{#2pt}%
  \fontfamily{#3}\fontseries{#4}\fontshape{#5}%
  \selectfont}%
\fi\endgroup%
{\renewcommand{\dashlinestretch}{30}
\begin{picture}(8169,2814)(0,-10)
\path(1182,447)(1452,987)
\path(1425.167,866.252)(1452.000,987.000)(1371.502,893.085)
\path(1182,447)(1632,1347) \path(2082,447)(1632,1347)
\path(2082,447)(1632,1347) \path(2082,447)(1812,987)
\path(2082,447)(1812,987)
\whiten\path(1892.498,893.085)(1812.000,987.000)(1838.833,866.252)(1881.765,847.469)(1892.498,893.085)
\path(1182,447)(732,1347) \path(1182,447)(732,1347)
\path(1182,447)(912,987) \path(1182,447)(912,987)
\whiten\path(992.498,893.085)(912.000,987.000)(938.833,866.252)(981.765,847.469)(992.498,893.085)
\path(3522,1347)(3792,807)
\path(3711.502,900.915)(3792.000,807.000)(3765.167,927.748)
\path(3522,1347)(3972,447) \path(3837,717)(3792,807)
\whiten\path(3872.498,713.085)(3792.000,807.000)(3818.833,686.252)(3872.498,713.085)
\path(3657,1077)(3612,1167)
\whiten\path(3692.498,1073.085)(3612.000,1167.000)(3638.833,1046.252)(3692.498,1073.085)
\path(3972,2247)(4242,1707)
\path(4161.502,1800.915)(4242.000,1707.000)(4215.167,1827.748)
\path(3972,2247)(4422,1347) \path(4287,1617)(4242,1707)
\whiten\path(4322.498,1613.085)(4242.000,1707.000)(4268.833,1586.252)(4322.498,1613.085)
\path(4107,1977)(4062,2067)
\whiten\path(4142.498,1973.085)(4062.000,2067.000)(4088.833,1946.252)(4142.498,1973.085)
\path(4422,1347)(4692,807)
\path(4611.502,900.915)(4692.000,807.000)(4665.167,927.748)
\path(4422,1347)(4872,447) \path(4737,717)(4692,807)
\whiten\path(4772.498,713.085)(4692.000,807.000)(4718.833,686.252)(4772.498,713.085)
\path(4557,1077)(4512,1167)
\whiten\path(4592.498,1073.085)(4512.000,1167.000)(4538.833,1046.252)(4592.498,1073.085)
\path(6762,2247)(6537,2697) \path(6762,2247)(6537,2697)
\path(6537,2697)(6582,2607) \path(6537,2697)(6582,2607)
\whiten\path(6501.502,2700.915)(6582.000,2607.000)(6555.167,2727.748)(6501.502,2700.915)
\path(6312,1347)(6087,1797) \path(6312,1347)(6087,1797)
\path(6087,1797)(6132,1707) \path(6087,1797)(6132,1707)
\whiten\path(6051.502,1800.915)(6132.000,1707.000)(6105.167,1827.748)(6051.502,1800.915)
\path(5862,447)(5637,897) \path(5862,447)(5637,897)
\path(5637,897)(5682,807) \path(5637,897)(5682,807)
\whiten\path(5601.502,900.915)(5682.000,807.000)(5655.167,927.748)(5601.502,900.915)
\path(7212,1347)(6762,2247) \path(6762,2247)(7032,1707)
\path(6951.502,1800.915)(7032.000,1707.000)(7005.167,1827.748)
\path(6762,447)(6312,1347) \path(6312,1347)(6582,807)
\path(6501.502,900.915)(6582.000,807.000)(6555.167,927.748)
\path(7662,447)(7212,1347) \path(7212,1347)(7482,807)
\path(7401.502,900.915)(7482.000,807.000)(7455.167,927.748)
\put(732,1347){\blacken\ellipse{46}{46}}
\put(732,1347){\ellipse{46}{46}}
\put(1182,2247){\blacken\ellipse{46}{46}}
\put(1182,2247){\ellipse{46}{46}}
\put(1632,1347){\blacken\ellipse{46}{46}}
\put(1632,1347){\ellipse{46}{46}}
\put(3522,1347){\blacken\ellipse{46}{46}}
\put(3522,1347){\ellipse{46}{46}}
\put(3972,2247){\blacken\ellipse{46}{46}}
\put(3972,2247){\ellipse{46}{46}}
\put(4422,1347){\blacken\ellipse{46}{46}}
\put(4422,1347){\ellipse{46}{46}}
\put(6312,1347){\blacken\ellipse{46}{46}}
\put(6312,1347){\ellipse{46}{46}}
\put(6762,2247){\blacken\ellipse{46}{46}}
\put(6762,2247){\ellipse{46}{46}}
\put(7212,1347){\blacken\ellipse{46}{46}}
\put(7212,1347){\ellipse{46}{46}} \path(732,1347)(462,1887)
\path(732,1347)(462,1887)
\whiten\path(542.498,1793.085)(462.000,1887.000)(488.833,1766.252)(542.498,1793.085)
\path(1182,2247)(912,2787) \path(1182,2247)(912,2787)
\whiten\path(992.498,2693.085)(912.000,2787.000)(938.833,2666.252)(992.498,2693.085)
\path(282,447)(12,987) \path(282,447)(12,987)
\whiten\path(92.498,893.085)(12.000,987.000)(38.833,866.252)(92.498,893.085)
\path(1632,1347)(1182,2247) \path(1632,1347)(1182,2247)
\path(1632,1347)(1362,1887) \path(1632,1347)(1362,1887)
\whiten\path(1442.498,1793.085)(1362.000,1887.000)(1388.833,1766.252)(1431.765,1747.469)(1442.498,1793.085)
\path(732,1347)(1002,1887)
\path(975.167,1766.252)(1002.000,1887.000)(921.502,1793.085)
\path(732,1347)(1182,2247) \path(282,447)(552,987)
\path(525.167,866.252)(552.000,987.000)(471.502,893.085)
\path(282,447)(732,1347) \path(2172,1527)(2667,1527)
\blacken\path(2547.000,1497.000)(2667.000,1527.000)(2547.000,1557.000)(2583.000,1527.000)(2547.000,1497.000)
\path(4962,1527)(5457,1527)
\blacken\path(5337.000,1497.000)(5457.000,1527.000)(5337.000,1557.000)(5373.000,1527.000)(5337.000,1497.000)
\path(3522,1347)(3252,1887) \path(3522,1347)(3252,1887)
\whiten\path(3332.498,1793.085)(3252.000,1887.000)(3278.833,1766.252)(3332.498,1793.085)
\path(3972,2247)(3702,2787) \path(3972,2247)(3702,2787)
\whiten\path(3782.498,2693.085)(3702.000,2787.000)(3728.833,2666.252)(3782.498,2693.085)
\path(3072,447)(2802,987) \path(3072,447)(2802,987)
\whiten\path(2882.498,893.085)(2802.000,987.000)(2828.833,866.252)(2882.498,893.085)
\path(3522,1347)(3792,1887)
\path(3765.167,1766.252)(3792.000,1887.000)(3711.502,1793.085)
\path(3522,1347)(3972,2247) \path(3072,447)(3342,987)
\path(3315.167,866.252)(3342.000,987.000)(3261.502,893.085)
\path(3072,447)(3522,1347) \path(3972,447)(4242,987)
\path(4215.167,866.252)(4242.000,987.000)(4161.502,893.085)
\path(3972,447)(4422,1347) \path(6312,1347)(6582,1887)
\path(6555.167,1766.252)(6582.000,1887.000)(6501.502,1793.085)
\path(6312,1347)(6762,2247) \path(5862,447)(6132,987)
\path(6105.167,866.252)(6132.000,987.000)(6051.502,893.085)
\path(5862,447)(6312,1347) \path(6762,447)(7032,987)
\path(7005.167,866.252)(7032.000,987.000)(6951.502,893.085)
\path(6762,447)(7212,1347) \path(7662,1527)(8157,1527)
\blacken\path(8037.000,1497.000)(8157.000,1527.000)(8037.000,1557.000)(8073.000,1527.000)(8037.000,1497.000)
\put(822,1302){\makebox(0,0)[lb]{\smash{{\SetFigFont{12}{14.4}{\rmdefault}{\mddefault}{\updefault}$y_1$}}}}
\put(1272,402){\makebox(0,0)[lb]{\smash{{\SetFigFont{12}{14.4}{\rmdefault}{\mddefault}{\updefault}$x_2$}}}}
\put(2172,402){\makebox(0,0)[lb]{\smash{{\SetFigFont{12}{14.4}{\rmdefault}{\mddefault}{\updefault}$x_3$}}}}
\put(1722,1302){\makebox(0,0)[lb]{\smash{{\SetFigFont{12}{14.4}{\rmdefault}{\mddefault}{\updefault}$y_2$}}}}
\put(1272,2202){\makebox(0,0)[lb]{\smash{{\SetFigFont{12}{14.4}{\rmdefault}{\mddefault}{\updefault}$z_3$}}}}
\put(372,402){\makebox(0,0)[lb]{\smash{{\SetFigFont{12}{14.4}{\rmdefault}{\mddefault}{\updefault}$x_1$}}}}
\put(867,87){\makebox(0,0)[lb]{\smash{{\SetFigFont{14}{16.8}{\rmdefault}{\mddefault}{\updefault}fig.
6a}}}}
\put(3612,1302){\makebox(0,0)[lb]{\smash{{\SetFigFont{12}{14.4}{\rmdefault}{\mddefault}{\updefault}$y_1$}}}}
\put(4062,402){\makebox(0,0)[lb]{\smash{{\SetFigFont{12}{14.4}{\rmdefault}{\mddefault}{\updefault}$x_2$}}}}
\put(4962,402){\makebox(0,0)[lb]{\smash{{\SetFigFont{12}{14.4}{\rmdefault}{\mddefault}{\updefault}$x_3$}}}}
\put(4512,1302){\makebox(0,0)[lb]{\smash{{\SetFigFont{12}{14.4}{\rmdefault}{\mddefault}{\updefault}$y_2$}}}}
\put(4062,2202){\makebox(0,0)[lb]{\smash{{\SetFigFont{12}{14.4}{\rmdefault}{\mddefault}{\updefault}$z_1$}}}}
\put(3162,402){\makebox(0,0)[lb]{\smash{{\SetFigFont{12}{14.4}{\rmdefault}{\mddefault}{\updefault}$x_1$}}}}
\put(3657,87){\makebox(0,0)[lb]{\smash{{\SetFigFont{14}{16.8}{\rmdefault}{\mddefault}{\updefault}fig.
6b}}}}
\put(6402,1302){\makebox(0,0)[lb]{\smash{{\SetFigFont{12}{14.4}{\rmdefault}{\mddefault}{\updefault}$y_1$}}}}
\put(6852,402){\makebox(0,0)[lb]{\smash{{\SetFigFont{12}{14.4}{\rmdefault}{\mddefault}{\updefault}$x_2$}}}}
\put(7752,402){\makebox(0,0)[lb]{\smash{{\SetFigFont{12}{14.4}{\rmdefault}{\mddefault}{\updefault}$x_3$}}}}
\put(7302,1302){\makebox(0,0)[lb]{\smash{{\SetFigFont{12}{14.4}{\rmdefault}{\mddefault}{\updefault}$y_2$}}}}
\put(6852,2202){\makebox(0,0)[lb]{\smash{{\SetFigFont{12}{14.4}{\rmdefault}{\mddefault}{\updefault}$z_1$}}}}
\put(5952,402){\makebox(0,0)[lb]{\smash{{\SetFigFont{12}{14.4}{\rmdefault}{\mddefault}{\updefault}$x_1$}}}}
\put(6447,87){\makebox(0,0)[lb]{\smash{{\SetFigFont{14}{16.8}{\rmdefault}{\mddefault}{\updefault}fig.
6c}}}} \path(7077,1617)(7122,1527)
\whiten\path(7041.502,1620.915)(7122.000,1527.000)(7095.167,1647.748)(7041.502,1620.915)
\path(6897,1977)(6942,1887)
\whiten\path(6861.502,1980.915)(6942.000,1887.000)(6915.167,2007.748)(6861.502,1980.915)
\path(6627,717)(6672,627)
\whiten\path(6591.502,720.915)(6672.000,627.000)(6645.167,747.748)(6591.502,720.915)
\path(6447,1077)(6492,987)
\whiten\path(6411.502,1080.915)(6492.000,987.000)(6465.167,1107.748)(6411.502,1080.915)
\path(7527,717)(7572,627)
\whiten\path(7491.502,720.915)(7572.000,627.000)(7545.167,747.748)(7491.502,720.915)
\path(7347,1077)(7392,987)
\whiten\path(7311.502,1080.915)(7392.000,987.000)(7365.167,1107.748)(7311.502,1080.915)
\end{picture}
} \\[8pt]
\end{figure}

\begin{figure}[h]
\noindent \setlength{\unitlength}{0.00075in}
\begingroup\makeatletter\ifx\SetFigFont\undefined%
\gdef\SetFigFont#1#2#3#4#5{%
  \reset@font\fontsize{#1}{#2pt}%
  \fontfamily{#3}\fontseries{#4}\fontshape{#5}%
  \selectfont}%
\fi\endgroup%
{\renewcommand{\dashlinestretch}{30}
\begin{picture}(7588,2817)(0,-10)
\path(1272,2247)(1497,2697) \path(1272,2247)(1497,2697)
\path(1497,2697)(1452,2607) \path(1497,2697)(1452,2607)
\whiten\path(1478.833,2727.748)(1452.000,2607.000)(1532.498,2700.915)(1478.833,2727.748)
\path(1722,1347)(1947,1797) \path(1722,1347)(1947,1797)
\path(1947,1797)(1902,1707) \path(1947,1797)(1902,1707)
\whiten\path(1928.833,1827.748)(1902.000,1707.000)(1982.498,1800.915)(1928.833,1827.748)
\path(2172,447)(2397,897) \path(2172,447)(2397,897)
\path(2397,897)(2352,807) \path(2397,897)(2352,807)
\whiten\path(2378.833,927.748)(2352.000,807.000)(2432.498,900.915)(2378.833,927.748)
\path(822,1347)(1092,1887)
\path(1065.167,1766.252)(1092.000,1887.000)(1011.502,1793.085)
\path(822,1347)(1272,2247) \path(1137,1977)(1092,1887)
\whiten\path(1118.833,2007.748)(1092.000,1887.000)(1172.498,1980.915)(1118.833,2007.748)
\path(957,1617)(912,1527)
\whiten\path(938.833,1647.748)(912.000,1527.000)(992.498,1620.915)(938.833,1647.748)
\path(372,447)(642,987)
\path(615.167,866.252)(642.000,987.000)(561.502,893.085)
\path(372,447)(822,1347) \path(687,1077)(642,987)
\whiten\path(668.833,1107.748)(642.000,987.000)(722.498,1080.915)(668.833,1107.748)
\path(507,717)(462,627)
\whiten\path(488.833,747.748)(462.000,627.000)(542.498,720.915)(488.833,747.748)
\path(1272,447)(1542,987)
\path(1515.167,866.252)(1542.000,987.000)(1461.502,893.085)
\path(1272,447)(1722,1347) \path(1587,1077)(1542,987)
\whiten\path(1568.833,1107.748)(1542.000,987.000)(1622.498,1080.915)(1568.833,1107.748)
\path(1407,717)(1362,627)
\whiten\path(1388.833,747.748)(1362.000,627.000)(1442.498,720.915)(1388.833,747.748)
\put(4242,1347){\blacken\ellipse{46}{46}}
\put(4242,1347){\ellipse{46}{46}}
\put(3342,1347){\blacken\ellipse{46}{46}}
\put(3342,1347){\ellipse{46}{46}}
\put(3792,2247){\blacken\ellipse{46}{46}}
\put(3792,2247){\ellipse{46}{46}} \path(4692,447)(4917,897)
\path(4692,447)(4917,897) \path(4917,897)(4872,807)
\path(4917,897)(4872,807)
\whiten\path(4898.833,927.748)(4872.000,807.000)(4952.498,900.915)(4898.833,927.748)
\path(4242,1347)(4467,1797) \path(4242,1347)(4467,1797)
\path(4467,1797)(4422,1707) \path(4467,1797)(4422,1707)
\whiten\path(4448.833,1827.748)(4422.000,1707.000)(4502.498,1800.915)(4448.833,1827.748)
\path(3792,2247)(4017,2697) \path(3792,2247)(4017,2697)
\path(4017,2697)(3972,2607) \path(4017,2697)(3972,2607)
\whiten\path(3998.833,2727.748)(3972.000,2607.000)(4052.498,2700.915)(3998.833,2727.748)
\path(4242,1347)(4512,807)
\path(4431.502,900.915)(4512.000,807.000)(4485.167,927.748)
\path(4242,1347)(4692,447) \path(3792,2247)(4062,1707)
\path(3981.502,1800.915)(4062.000,1707.000)(4035.167,1827.748)
\path(3792,2247)(4242,1347) \path(3342,1347)(3612,807)
\path(3531.502,900.915)(3612.000,807.000)(3585.167,927.748)
\path(3342,1347)(3792,447) \path(4242,1347)(3792,447)
\path(4242,1347)(3792,447) \path(4242,1347)(3972,807)
\path(4242,1347)(3972,807)
\whiten\path(3998.833,927.748)(3972.000,807.000)(4052.498,900.915)(4041.765,946.531)(3998.833,927.748)
\path(3792,2247)(3342,1347) \path(3792,2247)(3342,1347)
\path(3792,2247)(3522,1707) \path(3792,2247)(3522,1707)
\whiten\path(3548.833,1827.748)(3522.000,1707.000)(3602.498,1800.915)(3591.765,1846.531)(3548.833,1827.748)
\path(3342,1347)(2892,447) \path(3342,1347)(2892,447)
\path(3342,1347)(3072,807) \path(3342,1347)(3072,807)
\whiten\path(3098.833,927.748)(3072.000,807.000)(3152.498,900.915)(3141.765,946.531)(3098.833,927.748)
\put(3477,87){\makebox(0,0)[lb]{\smash{{\SetFigFont{14}{16.8}{\rmdefault}{\mddefault}{\updefault}fig.
6e}}}}
\put(3882,402){\makebox(0,0)[lb]{\smash{{\SetFigFont{12}{14.4}{\rmdefault}{\mddefault}{\updefault}$x_2$}}}}
\put(4782,402){\makebox(0,0)[lb]{\smash{{\SetFigFont{12}{14.4}{\rmdefault}{\mddefault}{\updefault}$x_3$}}}}
\put(3432,1347){\makebox(0,0)[lb]{\smash{{\SetFigFont{12}{14.4}{\rmdefault}{\mddefault}{\updefault}$y_1$}}}}
\put(4332,1347){\makebox(0,0)[lb]{\smash{{\SetFigFont{12}{14.4}{\rmdefault}{\mddefault}{\updefault}$y_2$}}}}
\put(3882,2202){\makebox(0,0)[lb]{\smash{{\SetFigFont{12}{14.4}{\rmdefault}{\mddefault}{\updefault}$z_1$}}}}
\put(2982,402){\makebox(0,0)[lb]{\smash{{\SetFigFont{12}{14.4}{\rmdefault}{\mddefault}{\updefault}$x_1$}}}}
\put(6852,1842){\blacken\ellipse{46}{46}}
\put(6852,1842){\ellipse{46}{46}}
\put(5952,1842){\blacken\ellipse{46}{46}}
\put(5952,1842){\ellipse{46}{46}}
\put(6402,942){\blacken\ellipse{46}{46}}
\put(6402,942){\ellipse{46}{46}} \path(7302,2742)(7527,2292)
\path(7302,2742)(7527,2292) \path(7527,2292)(7482,2382)
\path(7527,2292)(7482,2382)
\whiten\path(7562.498,2288.085)(7482.000,2382.000)(7508.833,2261.252)(7562.498,2288.085)
\path(6852,1842)(7077,1392) \path(6852,1842)(7077,1392)
\path(7077,1392)(7032,1482) \path(7077,1392)(7032,1482)
\whiten\path(7112.498,1388.085)(7032.000,1482.000)(7058.833,1361.252)(7112.498,1388.085)
\path(6402,942)(6627,492) \path(6402,942)(6627,492)
\path(6627,492)(6582,582) \path(6627,492)(6582,582)
\whiten\path(6662.498,488.085)(6582.000,582.000)(6608.833,461.252)(6662.498,488.085)
\path(6852,1842)(7122,2382)
\path(7095.167,2261.252)(7122.000,2382.000)(7041.502,2288.085)
\path(6852,1842)(7302,2742) \path(6402,942)(6672,1482)
\path(6645.167,1361.252)(6672.000,1482.000)(6591.502,1388.085)
\path(6402,942)(6852,1842) \path(5952,1842)(6222,2382)
\path(6195.167,2261.252)(6222.000,2382.000)(6141.502,2288.085)
\path(5952,1842)(6402,2742) \path(6852,1842)(6402,2742)
\path(6852,1842)(6402,2742) \path(6852,1842)(6582,2382)
\path(6852,1842)(6582,2382)
\whiten\path(6662.498,2288.085)(6582.000,2382.000)(6608.833,2261.252)(6651.765,2242.469)(6662.498,2288.085)
\path(6402,942)(5952,1842) \path(6402,942)(5952,1842)
\path(6402,942)(6132,1482) \path(6402,942)(6132,1482)
\whiten\path(6212.498,1388.085)(6132.000,1482.000)(6158.833,1361.252)(6201.765,1342.469)(6212.498,1388.085)
\path(5952,1842)(5502,2742) \path(5952,1842)(5502,2742)
\path(5952,1842)(5682,2382) \path(5952,1842)(5682,2382)
\whiten\path(5762.498,2288.085)(5682.000,2382.000)(5708.833,2261.252)(5751.765,2242.469)(5762.498,2288.085)
\put(6087,87){\makebox(0,0)[lb]{\smash{{\SetFigFont{14}{16.8}{\rmdefault}{\mddefault}{\updefault}fig.
6f}}}}
\put(6492,2697){\makebox(0,0)[lb]{\smash{{\SetFigFont{12}{14.4}{\rmdefault}{\mddefault}{\updefault}$x_2$}}}}
\put(7392,2697){\makebox(0,0)[lb]{\smash{{\SetFigFont{12}{14.4}{\rmdefault}{\mddefault}{\updefault}$x_3$}}}}
\put(6042,1797){\makebox(0,0)[lb]{\smash{{\SetFigFont{12}{14.4}{\rmdefault}{\mddefault}{\updefault}$y_1$}}}}
\put(6942,1797){\makebox(0,0)[lb]{\smash{{\SetFigFont{12}{14.4}{\rmdefault}{\mddefault}{\updefault}$y_2$}}}}
\put(6492,897){\makebox(0,0)[lb]{\smash{{\SetFigFont{12}{14.4}{\rmdefault}{\mddefault}{\updefault}$z_1$}}}}
\put(5592,2697){\makebox(0,0)[lb]{\smash{{\SetFigFont{12}{14.4}{\rmdefault}{\mddefault}{\updefault}$x_1$}}}}
\put(822,1347){\blacken\ellipse{46}{46}}
\put(822,1347){\ellipse{46}{46}}
\put(1722,1347){\blacken\ellipse{46}{46}}
\put(1722,1347){\ellipse{46}{46}}
\put(1272,2247){\blacken\ellipse{46}{46}}
\put(1272,2247){\ellipse{46}{46}} \path(822,1347)(1092,807)
\path(1011.502,900.915)(1092.000,807.000)(1065.167,927.748)
\path(1722,1347)(1992,807)
\path(1911.502,900.915)(1992.000,807.000)(1965.167,927.748)
\path(1272,2247)(1542,1707)
\path(1461.502,1800.915)(1542.000,1707.000)(1515.167,1827.748)
\path(1272,447)(822,1347) \path(2172,447)(1722,1347)
\path(1722,1347)(1272,2247) \path(2397,1527)(2892,1527)
\blacken\path(2772.000,1497.000)(2892.000,1527.000)(2772.000,1557.000)(2808.000,1527.000)(2772.000,1497.000)
\path(5007,1527)(5502,1527)
\blacken\path(5382.000,1497.000)(5502.000,1527.000)(5382.000,1557.000)(5418.000,1527.000)(5382.000,1497.000)
\path(12,1527)(507,1527)
\blacken\path(387.000,1497.000)(507.000,1527.000)(387.000,1557.000)(423.000,1527.000)(387.000,1497.000)
\put(912,1302){\makebox(0,0)[lb]{\smash{{\SetFigFont{12}{14.4}{\rmdefault}{\mddefault}{\updefault}$y_1$}}}}
\put(1362,402){\makebox(0,0)[lb]{\smash{{\SetFigFont{12}{14.4}{\rmdefault}{\mddefault}{\updefault}$x_2$}}}}
\put(2262,402){\makebox(0,0)[lb]{\smash{{\SetFigFont{12}{14.4}{\rmdefault}{\mddefault}{\updefault}$x_3$}}}}
\put(1812,1302){\makebox(0,0)[lb]{\smash{{\SetFigFont{12}{14.4}{\rmdefault}{\mddefault}{\updefault}$y_2$}}}}
\put(1362,2202){\makebox(0,0)[lb]{\smash{{\SetFigFont{12}{14.4}{\rmdefault}{\mddefault}{\updefault}$z_1$}}}}
\put(462,402){\makebox(0,0)[lb]{\smash{{\SetFigFont{12}{14.4}{\rmdefault}{\mddefault}{\updefault}$x_1$}}}}
\put(957,87){\makebox(0,0)[lb]{\smash{{\SetFigFont{14}{16.8}{\rmdefault}{\mddefault}{\updefault}fig.
6d}}}}
\end{picture}
}
\end{figure}

To obtain a graphical representation for the integral in~\eqref{A_orthog} we should join the diagrams
shown in the figs.~6a for $\psi_{\gamma}(x)$ and 6f for
$\overline{\psi_{\gamma}(x)}$ in the points $x_1,\ldots,x_N$ and
integrating over them. The diagram obtained like this
in the case $N=4$ is pictured in fig.~8a. Further we shall
simplify it using the following equalities.

\begin{lem}
 The equalities represented in the figures~7a and 7b are valid. 
\end{lem}

\begin{figure}[h]
\hspace{\fill} \setlength{\unitlength}{0.00087489in}
\begingroup\makeatletter\ifx\SetFigFont\undefined%
\gdef\SetFigFont#1#2#3#4#5{%
  \reset@font\fontsize{#1}{#2pt}%
  \fontfamily{#3}\fontseries{#4}\fontshape{#5}%
  \selectfont}%
\fi\endgroup%
{\renewcommand{\dashlinestretch}{30}
\begin{picture}(4366,2565)(0,-10)

\put(3210,1122){\makebox(0,0)[lb]{\smash{{\SetFigFont{10}{14.4}{\rmdefault}{\mddefault}{\updefault}$\gamma_k$}}}}
\put(3750,1572){\makebox(0,0)[lb]{\smash{{\SetFigFont{10}{12.0}{\rmdefault}{\mddefault}{\updefault}\makebox[0.8\width][s]{$(\gamma'_j-\gamma_k)$}}}}}
\path(3660,672)(3390,1212) \path(3660,672)(3390,1212)
\whiten\path(3470.498,1118.085)(3390.000,1212.000)(3416.833,1091.252)(3470.498,1118.085)
\path(3660,672)(3660,2472) \path(3660,672)(3660,2472)
\path(3660,672)(3660,1707) \path(3660,672)(3660,1707)
\whiten\path(3690.000,1587.000)(3660.000,1707.000)(3630.000,1587.000)(3660.000,1623.000)(3690.000,1587.000)
\put(3750,2427){\makebox(0,0)[lb]{\smash{{\SetFigFont{12}{14.4}{\rmdefault}{\mddefault}{\updefault}$y$}}}}
\put(3750,627){\makebox(0,0)[lb]{\smash{{\SetFigFont{12}{14.4}{\rmdefault}{\mddefault}{\updefault}$y'$}}}}
\put(15,2022){\makebox(0,0)[lb]{\smash{{\SetFigFont{10}{14.4}{\rmdefault}{\mddefault}{\updefault}$\gamma_k$}}}}
\put(555,1527){\makebox(0,0)[lb]{\smash{{\SetFigFont{12}{14.4}{\rmdefault}{\mddefault}{\updefault}$x$}}}}
\path(465,1572)(195,2112) \path(465,1572)(195,2112)
\whiten\path(275.498,2018.085)(195.000,2112.000)(221.833,1991.252)(275.498,2018.085)
\path(465,1572)(735,2112)
\path(708.167,1991.252)(735.000,2112.000)(654.502,2018.085)
\path(465,1572)(915,2472) \path(465,1572)(735,2112)
\path(708.167,1991.252)(735.000,2112.000)(654.502,2018.085)
\path(465,1572)(915,2472) \path(915,672)(465,1572)
\path(915,672)(465,1572) \path(915,672)(645,1212)
\path(915,672)(645,1212)
\whiten\path(725.498,1118.085)(645.000,1212.000)(671.833,1091.252)(714.765,1072.469)(725.498,1118.085)
\put(465,1122){\makebox(0,0)[lb]{\smash{{\SetFigFont{10}{14.4}{\rmdefault}{\mddefault}{\updefault}$\gamma'_j$}}}}
\put(1005,2427){\makebox(0,0)[lb]{\smash{{\SetFigFont{12}{14.4}{\rmdefault}{\mddefault}{\updefault}$y$}}}}
\put(1005,627){\makebox(0,0)[lb]{\smash{{\SetFigFont{12}{14.4}{\rmdefault}{\mddefault}{\updefault}$y'$}}}}
\put(465,1572){\blacken\ellipse{46}{46}}
\put(465,1572){\ellipse{46}{46}}
\put(1320,1527){\makebox(0,0)[lb]{\smash{{\SetFigFont{14}{16.8}{\rmdefault}{\mddefault}{\updefault}=}}}}
\put(1590,1527){\makebox(0,0)[lb]{\smash{{\SetFigFont{14}{16.8}{\rmdefault}{\mddefault}{\updefault}$\hbar^{\frac{\gamma'_j-\gamma_k}{i\hbar}}\Gamma\Bigl(\dfrac{\gamma'_j-\gamma_k}{i\hbar}\Bigr)$}}}}
\put(1995,87){\makebox(0,0)[lb]{\smash{{\SetFigFont{14}{16.8}{\rmdefault}{\mddefault}{\updefault}fig.
7a}}}}
\end{picture}
} \hspace{\fill} 
\end{figure}

\begin{figure}[h]
\hspace{31mm} \setlength{\unitlength}{0.00087489in}
\begingroup\makeatletter\ifx\SetFigFont\undefined%
\gdef\SetFigFont#1#2#3#4#5{%
  \reset@font\fontsize{#1}{#2pt}%
  \fontfamily{#3}\fontseries{#4}\fontshape{#5}%
  \selectfont}%
\fi\endgroup%
{\renewcommand{\dashlinestretch}{30}
\begin{picture}(4051,2430)(0,-10)
\put(1095,1392){\makebox(0,0)[lb]{\smash{{\SetFigFont{14}{16.8}{\rmdefault}{\mddefault}{\updefault}=}}}}
\put(1365,1392){\makebox(0,0)[lb]{\smash{{\SetFigFont{14}{16.8}{\rmdefault}{\mddefault}{\updefault}$\hbar^{\frac{\gamma_k-\gamma'_j}{i\hbar}}\Gamma\Bigl(\dfrac{\gamma_k-\gamma'_j}{i\hbar}\Bigr)$}}}}
\path(3345,537)(3345,2337) \path(3345,537)(3345,2337)
\path(3345,537)(3345,1572) \path(3345,537)(3345,1572)
\whiten\path(3375.000,1452.000)(3345.000,1572.000)(3315.000,1452.000)(3345.000,1488.000)(3375.000,1452.000)
\put(3660,1842){\makebox(0,0)[lb]{\smash{{\SetFigFont{10}{14.4}{\rmdefault}{\mddefault}{\updefault}$\gamma'_j$}}}}
\path(3345,2337)(3570,1887) \path(3345,2337)(3570,1887)
\path(465,1437)(15,2337) \path(465,1437)(15,2337)
\path(465,1437)(195,1977) \path(465,1437)(195,1977)
\whiten\path(275.498,1883.085)(195.000,1977.000)(221.833,1856.252)(264.765,1837.469)(275.498,1883.085)
\put(15,1887){\makebox(0,0)[lb]{\smash{{\SetFigFont{10}{14.4}{\rmdefault}{\mddefault}{\updefault}$\gamma_k$}}}}
\path(15,537)(285,1077)
\path(258.167,956.252)(285.000,1077.000)(204.502,983.085)
\path(15,537)(465,1437)
\put(465,987){\makebox(0,0)[lb]{\smash{{\SetFigFont{10}{14.4}{\rmdefault}{\mddefault}{\updefault}$\gamma'_j$}}}}
\put(105,2292){\makebox(0,0)[lb]{\smash{{\SetFigFont{12}{14.4}{\rmdefault}{\mddefault}{\updefault}$y$}}}}
\put(105,492){\makebox(0,0)[lb]{\smash{{\SetFigFont{12}{14.4}{\rmdefault}{\mddefault}{\updefault}$y'$}}}}
\put(555,1392){\makebox(0,0)[lb]{\smash{{\SetFigFont{12}{14.4}{\rmdefault}{\mddefault}{\updefault}$x$}}}}
\put(465,1437){\blacken\ellipse{46}{46}}
\put(465,1437){\ellipse{46}{46}} \path(465,1437)(690,987)
\path(465,1437)(690,987) \path(690,987)(645,1077)
\path(690,987)(645,1077)
\whiten\path(725.498,983.085)(645.000,1077.000)(671.833,956.252)(725.498,983.085)
\put(1725,87){\makebox(0,0)[lb]{\smash{{\SetFigFont{14}{16.8}{\rmdefault}{\mddefault}{\updefault}fig.
7b}}}}
\put(3435,2292){\makebox(0,0)[lb]{\smash{{\SetFigFont{12}{14.4}{\rmdefault}{\mddefault}{\updefault}$y$}}}}
\put(3435,1437){\makebox(0,0)[lb]{\smash{{\SetFigFont{10}{12.0}{\rmdefault}{\mddefault}{\updefault}\makebox[0.8\width][s]{$(\gamma_k-\gamma'_j)$}}}}}
\put(3435,492){\makebox(0,0)[lb]{\smash{{\SetFigFont{12}{14.4}{\rmdefault}{\mddefault}{\updefault}$y'$}}}}
\path(3570,1887)(3525,1977) \path(3570,1887)(3525,1977)
\whiten\path(3605.498,1883.085)(3525.000,1977.000)(3551.833,1856.252)(3605.498,1883.085)
\end{picture}
} \hspace{\fill} 
\end{figure}

\noindent{\bfseries Proof.} The fig.~7a and 7b mean the following
equalities, which are consequence of the integral representation
of the $\Gamma$-function
\begin{multline*}
 \int\limits_{-\infty}^{+\infty} dx\, Z_{\gamma_k}(x)I(x,y)J_{\gamma'_j}(y',x)=
 \int\limits_{-\infty}^{+\infty} dx\, \exp\Bigl\{\frac{i}{\hbar}[(\gamma_k-\gamma'_j)x+\gamma'_j y']
   -\frac{1}{\hbar}(e^{x-y}+e^{x-y'})\Bigr\}=\\
 =e^{\frac{i}{\hbar}\gamma'_j y'}
  \left(\frac{\hbar}{e^{-y}+e^{-y'}}\right)^{\frac{i}{\hbar}(\gamma_k-\gamma'_j)}
  \Gamma\Bigl(\frac{\gamma'_j-\gamma_k}{i\hbar}\Bigr)=
  \hbar^{\frac{\gamma'_j-\gamma_k}{i\hbar}}\Gamma\Bigl(\frac{\gamma'_j-\gamma_k}{i\hbar}\Bigr)
  Z_{\gamma_k}(y') Y_{\gamma'_j-\gamma_k}(y',y),
\end{multline*}
\begin{multline*}
 \int\limits_{-\infty}^{+\infty} dx\, J_{\gamma_k}(x,y)I(y',x)Z_{\gamma'_j}^{-1}(x)=
 \int\limits_{-\infty}^{+\infty} dx\, \exp\Bigl\{\frac{i}{\hbar}[(\gamma_k-\gamma'_j)x-\gamma_k y]
   -\frac{1}{\hbar}(e^{y-x}+e^{y'-x})\Bigr\}=\\
 =e^{-\frac{i}{\hbar}\gamma_k y}
  \left(\frac{\hbar}{e^{y}+e^{y'}}\right)^{\frac{i}{\hbar}(\gamma'_j-\gamma_k)}
  \Gamma\Bigl(\frac{\gamma_k-\gamma'_j}{i\hbar}\Bigr)=
  \hbar^{\frac{\gamma_k-\gamma'_j}{i\hbar}}\Gamma\Bigl(\frac{\gamma_k-\gamma'_j}{i\hbar}\Bigr)
  Y_{\gamma_k-\gamma'_j}(y',y) Z_{\gamma'_j}(y).
\end{multline*} \\

\qed \\

The integral represented graphically in the fig.~8a 
can be calculated by induction. We sequentially integrate over the
boundary points connected with $\gamma_1$, $\gamma'_1$. First, we
integrate over the very left point and the very right point, i.e.
over $x_1$ and $x_N$, using the fig.~7a and 7b respectively, and
we obtain the fig.~8b. Like fig.~4, the vertical line arising in
the left side is moved to right, where it is annihilated by the
line arising from the right side. This process exchanges
$\gamma_1$ with $\gamma'_1$. After the integration the factor
\begin{equation} \label{A_fac_1}
 \Gamma\Bigl(\frac{\gamma'_1-\gamma_1}{i\hbar}\Bigr)
 \Gamma\Bigl(\frac{\gamma_1-\gamma'_1}{i\hbar}\Bigr)
\end{equation}
arises, and thus one obtains the fig.~8c, where $\gamma_1$,
$\gamma'_1$ begin to be connected with another boundary points
(boundary in sense of fig.~8a), 
Now we integrate over these points. In $k$-th integration
($k=2,\ldots,N-1$) $\gamma_1$ exchanges with $\gamma'_k$ and
$\gamma'_1$ exchanges with $\gamma_k$ and one has the factor
\begin{equation} \label{A_fac_k}
 \Gamma\Bigl(\frac{\gamma'_1-\gamma_k}{i\hbar}\Bigr)
 \Gamma\Bigl(\frac{\gamma_k-\gamma'_1}{i\hbar}\Bigr)
 \Gamma\Bigl(\frac{\gamma'_k-\gamma_1}{i\hbar}\Bigr)
 \Gamma\Bigl(\frac{\gamma_1-\gamma'_k}{i\hbar}\Bigr).
\end{equation}
The $N$-th integration leads to the factor
\begin{equation} \label{A_fac_N}
 (2\pi\hbar)\delta(\gamma'_1-\gamma_N)\cdot(2\pi\hbar)\delta(\gamma'_N-\gamma_1).
\end{equation}
After this the variables $\gamma_1$, $\gamma'_1$, $\gamma_N$ and
$\gamma'_N$ disappear completely from diagram -- they have gone
away to the factors~\eqref{A_fac_1}, \eqref{A_fac_k} and
\eqref{A_fac_N}. The remaining diagram (the middle part of fig.~8e)
is exactly the initial diagram, but for $N-2$ which is depend on
$\gamma_2,\ldots,\gamma_{N-1}$, $\gamma'_2,\ldots,\gamma'_{N-1}$.
Thereby we have obtained the recurrent formula
\begin{multline*} 
 \int\limits_{\mathbb R^N}dx_1\ldots dx_N\,\overline{\psi_{\gamma_1,\ldots,\gamma_N}(x_1,\ldots,x_N)}
                                 \psi_{\gamma'_1,\ldots,\gamma'_N}(x_1,\ldots,x_N)=
 \Gamma\Bigl(\frac{\gamma'_1-\gamma_1}{i\hbar}\Bigr)\Gamma\Bigl(\frac{\gamma_1-\gamma'_1}{i\hbar}\Bigr)\times \\
 \times\prod_{k=2}^{N-1}\Gamma\Bigl(\frac{\gamma'_1-\gamma_k}{i\hbar}\Bigr)\Gamma\Bigl(\frac{\gamma_k-\gamma'_1}{i\hbar}\Bigr)
 \Gamma\Bigl(\frac{\gamma'_k-\gamma_1}{i\hbar}\Bigr)\Gamma\Bigl(\frac{\gamma_1-\gamma'_k}{i\hbar}\Bigr)\cdot
 (2\pi\hbar)\delta(\gamma'_1-\gamma_N)\times \\
 \times(2\pi\hbar)\delta(\gamma'_N-\gamma_1)
 \int\limits_{\mathbb R^{N-2}}dx_2\ldots dx_{N-1}\,\overline{\psi_{\gamma_2,\ldots,\gamma_{N-1}}(x_2,\ldots,x_{N-1})}
                                     \psi_{\gamma'_2,\ldots,\gamma'_{N-1}}(x_2,\ldots,x_{N-1}).
\end{multline*}
Continuing the calculation for the integral~\eqref{A_orthog} in the same
manner we obtain the following result.
\begin{multline} \label{A_int_psi}
 \int\limits_{\mathbb R^N}dx_1\ldots dx_N\,\overline{\psi_{\gamma_1,\ldots,\gamma_N}(x_1,\ldots,x_N)}
                                 \psi_{\gamma'_1,\ldots,\gamma'_N}(x_1,\ldots,x_N)=\\
 =\prod_{\substack{k,j=1\\k+j\le N}}^{N-1}\left[\Gamma\Bigl(\frac{\gamma'_j-\gamma_k}{i\hbar}\Bigr)
 \Gamma\Bigl(\frac{\gamma_k-\gamma'_j}{i\hbar}\Bigr)\right] \prod_{j=1}^N (2\pi\hbar) \delta(\gamma'_j-\gamma_{N+1-j}).
\end{multline}

\begin{figure}[h]
\noindent\hspace{-3mm} \setlength{\unitlength}{0.0004in} 
\begingroup\makeatletter\ifx\SetFigFont\undefined%
\gdef\SetFigFont#1#2#3#4#5{%
  \reset@font\fontsize{#1}{#2pt}%
  \fontfamily{#3}\fontseries{#4}\fontshape{#5}%
  \selectfont}%
\fi\endgroup%
{\renewcommand{\dashlinestretch}{30}
\begin{picture}(8272,7240)(0,-10)
\path(2265,3952)(2535,4492)
\path(2508.167,4371.252)(2535.000,4492.000)(2454.502,4398.085)
\path(2265,3952)(2715,4852) \path(2715,4852)(2265,5752)
\path(2715,4852)(2265,5752) \path(2715,4852)(2445,5392)
\path(2715,4852)(2445,5392)
\whiten\path(2525.498,5298.085)(2445.000,5392.000)(2471.833,5271.252)(2514.765,5252.469)(2525.498,5298.085)
\path(915,4852)(1185,5392)
\path(1158.167,5271.252)(1185.000,5392.000)(1104.502,5298.085)
\path(915,4852)(1365,5752) \path(1365,5752)(1635,6292)
\path(1608.167,6171.252)(1635.000,6292.000)(1554.502,6198.085)
\path(1365,5752)(1815,6652) \path(2265,5752)(1815,6652)
\path(2265,5752)(1815,6652) \path(2265,5752)(1995,6292)
\path(2265,5752)(1995,6292)
\whiten\path(2075.498,6198.085)(1995.000,6292.000)(2021.833,6171.252)(2064.765,6152.469)(2075.498,6198.085)
\put(-85,4402){\makebox(0,0)[lb]{\smash{{\SetFigFont{8}{14.4}{\rmdefault}{\mddefault}{\updefault}$\gamma_1$}}}}
\path(465,3952)(195,4492) \path(465,3952)(195,4492)
\whiten\path(275.498,4398.085)(195.000,4492.000)(221.833,4371.252)(275.498,4398.085)
\path(1815,4852)(2085,5392)
\path(2058.167,5271.252)(2085.000,5392.000)(2004.502,5298.085)
\path(1815,4852)(2265,5752)
\path(1815,4852)(1365,5752) \path(1815,4852)(1365,5752)
\path(1815,4852)(1545,5392) \path(1815,4852)(1545,5392)
\whiten\path(1625.498,5298.085)(1545.000,5392.000)(1571.833,5271.252)(1614.765,5252.469)(1625.498,5298.085)
\path(1365,3952)(1635,4492)
\path(1608.167,4371.252)(1635.000,4492.000)(1554.502,4398.085)
\path(1365,3952)(1815,4852) \path(2265,3952)(1815,4852)
\path(2265,3952)(1815,4852) \path(2265,3952)(1995,4492)
\path(2265,3952)(1995,4492)
\whiten\path(2075.498,4398.085)(1995.000,4492.000)(2021.833,4371.252)(2064.765,4352.469)(2075.498,4398.085)
\put(1815,4852){\blacken\ellipse{46}{46}}
\put(1815,4852){\ellipse{46}{46}}
\path(1365,3952)(915,4852) \path(1365,3952)(915,4852)
\path(1365,3952)(1095,4492) \path(1365,3952)(1095,4492)
\whiten\path(1175.498,4398.085)(1095.000,4492.000)(1121.833,4371.252)(1164.765,4352.469)(1175.498,4398.085)
\put(915,4852){\blacken\ellipse{46}{46}}
\put(915,4852){\ellipse{46}{46}} \path(465,3952)(735,4492)
\path(708.167,4371.252)(735.000,4492.000)(654.502,4398.085)
\path(465,3952)(915,4852) \path(465,3952)(735,4492)
\path(708.167,4371.252)(735.000,4492.000)(654.502,4398.085)
\path(465,3952)(915,4852) \path(915,3052)(1185,3592)
\path(1158.167,3471.252)(1185.000,3592.000)(1104.502,3498.085)
\path(915,3052)(1365,3952) \path(1815,3052)(1365,3952)
\path(1815,3052)(1365,3952) \path(1815,3052)(1545,3592)
\path(1815,3052)(1545,3592)
\whiten\path(1625.498,3498.085)(1545.000,3592.000)(1571.833,3471.252)(1614.765,3452.469)(1625.498,3498.085)
\path(1815,3052)(2085,3592)
\path(2058.167,3471.252)(2085.000,3592.000)(2004.502,3498.085)
\path(1815,3052)(2265,3952) \path(2715,3052)(2985,3592)
\path(2958.167,3471.252)(2985.000,3592.000)(2904.502,3498.085)
\path(2715,3052)(3165,3952) \path(2715,3052)(2265,3952)
\path(2715,3052)(2265,3952) \path(2715,3052)(2445,3592)
\path(2715,3052)(2445,3592)
\whiten\path(2525.498,3498.085)(2445.000,3592.000)(2471.833,3471.252)(2514.765,3452.469)(2525.498,3498.085)
\path(2265,2152)(2535,2692)
\path(2508.167,2571.252)(2535.000,2692.000)(2454.502,2598.085)
\path(2265,2152)(2715,3052) \path(2265,2152)(1815,3052)
\path(2265,2152)(1815,3052) \path(2265,2152)(1995,2692)
\path(2265,2152)(1995,2692)
\whiten\path(2075.498,2598.085)(1995.000,2692.000)(2021.833,2571.252)(2064.765,2552.469)(2075.498,2598.085)
\path(1365,2152)(1635,2692)
\path(1608.167,2571.252)(1635.000,2692.000)(1554.502,2598.085)
\path(1365,2152)(1815,3052) \path(1365,2152)(915,3052)
\path(1365,2152)(915,3052) \path(1365,2152)(1095,2692)
\path(1365,2152)(1095,2692)
\whiten\path(1175.498,2598.085)(1095.000,2692.000)(1121.833,2571.252)(1164.765,2552.469)(1175.498,2598.085)
\path(1815,1252)(2085,1792)
\path(2058.167,1671.252)(2085.000,1792.000)(2004.502,1698.085)
\path(1815,1252)(2265,2152) \path(1815,1252)(1365,2152)
\path(1815,1252)(1365,2152) \path(1815,1252)(1545,1792)
\path(1815,1252)(1545,1792)
\whiten\path(1625.498,1698.085)(1545.000,1792.000)(1571.833,1671.252)(1614.765,1652.469)(1625.498,1698.085)
\path(3165,3952)(2715,4852) \path(3165,3952)(2715,4852)
\path(3165,3952)(2895,4492) \path(3165,3952)(2895,4492)
\whiten\path(2975.498,4398.085)(2895.000,4492.000)(2921.833,4371.252)(2964.765,4352.469)(2975.498,4398.085)
\path(915,3052)(465,3952) \path(915,3052)(465,3952)
\path(915,3052)(645,3592) \path(915,3052)(645,3592)
\whiten\path(725.498,3498.085)(645.000,3592.000)(671.833,3471.252)(714.765,3452.469)(725.498,3498.085)
\path(6540,3952)(6810,4492)
\path(6783.167,4371.252)(6810.000,4492.000)(6729.502,4398.085)
\path(6540,3952)(6990,4852) \path(6990,4852)(6540,5752)
\path(6990,4852)(6540,5752) \path(6990,4852)(6720,5392)
\path(6990,4852)(6720,5392)
\whiten\path(6800.498,5298.085)(6720.000,5392.000)(6746.833,5271.252)(6789.765,5252.469)(6800.498,5298.085)
\path(5190,4852)(5460,5392)
\path(5433.167,5271.252)(5460.000,5392.000)(5379.502,5298.085)
\path(5190,4852)(5640,5752) \path(5640,5752)(5910,6292)
\path(5883.167,6171.252)(5910.000,6292.000)(5829.502,6198.085)
\path(5640,5752)(6090,6652) \path(6540,5752)(6090,6652)
\path(6540,5752)(6090,6652) \path(6540,5752)(6270,6292)
\path(6540,5752)(6270,6292)
\whiten\path(6350.498,6198.085)(6270.000,6292.000)(6296.833,6171.252)(6339.765,6152.469)(6350.498,6198.085)
\path(6090,4852)(6360,5392)
\path(6333.167,5271.252)(6360.000,5392.000)(6279.502,5298.085)
\path(6090,4852)(6540,5752)
\path(6090,4852)(5640,5752) \path(6090,4852)(5640,5752)
\path(6090,4852)(5820,5392) \path(6090,4852)(5820,5392)
\whiten\path(5900.498,5298.085)(5820.000,5392.000)(5846.833,5271.252)(5889.765,5252.469)(5900.498,5298.085)
\path(5190,3052)(5460,3592)
\path(5433.167,3471.252)(5460.000,3592.000)(5379.502,3498.085)
\path(5190,3052)(5640,3952) \path(6090,3052)(5640,3952)
\path(6090,3052)(5640,3952) \path(6090,3052)(5820,3592)
\path(6090,3052)(5820,3592)
\whiten\path(5900.498,3498.085)(5820.000,3592.000)(5846.833,3471.252)(5889.765,3452.469)(5900.498,3498.085)
\path(6090,3052)(6360,3592)
\path(6333.167,3471.252)(6360.000,3592.000)(6279.502,3498.085)
\path(6090,3052)(6540,3952) \path(6990,3052)(6540,3952)
\path(6990,3052)(6540,3952) \path(6990,3052)(6720,3592)
\path(6990,3052)(6720,3592)
\whiten\path(6800.498,3498.085)(6720.000,3592.000)(6746.833,3471.252)(6789.765,3452.469)(6800.498,3498.085)
\path(6540,2152)(6810,2692)
\path(6783.167,2571.252)(6810.000,2692.000)(6729.502,2598.085)
\path(6540,2152)(6990,3052) \path(6540,2152)(6090,3052)
\path(6540,2152)(6090,3052) \path(6540,2152)(6270,2692)
\path(6540,2152)(6270,2692)
\whiten\path(6350.498,2598.085)(6270.000,2692.000)(6296.833,2571.252)(6339.765,2552.469)(6350.498,2598.085)
\path(5640,2152)(5910,2692)
\path(5883.167,2571.252)(5910.000,2692.000)(5829.502,2598.085)
\path(5640,2152)(6090,3052) \path(5640,2152)(5190,3052)
\path(5640,2152)(5190,3052) \path(5640,2152)(5370,2692)
\path(5640,2152)(5370,2692)
\whiten\path(5450.498,2598.085)(5370.000,2692.000)(5396.833,2571.252)(5439.765,2552.469)(5450.498,2598.085)
\path(6090,1252)(6360,1792)
\path(6333.167,1671.252)(6360.000,1792.000)(6279.502,1698.085)
\path(6090,1252)(6540,2152) \path(6090,1252)(5640,2152)
\path(6090,1252)(5640,2152) \path(6090,1252)(5820,1792)
\path(6090,1252)(5820,1792)
\whiten\path(5900.498,1698.085)(5820.000,1792.000)(5846.833,1671.252)(5889.765,1652.469)(5900.498,1698.085)
\path(5640,3952)(5910,4492)
\path(5883.167,4371.252)(5910.000,4492.000)(5829.502,4398.085)
\path(5640,3952)(6090,4852) \path(6540,3952)(6090,4852)
\path(6540,3952)(6090,4852) \path(6540,3952)(6270,4492)
\path(6540,3952)(6270,4492)
\whiten\path(6350.498,4398.085)(6270.000,4492.000)(6296.833,4371.252)(6339.765,4352.469)(6350.498,4398.085)
\put(6090,4852){\blacken\ellipse{46}{46}}
\put(6090,4852){\ellipse{46}{46}}
\path(5640,3952)(5190,4852) \path(5640,3952)(5190,4852)
\path(5640,3952)(5370,4492) \path(5640,3952)(5370,4492)
\whiten\path(5450.498,4398.085)(5370.000,4492.000)(5396.833,4371.252)(5439.765,4352.469)(5450.498,4398.085)
\put(5025,4402){\makebox(0,0)[lb]{\smash{{\SetFigFont{8}{14.4}{\rmdefault}{\mddefault}{\updefault}$\gamma_1$}}}}
\path(6990,3052)(6990,4852) \path(6990,3052)(6990,4852)
\path(6990,3052)(6990,4087) \path(6990,3052)(6990,4087)
\whiten\path(7020.000,3967.000)(6990.000,4087.000)(6960.000,3967.000)(6990.000,4003.000)(7020.000,3967.000)
\put(7255,4357){\makebox(0,0)[lb]{\smash{{\SetFigFont{8}{14.4}{\rmdefault}{\mddefault}{\updefault}$\gamma'_1$}}}}
\put(4640,3502){\makebox(0,0)[lb]{\smash{{\SetFigFont{8}{14.4}{\rmdefault}{\mddefault}{\updefault}$\gamma_1$}}}}
\put(1815,6652){\blacken\ellipse{46}{46}}
\put(1815,6652){\ellipse{46}{46}}
\put(2265,5752){\blacken\ellipse{46}{46}}
\put(2265,5752){\ellipse{46}{46}}
\put(1365,5752){\blacken\ellipse{46}{46}}
\put(1365,5752){\ellipse{46}{46}}
\put(465,3952){\blacken\ellipse{46}{46}}
\put(465,3952){\ellipse{46}{46}}
\put(1365,3952){\blacken\ellipse{46}{46}}
\put(1365,3952){\ellipse{46}{46}}
\put(2265,3952){\blacken\ellipse{46}{46}}
\put(2265,3952){\ellipse{46}{46}}
\put(3165,3952){\blacken\ellipse{46}{46}}
\put(3165,3952){\ellipse{46}{46}}
\put(2715,3052){\blacken\ellipse{46}{46}}
\put(2715,3052){\ellipse{46}{46}}
\put(1815,3052){\blacken\ellipse{46}{46}}
\put(1815,3052){\ellipse{46}{46}}
\put(2265,2152){\blacken\ellipse{46}{46}}
\put(2265,2152){\ellipse{46}{46}}
\put(1815,1252){\blacken\ellipse{46}{46}}
\put(1815,1252){\ellipse{46}{46}}
\put(915,3052){\blacken\ellipse{46}{46}}
\put(915,3052){\ellipse{46}{46}}
\put(1365,2152){\blacken\ellipse{46}{46}}
\put(1365,2152){\ellipse{46}{46}}
\put(2715,4852){\blacken\ellipse{46}{46}}
\put(2715,4852){\ellipse{46}{46}}
\put(6090,6652){\blacken\ellipse{46}{46}}
\put(6090,6652){\ellipse{46}{46}}
\put(6540,5752){\blacken\ellipse{46}{46}}
\put(6540,5752){\ellipse{46}{46}}
\put(5640,5752){\blacken\ellipse{46}{46}}
\put(5640,5752){\ellipse{46}{46}}
\put(5640,3952){\blacken\ellipse{46}{46}}
\put(5640,3952){\ellipse{46}{46}}
\put(6540,3952){\blacken\ellipse{46}{46}}
\put(6540,3952){\ellipse{46}{46}}
\put(6990,3052){\blacken\ellipse{46}{46}}
\put(6990,3052){\ellipse{46}{46}}
\put(6090,3052){\blacken\ellipse{46}{46}}
\put(6090,3052){\ellipse{46}{46}}
\put(6540,2152){\blacken\ellipse{46}{46}}
\put(6540,2152){\ellipse{46}{46}}
\put(6090,1252){\blacken\ellipse{46}{46}}
\put(6090,1252){\ellipse{46}{46}}
\put(5190,3052){\blacken\ellipse{46}{46}}
\put(5190,3052){\ellipse{46}{46}}
\put(5640,2152){\blacken\ellipse{46}{46}}
\put(5640,2152){\ellipse{46}{46}}
\put(6990,4852){\blacken\ellipse{46}{46}}
\put(6990,4852){\ellipse{46}{46}}
\put(5190,4852){\blacken\ellipse{46}{46}}
\put(5190,4852){\ellipse{46}{46}} \path(915,4852)(645,5392)
\path(915,4852)(645,5392)
\whiten\path(725.498,5298.085)(645.000,5392.000)(671.833,5271.252)(725.498,5298.085)
\path(1365,5752)(1095,6292) \path(1365,5752)(1095,6292)
\whiten\path(1175.498,6198.085)(1095.000,6292.000)(1121.833,6171.252)(1175.498,6198.085)
\path(1815,6652)(1545,7192) \path(1815,6652)(1545,7192)
\whiten\path(1625.498,7098.085)(1545.000,7192.000)(1571.833,7071.252)(1625.498,7098.085)
\path(3165,3952)(3390,3502) \path(3165,3952)(3390,3502)
\path(3390,3502)(3345,3592) \path(3390,3502)(3345,3592)
\whiten\path(3425.498,3498.085)(3345.000,3592.000)(3371.833,3471.252)(3425.498,3498.085)
\path(2715,3052)(2940,2602) \path(2715,3052)(2940,2602)
\path(2940,2602)(2895,2692) \path(2940,2602)(2895,2692)
\whiten\path(2975.498,2598.085)(2895.000,2692.000)(2921.833,2571.252)(2975.498,2598.085)
\path(2265,2152)(2490,1702) \path(2265,2152)(2490,1702)
\path(2490,1702)(2445,1792) \path(2490,1702)(2445,1792)
\whiten\path(2525.498,1698.085)(2445.000,1792.000)(2471.833,1671.252)(2525.498,1698.085)
\path(1815,1252)(2040,802) \path(1815,1252)(2040,802)
\path(2040,802)(1995,892) \path(2040,802)(1995,892)
\whiten\path(2075.498,798.085)(1995.000,892.000)(2021.833,771.252)(2075.498,798.085)
\thicklines \path(3705,3817)(4335,3817)
\blacken\path(4095.000,3757.000)(4335.000,3817.000)(4095.000,3877.000)(4167.000,3817.000)(4095.000,3757.000)
\thinlines \path(5190,4852)(4920,5392) \path(5190,4852)(4920,5392)
\whiten\path(5000.498,5298.085)(4920.000,5392.000)(4946.833,5271.252)(5000.498,5298.085)
\path(5640,5752)(5370,6292) \path(5640,5752)(5370,6292)
\whiten\path(5450.498,6198.085)(5370.000,6292.000)(5396.833,6171.252)(5450.498,6198.085)
\path(6090,6652)(5820,7192) \path(6090,6652)(5820,7192)
\whiten\path(5900.498,7098.085)(5820.000,7192.000)(5846.833,7071.252)(5900.498,7098.085)
\path(6990,3052)(7215,2602) \path(6990,3052)(7215,2602)
\path(7215,2602)(7170,2692) \path(7215,2602)(7170,2692)
\whiten\path(7250.498,2598.085)(7170.000,2692.000)(7196.833,2571.252)(7250.498,2598.085)
\path(6540,2152)(6765,1702) \path(6540,2152)(6765,1702)
\path(6765,1702)(6720,1792) \path(6765,1702)(6720,1792)
\whiten\path(6800.498,1698.085)(6720.000,1792.000)(6746.833,1671.252)(6800.498,1698.085)
\path(6090,1252)(6315,802) \path(6090,1252)(6315,802)
\path(6315,802)(6270,892) \path(6315,802)(6270,892)
\whiten\path(6350.498,798.085)(6270.000,892.000)(6296.833,771.252)(6350.498,798.085)
\path(5190,3052)(4920,3592) \path(5190,3052)(4920,3592)
\whiten\path(5000.498,3498.085)(4920.000,3592.000)(4946.833,3471.252)(5000.498,3498.085)
\path(5190,3052)(5190,4852) \path(5190,3052)(5190,4852)
\path(5190,3052)(5190,4087) \path(5190,3052)(5190,4087)
\whiten\path(5220.000,3967.000)(5190.000,4087.000)(5160.000,3967.000)(5190.000,4003.000)(5220.000,3967.000)
\path(7215,4402)(7170,4492) \path(7215,4402)(7170,4492)
\whiten\path(7250.498,4398.085)(7170.000,4492.000)(7196.833,4371.252)(7250.498,4398.085)
\path(6990,4852)(7175,4482) \path(6990,4852)(7175,4482)
\put(345,5302){\makebox(0,0)[lb]{\smash{{\SetFigFont{8}{14.4}{\rmdefault}{\mddefault}{\updefault}$\gamma_2$}}}}
\put(795,6202){\makebox(0,0)[lb]{\smash{{\SetFigFont{8}{14.4}{\rmdefault}{\mddefault}{\updefault}$\gamma_3$}}}}
\put(1245,7102){\makebox(0,0)[lb]{\smash{{\SetFigFont{8}{14.4}{\rmdefault}{\mddefault}{\updefault}$\gamma_4$}}}}
\put(1135,82){\makebox(0,0)[lb]{\smash{{\SetFigFont{14}{16.8}{\rmdefault}{\mddefault}{\updefault}fig.
8a}}}}
\put(5410,82){\makebox(0,0)[lb]{\smash{{\SetFigFont{14}{16.8}{\rmdefault}{\mddefault}{\updefault}fig.
8b}}}}
\put(4620,5302){\makebox(0,0)[lb]{\smash{{\SetFigFont{8}{14.4}{\rmdefault}{\mddefault}{\updefault}$\gamma_2$}}}}
\put(5070,6202){\makebox(0,0)[lb]{\smash{{\SetFigFont{8}{14.4}{\rmdefault}{\mddefault}{\updefault}$\gamma_3$}}}}
\put(5520,7102){\makebox(0,0)[lb]{\smash{{\SetFigFont{8}{14.4}{\rmdefault}{\mddefault}{\updefault}$\gamma_4$}}}}
\put(5510,3502){\makebox(0,0)[lb]{\smash{{\SetFigFont{8}{14.4}{\rmdefault}{\mddefault}{\updefault}$\gamma'_1$}}}}
\put(365,3502){\makebox(0,0)[lb]{\smash{{\SetFigFont{8}{14.4}{\rmdefault}{\mddefault}{\updefault}$\gamma'_1$}}}}
\put(785,2602){\makebox(0,0)[lb]{\smash{{\SetFigFont{8}{14.4}{\rmdefault}{\mddefault}{\updefault}$\gamma'_2$}}}}
\put(5060,2602){\makebox(0,0)[lb]{\smash{{\SetFigFont{8}{14.4}{\rmdefault}{\mddefault}{\updefault}$\gamma'_2$}}}}
\put(1235,1702){\makebox(0,0)[lb]{\smash{{\SetFigFont{8}{14.4}{\rmdefault}{\mddefault}{\updefault}$\gamma'_3$}}}}
\put(1655,802){\makebox(0,0)[lb]{\smash{{\SetFigFont{8}{14.4}{\rmdefault}{\mddefault}{\updefault}$\gamma'_4$}}}}
\put(5510,1702){\makebox(0,0)[lb]{\smash{{\SetFigFont{8}{14.4}{\rmdefault}{\mddefault}{\updefault}$\gamma'_3$}}}}
\put(5960,802){\makebox(0,0)[lb]{\smash{{\SetFigFont{8}{14.4}{\rmdefault}{\mddefault}{\updefault}$\gamma'_4$}}}}
\end{picture}
} \hspace{-10mm} \setlength{\unitlength}{0.0004in}
\begingroup\makeatletter\ifx\SetFigFont\undefined%
\gdef\SetFigFont#1#2#3#4#5{%
  \reset@font\fontsize{#1}{#2pt}%
  \fontfamily{#3}\fontseries{#4}\fontshape{#5}%
  \selectfont}%
\fi\endgroup%
{\renewcommand{\dashlinestretch}{30}
\begin{picture}(8360,7129)(0,-10)
\put(1867,6517){\blacken\ellipse{46}{46}}
\put(1867,6517){\ellipse{46}{46}}
\put(2317,5617){\blacken\ellipse{46}{46}}
\put(2317,5617){\ellipse{46}{46}}
\put(1417,5617){\blacken\ellipse{46}{46}}
\put(1417,5617){\ellipse{46}{46}}
\put(1417,3817){\blacken\ellipse{46}{46}}
\put(1417,3817){\ellipse{46}{46}}
\put(2317,3817){\blacken\ellipse{46}{46}}
\put(2317,3817){\ellipse{46}{46}}
\put(1867,2917){\blacken\ellipse{46}{46}}
\put(1867,2917){\ellipse{46}{46}}
\put(2317,2017){\blacken\ellipse{46}{46}}
\put(2317,2017){\ellipse{46}{46}}
\put(967,2917){\blacken\ellipse{46}{46}}
\put(967,2917){\ellipse{46}{46}}
\put(1417,2017){\blacken\ellipse{46}{46}}
\put(1417,2017){\ellipse{46}{46}}
\put(2767,4717){\blacken\ellipse{46}{46}}
\put(2767,4717){\ellipse{46}{46}}
\put(967,4717){\blacken\ellipse{46}{46}}
\put(967,4717){\ellipse{46}{46}}
\put(2767,2917){\blacken\ellipse{46}{46}}
\put(2767,2917){\ellipse{46}{46}}
\put(1867,1117){\blacken\ellipse{46}{46}}
\put(1867,1117){\ellipse{46}{46}}
\put(1867,4717){\blacken\ellipse{46}{46}}
\put(1867,4717){\ellipse{46}{46}}
\put(4702,6517){\blacken\ellipse{46}{46}}
\put(4702,6517){\ellipse{46}{46}}
\put(5152,5617){\blacken\ellipse{46}{46}}
\put(5152,5617){\ellipse{46}{46}}
\put(4252,5617){\blacken\ellipse{46}{46}}
\put(4252,5617){\ellipse{46}{46}}
\put(4702,1117){\blacken\ellipse{46}{46}}
\put(4702,1117){\ellipse{46}{46}}
\put(4252,2017){\blacken\ellipse{46}{46}}
\put(4252,2017){\ellipse{46}{46}}
\put(5152,2017){\blacken\ellipse{46}{46}}
\put(5152,2017){\ellipse{46}{46}}
\put(4252,3817){\blacken\ellipse{46}{46}}
\put(4252,3817){\ellipse{46}{46}}
\put(5152,3817){\blacken\ellipse{46}{46}}
\put(5152,3817){\ellipse{46}{46}}
\put(4702,2917){\blacken\ellipse{46}{46}}
\put(4702,2917){\ellipse{46}{46}}
\put(4702,4717){\blacken\ellipse{46}{46}}
\put(4702,4717){\ellipse{46}{46}}
\put(7267,6517){\blacken\ellipse{46}{46}}
\put(7267,6517){\ellipse{46}{46}}
\put(7267,2917){\blacken\ellipse{46}{46}}
\put(7267,2917){\ellipse{46}{46}}
\put(7267,1117){\blacken\ellipse{46}{46}}
\put(7267,1117){\ellipse{46}{46}}
\put(7267,4717){\blacken\ellipse{46}{46}}
\put(7267,4717){\ellipse{46}{46}}
\put(6817,3817){\blacken\ellipse{46}{46}}
\put(6817,3817){\ellipse{46}{46}}
\put(7717,3817){\blacken\ellipse{46}{46}}
\put(7717,3817){\ellipse{46}{46}} \path(967,4717)(697,5257)
\path(967,4717)(697,5257)
\whiten\path(777.498,5163.085)(697.000,5257.000)(723.833,5136.252)(777.498,5163.085)
\path(1417,5617)(1147,6157) \path(1417,5617)(1147,6157)
\whiten\path(1227.498,6063.085)(1147.000,6157.000)(1173.833,6036.252)(1227.498,6063.085)
\path(1867,6517)(1597,7057) \path(1867,6517)(1597,7057)
\whiten\path(1677.498,6963.085)(1597.000,7057.000)(1623.833,6936.252)(1677.498,6963.085)
\path(967,2917)(697,3457) \path(967,2917)(697,3457)
\whiten\path(777.498,3363.085)(697.000,3457.000)(723.833,3336.252)(777.498,3363.085)
\thicklines \path(22,3817)(652,3817)
\blacken\path(412.000,3757.000)(652.000,3817.000)(412.000,3877.000)(484.000,3817.000)(412.000,3757.000)
\path(5782,3817)(6412,3817)
\blacken\path(6172.000,3757.000)(6412.000,3817.000)(6172.000,3877.000)(6244.000,3817.000)(6172.000,3757.000)
\path(3082,3817)(3712,3817)
\blacken\path(3472.000,3757.000)(3712.000,3817.000)(3472.000,3877.000)(3544.000,3817.000)(3472.000,3757.000)
\thinlines \path(967,4717)(1237,5257)
\path(1210.167,5136.252)(1237.000,5257.000)(1156.502,5163.085)
\path(967,4717)(1417,5617) \path(1867,4717)(1417,5617)
\path(1867,4717)(1417,5617) \path(1867,4717)(1597,5257)
\path(1867,4717)(1597,5257)
\whiten\path(1677.498,5163.085)(1597.000,5257.000)(1623.833,5136.252)(1666.765,5117.469)(1677.498,5163.085)
\path(2767,4717)(2992,4267) \path(2767,4717)(2992,4267)
\path(2317,3817)(2587,4357)
\path(2560.167,4236.252)(2587.000,4357.000)(2506.502,4263.085)
\path(2317,3817)(2767,4717) \path(2767,4717)(2317,5617)
\path(2767,4717)(2317,5617) \path(2767,4717)(2497,5257)
\path(2767,4717)(2497,5257)
\whiten\path(2577.498,5163.085)(2497.000,5257.000)(2523.833,5136.252)(2566.765,5117.469)(2577.498,5163.085)
\path(967,2917)(1237,3457)
\path(1210.167,3336.252)(1237.000,3457.000)(1156.502,3363.085)
\path(967,2917)(1417,3817) \path(1417,3817)(1687,4357)
\path(1660.167,4236.252)(1687.000,4357.000)(1606.502,4263.085)
\path(1417,3817)(1867,4717) \path(2767,2917)(2992,2467)
\path(2767,2917)(2992,2467) \path(1867,2917)(1417,3817)
\path(1867,2917)(1417,3817) \path(1417,3817)(967,4717)
\path(1867,2917)(2137,3457)
\path(2110.167,3336.252)(2137.000,3457.000)(2056.502,3363.085)
\path(1867,2917)(2317,3817) \path(2317,2017)(2587,2557)
\path(2560.167,2436.252)(2587.000,2557.000)(2506.502,2463.085)
\path(2317,2017)(2767,2917) \path(2317,2017)(2542,1567)
\path(2317,2017)(2542,1567) \path(1417,2017)(1687,2557)
\path(1660.167,2436.252)(1687.000,2557.000)(1606.502,2463.085)
\path(1417,2017)(1867,2917) \path(1867,1117)(2137,1657)
\path(2110.167,1536.252)(2137.000,1657.000)(2056.502,1563.085)
\path(1867,1117)(2317,2017) \path(1867,1117)(2092,667)
\path(1867,1117)(2092,667) \path(1867,1117)(1417,2017)
\path(1867,1117)(1417,2017) \path(4252,5617)(3982,6157)
\path(4252,5617)(3982,6157)
\whiten\path(4062.498,6063.085)(3982.000,6157.000)(4008.833,6036.252)(4062.498,6063.085)
\path(4702,6517)(4432,7057) \path(4702,6517)(4432,7057)
\whiten\path(4512.498,6963.085)(4432.000,7057.000)(4458.833,6936.252)(4512.498,6963.085)
\path(4252,2017)(3982,2557) \path(4252,2017)(3982,2557)
\whiten\path(4062.498,2463.085)(3982.000,2557.000)(4008.833,2436.252)(4062.498,2463.085)
\path(4702,1117)(4972,1657)
\path(4945.167,1536.252)(4972.000,1657.000)(4891.502,1563.085)
\path(4702,1117)(5152,2017) \path(4702,1117)(4927,667)
\path(4702,1117)(4927,667) \path(5152,2017)(4702,2917)
\path(5152,2017)(4702,2917) \path(5152,2017)(4882,2557)
\path(5152,2017)(4882,2557)
\whiten\path(4962.498,2463.085)(4882.000,2557.000)(4908.833,2436.252)(4951.765,2417.469)(4962.498,2463.085)
\path(5152,2017)(5377,1567) \path(5152,2017)(5377,1567)
\path(4702,1117)(4252,2017) \path(4702,1117)(4252,2017)
\path(4702,1117)(4432,1657) \path(4702,1117)(4432,1657)
\whiten\path(4512.498,1563.085)(4432.000,1657.000)(4458.833,1536.252)(4501.765,1517.469)(4512.498,1563.085)
\path(4252,2017)(4522,2557)
\path(4495.167,2436.252)(4522.000,2557.000)(4441.502,2463.085)
\path(4252,2017)(4702,2917) \path(4702,2917)(4252,3817)
\path(4702,2917)(4252,3817) \path(4702,2917)(4432,3457)
\path(4702,2917)(4432,3457)
\whiten\path(4512.498,3363.085)(4432.000,3457.000)(4458.833,3336.252)(4501.765,3317.469)(4512.498,3363.085)
\path(5152,3817)(5377,3367) \path(5152,3817)(5377,3367)
\path(4702,2917)(4972,3457)
\path(4945.167,3336.252)(4972.000,3457.000)(4891.502,3363.085)
\path(4702,2917)(5152,3817) \path(4252,3817)(3982,4357)
\path(4252,3817)(3982,4357)
\whiten\path(4062.498,4263.085)(3982.000,4357.000)(4008.833,4236.252)(4062.498,4263.085)
\path(5152,5617)(4702,6517) \path(5152,5617)(4702,6517)
\path(5152,5617)(4882,6157) \path(5152,5617)(4882,6157)
\whiten\path(4962.498,6063.085)(4882.000,6157.000)(4908.833,6036.252)(4951.765,6017.469)(4962.498,6063.085)
\path(4252,5617)(4522,6157)
\path(4495.167,6036.252)(4522.000,6157.000)(4441.502,6063.085)
\path(4252,5617)(4702,6517) \path(5152,5617)(5377,5167)
\path(5152,5617)(5377,5167) \path(4702,4717)(4252,5617)
\path(4702,4717)(4252,5617) \path(4702,4717)(4432,5257)
\path(4702,4717)(4432,5257)
\whiten\path(4512.498,5163.085)(4432.000,5257.000)(4458.833,5136.252)(4501.765,5117.469)(4512.498,5163.085)
\path(4702,4717)(4972,5257)
\path(4945.167,5136.252)(4972.000,5257.000)(4891.502,5163.085)
\path(4702,4717)(5152,5617) \path(4252,3817)(4522,4357)
\path(4495.167,4236.252)(4522.000,4357.000)(4441.502,4263.085)
\path(4252,3817)(4702,4717) \path(5152,3817)(4702,4717)
\path(5152,3817)(4702,4717) \path(5152,3817)(4882,4357)
\path(5152,3817)(4882,4357)
\whiten\path(4962.498,4263.085)(4882.000,4357.000)(4908.833,4236.252)(4951.765,4217.469)(4962.498,4263.085)
\path(7267,6517)(6997,7057) \path(7267,6517)(6997,7057)
\whiten\path(7077.498,6963.085)(6997.000,7057.000)(7023.833,6936.252)(7077.498,6963.085)
\path(7267,1117)(6997,1657) \path(7267,1117)(6997,1657)
\whiten\path(7077.498,1563.085)(6997.000,1657.000)(7023.833,1536.252)(7077.498,1563.085)
\path(7267,4717)(6997,5257) \path(7267,4717)(6997,5257)
\whiten\path(7077.498,5163.085)(6997.000,5257.000)(7023.833,5136.252)(7077.498,5163.085)
\path(7267,6517)(7492,6067) \path(7267,6517)(7492,6067)
\path(7717,3817)(7267,4717) \path(7717,3817)(7267,4717)
\path(7717,3817)(7447,4357) \path(7717,3817)(7447,4357)
\whiten\path(7527.498,4263.085)(7447.000,4357.000)(7473.833,4236.252)(7516.765,4217.469)(7527.498,4263.085)
\path(6817,3817)(7087,4357)
\path(7060.167,4236.252)(7087.000,4357.000)(7006.502,4263.085)
\path(6817,3817)(7267,4717) \path(6817,3817)(6547,4357)
\path(6817,3817)(6547,4357)
\whiten\path(6627.498,4263.085)(6547.000,4357.000)(6573.833,4236.252)(6627.498,4263.085)
\path(7267,2917)(7537,3457)
\path(7510.167,3336.252)(7537.000,3457.000)(7456.502,3363.085)
\path(7267,2917)(7717,3817) \path(7267,2917)(6817,3817)
\path(7267,2917)(6817,3817) \path(7267,2917)(6997,3457)
\path(7267,2917)(6997,3457)
\whiten\path(7077.498,3363.085)(6997.000,3457.000)(7023.833,3336.252)(7066.765,3317.469)(7077.498,3363.085)
\path(7717,3817)(7942,3367) \path(7717,3817)(7942,3367)
\path(7267,2917)(7492,2467) \path(7267,2917)(7492,2467)
\path(7267,1117)(7492,667) \path(7267,1117)(7492,667)
\path(1867,4717)(2137,5257)
\path(2110.167,5136.252)(2137.000,5257.000)(2056.502,5163.085)
\path(1867,4717)(2317,5617) \path(2317,5617)(1867,6517)
\path(2317,5617)(1867,6517) \path(2317,5617)(2047,6157)
\path(2317,5617)(2047,6157)
\whiten\path(2127.498,6063.085)(2047.000,6157.000)(2073.833,6036.252)(2116.765,6017.469)(2127.498,6063.085)
\path(1417,5617)(1687,6157)
\path(1660.167,6036.252)(1687.000,6157.000)(1606.502,6063.085)
\path(1417,5617)(1867,6517) \path(2317,3817)(1867,4717)
\path(2317,3817)(1867,4717) \path(2317,3817)(2047,4357)
\path(2317,3817)(2047,4357)
\whiten\path(2127.498,4263.085)(2047.000,4357.000)(2073.833,4236.252)(2116.765,4217.469)(2127.498,4263.085)
\put(1097,82){\makebox(0,0)[lb]{\smash{{\SetFigFont{14}{16.8}{\rmdefault}{\mddefault}{\updefault}fig.
8c}}}}
\put(842,4267){\makebox(0,0)[lb]{\smash{{\SetFigFont{8}{12.0}{\rmdefault}{\mddefault}{\updefault}$\gamma'_1$}}}}
\put(387,3367){\makebox(0,0)[lb]{\smash{{\SetFigFont{8}{12.0}{\rmdefault}{\mddefault}{\updefault}$\gamma_1$}}}}
\put(3932,82){\makebox(0,0)[lb]{\smash{{\SetFigFont{14}{16.8}{\rmdefault}{\mddefault}{\updefault}fig.
8d}}}}
\put(6452,82){\makebox(0,0)[lb]{\smash{{\SetFigFont{14}{16.8}{\rmdefault}{\mddefault}{\updefault}fig.
8e}}}}
\put(287,5167){\makebox(0,0)[lb]{\smash{{\SetFigFont{8}{12.0}{\rmdefault}{\mddefault}{\updefault}$\gamma_2$}}}}
\put(837,2467){\makebox(0,0)[lb]{\smash{{\SetFigFont{8}{12.0}{\rmdefault}{\mddefault}{\updefault}$\gamma'_2$}}}}
\put(1287,1567){\makebox(0,0)[lb]{\smash{{\SetFigFont{8}{12.0}{\rmdefault}{\mddefault}{\updefault}$\gamma'_3$}}}}
\put(1737,667){\makebox(0,0)[lb]{\smash{{\SetFigFont{8}{12.0}{\rmdefault}{\mddefault}{\updefault}$\gamma'_4$}}}}
\put(837,6067){\makebox(0,0)[lb]{\smash{{\SetFigFont{8}{12.0}{\rmdefault}{\mddefault}{\updefault}$\gamma_3$}}}}
\put(1287,6967){\makebox(0,0)[lb]{\smash{{\SetFigFont{8}{12.0}{\rmdefault}{\mddefault}{\updefault}$\gamma_4$}}}}
\put(4122,5167){\makebox(0,0)[lb]{\smash{{\SetFigFont{8}{12.0}{\rmdefault}{\mddefault}{\updefault}$\gamma'_1$}}}}
\put(3672,6067){\makebox(0,0)[lb]{\smash{{\SetFigFont{8}{12.0}{\rmdefault}{\mddefault}{\updefault}$\gamma_3$}}}}
\put(3672,4267){\makebox(0,0)[lb]{\smash{{\SetFigFont{8}{12.0}{\rmdefault}{\mddefault}{\updefault}$\gamma_2$}}}}
\put(4122,3367){\makebox(0,0)[lb]{\smash{{\SetFigFont{8}{12.0}{\rmdefault}{\mddefault}{\updefault}$\gamma'_2$}}}}
\put(3672,2467){\makebox(0,0)[lb]{\smash{{\SetFigFont{8}{12.0}{\rmdefault}{\mddefault}{\updefault}$\gamma_1$}}}}
\put(4122,1567){\makebox(0,0)[lb]{\smash{{\SetFigFont{8}{12.0}{\rmdefault}{\mddefault}{\updefault}$\gamma'_3$}}}}
\put(4572,667){\makebox(0,0)[lb]{\smash{{\SetFigFont{8}{12.0}{\rmdefault}{\mddefault}{\updefault}$\gamma'_4$}}}}
\put(6687,1567){\makebox(0,0)[lb]{\smash{{\SetFigFont{8}{12.0}{\rmdefault}{\mddefault}{\updefault}$\gamma_1$}}}}
\put(7137,667){\makebox(0,0)[lb]{\smash{{\SetFigFont{8}{12.0}{\rmdefault}{\mddefault}{\updefault}$\gamma'_4$}}}}
\put(6687,3367){\makebox(0,0)[lb]{\smash{{\SetFigFont{8}{12.0}{\rmdefault}{\mddefault}{\updefault}$\gamma'_2$}}}}
\put(7137,2467){\makebox(0,0)[lb]{\smash{{\SetFigFont{8}{12.0}{\rmdefault}{\mddefault}{\updefault}$\gamma'_3$}}}}
\put(6237,4267){\makebox(0,0)[lb]{\smash{{\SetFigFont{8}{12.0}{\rmdefault}{\mddefault}{\updefault}$\gamma_2$}}}}
\put(6687,5167){\makebox(0,0)[lb]{\smash{{\SetFigFont{8}{12.0}{\rmdefault}{\mddefault}{\updefault}$\gamma_3$}}}}
\put(7137,6067){\makebox(0,0)[lb]{\smash{{\SetFigFont{8}{12.0}{\rmdefault}{\mddefault}{\updefault}$\gamma'_1$}}}}
\put(6687,6967){\makebox(0,0)[lb]{\smash{{\SetFigFont{8}{12.0}{\rmdefault}{\mddefault}{\updefault}$\gamma_4$}}}}
\put(4122,6967){\makebox(0,0)[lb]{\smash{{\SetFigFont{8}{12.0}{\rmdefault}{\mddefault}{\updefault}$\gamma_4$}}}}
\path(2992,4267)(2947,4357) \path(2992,4267)(2947,4357)
\whiten\path(3027.498,4263.085)(2947.000,4357.000)(2973.833,4236.252)(3027.498,4263.085)
\path(2992,2467)(2947,2557) \path(2992,2467)(2947,2557)
\whiten\path(3027.498,2463.085)(2947.000,2557.000)(2973.833,2436.252)(3027.498,2463.085)
\path(1867,2917)(1597,3457) \path(1867,2917)(1597,3457)
\whiten\path(1677.498,3363.085)(1597.000,3457.000)(1623.833,3336.252)(1666.765,3317.469)(1677.498,3363.085)
\path(1417,3817)(1147,4357) \path(1417,3817)(1147,4357)
\whiten\path(1227.498,4263.085)(1147.000,4357.000)(1173.833,4236.252)(1216.765,4217.469)(1227.498,4263.085)
\path(2767,2917)(2317,3817) \path(2767,2917)(2317,3817)
\path(2542,1567)(2497,1657) \path(2542,1567)(2497,1657)
\whiten\path(2577.498,1563.085)(2497.000,1657.000)(2523.833,1536.252)(2577.498,1563.085)
\path(2092,667)(2047,757) \path(2092,667)(2047,757)
\whiten\path(2127.498,663.085)(2047.000,757.000)(2073.833,636.252)(2127.498,663.085)
\path(1417,2017)(967,2917) \path(1417,2017)(967,2917)
\path(2317,2017)(1867,2917) \path(2317,2017)(1867,2917)
\path(2767,2917)(2497,3457) \path(2767,2917)(2497,3457)
\whiten\path(2577.498,3363.085)(2497.000,3457.000)(2523.833,3336.252)(2566.765,3317.469)(2577.498,3363.085)
\path(2317,2017)(2047,2557) \path(2317,2017)(2047,2557)
\whiten\path(2127.498,2463.085)(2047.000,2557.000)(2073.833,2436.252)(2116.765,2417.469)(2127.498,2463.085)
\path(1417,2017)(1147,2557) \path(1417,2017)(1147,2557)
\whiten\path(1227.498,2463.085)(1147.000,2557.000)(1173.833,2436.252)(1216.765,2417.469)(1227.498,2463.085)
\path(1867,1117)(1597,1657) \path(1867,1117)(1597,1657)
\whiten\path(1677.498,1563.085)(1597.000,1657.000)(1623.833,1536.252)(1666.765,1517.469)(1677.498,1563.085)
\path(4927,667)(4882,757) \path(4927,667)(4882,757)
\whiten\path(4962.498,663.085)(4882.000,757.000)(4908.833,636.252)(4962.498,663.085)
\path(5377,1567)(5332,1657) \path(5377,1567)(5332,1657)
\whiten\path(5412.498,1563.085)(5332.000,1657.000)(5358.833,1536.252)(5412.498,1563.085)
\path(5377,3367)(5332,3457) \path(5377,3367)(5332,3457)
\whiten\path(5412.498,3363.085)(5332.000,3457.000)(5358.833,3336.252)(5412.498,3363.085)
\path(5377,5167)(5332,5257) \path(5377,5167)(5332,5257)
\whiten\path(5412.498,5163.085)(5332.000,5257.000)(5358.833,5136.252)(5412.498,5163.085)
\path(7492,6067)(7447,6157) \path(7492,6067)(7447,6157)
\whiten\path(7527.498,6063.085)(7447.000,6157.000)(7473.833,6036.252)(7527.498,6063.085)
\path(7942,3367)(7897,3457) \path(7942,3367)(7897,3457)
\whiten\path(7977.498,3363.085)(7897.000,3457.000)(7923.833,3336.252)(7977.498,3363.085)
\path(7492,2467)(7447,2557) \path(7492,2467)(7447,2557)
\whiten\path(7527.498,2463.085)(7447.000,2557.000)(7473.833,2436.252)(7527.498,2463.085)
\path(7492,667)(7447,757) \path(7492,667)(7447,757)
\whiten\path(7527.498,663.085)(7447.000,757.000)(7473.833,636.252)(7527.498,663.085)
\end{picture}
}
\end{figure}

The left and the right hand side of the equality~\eqref{A_int_psi}
are not agreed with each other because the right one is not
symmetric with respect to $\gamma_1,\ldots,\gamma_N$. It is
consequence of the fact that the singularities of the
$\Gamma$-functions were not taken into account. In order to treat
this problem we consider an equality~\eqref{A_orthog} as an equality
of distributions and check this equality on the basic functions
$\varphi(\gamma,\gamma')$ vanishing in some neighborhood of the
union of lines $\{\gamma_k=\gamma'_j\}_{k+j\le N}$. This
restriction allow us to multiply distributions by the
$\Gamma$-functions arising in~\eqref{A_int_psi}. But for these basic
functions only one term in the right hand side of~\eqref{A_orthog}
remain, which corresponds to the longest permutation
$\sigma={\,1\,\,2\,\ldots N \choose N\ldots\,2\,\,1\,}$. This
explains that a direct computation of the integral~\eqref{A_orthog}
using diagram technique produces only one (instead of $N!$)
$\delta$-function term. Other terms can be rebuilt by the reason
of symmetry. Analogous treatment of the similar calculation for
$XXX$-model was discussed in~\cite{Derkachov}. \\

After substitution of $\gamma'_j=\gamma_{N+1-j}$ in the product of
$\Gamma$-functions in right hand side of~\eqref{A_int_psi} we derive
the expression~\eqref{A_measure} for the Sklyanin measure. \qed \\

\section{Properties of the $\Lambda$-operators}
\label{A_Pr_Lambda}

Baxter's $Q$-operators $\Hat Q_N(u)$ described in~\cite{Pasquier}
for the periodic Toda chain satisfy the following properties.
\begin{itemize}
\item[(a)] These operators commute for different values of the spectral parameters: \\
$[\hat Q_N(u),\hat Q_N(v)]=0$.
\item[(b)] They commute with the transfer matrix $\hat
t_N(u)\hm=A_N(u)\hm+D_N(u)$ of the periodic Toda chain model:
$[\hat Q_N(u),\hat t_N(v)]=0$.
\item[(c)] $Q$-operator satisfies the Baxter equation
\begin{equation*} \label{A_Baxter_od_}
  \hat t_N(u)\hat Q_N(u)=i^N \hat Q_N(u+i\hbar)(x)+i^{-N} \hat Q_N(u-i\hbar).
\end{equation*}
\end{itemize}

In this section the similar properties for the operators
$\Lambda_N(u)$ defined in section~\ref{A_Eig_OTCh} will be
established.

\begin{prop}
 $\Lambda$-operator has the following properties:
\begin{itemize}
\item[{\scshape (i)}] $\Lambda_N(u)\Lambda_{N-1}(v)=\Lambda_N(v)\Lambda_{N-1}(u),$

\item[{\scshape (ii)}] $A_N(u)\Lambda_N(v)=(u-v)\Lambda_N(v)A_{N-1}(u),$

\item[{\scshape (iii)}] $C_N(u)\Lambda_N(u)=i^{-N-1}\Lambda_N(u-i\hbar),$

\item[{\scshape (iv)}] $B_N(u)\Lambda_N(u)=i^{N-1}\Lambda_N(u+i\hbar).$
\end{itemize}
\end{prop}

{\bfseries Proof.} The equality \textsc{(i)} was proved in the
theorem~\ref{A_Th_psi}.
Property \textsc{(ii)} have been proved in \cite{Kharchev_GG} by
the direct calculation. Below we give a simplified version of this
property. Let us
notice that \textsc{(ii)} is valid on arbitrary $N-1$-particle eigenfunction
$\psi_{\gamma_1,\ldots,\gamma_{N-1}}(y_1,\ldots, y_{N-1})$.
Indeed, in according with the equation~\eqref{A_psi_def} the left hand
side of \textsc{(ii)} is
\begin{multline} \label{A_lhs_pr_ii}
A_N(u)(\Lambda_N(v)\psi_{\gamma_1,\ldots,\gamma_{N-1}})(x_1,\ldots, x_N)=
A_N(u)\psi_{\gamma_1,\ldots,\gamma_{N-1},v}(x_1,\ldots, x_N)=\\
 =(u-v)\prod_{k=1}^{N-1}(u-\gamma_k)\psi_{\gamma_1,\ldots,\gamma_{N-1},v}(x_1,\ldots, x_N)
\end{multline}
while the right hand side is
\begin{multline} \label{A_rhs_pr_ii}
(u-v)(\Lambda_N(v)A_{N-1}(u)\psi_{\gamma_1,\ldots,\gamma_{N-1}})(x_1,\ldots,x_N)=\\
 =(u-v)\prod_{k=1}^{N-1}(u-\gamma_k)(\Lambda_N(v)\psi_{\gamma_1,\ldots,\gamma_{N-1}})(x_1,\ldots, x_N)=\\
  =(u-v)\prod_{k=1}^{N-1}(u-\gamma_k)\psi_{\gamma_1,\ldots,\gamma_{N-1},v}(x_1,\ldots, x_N).
\end{multline}
Comparing~\eqref{A_lhs_pr_ii} and \eqref{A_rhs_pr_ii} one concludes
that the equality \textsc{(ii)} is valid on the functions
$\psi_{\gamma}(x)$. Due to the completeness of the system of these
functions~\eqref{A_compl} any function that belongs to the domain of definition of
the operators $A_{N-1}(u)$ and $\Lambda_N(u)$ can be represented
in the form~\eqref{A_f_psi}. Therefore the property \textsc{(ii)} is
valid on any function of $N-1$ variables where the action of the
operators $A_{N-1}(u)$ and $\Lambda_N(v)$ is well defined.  \\

To prove \textsc{(iii)} we need the formula for an action of the
operators $A_m(u)$ on the kernel of the operator $\Lambda_N(u)$:
\begin{equation} \label{A_AmLam}
 A_m(u)\Lambda_{u}(x_1,\ldots,x_N;y_1,\ldots,y_{N-1})=
 (-i)^{m}e^{\sum\limits_{j=1}^m(x_j-y_j)}\Lambda_{u}(x_1,\ldots,x_N;y_1,\ldots,y_{N-1})
\end{equation}
for $m=0,\ldots,N-1$, ($A_0(u)\hm=1$ is implied). In the case
$m\hm=0$ and $m\hm=1$ this equality is obvious. For
$m=2,\ldots,N-1$ it can be proved by induction using the
relation~\eqref{A_AArec}. Indeed
\begin{multline} \label{A_iii_proof_ind}
 A_m(u)\Lambda_{u}(x;y)=(u+i\hbar\partial_{x_m})A_{m-1}(u)\Lambda_{u}(x;y)-e^{x_{m-1}-x_m}A_{m-2}(u)\Lambda_{u}(x;y)=\\
 =\bigl[(-i)^{m}e^{\sum\limits_{j=1}^{m-1}(x_j-y_j)}(e^{x_m-y_m}-e^{y_{m-1}-x_m})-(-i)^{m-2}e^{x_{m-1}-x_m}e^{\sum\limits_{j=1}^{m-2}(x_j-y_j)}\bigr]\Lambda_{u}(x;y)=\\
 =(-i)^{m}e^{\sum\limits_{j=1}^m(x_j-y_j)}\Lambda_{u}(x;y).
\end{multline}
The formula~\eqref{A_AmLam} for $m=N-1$ and the
relation~\eqref{A_Crec} give
\begin{multline*}
 C_N(u)\Lambda_{u}(x;y)=-e^{x_N}A_{N-1}(u)\Lambda_{u}(x;y)= \\
  =-(-i)^{N-1}e^{x_N+\sum_{j=1}^{N-1}(x_j-y_j)}\Lambda_{u}(x;y)=(-i)^{N+1}\Lambda_{u-i\hbar}(x;y),
\end{multline*}
what means in turn the equality \textsc{(iii)}. \\

Note that if one sets  $m=N$ in~\eqref{A_iii_proof_ind} then the
term containing $e^{x_m-y_m}$ does not arise and, consequently, we
obtain zero in the right hand side. Thus we have proved directly
that the kernel $\Lambda_{u}(x;y)$ satisfies to the
equation~\eqref{A_Apsi0}. \\

Analogously, to prove \textsc{(iv)} we used
\begin{equation} \label{A_BBrec}
 B_N(u)=(u-p_N)B_{N-1}(u)-e^{x_{N-1}-x_N}B_{N-2}(u),
\end{equation}
and the action of operator $B_m(u)$ on the kernel of the operator
$\Lambda_N(u)$
\begin{equation} \label{A_BmLam}
 B_m(u)\Lambda_{u}(x;y)=e^{-x_1}(-i)^{m-1}\sum_{k=1}^m(-1)^{k+1}\prod_{j=2}^k e^{y_{j-1}-x_j}
   \prod_{s=k+1}^m e^{x_s-y_s}\Lambda_{u}(x;y)
\end{equation}
which can be also proved by induction.
Exploring the relation~\eqref{A_BBrec} and the formula~\eqref{A_BmLam}
for $m=N-1, N-2$ we obtain
\begin{multline*}
 B_{N}(u)\Lambda_{u}(x;y)=e^{-x_1}(-i)^{N-1}\sum_{k=1}^{N-1} (-1)^{k+1}\prod_{j=2}^k e^{y_{j-1}-x_j}
    \prod_{s=k+1}^{N-1} e^{x_s-y_s}(-e^{y_{N-1}-x_N})\Lambda_{u}(x;y)-\\
  -e^{-x_1}(-i)^{N-3}\sum_{k=1}^{N-2} (-1)^{k+1}\prod_{j=2}^k e^{y_{j-1}-x_j}
    \prod_{s=k+1}^{N-2} e^{x_s-y_s} e^{x_{N-1}-x_N}\Lambda_{u}(x;y)=\\
  =-e^{-x_1}(-i)^{N-1}(-1)^N \prod_{j=2}^N e^{y_{j-1}-x_j}\Lambda_{u}(x;y)=i^{N-1}\Lambda_{u+i\hbar}(x;y).
\end{multline*}

Let us note that the equality \textsc{(iii)}, \textsc{(iv)} up to
the factors $(-i)^{\mp N-1}$ can be derived from \textsc{(ii)}
using R-matrix formalism (see for example~\cite{Kharchev_OP},
\cite{Kharchev_O}, \cite{Kharchev_P}).

\section{Conclusion}
\label{A_Conclusion}

As it was mentioned in the introduction the separation of variables is achieved by a special choice of the transition function. In the case of Toda chain this function should be chosen as follows~\cite{Kharchev_P}:
\begin{equation} \label{A_UEq}
 U_{\varepsilon,\gamma_1,\ldots,\gamma_{N-1}}(x_1,\ldots,x_N)=
  e^{\frac{i}{\hbar}\bigl(\varepsilon-\sum\limits_{j=1}^{N-1}\gamma_j\bigr)x_N}
   \psi_{\gamma_1,\ldots,\gamma_{N-1}}(x_1,\ldots,x_{N-1}).
\end{equation}
 The function~\eqref{A_UEq} can be calculated by induction: having an expression for the $(N-1)$-particle transition function one can yield the expression for the $N$-particle one. This method was proposed in the work~\cite{Kharchev_O}.  The authors of~\cite{Kharchev_O} obtain the recurrent formula integrating over the variables $\gamma_j$:
\begin{equation}  \label{A_psi_rec_gamma}
 \begin{split}
 \psi_{\lambda_1,\ldots,\lambda_N}&(x_1,\ldots,x_N)= \\
   &=\int\limits_{\mathbb R^{N-1}}
  e^{\frac{i}{\hbar}\bigl(\sum\limits_{k=1}^N\lambda_k-\sum\limits_{j=1}^{N-1}\gamma_j\bigr)x_N}
  \psi_{\gamma_1,\ldots,\gamma_{N-1}}(x_1,\ldots,x_{N-1})
  K(\lambda;\gamma) \mu(\gamma)d\gamma,
 \end{split}
\end{equation}
where $\mu(\gamma)$ is a Sklyanin measure described in
section~\ref{A_Int_meas} and $K(\lambda;\gamma)$ is some kernel.
This formula leads to the Mellin-Barns representation for the
transition function. \\

In the present paper it is shown that the recurrent integration
can be realized in terms of the coordinates $x_n$
\big(eq.~\eqref{A_psi_rec}\big). This leads in turn to the
Gauss-Givental representation~\eqref{A_psi_GG}, which was obtained
in~\cite{Givental}, \cite{Kharchev_GG} from the other circle of
ideas. The integration over the coordinates in the
formula~\eqref{A_psi_rec} implies actually an action of the operator
$\Lambda_N(\gamma_N)$ on the $(N-1)$-particle eigenfunction. The
function~\eqref{A_UEq} can be rewritten then in terms of the product
of $\Lambda$-operators like~\eqref{A_psi}. This results in the fact
that the function~\eqref{A_UEq} inherits the properties of the
$\Lambda$-operators discussed in section~\ref{A_Pr_Lambda}. For
example, the Weyl-invariance of this function is encoded in the
property \textsc{(i)}. Due to the properties \textsc{(iii)} and
\textsc{(iv)} the eigenfunction for the periodic Toda chain in the
new variables $\Phi_\varepsilon(\gamma)$ should satisfy the Baxter
equation. The last fact leads to the separation of variables, that
is this function decomposes in the product of one-variable
functions:
 $\Phi_\varepsilon(\gamma_1,\ldots,\gamma_{N-1})=\prod\limits_{j=1}^{N-1}c_\varepsilon(\gamma_j)$.
See details in~\cite{Kharchev_P}. \\

Availability of two kind of recurrent formulae -- of type~\eqref{A_psi_rec_gamma} and of type~\eqref{A_psi_rec} -- is explained by the fact that the function $\psi_\gamma(x)$ can be regarded as well as a function of $x_n$ satisfying the differential equations~\eqref{A_psi_def} and in other hand as a function of $\gamma_j$ satisfying difference equations in $\gamma_j$ \cite{Babelon}, i.e. as a wave function of some dual model. The duality of the same kind appears in the Representation Theory. The infinite-dimensional Gelfand-Zetlin representation of Lie algebra $\mathfrak{gl}(N)$ by shift operators in $\gamma_j$ allows to obtain the Mellin-Barns integral representation~\cite{Kharchev_GZ}, while the Gauss representation of the same Lie algebra by differential operators in $x_n$ leads to Gauss-Givental representation~\cite{Kharchev_GG}. \\

We hope the method proposed for the $XXX$-model in~\cite{Derkachov} and developed here for the Toda chain (including the use of diagram technique) can be applied to other more complicated integrable systems.

\section*{Acknowledgments}

This work is a part of PhD which the author is preparing at Bogoliubov Laboratory of Theoretical Physics (JINR, Dubna, Russia) and at LAREMA (UMR 6093 du CNRS, Univ.\,of Angers, France). He thanks both laboratories for excellent work conditions. He is thankful also to the French-Russian Network in Theoretical Physics and to prof. J.-M.\,Maillet for financial support for this joint PhD programme.\\

Author thanks to S.\,M.\,Kharchev, S.\,Z.\,Pakuliak and V.\,N.\,Rubtsov for useful advices and remarks.

\chapter{Classical elliptic current algebras. I}
\label{SA21}

\thispagestyle{empty} \pagestyle{myheadings}
\markboth{}{S.\,Pakuliak, V.\,Rubtsov, A.\,Silantyev\hfil{\itshape Classical elliptic current algebras. I}}

\newpage \thispagestyle{empty}

\vspace*{3cm}
\begin{center}
{\LARGE Classical elliptic current algebras. I} \\[17mm]
{\large Stanislav PAKULIAK~$^a$, Vladimir RUBTSOV~$^b$ and \\[4mm]
Alexey SILANTYEV~$^c$} \\[8mm]
$^a$ Institute of Theoretical \& Experimental Physics, 117259 Moscow, Russia \\
Laboratory of Theoretical Physics, JINR, 141980 Dubna, Moscow reg., Russia \\
~~E-mail: pakuliak@theor.jinr.ru \\[4mm]
$^b$ Institute of Theoretical \& Experimental Physics, 117259 Moscow, Russia  \\
D\'epartment de Math\'ematiques, Universit\'e d'Angers, 2 Bd. Lavoisier, 49045 Angers, France \\
~~E-mail: Volodya.Roubtsov@univ-angers.fr \\[4mm]
$^c$ Laboratory of Theoretical Physics, JINR, 141980 Dubna, Moscow reg., Russia \\
D\'epartment de Math\'ematiques, Universit\'e d'Angers, 2 Bd. Lavoisier, 49045 Angers, France \\
~~E-mail: silant@tonton.univ-angers.fr
\end{center}

\begin{center}
\bigskip
\bigskip
\emph{In memory of Leonid Vaksman}
\bigskip
\bigskip
\end{center}

\begin{abstract}
In this paper we discuss classical elliptic current algebras and
show that there are two different choices of commutative {\em
test function algebras} on a complex torus leading to two
different elliptic current algebras. Quantization of these
classical current algebras give rise to two classes of quantized
dynamical quasi-Hopf current algebras studied by
Enriquez-Felder-Rubtsov and
Arnaudon-Buffenoir-Ragoucy-Roche-Jimbo-Konno-Odake-Shi\-raishi.
\end{abstract}

\section{Introduction}
\label{B_secb1}

Classical elliptic algebras are "quasi-classical limits" of
quantum algebras whose structure is defined by an elliptic
$R$-matrix. The first elliptic $R$-matrix appeared as a matrix
of Boltzmann weights for the eight-vertex model~\cite{Bax1}. This
matrix satisfies the Yang-Baxter equation using which one proves
integrability of the model. An investigation of the eight-vertex
model~\cite{Bax2} uncovered its relation  to the so-called
generalized ice-type model -- the Solid-On-Solid (SOS) model. This is
a face type model with Boltzmann weights which form a matrix
satisfying a {\em dynamical} Yang-Baxter equation.

In this paper we restrict our attention to {\em classical current
algebras} (algebras which can be described by a collection of currents)
related to the classical $r$-matrices and which are quasi-classical limits of
SOS-type {\em quantized elliptic current algebras}. The latter
were introduced by Felder~\cite{F2} and the corresponding $R$-matrix
is called usually a {\em Felder $R$-matrix}. In \emph{loc. cit.}
the current algebras were defined by dynamical $RLL$-relations. At
the same time Enriquez and one of authors (V.R.) developed  a
theory of quantum current algebras related to arbitrary genus
complex curves (in particular to an elliptic curve) as a
quantization of certain (twisted) Manin pairs~\cite{ER1}
using {\em Drinfeld's new realization} of quantized current
algebras. Further, it was shown in~\cite{EF}
that the Felder algebra can be obtained by twisting of the
Enriquez-Rubtsov elliptic algebra. This twisted algebra will be
denoted by $E_{\tau,\eta}$ and it is a quasi-Hopf algebra.

Originally, the dynamical Yang-Baxter equation appeared in ~\cite{GN,F1}. The
fact that elliptic algebras could be obtained as quasi-Hopf
deformations of Hopf algebras was noted first in a special
case in ~\cite{BBB} and was discussed in ~\cite{Fr97}. The full potential
of this idea was realized in papers ~\cite{ABRR97} and
~\cite{JKOS1}. It was explained in these papers how to obtain
the universal dynamical Yang-Baxter equation for the twisted elliptic universal
$R$-matrix from the Yang-Baxter equation for the universal $R$-matrix of the quantum
affine algebra $U_q(\hat{\mathfrak g})$. It was also shown that the image of the twisted
$R$-matrix in finite-dimensional representations coincides with
SOS type $R$-matrix.

 Konno proposed in~\cite{K98} an RSOS type elliptic current algebra
 (which will be denoted by
$U_{p,q}(\hat{\mathfrak{sl}}_2)$) generalizing some ideas of
~\cite{KLP98}. This algebra was studied in detail in ~\cite{JKOS2}
where it was shown that
commutation relations for $U_{p,q}(\hat{\mathfrak{sl}}_2)$ expressed in terms of $L$-operators
coincide with the commutation relations of the Enriquez-Felder-Rubtsov
algebra up to a shift of the elliptic module by the central
element. It was observed in ~\cite{JKOS2} that this difference of
central charges can be explained by different choices of contours
on the elliptic curve entering in these extensions. In the case of the algebra
$E_{\tau,\eta}$ the elliptic module is fixed, while
in the case of $U_{p,q}(\hat{\mathfrak{sl}}_2),\ p=e^{i\pi\tau},$ it
turn out to be a dynamical parameter shifted by the central element.
Commutation relations for these algebras coincide when the central charge is zero,
but the algebras themselves are different. Furthermore, the difference between these two
algebras was interpreted in~\cite{EPR} as a difference in
definitions of half-currents (or Gauss coordinates) in $L$-operator
representation. The roots of this difference are related to
different decomposition types of so-called {\em Green kernels}
introduced in~\cite{ER1} for quantization of Manin pairs: they
are expanded into Taylor series in the case of the algebra
$E_{\tau,\eta}$ and into Fourier series for
$U_{p,q}(\hat{\mathfrak{sl}}_2)$.

Here, we continue a comparative study of different elliptic current
algebras. Since the Green kernel is the same in both the classical
and the quantum case we restrict ourselves only to the classical
case for the sake of  simplicity. The classical limits of
quasi-Hopf algebras $E_{\tau,\eta}$ and
$U_{p,q}(\hat{\mathfrak{sl}}_2)$ are {\em quasi-Lie bialgebras}
denoted by $\cEF$ and $\utau$ respectively. We will give an
"analytic" description of these algebras in terms of
distributions. Then, the different expansions of Green kernels
will be interpreted as the action of distributions on different test
function algebras. We will call them {\em Green distributions}.
The scalar products for test function algebras which define their
embedding in the corresponding space of distributions are defined by
integration over different contours on the surface.

Let us describe briefly the structure of the paper. Section 2
contains some basic notions and constructions which are used
throughout the paper. Here, we remind some
definitions from ~\cite{ER1}. Namely, we define test function
algebras on a complex curve $\Sigma$, a continuous non-degenerate
scalar product, distributions on the test functions and a
generalized notion of Drinfeld currents associated with these
algebras and with a (possibly infinite-dimensional) Lie algebra
$\mathfrak g$. Hence, our currents will be certain $\mathfrak
g$-valued distributions. Then we review the case when $\mathfrak
g$ is a loop algebra generated by a semi-simple Lie algebra
$\mathfrak a$. We also discuss a centrally and a co-centrally
extended version of $\mathfrak g$ and different bialgebra
structures. The latter are based on the notion of Green
distributions and related half-currents.

We describe in detail two different classical elliptic current
algebras which correspond to two different choices of the basic
test function algebras (in fact they correspond to two different
coverings of the underlying elliptic curve).

Section 3 is devoted to the construction and comparison of
classical elliptic  algebras $\cEF$ and $\utau$. In the first two
subsections we define elliptic Green distributions for both test
function algebras. We pay special attention to their properties
because they manifest the main differences between the
corresponding elliptic algebras. Further, we describe these
classical elliptic algebras in terms of the half-currents
constructed using the Green distributions. We use projections introducing
in ~\cite{ER2} to define these half-currents. We see how the
half-currents inherit the properties of Green distributions. In
the last subsection we show that the half-currents describe the
corresponding bialgebra structure. Namely, we recall the universal
classical $r$-matrices for both elliptic classical algebras $\cEF$
and $\utau$ and make explicit their relation to the $L$-operators.
Then, the corresponding co-brackets for half-currents are expressed
in a matrix form via the $L$-operators.

In the next paper \cite{S22} we will describe different degenerations of the
classical elliptic current algebras in terms of
degenerations of Green distributions. We will discuss also the inverse
problem of reconstruction of the trigonometric and elliptic
classical $r$-matrices from the rational and trigonometric $r$-matrices
using approach of \cite{RF}.

\section{Currents and half-currents}

Current realization of the quantum affine algebras and Yangiens
was introduced by Drinfeld in~\cite{D88}. In these cases they can
be understand as elements of the space $\A[[z,z^{-1}]]$, where
$\A$ is a corresponding algebra. Here we introduce a more general
notion of currents suitable even for the case when the currents
are expressed by integrals instead of formal series.

\bb{Test function algebras.} Let $\fK$ be a function algebra on a
one-dimensional complex manifold $\Sigma$ with a point-wise
multiplication and a continuous invariant (non-degenerate) scalar
product $\la\cdot,\cdot\ra\colon\fK\times\fK\to\CC$. We shall call
the pair $(\fK,\la\cdot,\cdot\ra)$ {\em a test function algebra}.
The non-degeneracy of the scalar product implies that the algebra
$\fK$ can be extended to a space $\fK'$ of linear continuous
functionals on $\fK$. We use the notation $\la a(u),s(u)\ra$ or
$\la a(u),s(u)\ra_u$ for the action of the distribution
$a(u)\in\fK'$ on a test function $s(u)\in\fK$. Let
$\{\epsilon^i(u)\}$ and $\{\epsilon_i(u)\}$ be  dual bases of
$\fK$. A typical example of the element from $\fK'$ is the series
$\delta(u,z)=\sum_i\epsilon^i(u)\epsilon_i(z)$. This is a
delta-function distribution on $\fK$ because it satisfies $\la
\delta(u,z),s(u)\ra_u=s(z)$ for any test function $s(u)\in\fK$.

\bb{Currents.} Consider an infinite-dimensional complex Lie algebra $\frg$ and an operator
$\hx\colon\fK\to\frg$. The expression
$
 x(u)=\sum_i \epsilon^i(u)\hx[\epsilon_i]
$ 
does not depend on a choice of dual bases in $\fK$ and is called a current corresponding
to the operator $\hx$ ($\hx[\epsilon_i]$ means an action of $\hx$ on $\epsilon_i$).
We should interpret the current $x(u)$ as a $\frg$-valued distribution such that
$
 \la x(u),s(u)\ra=\hx[s].\
$
That is the current $x(u)$ can be regarded as a kernel of the operator $\hx$ and the
latter formula\ 
 gives its invariant definition.

\bb{Loop algebras} Let $\{\hat x_k\},\ k=1,\ldots, n$ be a finite
number of operators $\hat x_k\colon\fK\to\frg$, where $\frg$ is an
infinite-dimensional space spanned by $\hat x_k[s]$, $s\in\fK$.
Consider the corresponding currents $x_k(u)$. For these currents
we impose the standard commutation relations
\begin{align}\label{B_cu-com-rel}
[x_k(u),x_l(v)]=C^m_{kl}x_m(u)\delta(u,v)\,
\end{align}
where $C^m_{kl}$ are structure constants of some semi-simple Lie
algebra $\mathfrak a,\ \dim\mathfrak a = n$
(equality~\r{B_cu-com-rel} is understood in sense of distributions).
These commutation relations equip $\frg$ with a Lie algebra
structure. The Lie algebra $\frg$ defined in such a way can be
viewed as a Lie algebra $\mathfrak a\otimes\fK$ with the brackets
$[x\otimes s(z),y\otimes t(z)]=[x,y]_{\mathfrak a}\otimes
s(z)t(z)$, where $x,y\in\mathfrak a$, $s,t\in\fK$. This algebra
possesses an invariant scalar product
$
 \la x\otimes s,y\otimes t\ra=(x,y)\la s(u),t(u)\ra_u\,,\
$
where $(\cdot,\cdot)$ an invariant scalar product on $\mathfrak a$ proportional to the Killing form.

\bb{Central extension.} The algebra $\frg=\mathfrak a\otimes\fK$
can be extended by introducing a central element $c$ and a
co-central element $d$. Let us consider the space $\hg=(\mathfrak
a\otimes\fK)\oplus\CC\oplus\CC$ and define an algebra structure on
this space. Let the element $c\equiv(0,1,0)$ commutes with
everything and the commutator of the element $d\equiv(0,0,1)$ with
the elements $\hat x[s]\equiv(x\otimes s,0,0)$, $x\in\mathfrak a$,
$s\in\mathfrak K$, is given by the formula
$
 [d,\hx[s]]=\hx[s'],\
$
where $s'$ is a derivation of $s$.
Define the Lie bracket between the elements of type $\hx[s]$ requiring the scalar product defined by formulae
\begin{align*}
 \la \hat x[s],\hat y[t]\ra&=\la x\otimes s,y\otimes t\ra, & \la c,\hx[s]\ra&=\la d,\hx[s]\ra=0
\end{align*}
to be invariant. It gives the formula
\begin{align}
 [\hat x[s],\hat y[t]]=([x\otimes s,y\otimes t]_0,0,0)+
 c\cdot B(x_1\otimes s_1,x_2\otimes s_2), \label{B_XXc}
\end{align}
where $[\cdot,\cdot]_0$ is the Lie bracket in the algebra
$\frg=\mathfrak a\otimes\mathfrak K$ and $B(\cdot,\cdot)$ is a
standard 1-cocycle:
$
 B\big(x\otimes s,y\otimes t\big)=(x,y)\la s'(z),t(z)\ra_z.\
$
The expression $\hat x[s]$ depends linearly on $s\in\fK$ and,
therefore, can be regarded as an action of operator
$\hx\colon\fK\to\hg$. The commutation relations for the algebra
$\hg$ in terms of currents $x(u)$ corresponding to these operators
can be written in the standard form: $[c,x(u)]=[c,d]=0$ and
\begin{align}
 [x_1(u),x_2(v)]&=x_3(u)\delta(u,v)-c\cdot(x_1,x_2){d\delta(u,v)}/{du}, & [d,x(u)]&=-d x(u)/du, \label{B_xxtc}
\end{align}
where $x_1,x_2\in\mathfrak a$, $x_3=[x_1,x_2]_{\mathfrak a}$.

\bb{Half-currents.} To describe different bialgebra structures in
the current algebras we have to decompose  the currents in these
algebras into difference of the currents which have good
analytical properties in certain domains:
$
 x(u)=x^+(u)-x^-(u).\
$
The $\frg$-valued distributions $x^+(u)$, $x^-(u)$ are called {\em
half-currents}. To perform such a decomposition we will use
so-called Green distributions~\cite{ER1}. Let
$\Omega^+,\Omega^-\subset\Sigma\times\Sigma$ be two domain
separated by a hypersurface $\bar\Delta\subset\Sigma\times\Sigma$
which contains the diagonal $\Delta=\{(u,u)\mid
u\in\Sigma\}\subset\bar\Delta$. Let there exist distributions
$G^+(u,z)$ and $G^-(u,z)$ regular in $\Omega^+$ and $\Omega^-$
respectively such that $ \delta(u,z)=G^+(u,z)-G^-(u,z)$. To define
half-currents corresponding to these Green distributions we
decompose them as $G^+(u,z)=\sum_i\alpha^+_i(u)\beta^+_i(z)$ and
$G^-(u,z)=\sum_i\alpha^-_i(u)\beta^-_i(z)$. Then the half-currents
are defined as $x^+(u)=\sum_i \alpha^+_i(u)\hx[\beta^+_i]$ and
$x^-(u)=\sum_i \alpha^-_i(u)\hx[\beta^-_i]$. This definition does
not depend on a choice of decompositions of the Green
distributions. The half-currents are currents corresponding to the
operators $\hx^\pm=\pm\,\hx\cdot P^\pm$, where $P^\pm[s](z)=\pm\la
G^\pm(u,z),s(u)\ra$, $s\in\fK$. One can express the half-currents
through the current $x(u)$, which we shall call {\em a total
current} in contrast with the half ones:
\begin{align}\label{B_hc_tc}
 x^+(u)=\la G^+(u,z)x(z)\ra_z,\qquad x^-(u)=\la G^-(u,z)x(z)\ra_z.
\end{align}
Here $\la a(z) \ra_z\equiv \la a(z),1 \ra_z$.

\bb{Two elliptic classical current algebras.} In this paper we
will consider the case when $\Sigma$ is a covering of an elliptic
curve and Green distributions are regularization of certain
quasi-doubly periodic meromorphic functions. We will call the
corresponding centrally extended algebras of currents by {\em
elliptic classical current algebras}. The main aim of this paper
is to show the following facts:
\begin{itemize}
\item There are two essentially different choices of the test function algebras $\fK$ in this case
corresponding to the different covering $\Sigma$.
\item The same quasi-doubly periodic meromorphic functions regularized with respect to the different
 test function algebras define the different quasi-Lie bialgebra structures and, therefore,
 the different classical elliptic current algebras.
\item The internal structure of these two elliptic algebras is essentially different
in spite of a similarity in the commutation relations between
their half-currents.
\end{itemize}

The first choice corresponds to $\fK=\lfK_0$, where $\lfK_0$
consists of complex-valued one-variable functions defined in a
vicinity of origin  equipped
with the scalar product
\begin{align}
 \la s_1(u),s_2(u)\ra=\oint\limits_{C_0}\frac{du}{2\pi i}s_1(u)s_2(u). \label{B_lfK_sp}
\end{align}
Here $C_0$ is a contour encircling zero and belonging in the intersection of domains
of functions $s_1(u)$, $s_2(u)$, such that the scalar product is a residue in zero.
 These functions can be
extended up to meromorphic functions on the covering
$\Sigma=\mathbb C$. The regularization domain $\Omega^+$,
$\Omega^-$ for Green distributions in this case consist of the
pairs $(u,z)$ such that $\max(1,|\tau|)>|u|>|z|>0$ and
$0<|u|<|z|<\max(1,|\tau|)$ respectively, where $\tau$ is an
elliptic module, and $\bar\Delta=\{(u,z)\mid |u|=|z|\}$.

The second choice corresponds to $\fK=K=K(\Cyl)$. The algebra $K$ consists of entire periodic functions
$s(u)=s(u+1)$ on $\CC$ decaying exponentially at $\Im u\to\pm\infty$ equipped with an invariant scalar
product
\begin{align}
 \la s(u),t(u)\ra=\int\limits_{-\frac12}^{\frac12}\frac{du}{2\pi i}s(u)t(u),
   \qquad s,t\in K. \label{B_spJ}
\end{align}
This functions can be regarded as functions on cylinder $\Sigma=\Cyl$.
The regularization domain $\Omega^+$, $\Omega^-$ for Green distributions
consist of the pairs $(u,z)$ such that $-\Im\tau<\Im(u-z)<0$ and $0<\Im (u-z)<\Im\tau$ respectively
and $\bar\Delta=\{(u,z)\mid \Im u=\Im z\}$.

\bb{Integration contour.} The geometric roots of the difference
between these two choices can be explained as follows. These
choices of test functions on different coverings $\Sigma$ of
elliptic curve correspond to the homotopically different contours
on the elliptic curve. Each test function can be considered as an
analytical continuation of a function from this contour -- a real
manifold -- to the corresponding covering. This covering should be
chosen as a most homotopically simple covering which permits to
obtain a bigger source of test functions. In the first case, this
contour is a homotopically trivial and coincides with a small
contour around fixed point on the torus. We can always choose a
local coordinate $u$ such that $u=0$ in this point. This explains
the notation $\lfK_0$. This contour corresponds to the covering
$\Sigma=\CC$ and it enters in the pairing~\r{B_lfK_sp}. In the
second case, it goes along a cycle and it can not be represented
as a closed contour on $\CC$. Hence the most simple covering in
this case is a cylinder $\Sigma=\Cyl$ and the contour is that one
in the pairing~\r{B_spJ}. This leads to essentially different
properties of the current elliptic algebras based on the test
function algebras $\fK=\lfK_0$ and $\fK=K(\Cyl)$.

\bb{Restriction to the $\slt$ case.} To make these differences more transparent we shall consider
only the simplest case of Lie algebra $\mathfrak a=\slt$ defined as a three-dimensional complex
Lie algebra with commutation relations $[h,e]=2e$, $[h,f]=-2f$ and $[e,f]=h$. We denote the
constructed current algebra $\hg$ for the case $\fK=\lfK_0$ as $\cEF$ and for $\fK=K=K(\Cyl)$ as
$\utau$. These current algebras may be identified with classical limits of the quantized currents
algebra $E_{\tau,\eta}(\slt)$ of  \cite{EF} and $U_{p,q}(\widehat{\mathfrak{sl}}_2)$
of \cite{JKOS2} respectively. The Green
distributions appear in the algebras  $\cEF$ and $\utau$ as a regularization of the same meromorphic
quasi-doubly periodic functions but in different spaces: $(\lfK_0\otimes\lfK_0)'$ and $(K\otimes K)'$
respectively. Primes mean the extension to the space of the distributions. We call them {\em elliptic
Green distributions}. We define the algebras $\cEF$ and $\utau$ to be {\itshape a priori} different,
because the main component of our construction, elliptic Green distributions are {\itshape a priori}
different being understood as distributions of different types: related to algebras $\lfK_0$ and $K$
respectively. It means, in particular, that their quantum analogs, the algebras
$E_{\tau,\eta}(\mathfrak{sl}_2)$ and $U_{p,q}(\widehat{\mathfrak{sl}}_2)$ are different.

\section{Half-currents and co-structures}
\label{B_secb3}
\setcounter{bbcount}{0}

We start with a suitable definition of theta-functions and a
conventional choice of standard bases. This choice is motivated
and corresponds to definitions and notations of ~\cite{EPR}.

\bb{Theta-function.} Let $\tau\in\CC$, $\Im\tau>0$ be a module of
the elliptic curve $\mathbb C/\Gamma$, where $\Gamma=\mathbb
Z+\tau\mathbb Z$ is a period lattice. The odd theta function
$\theta(u)=-\theta(-u)$ is defined as a holomorphic function on
$\CC$ with the properties
\begin{align}
 \theta(u+1)&=-\theta(u), &\theta(u+\tau)&=-e^{-2\pi i u-\pi i\tau}\theta(u), &\theta'(0)&=1.
\end{align}

\subsection{Elliptic Green distributions on $\lfK_0$}
\label{B_subsec31}
\setcounter{bbcount}{1}

\bb{Dual bases.} Fix a complex number $\lambda$.  Consider the following bases in $\lfK_0$ ($n\geq0$):
$\epsilon_{n;\lambda}(u)=(-u)^n$, $\epsilon^{-n-1;\lambda}(u)=u^n$,
\begin{align}
 \epsilon^{n;\lambda}(u)=\frac1{n!}\left(\frac{\theta(u+\lambda)}
 {\theta(u)\theta(\lambda)}\right)^{(n)}\,,\quad
 \epsilon_{-n-1;\lambda}(u)=\frac{(-1)^n}{n!}\left(\frac{\theta(u-\lambda)}
 {\theta(u)\theta(-\lambda)}\right)^{(n)}\,,  \nn
\end{align}
for $\lambda\not\in\Gamma$ and the bases $\epsilon_{n;0}(u)=(-u)^n$, $\epsilon^{-n-1;0}(u)=u^n$,
\begin{align}
 \epsilon^{n;0}(u)=\frac1{n!}\left(\frac{\theta'(u)}{\theta(u)}\right)^{(n)}\,,\quad
 \epsilon_{-n-1;0}(u)=\frac{(-1)^n}{n!}\left(\frac{\theta'(u)}
 {\theta(u)}\right)^{(n)}\,, \nn
\end{align}
for $\lambda=0$. Here $(\cdot)^{(n)}$ means $n$-times derivation. These bases are dual:
$\la\epsilon^{n;\lambda}(u),\epsilon_{m;\lambda}(u)\ra=\delta^n_m$ and
$\la\epsilon^{n;0}(u),\epsilon_{m;0}(u)\ra=\delta^n_m$
with respect to the scalar product~\r{B_lfK_sp}
which means
\begin{align}
 &\sum_{n\in\mathbb Z}\epsilon^{n;\lambda}(u)\epsilon_{n;\lambda}(z)=\delta(u,z),
 & &\sum_{n\in\mathbb Z}\epsilon^{n;0}(u),\epsilon_{n;0}(z)=\delta(u,z). \label{B_dbs}
\end{align}

\bb{Green distributions for $\lfK_0$ and the addition theorems.}
Here we follow the ideas of ~\cite{ER1} and ~\cite{EPR}. We
define the following distribution
\begin{align}
 \G+(u,z)&=\sum_{n\ge0}\epsilon^{n;\lambda}(u)\epsilon_{n;\lambda}(z)\,,\quad
 \G-(u,z)=-\sum_{n<0}\epsilon^{n;\lambda}(u)\epsilon_{n;\lambda}(z)\,,  \label{B_Glpmdef} \\
 G(u,z)&=\sum_{n\ge0}\epsilon^{n;0}(u)\epsilon_{n;0}(z)=\sum_{n<0}\epsilon^{n;0}(z)\epsilon_{n;0}(u)\,.
 \label{B_Gdef} \end{align}
One can check that these series converge in sense of distributions and, therefore, define
continuous functionals on $\lfK_0$ called Green distributions. Their action on a test function $s(u)$ reads
\begin{align}
\la\G\pm(u,z),s(u)\ra_u&=
  \oint\limits_{|u|>|z|\atop |u|<|z|}\frac{du}{2\pi i}\frac{\theta(u-z+\lambda)}{\theta(u-z)
  \theta(\lambda)}s(u)\,,
  \label{B_Glpmact} \\
 \la G(u,z),s(u)\ra_u&=
  \oint\limits_{|u|>|z|}\frac{du}{2\pi i}\frac{\theta'(u-z)}{\theta(u-z)}s(u)\,, \label{B_Gact}
\end{align}
where integrations are taken over circles around zero which are
small enough such that the corresponding inequality takes place.

One can define a 'rescaling' of a test function $s(u)$ as a function $s\big(\frac u\alpha\big)$,
where $\alpha\in\mathbb C$, and therefore a 'rescaling' of distributions by the formula
$\la a(\frac u\alpha\big),s(u)\ra=\la a(u),s(\alpha u)\ra$. On the contrary, we are unable to define a
'shift' of test functions by a standard rule, because the operator $s(u)\mapsto s(u+z)$
is not a continuous one~\footnote{Consider, for example, the sum
$s_N(u)=\sum_{n=0}^N(\frac{u}{\alpha})^n$. For each $z$ there exist $\alpha$ such that the
sum $s_N(u+z)$ diverges, when $N\to\infty$.}. Nevertheless we use distributions 'shifted' in some sense.
Namely, we say that a two-variable distribution $a(u,z)$ (a linear continuous functional
$a\colon\lfK_0\otimes\lfK_0\to\mathbb C$) is 'shifted' if it possesses the properties:
(i) for any $s\in\lfK_0$ the functions $s_1(z)=\la a(u,z),s(u)\ra_u$ and $s_2(u)=\la a(u,z),s(z)\ra_z$
belong to $\lfK_0$; (ii) $\frac{\partial}{\partial u}a(u,z)=-\frac{\partial}{\partial z}a(u,z)$.
Here the subscripts $u$ and $z$ mean the corresponding partial action, for instance,
$\la a(u,z),s(u,z)\ra_u$ is a distribution acting on $\lfK_0$ by the formula
\begin{align*}
 \La\la a(u,z),s(u,z)\ra_u,t(z)\Ra=\la a(u,z),s(u,z)t(z)\ra.
\end{align*}
The condition (ii) means the equality $\la a(u,z),s'(u)t(z)\ra=-\la a(u,z),s(u)t'(z)\ra$.
The condition (i) implies that for any $s\in\lfK_0\otimes\lfK_0$ the expression
\begin{align}
 \la a(u,z),s(u,z)\ra_u=\sum_i\la a(u,z),p_i(u)\ra_u q_i(z), \label{B_auz_suz}
\end{align}
where $s(u,z)=\sum_i p_i(u)q_i(z)$, belongs to $\lfK_0$ (as a function of $z$).

The Green distributions ~\r{B_Glpmdef} and \r{B_Gdef} are examples of the `shifted'
distributions. The formula~\r{B_dbs} implies that
\begin{align}
 \G+(u,z)-\G-(u,z)=\delta(u,z)\,,\quad
 G(u,z)+G(z,u)=\delta(u,z)\,. \label{B_GGdelta}
\end{align}
The last formulae can be also obtained from~\r{B_Glpmact},
\r{B_Gact} taking into account that the function $s(u)$ has poles only in the points $u=0$.
As it is seen from~\r{B_Glpmact}, the oddness of function $\theta(u)$
leads to the following connection between the $\lambda$-depending Green distributions:
$
 \G+(u,z)=-G_{-\lambda}^-(z,u).
$

Now we are able to define a \emph{semidirect product of
two 'shifted' distributions} $a(u,z)$ and $b(v,z)$
as a linear continuous functional $a(u,z)b(v,z)$ acting on $s\in\lfK_0\otimes\lfK_0\otimes\lfK_0$
by the rule
\begin{align*}
 \la a(u,z)b(v,z),s(u,v,z)\ra=\La a(u,z),\la b(v,z),s(u,v,z)\ra_v\Ra_{u,z}.
\end{align*}

\begin{prop} \label{B_prop1}
The {\em semi-direct products} of Green
distributions are related by the following addition formulae
\begin{align}
\begin{split}
 \G+(u,z)\G-(z,v)&=\G+(u,v)G(u,z)-\G+(u,v)G(v,z)-\frac{\partial}{\partial\lambda}\G+(u,v)\,,
\end{split} \label{B_llpp} \\
\begin{split}
 \G+(u,z)\G+(z,v)&=\G+(u,v)G(u,z)+\G+(u,v)G(z,v)-\frac{\partial}{\partial\lambda}\G+(u,v)\,,
\end{split} \label{B_llpm} \\
\begin{split}
 \G-(u,z)\G-(z,v)&=-\G-(u,v)G(z,u)-\G-(u,v)G(v,z)-\frac{\partial}{\partial\lambda}\G+(u,v)\,,
\end{split} \label{B_llmp} \\
\begin{split}
 \G-(u,z)\G+(z,v)&=-\G+(u,v)G(z,u)+\G+(u,v)G(z,v)-\frac{\partial}{\partial\lambda}\G+(u,v)\,.
\end{split} \label{B_llmm}
\end{align}
\end{prop}
\begin{proof}\
 The actions of both hand sides of~\r{B_llpp}, for example, can be reduced to the integration
 over the same contours with some kernels. One can check the equality of these kernels using the
 degenerated Fay's identity~\cite{Fay}
\begin{align}
\frac{\theta(u-z+\lambda)}{\theta(u-z)\theta(\lambda)} \frac{\theta(z+\lambda)}{\theta(z)\theta(\lambda)}=
 \frac{\theta(u+\lambda)}{\theta(u)\theta(\lambda)} \frac{\theta'(u-z)}{\theta(u-z)}
 +\frac{\theta(u+\lambda)}{\theta(u)\theta(\lambda)} \frac{\theta'(z)}{\theta(z)}
 -\frac{\partial}{\partial\lambda} \frac{\theta(u+\lambda)}{\theta(u)\theta(\lambda)}. \label{B_Fayid}
\end{align}
The other formulae can be proved in the same way
if one takes into account $a(u,v)\delta(u,z)=a(z,v)\delta(u,z)$ and~\r{B_GGdelta}.
\end{proof}

\bb{Projections.} Let us notice that the vectors
$\epsilon_{n;\lambda}(u)$ and $\epsilon^{-n-1;\lambda}(u)$ span
two complementary subspaces  of $\lfK_0$. The
formulae~\eqref{B_Glpmdef} mean that the distributions $\G+(u,z)$
and $\G-(u,z)$ define orthogonal projections $P_\lambda^+$ and
$P_\lambda^-$ onto these subspaces. They act as
$
 P_\lambda^+[s](z)=\la\G+(u,z),s(u)\ra_u$ and   $P_\lambda^-[s](z)=-\la\G-(u,z),s(u)\ra_u.\
$
Similarly, the operators
$
 P^+[s](z)=\la G(u,z),s(u)\ra_u$ and   $P^-[s](z)=\la G(z,u),s(u)\ra_u\
$
are projections onto the lagrangian (involutive) subspaces spanned
by vectors $\epsilon_{n;0}(u)$ and $\epsilon_{-n-1;0}(u)$,
respectively. The fact that the corresponding spaces are
complementary to each other is encoded in the
formulae~\eqref{B_GGdelta}, which can be rewritten as
$P_\lambda^++P_\lambda^-=\id$, $ P^++P^-=\id$. The idempotent
properties and orthogonality of these projection
\begin{align*}
P_\lambda^\pm\cdot P_\lambda^\pm&=P_\lambda^\pm, &
P^\pm\cdot P^\pm&=P^\pm, &  P_\lambda^+\cdot P_\lambda^-&=P_\lambda^-\cdot P_\lambda^+\hm=0, &
P^+\cdot P^-&=P^-\cdot P^+&=0
\end{align*}
are encoded in the formulae
\begin{align}
 \la\G+(u,z)\G+(z,v)\ra_z&=\G+(u,v)\,, & \la\G+(u,z)\G-(z,v)\ra_z&=0\,, \label{B_GpGpconv}\\
 \la\G-(u,z)\G-(z,v)\ra_z&=-\G-(u,v)\,, & \la\G-(u,z)\G+(z,v)\ra_z&=0\,, \label{B_GmGmconv} \\
 \la G(u,z)G(z,v)\ra_z&=G(u,v)\,, & \la G(u,z)G(v,z)\ra_z&=0\,, \label{B_GGconv}
\end{align}
which immediately follow from~\eqref{B_Glpmdef} and also can be obtained from the
relations~\eqref{B_llpp} -- \eqref{B_llmm} if one takes into account $\la G(u,z)\ra_z=0$, $\la G(z,u)\ra_z=1$.

\subsection{Elliptic Green distributions on $K$}
\label{B_subsec32}

\bb{Green distributions and dual bases for $K$.} The analogs of
the Green distributions $\G+(u,z)$, $\G-(u,z)$ are defined in this case
by the following action on the space $K$
\begin{align}
 \la\Gg\pm(u-z),s(u)\ra_u&=\int\limits_{-\Im\tau<\Im(u-z)<0\atop 0<\Im(u-z)<\Im\tau}
 \frac{du}{2\pi i} \frac{\theta(u-z+\lambda)}{\theta(u-z)\theta(\lambda)}s(u), \label{B_GlpmdefJ} \\
 \la{\cal G}(u-z),s(u)\ra_u
   &=\int\limits_{-\Im\tau<\Im(u-z)<0}\frac{du}{2\pi i}\frac{\theta'(u-z)}{\theta(u-z)}s(u). \label{B_GdefJ}
\end{align}
where we integrate over line segments of unit length (cycles of cylinder) such that
the corresponding inequality takes place. The role of dual bases in the algebra $K$ is played by
$\{j_n(u)=e^{2\pi inu}\}_{n\in\mathbb Z}$ and
$\{j^n(u)=2\pi i e^{-2\pi inu}\}_{n\in\mathbb Z}$, a decomposition to
these bases is the usual Fourier expansion. The Fourier expansions for the Green distributions
are~\footnote{Fourier expansions presented in this subsection are obtained considering integration
around boundary of fundamental domain.}
\begin{gather}
 \Gg\pm(u-z)=\pm2\pi i\sum_{n\in\mathbb Z}\frac{e^{-2\pi in(u-z)}}{1-e^{\pm2\pi i(n\tau-\lambda)}},
 \label{B_Ggpmdec} \\
{\cal G}(u-z)=\pi i+2\pi i\sum_{n\ne0}\frac{e^{-2\pi in(u-z)}}{1-e^{2\pi in\tau}}. \label{B_GexpanJ}
\end{gather}
These expansions are in according with formulae
\begin{align}
 \Gg+(u-z)-\Gg-(u-z)&=\delta(u-z), \label{B_GpGmdeltaJ} \\
  {\cal G}(u-z)+{\cal G}(z-u)&=\delta(u-z), \label{B_GGdeltaJ}
\end{align}
where $\delta(u-z)$ is a delta-function on $K$, given by the
expansion
\begin{align}
 \delta(u-z)=\sum_{n\in\mathbb Z}j^n(u)j_n(z)=2\pi i\sum_{n\in\mathbb Z} e^{-2\pi in(u-z)}. \label{B_dffeJ}
\end{align}

\bb{Addition theorems.} Now we obtain some properties of these
Green distributions and compare them with the properties of their
analogs $\G+(u,z)$, $\G-(u,z)$, $G(u,z)$ described in
subsection~\ref{B_subsec31}. In particular, we shall see that some
properties are essentially different. Let us start with the
properties of Green distribution which are similar to the case of
algebra $\lfK_0$. They satisfy the same addition theorems that was
described in~\ref{B_subsec31}:
\begin{prop} \label{B_prop2}
 The semi-direct product of Green distributions for algebra $K$ is related by the
 formulae~\r{B_llpp}--\r{B_llmm} with the distributions $\Gg\pm(u-z)$,
 ${\cal G}(u-z)$ instead of $\G\pm(u-z)$, $G(u-z)$ respectively.
\end{prop}

\noindent{\bf Proof.} The kernels of these distributions are the same and therefore
the addition formula in this case is also based on the Fay's identity~\r{B_Fayid}.
\qed

\bb{Analogs of projections.} The Green distributions define the
operators on $K$:
\begin{align*}
 {\cal P}_\lambda^+[s](z)&=\la\Gg+(u-z),s(u)\ra_u, & {\cal P}_\lambda^-[s](z)&=\la\Gg-(u-z),s(u)\ra_u, \\
  {\cal P}^+[s](z)&=\la {\cal G}(u-z),s(u)\ra_u, & {\cal P}^-[s](z)&=\la {\cal G}(z-u),s(u)\ra_u
\end{align*}
which are similar to their analogs $P_\lambda^\pm$, $P^\pm$ and
satisfy ${\cal P}_\lambda^++{\cal P}_\lambda^-=\id$, ${\cal
P}^++{\cal P}^-=\id$ (due to~\r{B_GpGmdeltaJ}), but they are not
projections. This fact is reflected in the following relations,
which are consequence of the formulae~\r{B_llpp}--\r{B_llmp} and
$\la{\cal G}(u-z)\ra_z=\dfrac12$, $\la{\cal
G}(z-u)\ra_z=\dfrac12$,
\begin{align}
 \la\Gg+(u-z)\Gg+(z-v)\ra_z&=\Gg+(u-v)-\frac1{2\pi i}\frac{\partial}{\partial\lambda}\Gg+(u-v)\,,
 \label{B_GpGpconvJ} \\
 \la\Gg+(u-z)\Gg-(z-v)\ra_z&=-\frac1{2\pi i}\frac{\partial}{\partial\lambda}\Gg+(u-v)\,,
 \label{B_GpGmconvJ} \\
 \la\Gg-(u-z)\Gg+(z-v)\ra_z&=-\frac1{2\pi i}\frac{\partial}{\partial\lambda}\Gg+(u-v)\,,
 \nn\\ 
 \la\Gg-(u-z)\Gg-(z-v)\ra_z&=-\Gg-(u-v)-\frac1{2\pi i}\frac{\partial}{\partial\lambda}\Gg+(u-v)\,,
 \label{B_GmGmconvJ}
\end{align}
where $\frac{\partial}{\partial\lambda}\Gg+(u-v)=\frac{\partial}{\partial\lambda}\Gg-(u-v)$ and
\begin{align}
\la{\cal G}(u-z){\cal G}(z-v)\ra_z=&{\cal G}(u-v)
 -\frac1{4\pi i}\gamma(u-v)\,, \label{B_GGppconvJ} \\
\la{\cal G}(u-z){\cal G}(v-z)\ra_z=&\la{\cal G}(z-u){\cal G}(z-v)\ra_z=
 \frac1{4\pi i}\gamma(u-v)\,. \label{B_GGpmconvJ}
\end{align}
$\gamma(u-z)$ is a distribution which has the following action and expansion
\begin{align}
 \la&\gamma(u-z),s(u)\ra=-\frac{\theta'''(0)+4\pi^2}3+\int\limits_{-\frac12}^{+\frac12}\frac{du}{2\pi i}
 \frac{\theta''(u-z)}{\theta(u-z)}s(u)\,, \nn\\
&\gamma(u-z)=-2\pi^2+8\pi^2\sum_{n\ne0}\frac{e^{-2\pi in(u-z)+2\pi in\tau}}{(1-e^{2\pi in\tau})^2}\,.
\nn
\end{align}

\bb{Comparison of the Green distributions.} Contrary
to~\r{B_GpGpconv}, \r{B_GmGmconv} the
formulae~\eqref{B_GpGpconvJ}--\eqref{B_GGpmconvJ} contain some
additional terms in the right hand sides obstructed the operators
${\cal P}_\lambda^\pm$, ${\cal P}^\pm$ to be projections. They do
not decompose the space $K(\Cyl)$ in a direct sum of subspaces
as it would be in the case of projections $P_\lambda^\pm$, $P^\pm$
acting on $\lfK_0$. Moreover, as one can see from the Fourier
expansions~\r{B_Ggpmdec}, \r{B_GexpanJ} of Green distributions the
images of the operators coincide with whole algebra $K$:
${\cal P}_\lambda^\pm\big(K(\Cyl)\big)=K(\Cyl)$, ${\cal
P}^\pm\big(K(\Cyl)\big)=K(\Cyl)$. As we shall see this fact has a
deep consequence for the half-currents of the corresponding Lie
algebra $\utau$. As soon as we are aware that the positive
operators ${\cal P}_\lambda^+$, ${\cal P}^+$ as well as negative
ones ${\cal P}_\lambda^-$, ${\cal P}^-$
 transform the algebra $K$ to itself, we can surmise that
they can be related to each other. This is actually true. From formulae~\r{B_Ggpmdec}, \r{B_GexpanJ} we
conclude that
\begin{align}
\Gg+(u-z-\tau)&=e^{2\pi i\lambda}\Gg-(u-z),  &{\cal G}(u-z-\tau)&=2\pi i-{\cal G}(z-u). \label{B_shift_tau}
\end{align}
In terms of operator's composition these properties look as
\begin{align}
 {\cal T}_\tau\circ{\cal P}_\lambda^+&={\cal P}_\lambda^+\circ{\cal T}_\tau=-e^{2\pi i\lambda}
{\cal P}_\lambda^-, &
{\cal T}_\tau\circ{\cal P}^+&={\cal P}^+\circ{\cal T}_\tau=2\pi i{\cal I}-{\cal P}^-,
\end{align}
where ${\cal T}_t$ is a shift operator: ${\cal T}_t[s](z)=s(z+t)$,
and ${\cal I}$ is an integration operator: ${\cal
I}[s](z)=\int_{-\frac12}^{\frac12}\frac{du}{2\pi i}s(u)$. This
property is no longer true for the case of
 Green distributions from section~\ref{B_subsec31}.

\subsection{Elliptic half-currents}
\label{B_subsec33}
\setcounter{bbcount}{4}

\bb{Tensor subscripts.} First introduce the following
notation. Let $\bU={\cal U}(\mathfrak{g})$ be a universal
enveloping algebra of the considering Lie algebra $\mathfrak{g}$
and $V$ be a $\bU$-module. For an element
$t=\sum_{k}a_1^k\otimes\ldots\otimes a_n^{k}\otimes
u_1^k\otimes\ldots\otimes u_m^{k} \in\End V^{\otimes
n}\otimes\bU^{\otimes m}$, where $n,m\ge0$,
$a_1^k,\ldots,a_n^k\in\End V$, $u_1^k,\ldots,u_m^k\in\bU$ we shall
use the following notation for an element of $\End V^{\otimes
N}\otimes\bU^{\otimes M}$, $N\ge n, M\ge m$,
\begin{align*}
 t_{i_1,\ldots,i_n,{\bf j_1},\ldots,{\bf j_m}}=&
  \sum_k \id_V\otimes\cdots\otimes\id_V\otimes a_1^k\otimes\id_V\otimes\cdots\otimes\id_V\otimes
  a_n^k\otimes\id_V\otimes\cdots\otimes\id_V\otimes \\
  &\otimes1\otimes\ldots\otimes1\otimes
  u_1^k\otimes1\otimes\ldots\ldots\otimes u_m^{k}
  \otimes1\otimes\ldots\otimes1,
\end{align*}
where $a_s^k$ stays in the $i_s$-th position in the tensor product and $u_s^k$ stays in the $j_s$-th position.

\bb{Half-currents.} The total currents $h(u)$, $e(u)$ and $f(u)$ of the algebra $\cEF$ can be divided into
half-currents using the Green distributions $G(u,z)$, $-G(z,u)$ for $h(u)$; $\G+(u,z)$,
$\G-(u,z)$ for $e(u)$; and $G^+_{-\lambda}(u,z)=-\G-(z,u)$, $G^-_{-\lambda}(u,z)=-\G+(z,u)$.
The relations of type~\r{B_hc_tc}, then, looks as
\begin{align}
 h^+(u)&=\la G(u,v)h(v)\ra_v, & h^-(u)&=-\la G(v,u)h(v)\ra_v,         \label{B_hpmht} \\
 \e^+(u)&=\la \G+(u,v)e(v)\ra_v, & \e^-(u)&=\la \G-(u,v)e(v)\ra_v,     \label{B_epmht} \\
 \f^+(u)&=\la G_{-\lambda}^+(u,v)f(v)\ra_v, &  \f^-(u)&=\la G_{-\lambda}^-(u,v)f(v)\ra_v,  \label{B_fpmht}
\end{align}
so that
  $h(u)=h^+(u)-h^-(u)$, $e(u)=\e^+(u)-\e^-(u)$, $f(u)=\f^+(u)-\f^-(u)$.  

\bb{$rLL$-relations for $\cEF$.} The commutation relation between
the half-currents can be written in a matrix form. Let us
introduce the matrices of {\em $L$-operators:}
\begin{gather}
\L_\lambda^\pm(u)=
 \begin{pmatrix}
  \frac12h^\pm(u) & \f^\pm(u) \\
  \e^\pm(u)       & -\frac12h^\pm(u)
 \end{pmatrix}, \\
\end{gather}
as well as the $r$-matrices:
\begin{gather} r_\lambda^+(u,v)=
 \begin{pmatrix}
   \frac12G(u,v) & 0              & 0              & 0              \\
  0              & -\frac12G(u,v) & G^+_{-\lambda}(u,v)      & 0              \\
  0              & \G+(u,v)       & -\frac12G(u,v) & 0              \\
  0              & 0              & 0              & \frac12G(u,v)
 \end{pmatrix}.
\end{gather}

\begin{prop} \label{B_prop3}
 The commutation relations of the algebra $\cEF$ in terms of half-currents
 can be written in the form:
\begin{align}
   [d,\L_{\lambda}^\pm(u)]=-\frac{\partial}{\partial u}\L_{\lambda}^\pm(u), \label{B_dL}
\end{align}
\begin{multline} \label{B_rLpmLpm}
 [\L_{\lambda,1}^\pm(u),\L_{\lambda,2}^\pm(v)]=[\L_{\lambda,1}^\pm(u)+\L_{\lambda,2}^\pm
 (v),r_\lambda^+(u-v)]+ \\
 +H_1\frac{\partial}{\partial\lambda}\L_{\lambda,2}^\pm(v)-H_2\frac{\partial}{\partial\lambda}
 \L_{\lambda,1}^\pm(u)
 +h\frac{\partial}{\partial\lambda}r_\lambda^+(u-v),
\end{multline}
\begin{multline} \label{B_rLpLm}
 [\L_{\lambda,1}^+(u),\L_{\lambda,2}^-(v)]=[\L_{\lambda,1}^+(u)+\L_{\lambda,2}^-(v),r_\lambda^+(u-v)]+ \\
 +H_1\frac{\partial}{\partial\lambda}\L_{\lambda,2}^-(v)-H_2\frac{\partial}{\partial
 \lambda}\L_{\lambda,1}^+(u)
 +h\frac{\partial}{\partial\lambda}r_\lambda^+(u-v)+c\cdot\frac{\partial}{\partial u}r_\lambda^+(u-v),
\end{multline}
where
$H=
 \begin{pmatrix}
  1 & 0 \\
  0  & -1
 \end{pmatrix}$ and  $h=\hat h[\epsilon_{0;0}]$.
The $L$-operators satisfy an important relation
\begin{align}
 [H+h,\L^\pm(u)]=0\,. \label{B_HhL}
\end{align}
\end{prop}

\begin{proof}\  Using the formulae~\eqref{B_GpGpconv} -- \eqref{B_GGconv} we calculate the
scalar products on the half-currents:
  $\la \L_{\lambda,1}^\pm(u),\L_{\lambda,2}^\pm(v)\ra=0$,
    $\la \L_{\lambda,1}^+(u),\L_{\lambda,2}^-(v)\ra=-r_\lambda^+(u,v)$.
Differentiating these formulae by $u$  we can  obtain the values of the standard co-cycle
on the half-currents:
  $B\big(\L_{\lambda,1}^\pm(u),\L_{\lambda,2}^\pm(v)\big)=0$,
  $B\big(\L_{\lambda,1}^+(u),\L_{\lambda,2}^-(v)\big)=\frac{\partial}{\partial u}r_\lambda^+(u,v)$.
Using the formulae~\r{B_llpp}--\r{B_llmm} one can calculate the
brackets $[\cdot,\cdot]_0$ on the half-currents. Representing them
in the matrix form and adding the co-cycle term one can derive the
relations~\r{B_rLpmLpm}, \eqref{B_rLpLm}. Using the formulae
$[h,\L_{\lambda}^\pm(v)]=\tr_1\la
H_1[\L_{\lambda,1}^+(u),\L_{\lambda,2}^\pm(v)]\ra_u$, $\tr_1\la
H_1r_{\lambda}^+(u,v)\ra_u=H$,
$\tr_1\la[H_1,\L_{\lambda,1}^+(u)]r_{\lambda}^+(u,v)\ra_u=0$ we
obtain the relation~\eqref{B_HhL} from~\eqref{B_rLpmLpm},
\eqref{B_rLpLm}.
\end{proof}

\bb{$rLL$-relations for $\utau$.} Now consider the case of the
algebra $\utau$. The half-currents, $L$-operators
$\LL_\lambda^\pm(u)$ and $r$-matrix $\rr_\lambda^+(u-v)$ are
defined by the same formulas as above with distributions
${G}(u,v)$ and ${G}^\pm_{\lambda}(u,v)$ replaced everywhere by the
distributions ${\cal G}(u,v)$  and ${\cal G}^\pm_{\lambda}(u,v)$.
We have
\begin{prop} \label{B_prop4}
The commutation relations of algebra $\cEF$ in terms of
half-currents  can be written in the form:
\begin{multline}
 [\LL_{\lambda,1}^\pm(u),\LL_{\lambda,2}^\pm(v)]=[\LL_{\lambda,1}^\pm(u)+
 \LL_{\lambda,2}^\pm(v),\rr_\lambda^+(u-v)]+ \\
 +H_1\frac{\partial}{\partial\lambda}\LL_{\lambda,2}^\pm(v)-H_2\frac{\partial}
 {\partial\lambda}\LL_{\lambda,1}^\pm(u)
 +h\frac{\partial}{\partial\lambda}\rr_\lambda^+(u-v)-c\cdot\frac{\partial}
 {\partial\tau}\rr_\lambda^+(u-v)\,, \label{B_rLpmLpmJ}
\end{multline}
\begin{multline}
 [\LL_{\lambda,1}^+(u),\LL_{\lambda,2}^-(v)]=[\LL_{\lambda,1}^+(u)+
 \LL_{\lambda,2}^-(v),\rr_\lambda^+(u-v)]+ \\
 +H_1\frac{\partial}{\partial\lambda}\LL_{\lambda,2}^-(v)-H_2\frac{\partial}
 {\partial\lambda}\LL_{\lambda,1}^+(u)
 +h\frac{\partial}{\partial\lambda}\rr_\lambda^+(u-v)
 +c\cdot\bigg(\frac{\partial}{\partial u}-\frac{\partial}{\partial\tau}\bigg)
 \rr_\lambda^+(u-v)\,, \label{B_rLpLmJ}
\end{multline}
where $h=\hat h[j_0]$. We also have in this case the relation
\begin{align}
 [H+h,\LL^\pm(u)]=0\,. \label{B_HhLJ}
\end{align}
\end{prop}

\noindent{\bf Proof.} To express the standard co-cycle on the half currents through
the derivatives of the $r$-matrix we need the following formulae
\begin{align}
 \frac1{2\pi i}\frac{\partial}{\partial u}\frac{\partial}{\partial\lambda}\Gg+(u-v)
     &=\frac{\partial}{\partial\tau}\Gg+(u-v), \nn\\ 
 \frac1{2\pi i}\frac{\partial}{\partial u}\frac{\partial}{\partial\lambda}\Gg-(u-v)
     &=\frac{\partial}{\partial\tau}\Gg-(u-v)=\frac{\partial}{\partial\tau}\Gg+(u-v),
     \nn\\ 
 \frac1{4\pi i}\frac{\partial}{\partial u}\gamma(u-v)&=\frac{\partial}{\partial\tau}{\cal G}(u-v).
 \nn  
\end{align}
Using these formulae we obtain
\begin{align}
  B\big(\LL_{\lambda,1}^\pm(u),\LL_{\lambda,2}^\pm(v)\big)&=-\frac{\partial}
  {\partial\tau}\rr_\lambda^+(u-v),
    & B\big(\LL_{\lambda,1}^+(u),\LL_{\lambda,2}^-(v)\big)
      &=\bigg(\frac{\partial}{\partial u}-\frac{\partial}{\partial\tau}\bigg)\rr_\lambda^+(u-v).
    \notag
\end{align}
Using the formulae $[h,\LL_{\lambda}^\pm(v)]=2\tr_1\la
H_1[\LL_{\lambda,1}^+(u),\L_{\lambda,2}^\pm(v)]\ra_u$, $\tr_1\la
H_1r_{\lambda}^+(u,v)\ra_u=H/2$,
$\tr_1\la[H_1,\L_{\lambda,1}^+(u)]r_{\lambda}^+(u,v)\ra_u
=\frac{i}{\pi
}\frac{\partial}{\partial\lambda}\LL_{\lambda}^\pm(v)$ we get the
relation~\r{B_HhLJ} from~\r{B_rLpmLpmJ}, \r{B_rLpLmJ}.
\qed

\bb{Peculiarities of half-currents for $\utau$. } To conclude this subsection we discuss the
implication of  the properties of Green distributions described in the end of the previous section
to the Lie algebra $\utau$. The fact that the images of the operators ${\cal P}_\lambda^\pm$,
${\cal P}^\pm$ coincide with all the space $K$ means that the
commutation relations between the positive (or negative) half-currents are sufficient to describe
all the Lie algebra $\utau$.
This is a consequence of construction of the Lie algebra $\utau$ as the
central extension of $\slt\otimes K$. To obtain all commutation relations given in
Proposition~\ref{B_prop4}
from relations between only positive (or negative) half-currents one can use, first,
the connection between positive and negative ones:
\begin{align*}
 h^+(u-\tau)&=2\pi i h+h^-(u), & e^+(u-\tau)&=e^{2\pi i\lambda}e^-(u), & f^+(u-\tau)&=e^{-2\pi i\lambda}f^-(u),
\end{align*}
which follows from the properties of Green distributions expressed in formulae~\r{B_shift_tau}.
Second, relations~\r{B_HhLJ}, which also follow from the relations between only positive (respectively
negative) half-currents, and finally, one needs to use the equality
$\frac{\partial}{\partial\tau}\Gg\pm(u-z-\tau)=
  e^{2\pi i\lambda}(-\frac{\partial}{\partial u}+\frac{\partial}{\partial\tau})\Gg-(u-z)$.
At this point we see the essential difference of the Lie algebra $\utau$ with the Lie algebra $\cEF$.

\subsection{Coalgebra  structures of $\cEF$ and $\utau$}
\label{B_subsec34}
We describe here the structure of {\emph quasi-Lie bialgebras} for our Lie algebras $\cEF$ and $\utau$.
We will start with an explicit expression for universal (dynamical) $r$-matrices for both Lie algebras.

\begin{prop} \label{B_prop_CDYBE}
The universal $r$-matrix for the Lie algebra $\cEF$ defined as
\begin{align}
 r_\lambda=\frac12\sum_{n\ge0}\hat h[\epsilon^{n;0}]\otimes\hat h[\epsilon_{n;0}]
  +\sum_{n\ge0}\hat f[\epsilon^{n;\lambda}]\otimes\hat e[\epsilon_{n;\lambda}]
  +\sum_{n<0}\hat e[\epsilon_{n;\lambda}]\otimes\hat f[\epsilon^{n;\lambda}]+c\otimes d
  \notag 
\end{align}
satisfies the Classical Dynamical Yang-Baxter Equation (CDYBE)
\begin{align}
  [r_{\lambda,{\bf 12}},r_{\lambda,{\bf 13}}]
 +[r_{\lambda,{\bf 12}},r_{\lambda,{\bf 23}}]
 +[r_{\lambda,{\bf 13}},r_{\lambda,{\bf 23}}]=
   h_{{\bf 1}}\frac{\partial}{\partial\lambda}r_{\lambda,{\bf 23}}
  -h_{{\bf 2}}\frac{\partial}{\partial\lambda}r_{\lambda,{\bf 13}}
  +h_{{\bf 3}}\frac{\partial}{\partial\lambda}r_{\lambda,{\bf 12}}\,. \label{B_CDYBE}
\end{align}
\end{prop}

Denote by $\Pi_u$  the evaluation representation $\Pi_u\colon\cEF\to\End V_u$,
where $V_u=\mathbb C^2\otimes\lfK_0$ and the subscript $u$ means the argument
of the functions belonging to $\lfK_0$:
\begin{align}
 &\Pi_u\colon \hat h[s]\mapsto s(u)H, & &\Pi_u\colon \hat e[s]\mapsto s(u)E, &
 &\Pi_u\colon \hat f[s]\mapsto s(u)F, \label{B_ev_repr}
\end{align}
and
$\Pi_u\colon c\mapsto 0$, $\Pi_u\colon d\mapsto \frac{\partial}{\partial u}$, where $ H=\left(
 \begin{smallmatrix}
  1 & 0 \\
  0  & -1
 \end{smallmatrix}\right)$,
 $ E=\left(
 \begin{smallmatrix}
  0 & 1 \\
  0  & 0
 \end{smallmatrix}\right)$,
 $ F=\left(
 \begin{smallmatrix}
  0 & 0 \\
  1  & 0
 \end{smallmatrix}\right)$,
$s\in\lfK_0$. The relations between $L$-operators and the universal $r$-matrix are given by the formulae
\begin{align}
\L_\lambda^+(u)&=(\Pi_u\otimes\id) r_\lambda, & 
\L_\lambda^-(u)-c\frac{\partial}{\partial u}&=-(\Pi_u\otimes\id) r_{\lambda,\textbf{21}},  \label{B_Lm_ev_r}
\end{align}
and $r_\lambda^+(u-v)=(\Pi_u\otimes\Pi_v) r_\lambda$.
Taking into account these formulae  and applying $(\Pi_u\otimes\Pi_v\otimes\id)$,
$(\id\otimes\Pi_u\otimes\Pi_v)$, $(\Pi_u\otimes\id\otimes\Pi_v)$ to the
equation~\eqref{B_CDYBE} we derive the relation~\eqref{B_rLpmLpm} for the sign `$+$',
the relation~\eqref{B_rLpmLpm} for the sign `$-$' and the relation~\eqref{B_rLpLm} respectively.
Applying $(\Pi_u\otimes\id)$ or $(\id\otimes\Pi_u)$ to the identity $[\Delta h,\rr_\lambda]=0$
we derive the relation~\eqref{B_HhL}.

The co-bracket $\delta\colon\cEF\to\cEF\wedge\cEF$ and an element $\varphi\in\cEF\wedge\cEF\wedge\cEF$
are defined as
$\delta x=[\Delta x,r_\lambda]=[x\otimes1+1\otimes x,r_\lambda]$, for
 $x\in\cEF$ and
\begin{align}
 \varphi=-[r_{\lambda,\textbf{12}},r_{\lambda,\textbf{13}}]
 -[r_{\lambda,\textbf{12}},r_{\textbf{23}}]-[r_{\lambda,\textbf{13}},r_{\lambda,\textbf{23}}]
 =
-h_{\textbf{1}}\frac{\partial}{\partial\lambda}r_{\lambda,\textbf{23}}
  +h_{\textbf{2}}\frac{\partial}{\partial\lambda}r_{\lambda,\textbf{13}}
  -h_{\textbf{3}}\frac{\partial}{\partial\lambda}r_{\lambda,\textbf{12}}\,. \notag
\end{align}
They equip the Lie algebra $\cEF$ with a structure of a quasi-Lie bialgebra~\cite{D90}.
This fact follows from the equality $r_{\textbf{12}}+r_{\textbf{21}}=\Omega$,
where $\Omega$ is a tensor Casimir element of algebra $\cEF$.
 To calculate this co-bracket on the half-currents in the matrix
 form we apply $(\Pi_u\otimes\id\otimes\id)$, $(\id\otimes\id\otimes\Pi_u)$
 to the equation~\eqref{B_CDYBE} and derive
\begin{align*}
 \delta L^+_{\lambda}(u)=-[L^+_{\lambda,\textbf{1}}(u)&,L^+_{\lambda,\textbf{2}}(u)]
  +H\frac{\partial}{\partial\lambda}r_{\lambda}
  -h\wedge\frac{\partial}{\partial\lambda}L^+_{\lambda}(u)\,, \\
 \delta L^-_{\lambda}(u)=-[L^-_{\lambda,\textbf{1}}(u)&,L^-_{\lambda,\textbf{2}}(u)]
  +H\frac{\partial}{\partial\lambda}r_{\lambda}
  -h\wedge\frac{\partial}{\partial\lambda}L^-_{\lambda}(u)
  -c\wedge\frac{\partial}{\partial u}L^-_{\lambda}(u)\,.
\end{align*}
We can see also that
 $\delta h=0$, $\delta c=0$, $\delta d=0$.

\begin{prop} \label{B_prop_CDYBEJ}
The universal $r$-matrix for the Lie algebra $\utau$ defined by formula
\begin{align}
 \rr_\lambda=\frac14\hat h[j^0]\otimes\hat h[j_0]&+\frac12\sum_{n\ne0}\frac{\hat h[j^n]\otimes\hat
  h[j_n]}{1-e^{2\pi in\tau}}+ \notag \\
 &+\sum_{n\in\mathbb Z}\frac{\hat e[j^n]\otimes\hat f[j_n]}{1-e^{2\pi i(n\tau+\lambda)}}
 +\sum_{n\in\mathbb Z}\frac{\hat f[j^n]\otimes\hat e[j_n]}{1-e^{2\pi i(n\tau-\lambda)}}
 +c\otimes d. \notag 
\end{align}
satisfies the equation
\begin{multline}
  [\rr_{\lambda,\bf{12}},\rr_{\lambda,\bf{13}}]
 +[\rr_{\lambda,\bf{12}},\rr_{\lambda,\bf{23}}]
 +[\rr_{\lambda,\bf{13}},\rr_{\lambda,\bf{23}}]=\\
  =h_{\bf{1}}\frac{\partial}{\partial\lambda}\rr_{\lambda,\bf{23}}
  -h_{\bf{2}}\frac{\partial}{\partial\lambda}\rr_{\lambda,\bf{13}}
  +h_{\bf{3}}\frac{\partial}{\partial\lambda}\rr_{\lambda,\bf{12}}
  -c_{\bf{1}}\frac{\partial}{\partial\tau}\rr_{\lambda,\bf{23}}
  +c_{\bf{2}}\frac{\partial}{\partial\tau}\rr_{\lambda,\bf{13}}
  -c_{\bf{3}}\frac{\partial}{\partial\tau}\rr_{\lambda,\bf{12}}. \notag
\end{multline}
\end{prop}

The relations between the universal matrix $\rr_{\lambda}$ and $L$-operators of the
algebra $\utau$ are the same as for the algebra $\cEF$ with a proper modification of the
evaluation representation
$\Pi_u\colon\utau\to\End {\cal V}_u$, ${\cal V}_u=\mathbb C^2\otimes K$ defined by the same
formulas~\r{B_ev_repr} as above for
$s\in K$.

The bialgebra structure of $\utau$ is defined in analogous way as for the algebra $\cEF$
and  can be presented in the form
\begin{align}
 &\delta\LL^+_{\lambda}(u)=-[\LL^+_{\lambda,\textbf{1}}(u),\LL^+_{\lambda,\textbf{2}}(u)]
  +H\frac{\partial}{\partial\lambda}\rr_{\lambda}
  -h\wedge\frac{\partial}{\partial\lambda}\LL^+_{\lambda}(u)+c\wedge\frac{\partial}{\partial\tau}
  \LL^+_{\lambda}(u)\,, \notag \\
 &\delta\LL^-_{\lambda}(u)=-[\LL^-_{\lambda,\textbf{1}}(u),\LL^-_{\lambda,\textbf{2}}(u)]
  +H\frac{\partial}{\partial\lambda}\rr_{\lambda}
  -h\wedge\frac{\partial}{\partial\lambda}\LL^-_{\lambda}(u)
  -c\wedge\Big(\frac{\partial}{\partial u}-\frac{\partial}{\partial\tau}\Big)\LL^-_{\lambda}(u)\,.\notag
\end{align}

\section*{Acknowledgements} This paper is a part of PhD thesis
of A.S. which he is prepared in co-direction of S.P. and V.R. in
the Bogoliubov  Laboratory of Theoretical Physics, JINR, Dubna and in LAREMA,
D\'epartement de Math\'ematics, Universit\'e d'Angers. He is
grateful to the CNRS-Russia exchange program on mathematical physics  and  personally to J.-M. Maillet
for financial and general support of this thesis project. V.R. are
thankful to B.Enriquez for discussions. He had used during the
project a partial financial support by ANR GIMP, Grant for support of scientific schools
NSh-8065.2006.2
and a support of INFN-RFBR "Einstein" grant (Italy-Russia). He acknowledges a warm hospitality
of Erwin Schr\"{o}dinger Institute for Mathematical Physics and the Program
"Poisson Sigma Models, Lie Algebroids, deformations and higher analogues" where this paper was finished.
S.P. was supported in part by RFBR grant 06-02-17383.

Leonid Vaksman, an excellent mathematician, one of Quantum Group "pioneers", patient
teacher and a bright person,
had passed away after coward disease when this paper was finished. We dedicate it to
his memory with a sad and sorrow.

\chapter{Classical elliptic current algebras. II}
\label{SA22}

\thispagestyle{empty} \pagestyle{myheadings}
\markboth{}{S.\,Pakuliak, V.\,Rubtsov, A.\,Silantyev\hfil{\itshape Classical elliptic current algebras. II}}

\newpage \thispagestyle{empty}

\vspace*{3cm}
\begin{center}
{\LARGE Classical elliptic current algebras. II} \\[17mm]
{\large Stanislav PAKULIAK~$^a$, Vladimir RUBTSOV~$^b$ and \\[4mm]
Alexey SILANTYEV~$^c$} \\[8mm]
$^a$ Institute of Theoretical \& Experimental Physics, 117259 Moscow, Russia \\
Laboratory of Theoretical Physics, JINR, 141980 Dubna, Moscow reg., Russia \\
~~E-mail: pakuliak@theor.jinr.ru \\[4mm]
$^b$ Institute of Theoretical \& Experimental Physics, 117259 Moscow, Russia  \\
D\'epartment de Math\'ematiques, Universit\'e d'Angers, 2 Bd. Lavoisier, 49045 Angers, France \\
~~E-mail: Volodya.Roubtsov@univ-angers.fr \\[4mm]
$^c$ Laboratory of Theoretical Physics, JINR, 141980 Dubna, Moscow reg., Russia \\
D\'epartment de Math\'ematiques, Universit\'e d'Angers, 2 Bd. Lavoisier, 49045 Angers, France \\
~~E-mail: silant@tonton.univ-angers.fr
\end{center}

\begin{center}
\bigskip
\bigskip
\emph{In memory of Leonid Vaksman}
\bigskip
\bigskip
\end{center}

\begin{abstract}
This is a continuation of the previous paper~\cite{S21}.
We describe different degenerations of the classical elliptic algebras.
They yield different versions of rational and
trigonometric current algebras.
 We also review the
averaging method of Faddeev-Reshetikhin, which allows to restore elliptic algebras from
the trigonometric ones.
\end{abstract}

\section{Introduction}
\label{C_sec1}

We continue the investigation of the  {\em classical current
algebras} (algebras which can be described by a collection of currents)
related to the classical $r$-matrices and which are quasi-classical limits of
SOS-type {\em quantized elliptic current algebras}. The latter
were introduced by Felder~\cite{F1}. We describe shortly the results of the
previous paper \cite{S21}.

Let $\fK$ be a function algebra on a
one-dimensional complex manifold $\Sigma$ with a point-wise
multiplication and a continuous invariant (non-degenerate) scalar
product $\la\cdot,\cdot\ra\colon\fK\times\fK\to\CC$. We shall call
the pair $(\fK,\la\cdot,\cdot\ra)$ {\em a test function algebra}.
The non-degeneracy of the scalar product implies that the algebra
$\fK$ can be extended to a space $\fK'$ of linear continuous
functionals on $\fK$. We use the notation $\la a(u),s(u)\ra$ or
$\la a(u),s(u)\ra_u$ for the action of the distribution
$a(u)\in\fK'$ on a test function $s(u)\in\fK$. Let
$\{\epsilon^i(u)\}$ and $\{\epsilon_i(u)\}$ be  dual bases of
$\fK$. A typical example of the element from $\fK'$ is the series
$\delta(u,z)=\sum_i\epsilon^i(u)\epsilon_i(z)$. This is a
delta-function distribution on $\fK$ because it satisfies $\la
\delta(u,z),s(u)\ra_u=s(z)$ for any test function $s(u)\in\fK$.

Consider an infinite-dimensional complex Lie algebra $\frg$ and an operator
$\hx\colon\fK\to\frg$. The expression
$
 x(u)=\sum_i \epsilon^i(u)\hx[\epsilon_i]
$ 
does not depend on a choice of dual bases in $\fK$ and is called a current corresponding
to the operator $\hx$ ($\hx[\epsilon_i]$ means an action of $\hx$ on $\epsilon_i$).
We should interpret the current $x(u)$ as a $\frg$-valued distribution such that
$
 \la x(u),s(u)\ra=\hx[s].\
$
That is the current $x(u)$ can be regarded as a kernel of the operator $\hx$ and the
latter formula\ 
 gives its invariant definition.

To describe different bialgebra structures in
the current algebras we have to decompose  the currents in these
algebras into difference of the currents which have good
analytical properties in certain domains:
$
 x(u)=x^+(u)-x^-(u).\
$
The $\frg$-valued distributions $x^+(u)$, $x^-(u)$ are called {\em
half-currents}. To perform such a decomposition we will use
so-called Green distributions~\cite{ER1}. Let
$\Omega^+,\Omega^-\subset\Sigma\times\Sigma$ be two domain
separated by a hypersurface $\bar\Delta\subset\Sigma\times\Sigma$
which contains the diagonal $\Delta=\{(u,u)\mid
u\in\Sigma\}\subset\bar\Delta$. Let there exist distributions
$G^+(u,z)$ and $G^-(u,z)$ regular in $\Omega^+$ and $\Omega^-$
respectively such that $ \delta(u,z)=G^+(u,z)-G^-(u,z)$. To define
half-currents corresponding to these Green distributions we
decompose them as $G^+(u,z)=\sum_i\alpha^+_i(u)\beta^+_i(z)$ and
$G^-(u,z)=\sum_i\alpha^-_i(u)\beta^-_i(z)$. Then the half-currents
are defined as $x^+(u)=\sum_i \alpha^+_i(u)\hx[\beta^+_i]$ and
$x^-(u)=\sum_i \alpha^-_i(u)\hx[\beta^-_i]$. This definition does
not depend on a choice of decompositions of the Green
distributions. The half-currents are currents corresponding to the
operators $\hx^\pm=\pm\,\hx\cdot P^\pm$, where $P^\pm[s](z)=\pm\la
G^\pm(u,z),s(u)\ra$, $s\in\fK$. One can express the half-currents
through the current $x(u)$, which we shall call {\em a total
current} in contrast with the half ones:
\begin{align}\label{C_hc_tc}
 x^+(u)=\la G^+(u,z)x(z)\ra_z,\qquad x^-(u)=\la G^-(u,z)x(z)\ra_z.
\end{align}
Here $\la a(z) \ra_z\equiv \la a(z),1 \ra_z$.

When $\Sigma$ is a covering of an elliptic
curve  we showed in \cite{S21} that:
\begin{itemize}
\item there are two essentially different choices of the test function algebras $\fK$ in this case
corresponding to the different covering $\Sigma$;
\item the same quasi-doubly periodic meromorphic functions regularized with respect to the different
 test function algebras define the different quasi-Lie bialgebra structures and, therefore,
 the different classical elliptic current algebras;
\item the internal structure of these two elliptic algebras is essentially different
in spite of a similarity in the commutation relations between
their half-currents.
\end{itemize}

Let $\tau\in\CC$, $\Im\tau>0$ be a module of
the elliptic curve $\mathbb C/\Gamma$, where $\Gamma=\mathbb
Z+\tau\mathbb Z$ is a period lattice. The odd theta function
$\theta(u)=-\theta(-u)$ is defined as a holomorphic function on
$\CC$ with the properties
\begin{align}
 \theta(u+1)&=-\theta(u), &\theta(u+\tau)&=-e^{-2\pi i u-\pi i\tau}\theta(u), &\theta'(0)&=1.
\end{align}

The first choice of the test function algebra corresponds to $\fK=\lfK_0$, where $\lfK_0$
consists of complex-valued one-variable functions defined in a
vicinity of origin (see details in Appendix~\ref{C_apA}) equipped
with the scalar product~\r{C_lfK_sp}. These functions can be
extended up to meromorphic functions on the covering
$\Sigma=\mathbb C$. The regularization domain $\Omega^+$,
$\Omega^-$ for Green distributions in this case consist of the
pairs $(u,z)$ such that $\max(1,|\tau|)>|u|>|z|>0$ and
$0<|u|<|z|<\max(1,|\tau|)$ respectively, where $\tau$ is an
elliptic module, and $\bar\Delta=\{(u,z)\mid |u|=|z|\}$.
We denote the corresponding elliptic Green distributions as
$\G\pm(u,z)$ and $G(u,z)$ and their action on a test function $s(u)$ is defined as
\begin{align}
\la\G\pm(u,z),s(u)\ra_u&=
  \oint\limits_{|u|>|z|\atop |u|<|z|}\frac{du}{2\pi i}\frac{\theta(u-z+\lambda)}{\theta(u-z)
  \theta(\lambda)}s(u)\,,
  \label{C_Glpmact} \\
 \la G(u,z),s(u)\ra_u&=
  \oint\limits_{|u|>|z|}\frac{du}{2\pi i}\frac{\theta'(u-z)}{\theta(u-z)}s(u)\,, \label{C_Gact}
\end{align}
where integrations are taken over circles around zero which are
small enough such that the corresponding inequality takes place.
The Green distributions are examples of the `shifted'
distributions (see Appendix~\ref{C_apA}) satisfying
\begin{align}
 \G+(u,z)-\G-(u,z)=\delta(u,z)\,,\quad
 G(u,z)+G(z,u)=\delta(u,z)\,. \label{C_GGdelta}
\end{align}
The oddness of function $\theta(u)$
leads to the following connection between the $\lambda$-depending Green distributions:
$
 \G+(u,z)=-G_{-\lambda}^-(z,u).
$

The second choice of the test function algebra
corresponds to $\fK=K=K(\Cyl)$. The algebra $K$ consists of entire periodic functions
$s(u)=s(u+1)$ on $\CC$ decaying exponentially at $\Im u\to\pm\infty$ equipped with an invariant scalar
product~\r{C_spJ}. This functions can be regarded as functions on cylinder $\Sigma=\Cyl$
(see Appendix~\ref{C_apA}). The regularization domain $\Omega^+$, $\Omega^-$ for Green distributions
consist of the pairs $(u,z)$ such that $-\Im\tau<\Im(u-z)<0$ and $0<\Im (u-z)<\Im\tau$ respectively
and $\bar\Delta=\{(u,z)\mid \Im u=\Im z\}$.
We denote the corresponding distributions as $\Gg\pm(u-z)$ and ${\cal G}(u-z)$ and
their action on the space $K$ is given by the formulas
\begin{align}
 \la\Gg\pm(u-z),s(u)\ra_u&=\int\limits_{-\Im\tau<\Im(u-z)<0\atop 0<\Im(u-z)<\Im\tau}
 \frac{du}{2\pi i} \frac{\theta(u-z+\lambda)}{\theta(u-z)\theta(\lambda)}s(u), \label{C_GlpmdefJ} \\
 \la{\cal G}(u-z),s(u)\ra_u
   &=\int\limits_{-\Im\tau<\Im(u-z)<0}\frac{du}{2\pi i}\frac{\theta'(u-z)}{\theta(u-z)}s(u), \label{C_GdefJ}
\end{align}
where the integration goes  over the line segments of unit length (cycles of cylinder) such that
the corresponding inequality takes place. The role of dual bases in the algebra $K$ is played by
$\{j_n(u)=e^{2\pi inu}\}_{n\in\mathbb Z}$ and
$\{j^n(u)=2\pi i e^{-2\pi inu}\}_{n\in\mathbb Z}$ (see Appendix~\ref{C_apA}), a decomposition to
these bases is the usual Fourier expansion. The Fourier expansions for the Green distributions
are
\begin{gather}
 \Gg\pm(u-z)=\pm2\pi i\sum_{n\in\mathbb Z}\frac{e^{-2\pi in(u-z)}}{1-e^{\pm2\pi i(n\tau-\lambda)}},
\qquad 
{\cal G}(u-z)=\pi i+2\pi i\sum_{n\ne0}\frac{e^{-2\pi in(u-z)}}{1-e^{2\pi in\tau}}. \label{C_GexpanJ}
\end{gather}
These expansions are in according with formulae
\begin{align}
 \Gg+(u-z)-\Gg-(u-z)=\delta(u-z),\qquad
  {\cal G}(u-z)+{\cal G}(z-u)=\delta(u-z),
  \notag
\end{align}
where $\delta(u-z)$ is a delta-function on $K$, given by the
expansion~\r{C_dffeJ}.

Using these two types of distributions we define in \cite{S21} two {\em different}
{\em quasi-Lie bialgebras} $\cEF$ and $\utau$,
which are classical limits of
quasi-Hopf algebras $E_{\tau,\eta}$~\cite{EF} and
$U_{p,q}(\hat{\mathfrak{sl}}_2)$~\cite{K98}
 respectively. The algebraic and coalgebraic structures of these quasi-Lie bialgebras
were described in terms of $L$-operators and classical $r$-matrices.
In case of the algebra $\cEF$ these objects are
\begin{gather}
\L_\lambda^\pm(u)=
 \begin{pmatrix}
  \frac12h^\pm(u) & \f^\pm(u) \\
  \e^\pm(u)       & -\frac12h^\pm(u)
 \end{pmatrix}, \\
\end{gather}
\begin{gather} r_\lambda^+(u,v)=
 \begin{pmatrix}
   \frac12G(u,v) & 0              & 0              & 0              \\
  0              & -\frac12G(u,v) & G^+_{-\lambda}(u,v)      & 0              \\
  0              & \G+(u,v)       & -\frac12G(u,v) & 0              \\
  0              & 0              & 0              & \frac12G(u,v)
 \end{pmatrix}
\end{gather}
and can be obtained from the corresponding classical universal $r$-matrix using evaluation
map. We denote the same objects in case of the algebra $\utau$ as
$\LL_\lambda^\pm(u)$ and $r$-matrix $\rr_\lambda^+(u-v)$  with distributions
${G}(u,v)$ and ${G}^\pm_{\lambda}(u,v)$ replaced everywhere by the
distributions ${\cal G}(u,v)$  and ${\cal G}^\pm_{\lambda}(u,v)$.

 In Section~\ref{C_sec4} we describe different degenerations of the
classical elliptic current algebras  $\cEF$ and $\utau$  in terms of
degenerations of Green distributions entering the $r$-matrix. The
degenerate Green functions define the $rLL$-relations, the
bialgebra structure and the analytic structure of half-currents.
We do not write out explicitly the bialgebra structure related to
the half-currents of the second classical elliptic algebra: it can
be reconstructed along the lines of the paper~\cite{S21}.

We discuss the inverse problem in Section~\ref{C_sec5}. A
way to present the trigonometric and elliptic solutions of
a Classical Yang-Baxter Equation (CYBE) by averaging of the
rational ones was introduced in~\cite{RF}. Faddeev and Reshetikhin applied the
averaging method to a description of corresponding algebras. Here, we only
represent the elliptic $r$-matrix $\rr^+_\lambda(u-v)$ and
the trigonometric $r$-matrix $\rr^{(c)+}(u-v)$ as an average of
the trigonometric $r$-matrix $\rr^{(b)+}(u-v)$ and the rational $r$-matrix
$\rr^{(a)+}(u-v)$ respectively, for some domains of parameters.

Finally, in the Appendix, we have collected technical and "folklore" definitions, results concerning the
test and distribution algebras on Riemann surfaces. Although some of the results
can be extracted from standard textbooks~\cite{GSh1}, \cite{Vl}, we were not able to
find them in the literature in the form suited for our goals and we have decided to keep
them for the sake of completeness.

\section{Degenerated classical elliptic algebras}
\label{C_sec4}

We will describe a behavior of our algebras while one or both periods of the elliptic curve become infinite.
The corresponding `degenerated' Green distributions, $r$-matrix and $L$-operators give us
a classical rational or a classical trigonometric `limit' of corresponding elliptic current algebras.

\subsection{Degenerations of the quasi-Lie bialgebra $\cEF$}
\label{C_subsec41}

There are two different degenerations denoted {\bf(a)} and {\bf(b)} for $\cEF$.
{\bf(a)} corresponds to the case when both periods are infinite ($\omega\to\infty$,
$\omega'\to\infty$). This is a rational degeneration. In the case {\bf(b)} one of the periods
is infinite ($\omega'\to\infty$) while another ($\omega$) rests finite. This is a case of trigonometric
degeneration. A situation when $\omega\to\infty$ and $\omega'$ is finite,
is equivalent to {\bf(b)} due to the symmetry of integration contour and, therefore,
we do not consider it separately.

\bb{Case (a): $\omega\to\infty$, $\omega'\to\infty$,
($\Im\frac{\omega'}{\omega}>0$)}. In order to turn to the lattice
of periods $\Gamma=\mathbb Z \omega+\mathbb Z \omega'$, with
$\frac{\omega'}{\omega}=\tau$, we need to re-scale the variables
like $u\to\frac{u}{\omega}$. Let us introduce the following
notations for rational Green distributions
\begin{align*}
\la\ph\pm(u,z),s(u)\ra_u=\oint\limits_{\substack{|u|>|z| \\|u|<|z|}}\frac{du}{2\pi i}\frac1{u-z}s(u).
\end{align*}
These distributions are degenerations of elliptic Green
distributions:
\begin{align*}
 \frac1{\omega}G\big(\frac{u}{\omega},\frac{z}{\omega}\big)&\to\ph+(u,z), &
 \frac1{\omega}G_{\frac\lambda{\omega}}^\pm\big(\frac{u}{\omega},\frac{z}{\omega}\big)&\to
 \frac1{\lambda}+\ph\pm(u,z),
\end{align*}
and the $r$-matrix tends to
\begin{align*}
r_\lambda^{(a)+}(u,v)&=\lim_{\omega,\omega'\to\infty}\frac1{\omega}r_{\frac\lambda{\omega}}^+
\Big(\frac{u}{\omega},\frac{v}{\omega}\Big)= \nn \\
 &=\begin{pmatrix}
   \frac12\ph+(u,v) & 0              & 0              & 0              \\
  0              & -\frac12\ph+(u,v) & -\frac1{\lambda}+\ph+(u,v)      & 0              \\
  0              & \frac1{\lambda}+\ph+(u,v)       & -\frac12\ph+(u,v) & 0              \\
  0              & 0              & 0              & \frac12\ph+(u,v)
 \end{pmatrix}. \nn
\end{align*}
Actually the quasi-Lie bialgebras obtained as rational
degenerations of $\cEF$ for different $\lambda$ is related to each
other by very simple twist. Therefore we shall consider only one
value of the parameter $\lambda$, namely we shall consider the
limited value $\lambda\to\infty$. The $r$-matrix and $L$-operators
looks then as follows
\begin{align*}
 \L^{(a)\pm}(u)&=\lim\limits_{\lambda\to\infty}\lim\limits_{\omega,\omega'\to\infty}\dfrac1{\omega}
   \L_{\frac{\lambda}{\omega}}^\pm\Big(\dfrac{u}{\omega}\Big)=
 \begin{pmatrix}
  \frac12h^{(a)\pm}(u) & f^{(a)\pm}(u) \\
  e^{(a)\pm}(u)       & -\frac12h^{(a)\pm}(u)
 \end{pmatrix},
 \end{align*}
\begin{align}
r^{(a)+}(u,v)&=\lim_{\lambda\to\infty}\frac1{\omega}r_\lambda^{(a)+}(u,v)= \nn \\
 &=\begin{pmatrix}
   \frac12\ph+(u,v) & 0              & 0                 & 0              \\
  0              & -\frac12\ph+(u,v) & \ph+(u,v)         & 0              \\
  0              & \ph+(u,v)         & -\frac12\ph+(u,v) & 0              \\
  0              & 0                 & 0                 & \frac12\ph+(u,v)
 \end{pmatrix}. \label{C_ra}
\end{align}

Substituting $u\to\frac{u}{\omega}$, $v\to\frac{v}{\omega}$, $\lambda\to\frac{\lambda}{\omega}$
into commutation relations of the algebra $\cEF$ given by
the formulas (3.38) and (3.39) of the paper \cite{S21},
 multiplying it by $\frac1{\omega^2}$ and passing to
the limits we obtain
\begin{align}
 [\L_1^{(a),\pm}(u),\L_2^{(a)\pm}(v)]&=[\L_1^{(a)\pm}(u)+\L_2^{(a)\pm}(v),r^{(a)+}(u,v)]. \label{C_rLpmLpma} \\
 [\L_1^{(a)+}(u),\L_2^{(a)-}(v)]&=[\L_1^{(a)+}(u)+\L_2^{(a)-}(v),r^{(a)+}(u,v)]
   +c\cdot\frac{\partial}{\partial u}r^{(a)+}(u,v). \label{C_rLpLma}
\end{align}
The half-currents have decompositions
\begin{align*}
   x^{(a)+}(u)&=\sum_{n\ge0}x^{(a)}_n u^{-n-1}, &   x^{(a)-}(u)&=-\sum_{n<0}x^{(a)}_n u^{-n-1},
\end{align*}
where $x^{(a)}_n=(x\otimes z^n,0,0)$ for $n\in\mathbb Z$, $x\in\{h,e,f\}$. This means that this
algebra coincides with a classical limit of the central extension of the Yangien double
$\widehat{DY(\slt)}$~\cite{Kh}.

\bb{Case (b): $\tau\to i\infty$ ($\omega=1$, $\omega'\to\infty$, $\tau=\omega'/\omega$, $\Im\tau>0$).}
In this case the degenerations of elliptic Green distributions
look as follows:
\begin{align*}
 G(u,z)&\to\psi^+(u-z), \\
 \G\pm(u,z)&\to\pi\ctg\pi\lambda+\psi^\pm(u,z),
\end{align*}
where
\begin{align*}
\la\psi^\pm(u,z),s(u)\ra_u=\oint\limits_{\substack{|u|>|z| \\ |u|<|z|}}\frac{du}{2\pi i}\pi\ctg\pi(u-z)s(u).
\end{align*}
By the same reason the degenerated algebras are isomorphic for
different $\lambda$ and we shall consider this bialgebra only in
the limit $\lambda\to-i\infty$. The $r$-matrix, $L$-operators and
$rLL$-relations in this case take the form
\begin{align}
r^{(b)+}(u,v)
 =\begin{pmatrix}
   \frac12\psi^+(u,v) & 0              & 0              & 0              \\
  0              & -\frac12\psi^+(u,v) & -\pi i+\psi^+(u,v)      & 0              \\
  0              & \pi i+\psi^+(u,v)       & -\frac12\psi^+(u,v) & 0              \\
  0              & 0              & 0              &  \frac12\psi^+(u,v)
 \end{pmatrix}. \label{C_rb}
\end{align}
\begin{align*}
 \L^{(b)\pm}(u)&=\lim\limits_{\lambda\to-i\infty}\lim\limits_{\tau\to i\infty}\L_{\lambda}^\pm(u)=
 \begin{pmatrix}
  \frac12h^{(b)\pm}(u) & f^{(b)\pm}(u) \\
  e^{(b)\pm}(u)       & -\frac12h^{(b)\pm}(u)
 \end{pmatrix},
\end{align*}
\begin{align*}
 [\L_1^{(b),\pm}(u),\L_2^{(b)\pm}(v)]&=[\L_1^{(b)\pm}(u)+\L_2^{(b)\pm}(v),r^{(b)+}(u,v)], \\
 [\L_1^{(b)+}(u),\L_2^{(b)-}(v)]&=[\L_1^{(b)+}(u)+\L_2^{(b)-}(v),r^{(b)+}(u,v)]
 +c\cdot\frac{\partial}{\partial u}r^{(b)+}(u,v)
\end{align*}
The half-currents have the following decompositions
\begin{align*}
  h^{(b)+}(u)&=\sum_{n\ge0}h^{(b)}_n \dfrac{\partial^n}{\partial u^n}\ctg\pi u, &
  h^{(b)-}(u)&=-\sum_{n\ge0}h^{(b)}_{-n-1} u^n, \\
  e^{(b)+}(u)&=ie^{(b)}_0+\sum_{n\ge0}e^{(b)}_n \dfrac{\partial^n}{\partial u^n}\ctg\pi u, &
  e^{(b)-}(u)&=ie^{(b)}_0-\sum_{n\ge0}e^{(b)}_{-n-1} u^n, \\
  f^{(b)+}(u)&=-if^{(b)}_0+\sum_{n\ge0}f^{(b)}_n \dfrac{\partial^n}{\partial u^n}\ctg\pi u, &
  f^{(b)-}(u)&=-if^{(b)}_0-\sum_{n\ge0}f^{(b)}_{-n-1} u^n,
\end{align*}
where $x^{(b)}_n=(x\otimes\pi\dfrac{(-1)^n}{n!}z^n,0,0)$,
$x^{(b)}_{-n-1}=(x\otimes\pi\dfrac{(-1)^n}{n!}\dfrac{\partial^n}{\partial z^n}\ctg\pi z,0,0)$
for $n\ge0$, $x\in\slt$.

\subsection{Degeneration of the quasi-Lie bialgebra $\utau$}
\label{C_subsec42}

In the case of algebra $\utau$ there are three cases of degenerations: {\bf(a)}, {\bf(b)} and {\bf(c)}.
The rational degeneration {\bf(a)} and trigonometric degeneration {\bf(b)} are analogous to the
corresponding degenerations of $\cEF$. Additionally there is one more trigonometric case {\bf(c)},
when $\omega\to\infty$ and $\omega'$ is finite. It is not equivalent to the case {\bf(b)} because
the integration contour for $\utau$ is not symmetric in this case.
In the cases {\bf(a)} and {\bf(c)}
the degeneration of elliptic Green distributions  acts on another test function algebra $Z$.
This is an algebra of entire functions $s(u)$ subjected to the inequalities
$|u^n s(u)|<C_n e^{p|\Im u|}$, $n\in\mathbb Z_+$, for some constants $C_n,p>0$
depending on $s(u)$~\cite{GSh1}. The scalar product in $Z$ is $\la s(u),
t(u)\ra_u=\int_{-\infty}^{+\infty}\frac{du}{2\pi i}s(u)t(u)$. The distributions acting on
$K$ can be considered as periodic distributions acting on $Z$.

\bb{Case (a): $\omega\to\infty$, $\omega'\to\infty$, ($\Im\frac{\omega'}{\omega}>0$).}
The degenerating of the elliptic Green distributions in this case reads as
\begin{align*}
 \frac1{\omega}{\cal G}\big(\frac{u-z}{\omega}\big)&\to\frac1{u-z-i0}=\Ph+(u-z), \\
 \frac1{\omega}{\cal G}_{\frac\lambda{\omega}}^\pm\big(\frac{u-z}{\omega}\big)&\to
   \frac{u-z+\lambda}{(u-z\mp i0)\lambda}=\frac1{\lambda}+\Ph\pm(u-z),
\end{align*}
where we introduced the rational Green distributions acting on the test function algebra $Z$ by formula
\begin{align*}
\la\Ph\pm(u-z),s(u)\ra=\int\limits_{\substack{\Im u<\Im z \\ \Im u>\Im z}}\frac{du}{2\pi i}\frac1{u-z}s(u),
\end{align*}
with infinite horizontal integration lines. They can be represented as integrals
\begin{align*}
 \Ph+(u-z)&=2\pi i\int\limits_{-\infty}^0e^{2\pi i k(u-z)}dk=2\pi i\int
 \limits_0^{+\infty}e^{-2\pi i k(u-z)}dk, \\
 \Ph-(u-z)&=-2\pi i\int\limits_0^{+\infty}e^{2\pi i k(u-z)}dk=-2\pi i\int
 \limits_{-\infty}^0e^{-2\pi i k(u-z)}dk, 
\end{align*}
These formulae are degenerations of the expansions~\eqref{C_GexpanJ}.

As above all the algebras that are obtained from the limit $\omega,\omega'\to\infty$
are isomorphic for the different values of the parameter $\lambda$ and it is sufficient
to describe the limit case $\lambda=\infty$. The $rLL$-relations are the
same as~\eqref{C_rLpmLpma},~\eqref{C_rLpLma}
\begin{align*}
 [\LL_{1}^{(a),\pm}(u),\LL_{2}^{(a)\pm}(v)]&=[\LL_{1}^{(a)\pm}(u)+\LL_{2}^{(a)\pm}(v),\rr^{(a)+}(u-v)], \\
 [\LL_{1}^{(a)+}(u),\LL_{2}^{(a)-}(v)]&=[\LL_{1}^{(a)+}(u)+\LL_{2}^{(a)-}(v),\rr^{(a)+}(u-v)]
 +c\cdot\frac{\partial}{\partial u}\rr^{(a)+}(u-v)
\end{align*}
with similar $r$-matrix
\begin{align}
\rr^{(a)+}(u-v)&=\lim_{\lambda\to\infty}\lim_{\omega,\omega'\to\infty}\frac1{\omega}
\rr_{\frac\lambda{\omega}}^+\Big(\frac{u-v}{\omega}\Big)= \notag \\
 &=\begin{pmatrix}
   \frac12\Ph+(u-v) & 0              & 0              & 0              \\
  0              & -\frac12\Ph+(u-v) & \Ph+(u-v)      & 0              \\
  0              & \Ph+(u-v)       & -\frac12\Ph+(u-v) & 0              \\
  0              & 0              & 0              & \frac12\Ph+(u-v)
 \end{pmatrix}, \label{C_raJinf}
\end{align}
but the entries of the $L$-matrix
\begin{align*}
 \LL^{(a)\pm}(u)&=\lim_{\lambda\to\infty}\lim\limits_{\omega,\omega'\to\infty}
 \dfrac1{\omega}\LL_{\frac{\lambda}{\omega}}^\pm\Big(\dfrac{u}{\omega}\Big)=
 \begin{pmatrix}
  \frac12h^{(a)\pm}(u) & f^{(a)\pm}(u) \\
  e^{(a)\pm}(u)       & -\frac12h^{(a)\pm}(u)
 \end{pmatrix}
\end{align*}
are decomposed to the integrals instead of the series:
\begin{align}
  x^{(a)+}(u)&=\int\limits_0^{+\infty}x^{(a)}_k e^{-2\pi i k u}dk, &
  x^{(a)-}(u)&=-\int\limits_{-\infty}^0x^{(a)}_k e^{-2\pi i k u}dk,
\end{align}
where $x^{(a)}_k=(x\otimes 2\pi i e^{2\pi i k z},0,0)$, $x\in\{h,e,f\}$.
These half-currents form a quasi-classical degeneration of the algebra
${\cal A}_\hbar(\widehat{\slt})$~\cite{KLP99}. The difference between
algebras $\widehat{DY(\slt)}$ and ${\cal A}_\hbar(\widehat{\slt})$
is considered in details on the quantum level in this paper.

\bb{Case (b). $\tau\to i\infty$, ($\omega=1$, $\omega'\to\infty$, $\tau=\omega'/\omega$, $\Im\tau>0$).}
Taking the limit $\tau\to i\infty$ in the formula~\eqref{C_GexpanJ} we obtain
\begin{align*}
 {\cal G}(u-z)&\to\pi i+2\pi i\sum_{n>0}e^{-2\pi in(u-z)}=\tilde\psi^+(u-z), \\
 \Gg\pm(u-z)&\to\pi\ctg\pi\lambda-\pi i
    \pm2\pi i\sum_{\substack{n\ge0 \\ n<0}}e^{-2\pi in(u-z)}=\pi\ctg\pi\lambda+\tilde\psi^\pm(u-z).
\end{align*}
where
\begin{align*}
\la\tilde\psi^\pm(u-z),s(u)\ra=\int\limits_{\substack{\Im u<\Im z \\ \Im u>\Im z}}
\frac{du}{2\pi i}\pi\ctg\pi(u-z)s(u), 
\end{align*}
where $s\in K$ and the integration is taken over a horizontal line segments with unit length.
In these notations the $r$-matrix (in the limit $\lambda\to-i\infty$) can be written as
\begin{align}
&\rr^{(b)+}(u-v)=\lim_{\lambda\to-i\infty}\lim_{\tau\to i\infty}\rr_\lambda^+(u-v)= \notag \\ 
 &=\begin{pmatrix}
   \frac12\tilde\psi^+(u-v) & 0              & 0              & 0              \\
  0              & -\frac12\tilde\psi^+(u-v) & -\pi i+\tilde\psi^+(u-v)      & 0              \\
  0              & \pi i+\tilde\psi^+(u-v)   & -\frac12\tilde\psi^+(u-v) & 0              \\
  0              & 0              & 0              & \frac12\tilde\psi^+(u-v)
 \end{pmatrix}. \label{C_rbJ}
\end{align}
Setting
\begin{align*}
 \LL^{(b)\pm}(u)&=\lim_{\lambda\to-i\infty}\lim\limits_{\tau\to i\infty}\LL_{\lambda}^\pm(u)=
 \begin{pmatrix}
  \frac12h^{(b)\pm}(u) & f^{(b)\pm}(u) \\
  e^{(b)\pm}(u)       & -\frac12h^{(b)\pm}(u)
 \end{pmatrix}.
\end{align*}
one derives
\begin{align*}
 [\LL_{1}^{(b),\pm}(u),\LL_{2}^{(b)\pm}(v)]=&[\LL_{1}^{(b)\pm}(u)+\LL_{2}^{(b)\pm}(v),\rr^{(b)+}(u-v)], \\
 [\LL_{1}^{(b)+}(u),\LL_{2}^{(b)-}(v)]=&[\LL_{1}^{(b)+}(u)+\LL_{2}^{(b)-}(v),\rr^{(b)+}(u-v)]
 +c\cdot\frac{\partial}{\partial u}\rr^{(b)+}(u-v),
\end{align*}
which define some Lie algebra together with half-current decompositions
\begin{align*}
 h^{(b)+}(u)&=-\frac12h^{(a)}_0+\sum_{n\ge0}h^{(a)}_n e^{-2\pi inu}, &
 h^{(b)-}(u)&=-\frac12h^{(a)}_0-\sum_{n<0}h^{(a)}_n e^{-2\pi inu}, \\
 e^{(b)+}(u)&=\sum_{n\ge0}e^{(a)}_n e^{-2\pi inu}, &
 e^{(b)-}(u)&=-\sum_{n<0}e^{(a)}_n e^{-2\pi inu},\\
 f^{(b)+}(u)&=\sum_{n>0}f^{(a)}_n e^{-2\pi inu}, &
 f^{(b)-}(u)&=-\sum_{n\le0}f^{(a)}_n e^{-2\pi inu}.
\end{align*}
where $x^{(b)}_k=(x\otimes 2\pi i e^{2\pi inz},0,0)$, $x\in\{h,e,f\}$. This is exactly
an affine Lie algebra $\widehat{\slt}$ with a bialgebra structure inherited from the quantum
affine algebra $U_q(\widehat{\slt})$.

\bb{Case (c): $\omega=\to\infty$, $\omega'=\tau\omega=\frac{i}\eta=const$, ($\Re\eta>0$).}
Substituting $u\to\frac{u}{\omega}$, $z\to\frac{z}{\omega}$ to the expansion for ${\cal G}(u-z)$
we yield the following degeneration
\begin{align}
 \frac1\omega{\cal G}\big(\frac{u-z}\omega\big)&\to
  \Psi(u-z)\stackrel{\tiny{\rm def}}{=}2\pi i\vpint\limits_{-\infty}^{+\infty}
   \frac{e^{-2\pi i k(u-z)}dk}{1-e^{-\frac{2\pi k}\eta}}. \label{C_GdaJ}
\end{align}
Here $\vpints$ means an integral in sense of principal value. The integral in the
formula~\eqref{C_GdaJ} converges in the domain $\Re\eta^{-1}<\Im(u-z)<0$ and is equal
to $\pi\eta\cth\pi\eta(u-z)$ in this domain. It means that the distribution $\Psi(u-z)$
defined by formula~\eqref{C_GdaJ} acts on $Z$ as follows
\begin{align}
\la\Psi(u-z),s(u)\ra=\int\limits_{-\Re\eta^{-1}<\Im (u-z)<0}\frac{du}{2\pi i}
\pi\eta\cth\pi\eta(u-z)s(u).\label{C_Psitg}
\end{align}

The degeneration of Green distributions parametrized by $\lambda$ can be performed in different ways.
We can consider a more general substitution $\lambda\to\mu+\frac\lambda\omega$ instead of
$\lambda\to\frac\lambda\omega$ used above. Substituting $u\to\frac{u}{\omega}$,
$z\to\frac{z}{\omega}$, $\lambda\to\mu+\frac{\lambda}{\omega}$ to the formula~\eqref{C_GexpanJ}
and taking the limits $\omega\to\infty$ and $\lambda\to\infty$  we obtain
\begin{align}
 \frac1{\omega}{\cal G}_{\mu+\frac\lambda{\omega}}^+\big(\frac{u-z}{\omega}\big)&\to
 2\pi i\int\limits_{-\infty}^{\infty}\frac{e^{-2\pi i k(u-z)}dk}{1-e^{-\frac{2\pi k}\eta-2\pi i\mu}},
 \label{C_GpdaJmu}\\
 \frac1{\omega}{\cal G}_{\mu+\frac\lambda{\omega}}^-\big(\frac{u-z}{\omega}\big)&\to
 2\pi i\int\limits_{-\infty}^{\infty}\frac{e^{-2\pi i k(u-z)}dk}{e^{\frac{2\pi k}\eta+2\pi i\mu}-1}.
 \label{C_GmdaJmu}
\end{align}
Left hand sides of~\eqref{C_GpdaJmu} and \eqref{C_GmdaJmu} as well as right hand sides are invariant under
$\mu\to\mu+1$, but the right hand side is not holomorphic with respect to $\mu$ because of
integrand poles. The complex plane split up to the following  analyticity zones
$\frac{\Im\eta\Im\mu}{\Re\eta}+n<\Re\mu<\frac{\Im\eta\Im\mu}{\Re\eta}+n+1$, $n\in\mathbb Z$,
and due to periodicity with respect to $\mu$ one can consider only one of these zones.

Integrals in the formulae~\eqref{C_GpdaJmu} and \eqref{C_GmdaJmu} converge in the domain
$\Re\eta^{-1}<\Im(u-z)<0$ and $0<\Im (u-z)<\Re\eta^{-1}$ respectively and they can be
calculated like the integral in~\eqref{C_GdaJ} for the chosen zone. Denote by $\Psi_\mu^+(u-z)$
and $\Psi_\mu^-(u-z)$ the analytic continuation with respect to $\mu$ of the right hand sides
of~\eqref{C_GpdaJmu} and \eqref{C_GmdaJmu} respectively from the zone
\begin{align}
 \frac{\Im\eta\Im\mu}{\Re\eta}<\Re\mu<\frac{\Im\eta\Im\mu}{\Re\eta}+1. \label{C_zone_mu}
\end{align}
Thus, this degeneration of Green distributions can be rewritten as
\begin{align}
 \lim_{\omega\to\infty}{\cal G}_{\mu+\frac\lambda{\omega}}^\pm\big(\frac{u-z}{\omega}\big)
    &=\Psi_\mu^\pm(u-z), \notag \\
\la\Psi_\mu^\pm(u-z),s(u)\ra&=\int\limits_{\substack{-\Re\eta^{-1}<\Im (u-z)<0 \\
0<\Im(u-z)<\Re\eta^{-1}}}\frac{du}{2\pi i}2\pi\eta\frac{e^{-2\pi\eta\mu(u-z)}}{1-e^{-2\pi\eta(u-z)}}s(u)\ ,
 \label{C_Psipsh}
\end{align}
where $s\in Z$ and the integrals are taken over the horizontal lines.

For the values $\Re\mu=\frac{\Im\eta\Im\mu}{\Re\eta}+n$, $n\in\mathbb Z$, integrands
in~\eqref{C_GpdaJmu}, \eqref{C_GmdaJmu} have a pole on the real axis and the distributions
$\Psi_\mu^+(u-z)$ and $\Psi_\mu^-(u-z)$ regularize these integrals as analytical continuation
(see~\cite{GSh1}). The $r$-matrix obtained by another regularization does not satisfy the CYBE.

The degeneration of $r$-matrix is~\footnote{Let us remark that the degeneration of
the entry $r^+(u-v)_{12,21}={\cal G}^+_{-\lambda}(u-v)$ in the zone~\eqref{C_zone_mu} is
$\Psi^+_{1-\mu}(u-v)$, but is not $\Psi^+_{-\mu}(u-v)$ as one could expect, because of
periodicity with respect to $\mu$ and the fact that $\mu$ belongs to the zone~\eqref{C_zone_mu}
if and only if $1-\mu$ belongs to the zone~\eqref{C_zone_mu}. One can also use the relations
${\cal G}^+_{-\mu-\frac\lambda\omega}(\frac{u-v}\omega)=-{\cal G}^-_{\mu+\frac\lambda\omega}
(\frac{v-u}\omega)\longrightarrow-\Psi^-_{\mu}(v-u)=\Psi^+_{1-\mu}(u-v)$.}
\begin{align}
&\rr^{(c)+}(u-v)=\lim_{\lambda\to\infty}\lim_{\omega\to\infty}\frac1{\omega}
\rr_{\mu+\frac\lambda{\omega}}^+\Big(\frac{u-v}{\omega}\Big)= \notag \\ 
 &=\begin{pmatrix}
   \frac12\Psi(u-v) & 0              & 0              & 0              \\
  0              & -\frac12\Psi(u-v) & -\Psi_\mu^-(v-u)      & 0              \\
  0              & \Psi_\mu^+(u-v)   & -\frac12\Psi(u-v) & 0              \\
  0              & 0              & 0              & \frac12\Psi(u-v)
 \end{pmatrix}. \label{C_rcJ}
\end{align}
The $L$-operators
\begin{align*}
 \LL^{(c)\pm}(u)&=\lim_{\lambda\to\infty}\lim\limits_{\omega\to\infty}\dfrac1
 {\omega}\LL_{\mu+\frac{\lambda}{\omega}}^\pm\Big(\dfrac{u}{\omega}\Big)=
 \begin{pmatrix}
  \frac12h^{(c)\pm}(u) & f^{(c)\pm}(u) \\
  e^{(c)\pm}(u)       & -\frac12h^{(c)\pm}(u)
 \end{pmatrix}.
\end{align*}
with this $r$-matrix satisfy the dynamical $rLL$-relations
\begin{align}
 [\LL_{1}^{(c)\pm}(u),\LL_{2}^{(c)\pm}(v)]=[\LL_{1}^{(c)\pm}(u)&+\LL_{2}^{(c)\pm}(v),\rr^{(c)+}(u-v)]
 -c\cdot i\eta^2\frac{\partial}{\partial\eta}\rr^{(c)+}(u-v)\ , \nn \\
 [\LL_{1}^{(c)+}(u),\LL_{2}^{(c)-}(v)]=[\LL_{1}^{(c)+}(u)&+\LL_{2}^{(c)-}(v),\rr^{(c)+}(u-v)]+
  \label{C_rLLcJ} \\
 &+c\cdot\bigg(\frac{\partial}{\partial u}-i\eta^2\frac{\partial}{\partial\eta}\bigg)\rr^{(c)+}(u-v)\ . \nn
\end{align}
Decompositions of the half-current in this degeneration are
\begin{align*}
 h^{(c)+}(u)&=\vpint\limits_{-\infty}^{+\infty}h^{(c)}_k
 \frac{e^{-2\pi i k u}dk}{1-e^{-\frac{2\pi k}\eta}}\ , &
 h^{(c)-}(u)&=\vpint\limits_{-\infty}^{+\infty}h^{(c)}_k
 \frac{e^{-2\pi i k u}dk}{e^{\frac{2\pi k}\eta}-1}\ , \\
 e^{(c)+}(u)&=\vpint\limits_{-\infty}^{+\infty}e^{(c)}_k
 \frac{e^{-2\pi i k u}dk}{1-e^{-\frac{2\pi k}\eta-2\pi i\mu}}\ , &
 e^{(c)-}(u)&=\vpint\limits_{-\infty}^{+\infty}e^{(c)}_k
 \frac{e^{-2\pi i k u}dk}{e^{\frac{2\pi k}\eta+2\pi i\mu}-1}\ , \\
 f^{(c)+}(u)&=\vpint\limits_{-\infty}^{+\infty}f^{(c)}_k
 \frac{e^{-2\pi i k u}dk}{1-e^{-\frac{2\pi k}\eta+2\pi i\mu}}\ , &
 f^{(c)-}(u)&=\vpint\limits_{-\infty}^{+\infty}f^{(c)}_k
 \frac{e^{-2\pi i k u}dk}{e^{\frac{2\pi k}\eta-2\pi i\mu}-1}\ ,
\end{align*}
where $x^{(a)}_k=x\otimes 2\pi i e^{2\pi i k z}$, $x\in\{h,e,f\}$.
We do not make explicit the dependence of the parameter $\mu$ because, contrary
to $\lambda$, it is not a dynamical parameter. We also omit dependence on the parameter $\eta$
which provides the dynamics over $c$ just as we omitted its analogue $\tau$ in the elliptic case.
The case $\mu=\frac12$, $\Im\eta=0$ coincides with the quasi-classical limit of the
quantum current algebra ${\cal A}_{\hbar,\eta}(\widehat{\slt})$ \cite{KLP98,CKP}.
This algebra was investigated in~\cite{KLPST} in detail. Other degenerations
{\bf(c)} seem to be unknown, though the matrices $\rr^{(c)+}(u)$ fit the Belavin-Drinfeld
classification~\cite{BD}.

\section{Averaging of $r$-matrices}
\label{C_sec5}

Now we will use the averaging method of Faddeev-Reshetikhin~\cite{RF} and will write down trigonometric and
elliptic $r$-matrices starting with a rational solution of the Classical Yang-Baxter Equation. We show that the
$r$-matrices satisfying to a Dynamical Classical Yang-Baxter Equation can be also obtained by this method.

\setcounter{bbcount}{0}

\bb{CYBE.} A meromorphic $\mathfrak{a}\otimes\mathfrak{a}$-valued function $X(u)$
(in our case $\mathfrak a=\slt$) is called solution of the CYBE if it satisfies the equation
\begin{multline}
  [X_{12}(u_1-u_2),X_{13}(u_1-u_3)]+[X_{12}(u_1-u_2),X_{23}(u_2-u_3)]+ \\
  +[X_{13}(u_1-u_3),X_{23}(u_2-u_3)]=0. \label{C_CYBE}
\end{multline}
The $r$-matrices $\rr^{(a)+}(u-v)$, $\rr^{(b)+}(u-v)$, $\rr^{(c)+}(u-v)$
defined by formulae~\eqref{C_raJinf}, \eqref{C_rbJ} and \eqref{C_rcJ} satisfy CYBE~\eqref{C_CYBE},
what follows from the fact that they are regularization of the corresponding rational
and trigonometric solutions of CYBE in the domain $\Im u<\Im v$. Indeed, in order to
check the equation~\eqref{C_CYBE} for these $r$-matrices it is sufficient to check it
in the domain $\Im u_1<\Im u_2<\Im u_3$. The regularization of the same first
two solutions of CYBE ({\bf(a)} and {\bf(b)} cases) but in domain $|u|>|v|$ are
$r$-matrices $\rr^{(a)+}(u-v)$ and $\rr^{(b)+}(u-v)$ (formulae~\eqref{C_rb} and \eqref{C_rb})
respectively. Hence they also satisfy~\eqref{C_CYBE}, but where $X_{ij}(u_i-u_j)$ replaced
by $X_{ij}(u_i,u_j)$. The elliptic $r$-matrix $\rr^+_\lambda(u-v)$ satisfies
{\itshape Dynamical} CYBE, but it can be also obtained by the averaging method.

\bb{Basis of averaging.} As it was shown in~\cite{BD} each solution of CYBE $X(u)$ is a
rational, trigonometric or elliptic (doubly periodic) function of $u$, the poles of $X(u)$
form a lattice $\mathfrak R\subset\mathbb C$ and there is a group homomorphism
$A\colon\mathfrak R\to\Aut\mathfrak{g}$ such that for each $\gamma\in\mathfrak R$
one has the relation $X(u+\gamma)=(A_\gamma\otimes\id)X(u)$. Having a rational solution
$X(u)$, for which $\mathfrak R=\{0\}$, and choosing an appropriate automorphisms
$A=A_{\gamma_0}$ we can construct the trigonometric solution with $\mathfrak R=\gamma_0\mathbb Z$
in the form
\begin{align}
 \sum_{n\in\mathbb Z}(A^n\otimes\id)X(u-n\gamma_0). \label{C_aver}
\end{align}
Applying the same procedure for a trigonometric solution with $\mathfrak R=\gamma_1\mathbb Z$,
 where $\gamma_1/\gamma_0\not\in\mathbb R$, we obtain an elliptic (doubly periodic) solution of
 CYBE with $\mathfrak R=\gamma_1\mathbb Z+\gamma_0\mathbb Z$. The convergence of series in the
 formula~\eqref{C_aver} should be understood in the principal value sense (below we will detail it).

\bb{Quasi-doubly periodic case.} The entries of elliptic $r$-matrix $\rr^+_\lambda(u)$ --
elliptic Green distributions -- are regularizations of {\itshape quasi-}doubly periodic functions.
This is a direct consequence of those fact that this $r$-matrix satisfies {\itshape Dynamical} CYBE and
therefore does not belong to the Belavin-Drinfeld classification~\cite{BD}. Nevertheless, these
functions have the elliptic type of the pole lattice $\mathfrak R=\Gamma=\mathbb Z+\mathbb Z\tau$
and one can expect that the $r$-matrix $\rr^+_\lambda(u)$ can be represent by formula~\eqref{C_aver}
with $\gamma_0=\tau$ and $X(u)$ replaced by some trigonometric $r$-matrix with $\mathfrak R=\mathbb Z$.
To pass on from the averaging of meromorphic functions to the averaging of distributions we should
choose the proper regularization. Actually the regularization of this trigonometric $r$-matrix in
this formula can depend on $n$ (see~\eqref{C_rerb}). The $r$-matrices $r^{(a)+}(u,v)$, $r^{(b)+}(u,v)$,
$r^+_\lambda(u,v)$ can be also regarded as a regularization of the same meromorphic $\slt\otimes\slt$-valued
function, but they depend on $u$, $v$ in more general way than on the difference $(u-v)$.
This makes their averaging
 to be more complicated. By this reason we shall not consider these matrices
 in this section.

\bb{Dynamical elliptic $r$-matrix as an averaging of $\rr^{(b)\pm}(u)$.} To represent
the $r$-matrix $\rr^+_\lambda(u)$ as an averaging of trigonometric matrix~\eqref{C_rbJ}
we need the following formulae
\begin{align}
 \frac{\theta'(u)}{\theta(u)}&=v.p.\sum_{n\in\mathbb Z}\pi\ctg\pi(u-n\tau)\ , \label{C_avctg} \\
 \frac{\theta(u+\lambda)}{\theta(u)\theta(\lambda)}&=\frac{\theta'(\lambda)}{\theta(\lambda)}+
 v.p.\sum_{n\in\mathbb Z}\big(\pi e^{-2n\pi i\lambda}\ctg\pi(u-n\tau)
   +(1-\delta_{n0})\pi e^{-2n\pi i\lambda}\ctg\pi n\tau\big)\ , \label{C_avctg_}
\end{align}
where  $|\Im\lambda|<\Im\tau$, $\lambda\notin\mathbb Z$ and the symbol $v.p.$ means convergence
of the series in the principal value sense:
\begin{align*}
 v.p.\sum_{n\in\mathbb Z}x_n=\lim_{N\to\infty}\sum_{n=-N}^N x_n\ .
\end{align*}
The Fourier expansion of the function $\frac{\theta'(\lambda)}{\theta(\lambda)}$ has
the form~\eqref{C_GexpanJ} (with $(u-z)$ replaced by $\lambda$) in the domain $-\Im\tau<\Im\lambda<0$.
Substituting this expansion to the right hand side of~\eqref{C_avctg_} one yields
\begin{align}
 \frac{\theta(u+\lambda)}{\theta(u)\theta(\lambda)}=
  v.p.\sum_{n\in\mathbb Z}\pi e^{-2n\pi i\lambda}\big(\ctg\pi(u-n\tau)+i\big)\ , \label{C_avctg_p}\\
 \frac{\theta(u-\lambda)}{\theta(u)\theta(-\lambda)}=
  v.p.\sum_{n\in\mathbb Z}\pi e^{2n\pi i\lambda}\big(\ctg\pi(u-n\tau)-i\big)\ .  \label{C_avctg_m}
\end{align}
The formula~\eqref{C_avctg_m} is obtained from~\eqref{C_avctg_p} by replacing $u\to-u$, $n\to-n$,
hence both formulae are valid in the domain $-\Im\tau<\Im\lambda<0$.
Let us choose an automorphism $A=A_\tau$ as follows
\begin{align*}
 &A\colon h \mapsto h, & &A\colon e \mapsto e^{2\pi i\lambda} e,  & &A\colon f \mapsto e^{-2\pi i\lambda} f,
\end{align*}
and define $\vartheta_n=+$ for $n\ge0$ and $\vartheta_n=-$ for $n<0$.
Then the formulae~\eqref{C_avctg}, \eqref{C_avctg_p}, \eqref{C_avctg_m} imply
\begin{align}
 \rr^+_\lambda(u)=v.p.\sum_{n\in\mathbb Z}(A^n\otimes\id)\rr^{(b),\vartheta_n}(u-\tau n)\ ,  \label{C_rerb}
\end{align}
where $-\Im\tau<\Im\lambda<0$ and $\rr^{(b),-}(u)$ defined by formula~\eqref{C_rbJ} with
$\tilde\psi^+(u)$ replaced by $\tilde\psi^-(u)$. Let us notice that $r$-matrix $\rr^+_\lambda(u)$
and $\rr^{(b)\pm}(u)$ act as distributions on the same space $K$ and hence belong to the same space.
Thus we do not have any problem with interpretation of the averaging formula in sense of distributions.

\bb{The matrix $\rr^{(c)+}(u)$ as an averaging of $\rr^{(a)\pm}(u)$.} We restrict our
attention to the case $\frac{\Im\eta\Im\mu}{\Re\eta}\le\Re\mu<\frac{\Im\eta\Im\mu}{\Re\eta}+1$.
In this case Green distributions $\Psi^+_\mu(u)$, $\Psi^-_\mu(-u)$, $\Psi(u)$ entering into the
$r$-matrix $\rr^{(c)+}(u)$ are defined by~\eqref{C_Psitg}, \eqref{C_Psipsh}. One has the formula
\begin{align*}
 2\pi\eta\frac{e^{2\pi\eta\mu u}}{e^{2\pi\eta u}-1}
  =v.p.\sum_{n\in\mathbb Z}\frac{e^{2\pi i\mu n}}{u-i\eta^{-1}n}. 
\end{align*}
Replacing $u\to-u$, $n\to-n$ in both sides one yields
\begin{align*}
 2\pi\eta\frac{e^{-2\pi\eta\mu u}}{1-e^{-2\pi\eta u}}
  =v.p.\sum_{n\in\mathbb Z}\frac{e^{-2\pi i\mu n}}{u-i\eta^{-1}n}. 
\end{align*}
Let us choose the automorphism $A=A_{i\eta^{-1}}$ in the form
\begin{align*}
 &A\colon H\mapsto H, & &A\colon E\mapsto e^{2\pi i\mu}E, & &A\colon F\mapsto e^{-2\pi i\mu}F\ .
\end{align*}
Then the formulae imply the averaging $r$-matrix
\begin{align*}
 \rr^{(c)+}(u)=\sum_{n\in\mathbb Z}(A^n\otimes\id)\rr^{(a),\vartheta_n}(u-i\eta^{-1}n)\ ,
\end{align*}
where $\frac{\Im\eta\Im\mu}{\Re\eta}\le\Re\mu<\frac{\Im\eta\Im\mu}{\Re\eta}+1$ and $\rr^{(a),-}(u)$
is defined by formula~\eqref{C_raJinf} with $\Phi^+(u)$ substituted by $\Phi^-(u)$. These obtained
averaged $r$-matrices $\rr^{(c)+}(u)$ and $\rr^{(a)\pm}(u)$ act also on the same space -- on the
algebra $Z$ from 4.2.

\section*{Acknowledgements} This paper is a part of PhD thesis
of A.S. which he is prepared in co-direction of S.P. and V.R. in
the Bogoliubov  Laboratory of Theoretical Physics, JINR, Dubna and in LAREMA,
D\'epartement de Math\'ematics, Universit\'e d'Angers. He is
grateful to the CNRS-Russia exchange program on mathematical physics  and  personally to J.-M. Maillet
for financial and general support of this thesis project. V.R. are
thankful to B.Enriquez for discussions. He had used during the
project a partial financial support by ANR GIMP, Grant for support of scientific schools
NSh-8065.2006.2
and a support of INFN-RFBR "Einstein" grant (Italy-Russia). He acknowledges a warm hospitality
of Erwin Schr\"{o}dinger Institute for Mathematical Physics and the Program
"Poisson Sigma Models, Lie Algebroids, deformations and higher analogues" where this paper was finished.
S.P. was supported in part by RFBR grant 06-02-17383.

\section*{Appendix}

\setcounter{section}{0} \setcounter{subsection}{0}
\renewcommand{\thesection}{\Alph{section}}

\section{Test function algebras $\lfK_0$ and $K=K(\Cyl)$}\label{C_apA}

\setcounter{bbcount}{0}

\bb{Test function algebra $\lfK_0$.} Let $\lfK_0$ be a set of complex-valued meromorphic
functions defined in some vicinity of origin which have an only pole in the origin. If $s_1(u)$
and $s_2(u)$ are two such functions with domains $U_1$ and $U_2$ then their sum $s_1(u)+s_2(u)$
and their product $s_1(u)\times s_2(u)$ are also functions of this type which are defined
in the intersection $U_1\cap U_2$. Moreover if $s(u)$ is a function from $\lfK_0$
which is not identically zero then there exists a neighborhood $U$ of the origin such that
the domain $U\bs0$ does not contain zeros of function $s(u)$ and, therefore,
the function $\dfrac1{s(u)}$ is a function from $\lfK_0$ with the domain $U$. This
means that the set $\lfK_0$ can be endowed with a structure of a function field. We shall consider
$\lfK_0$ as an associative unital algebra over $\mathbb C$ equipped with the invariant scalar product
\begin{align}
 \la s_1(u),s_2(u)\ra=\oint\limits_{C_0}\frac{du}{2\pi i}s_1(u)s_2(u), \label{C_lfK_sp}
\end{align}
where $C_0$ is a contour encircling zero and belonging in the intersection of domains
of functions $s_1(u)$, $s_2(u)$, such that the scalar product is a residue in zero. We
consider the algebra $\lfK_0$ as an algebra of test functions. A convergence in $\lfK_0$
is defined as follows: a sequence of functions $\{s_n(u)\}$ converges to
zero if there exists a number $N$ such that all the function $z^Ns_n(u)$ are regular
in origin and all the coefficients in their Laurent expansion tend to zero.
One can consider (instead of the algebra $\lfK_0$ defined in this way one) the completion
$\overline{\lfK}_0=\mathbb C[u^{-1}][[u]]$. Linear continuous functionals on $\lfK_0$
 are called distributions and form the space $\lfK'_0$ (which coincide with $\overline{\lfK}'_0$).
 The scalar product~\eqref{C_lfK_sp} being continues defines a continuous injection
 $\lfK_0\to\lfK'_0$. We use the notation $\la a(u),s(u)\ra$ for the action of a distribution $a(u)$
 on a test function $s(u)$ and also the notation $\la a(u)\ra_u=\la a(u),1\ra$, where $1$ is a
 function which identically equals to the unit.

One can define a 'rescaling' of a test function $s(u)$ as a function $s\big(\frac u\alpha\big)$,
where $\alpha\in\mathbb C$, and therefore a 'rescaling' of distributions by the formula
$\la a(\frac u\alpha\big),s(u)\ra=\la a(u),s(\alpha u)\ra$. On the contrary, we are unable to define a
'shift' of test functions by a standard rule, because the operator $s(u)\mapsto s(u+z)$
is not a continuous one~\footnote{Consider, for example, the sum
$s_N(u)=\sum_{n=0}^N(\frac{u}{\alpha})^n$. For each $z$ there exist $\alpha$ such that the
sum $s_N(u+z)$ diverges, when $N\to\infty$.}. Nevertheless we use distributions 'shifted' in some sense.
Namely, we say that a two-variable distribution $a(u,z)$ (a linear continuous functional
$a\colon\lfK_0\otimes\lfK_0\to\mathbb C$) is 'shifted' if it possesses the properties:
(i) for any $s\in\lfK_0$ the functions $s_1(z)=\la a(u,z),s(u)\ra_u$ and $s_2(u)=\la a(u,z),s(z)\ra_z$
belong to $\lfK_0$; (ii) $\frac{\partial}{\partial u}a(u,z)=-\frac{\partial}{\partial z}a(u,z)$.
Here the subscripts $u$ and $z$ mean the corresponding partial action, for instance,
$\la a(u,z),s(u,z)\ra_u$ is a distribution acting on $\lfK_0$ by the formula
\begin{align*}
 \La\la a(u,z),s(u,z)\ra_u,t(z)\Ra=\la a(u,z),s(u,z)t(z)\ra.
\end{align*}
The condition (ii) means the equality $\la a(u,z),s'(u)t(z)\ra=-\la a(u,z),s(u)t'(z)\ra$.
The condition (i) implies that for any $s\in\lfK_0\otimes\lfK_0$ the expression
\begin{align}
 \la a(u,z),s(u,z)\ra_u=\sum_i\la a(u,z),p_i(u)\ra_u q_i(z), \label{C_auz_suz}
\end{align}
where $s(u,z)=\sum_i p_i(u)q_i(z)$, belongs to $\lfK_0$ (as a function of $z$).

\bb{Semidirect product.} Now we are able to define a \emph{semidirect product of
two 'shifted' distributions} $a(u,z)$ and $b(v,z)$
as a linear continuous functional $a(u,z)b(v,z)$ acting on $s\in\lfK_0\otimes\lfK_0\otimes\lfK_0$
by the rule
\begin{align*}
 \la a(u,z)b(v,z),s(u,v,z)\ra=\La a(u,z),\la b(v,z),s(u,v,z)\ra_v\Ra_{u,z}.
\end{align*}

The 'shifted' distribution $a(u,z)$ acting on $\lfK_0\otimes\lfK_0$ can be defined by one
of its partial actions on the function of one variable. For instance, if the partial action
of the type $\la a(u,v),s(u)\ra_u$ is defined for any test function $s(u)$ then one can
calculate the left hand side of~\eqref{C_auz_suz} and, then, obtain the total action of $a(u,z)$
on the test function $s(u,z)$. This means that a 'shifted' distribution (more generally, a
distribution satisfying condition (i)) define a continuous operator on $\lfK_0$.

The main example of a 'shifted' distribution is a delta-function $\delta(u,z)$ defined
by one of the formulae
\begin{align*}
 &\la\delta(u,z),s(u)\ra_u=s(z), & &\la\delta(u,z),s(z)\ra_z=s(u), &
 &\la\delta(u,z),s(u,z)\ra_{u,z}=\la s(z,z)\ra_z.
\end{align*}
It is symmetric: $\delta(u,z)=\delta(z,u)$ and one can show that any 'shifted'
distribution $a(u,v)$ satisfies
\begin{align}
 a(u,v)\delta(u,z)=a(z,v)\delta(u,z). \label{C_a_delta}
\end{align}
The distribution $\delta(u,z)$ defines an identical operator on $\lfK_0$.

\bb{Test function algebra $K$.} We define the algebra $K$ as an algebra of entire functions
on $\mathbb C$ subjected to the periodicity condition $s(u+1)=s(u)$ and
to the condition $|s(u)|\le C e^{p|\Im u|}$, where the constants $C,p>0$ depend on the function $s(u)$.
The periodicity of these functions means that they can be considered as functions on the
cylinder $\Cyl=\mathbb C/{\mathbb Z}$: $K=K(\Cyl)$. We equip the algebra $K$ with an
invariant scalar product
\begin{align}
 \la s(u),t(u)\ra=\int\limits_{-\frac12+\alpha}^{\frac12+\alpha}\frac{du}{2\pi i}s(u)t(u),
   \qquad s,t\in K, \label{C_spJ}
\end{align}
which does not depend on a choice of the complex number $\alpha\in\mathbb C$. A
convergence in $K$ is given as follows: a sequence $\{s_n\}\subset K$ tends to zero
if there exist such constants $C,p>0$ that $|s_n(u)|\le C e^{p|\Im u|}$ and for all
$u\in\mathbb C$ the sequence $s_n(u)\to0$. In particular, if $s_n\to0$ then the
functions $s_n(u)$ tends uniformly to zero. Therefore, the scalar product~\eqref{C_spJ}
is continuous with respect to this topology and it defines a continuous embedding
of $K=K(\Cyl)$ to the space of distributions $K'=K'(\Cyl)$.

Each function $s\in K$ can be restricted to the line segment
$[-\frac12+\alpha;\frac12+\alpha]$, be expanded in this line segment to a Fourier
series and, then, this expansion can be uniquely extend to all the
$\mathbb C$ by the analyticity principle. It means that $\{j_n(u)=e^{2\pi inu}\}_{n\in\mathbb Z}$ is
a basis of $K$ and $\{j^n(u)=2\pi i e^{-2\pi inu}\}_{n\in\mathbb Z}$ is its dual
one with respect to the scalar product~\eqref{C_spJ}.

The functions belonging to the space $K$ can be correctly shifted because for
all $z\in\mathbb C$ the operator ${\cal T}_z\colon s(u)\mapsto s(u+z)$ is continuous
and maps a periodic function to a periodic one. Hence we can define sifted distributions
in the usual way: $\la a(u-z),s(u)\ra\hm=\la a(u),s(u+z)\ra$. Thereby defined shifted distributions
$a(u-z)$ can be considered as two-variable distributions with properties (i) and (ii) as well as
distributions depending of one of variables as of an argument and of another as of a parameter.
For example the distribution $\delta(u)\in K'$ defined by the formula $\la\delta(u),s(u)\ra=s(0)$
can be shifted by variable $z$ and consider as a distribution of variables $u$ and $z$. This
shifted distribution is called delta-function. Their Fourier expansion looks as follows
\begin{align}
 \delta(u-z)=\sum_{n\in\mathbb Z}j^n(u)j_n(z)=2\pi i\sum_{n\in\mathbb Z} e^{-2\pi in(u-z)}. \label{C_dffeJ}
\end{align}

\renewcommand{\thesection}{\arabic{section}}

\chapter{SOS model partition function and the elliptic weight functions}
\label{SA3}

\thispagestyle{empty} \pagestyle{myheadings}
\markboth{}{S.\,Pakuliak, V.\,Rubtsov, A.\,Silantyev\hfil{\itshape SOS model partition function}}

\newpage \thispagestyle{empty}

\renewcommand{\thefigure}{\arabic{figure}}

\vspace*{3cm}
\begin{center}
{\LARGE SOS model partition function and the elliptic weight functions} \\[17mm]
{\large Stanislav PAKULIAK~$^{\dag\sharp}$, Vladimir RUBTSOV~$^{\ddag\sharp}$ and \\[4mm]
Alexey SILANTYEV~$^{\dag\ddag}$} \\[8mm]
$^\dag$\ Laboratory of Theoretical Physics, JINR, \\
141980 Dubna, Moscow reg., Russia \\[4mm]
$^\ddag$\ D\'epartment de Math\'ematiques, Universit\'e d'Angers, \\
2 Bd. Lavoisier, 49045 Angers, France \\[4mm]
$^\sharp$\ Institute of Theoretical and Experimental Physics, \\
Moscow 117259, Russia \\[6mm]
 pakuliak@theor.jinr.ru, \\  Volodya.Roubtsov@univ-angers.fr, \\  silant@tonton.univ-angers.fr
\end{center}

\begin{abstract}
We generalized a recent observation \cite{KhP} that the partition function of the 6-vertex
model with domain wall boundary conditions can be obtained from a calculation of projections of
the product of total currents in the quantum affine algebra $U_{q}(\widehat{\mathfrak{sl}}_{2})$
in its current realization. A generalization is done for the elliptic current algebra \cite{EF,ER1}.
The projections of the product of total currents in this case are calculated explicitly and are
presented as integral transforms of a product of the total currents. It is proved that the integral kernel of
this transform is proportional to the partition function of the SOS model with domain wall boundary
conditions.
\end{abstract}

\section{Introduction}

The main aim of this paper is to apply the elliptic current projection method to calculate
the universal elliptic weight functions. The projections of currents first appeared in the works of
B.Enriquez and second author~\cite{ER2}, \cite{ER3}.
This was a method to construct a higher genus analog of quantum groups in terms of Drinfeld
currents~\cite{D88}. The current (or ``new") realization supplies the quantum affine algebra with another
co-product (the ``Drinfeld" co-product). The standard and Drinfeld co-products are related by a ``twist"
(see ~\cite{ER2}). The quantum algebra is decomposed (in two different ways) in a product of two ``Borel
 subalgebras''.
We can consider ( for each subalgebra) its intersection with two another Borel subalgebras and
express it as their
product. Thus we obtain to each subalgebra a pair of projection operators from the subalgebra to
 each of these intersections.
The above-mentioned twist is defined by a Hopf pairing of the subalgebras and the projection
operators (see Sect.4 where
we remind an elliptic version of this construction).

Further, S.Khoroshkin and first author have applied this method
 for a factorization of the universal $R$-matrix~\cite{DKhP} in the
quantum affine algebras and to obtain
a universal weight functions~\cite{KhP,EKhP} for arbitrary quantum affine algebra.
The weight functions play a fundamental role in the theory of deformed Knizhnik-Zamolodchikov and
Knizhnik-Zamolodchikov-Bernard equations.
In particular, in the case of $U_q(\widehat{\glt})$, acting by the projections of Drinfeld
currents on the highest
weight vectors of irreducible finite-dimensional representations, one obtains exactly the
(trigonometric) weight functions
or off-shell Bethe vectors. In the canonical nested Bethe ansatz these objects are defined
implicitly by the recursive relations. Calculations of the projections are an effective
way to resolve the hierarchical relations of the nested Bethe ansatz.

It was observed in \cite{KhP} that the projections for the
algebra $U_q(\widehat{\slt})$ can be presented as an integral transform and the integral kernel of this
transform is proportional to the partition function of the finite 6-vertex model with
domain wall boundary conditions (DWBC)~\cite{KhP}. Here we prove that the elliptic projections
described in~\cite{EF} can help to derive the partition function for the elliptic models.
It was shown in this paper that the
calculation of the projections in the current elliptic algebra \cite{EF,ER1}
yields the partition function of   the Solid-On-Solid (SOS) model with domain wall boundary
conditions.

In~\cite{Kor} Korepin derived recurrent relations for the partition function for the finite 6-vertex model with domain wall boundary conditions. Further Izergin used them to find the expression for the partition function in a determinant form~\cite{I87}. The integral kernel of projections satisfies the same recursive
relations and gives another formula for the partition function.

The problem of generalization of the Izergin's determinant formula to the elliptic case
was extensively discussed in the last two decades.
One can prove that the statistical sum of the SOS model with DWBC  cannot be presented in
the form of the single determinant. When this paper was prepared H. Rosengren \cite{Ros08} has shown
that this statistical sum for $n\times n$ lattice can be written as a sum of $2^n$ determinant
generalizing the Izergin's determinant formula. His approach relates to some (dynamical) generalization
of Alternating-Sign Matrices and goes along with the famous Kuperberg combinatorial demonstration
\cite{Kup}.

We expect that the projection method gives a universal form of the
elliptic weight function \cite{TV} as it does in the case of the quantum affine
algebras \cite{KhPT}. When this universal weight function is presented as
an integral transform of the product of the elliptic currents we show that the integral kernel
of this transform gives an expression of the partition function for the SOS model.
In one hand we generalize Korepin's recurrent relations and in other hand we generalize the method
proposed in \cite{KhP} for calculating the projections to the elliptic case.
We check that the integral kernel extracted from the universal weight function
and multiplied by certain  factor satisfies
the obtained recurrent relations, which uniquely define the partition function for SOS model with DWBC.
Our formula given by the projection method coincides with the Rosengren's one.

An interesting open problem which still deserves more extensive studies is a relation of the projection method
with Elliptic Sklyanin-Odesskii-Feigin algebras. It was observed in the pioneering paper \cite{ER1} that
the half of elliptic current generators satisfies the commutation relations of $W$-elliptic algebras of
Feigin. Another intriguing relation was observed in \cite{FO}: a certain subalgebra in the
``$\lambda$-generalization"
of Sklyanin algebra (a graded algebra of meromorphic functions with ``$\lambda$-twisted"
(anti)symmetrization product) satisfies the Felder $R$-matrix
quadratic relations \cite{FVT}. The latter paper gives a description of the elliptic
Bethe eigenvectors or the elliptic weight functions.
This is a strong indication that the projection method should be considered and interpreted in the
framework of (generalized) Sklyanin-Odesskii-Feigin algebras. We hope to  discuss it elsewhere.

The main results of the paper were reported in 7-th International Workshop ``Supersymmetry and Quantum
Symmetry" in JINR, Dubna (Russia), July 30 - August 4, 2007.

The paper is organized as follows. In section~\ref{D_secd2} we briefly review the finite 6-vertex model
with DWBC and present the formulae for the partition function: the Izergin's determinant formula
and the formula obtained by the projection method. Section~\ref{D_secd3} is devoted to the SOS model
with DWBC. We introduce the model without great details and pose a problem how to calculate
the partition function of this model. We obtain analytical properties of the partition function
and prove that they allow to reconstruct the partition function exactly.
In section~\ref{D_secd4} we introduce the projections in terms of the currents
for the elliptic algebra following~\cite{EF}. We generalize the method proposed in~\cite{KhP}
to our case to obtain the integral representation of the projections of product of currents.
Then, using a Hopf pairing we extract the integral kernel and check
that it  satisfies all necessary analytical properties
of the partition function of the SOS model with DWBC.
In section~\ref{D_secd5} we investigate trigonometric degeneration of the elliptic model
and the partition function with DWBC. We arrive to the 6-vertex model case by two steps.
The model obtained after the first step is a trigonometric SOS model. Then we show that
the degeneration of the
expression derived in the section~\ref{D_secd4} coincides with known expression for the
6-vertex model partition function with DWBC.
An appendix contains the necessary information on the properties of the elliptic polynomials.

\section{Partition function of the finite 6-vertex model}
\label{D_secd2}

Let us consider a statistical system on a  $n\times n$ square lattice, where the
columns and rows are enumerated from $1$ to $n$ from right to left and upward respectively.
This is a 6-vertex model where vertices on the lattice are associated with Boltzmann weights
which depend on the configuration of the arrows around a given vertex.
Six possible configurations are shown in the Fig.~\ref{D_figd1} and weights are functions of
two spectral parameters $z$, $w$ and anisotropy parameters $q$:
\begin{align} \label{D_Bw}
a(z,w)&=qz-q^{-1}w,  &  b(z,w)&=z-w, \\
c(z,w)&=(q-q^{-1})z, &  \bar c(z,w)&=(q-q^{-1})w.
\end{align}

\begin{figure}[h]
\begin{center}
\begin{picture}(180,130)
\put(05,63){$a(z,w)$}
\put(05,-15){$a(z,w)$}
\put(23,118){$+$}
\put(0,104){$+$}
\put(33,104){$+$}
\put(23,78){$+$}
\put(40,100){\vector(-1,0){40}}
\put(40,100){\vector(-1,0){3}}
\put(20,120){\vector(0,-1){40}}
\put(20,120){\vector(0,-1){3}}

\put(85,63){$b(z,w)$}
\put(85,-15){$b(z,w)$}
\put(103,118){$+$}
\put(113,104){$-$}
\put(80,104){$-$}
\put(103,78){$+$}
\put(80,100){\vector(1,0){40}}
\put(80,100){\vector(1,0){3}}
\put(100,120){\vector(0,-1){40}}
\put(100,120){\vector(0,-1){3}}

\put(165,63){$c(z,w)$}
\put(165,-15){$\bar c(z,w)$}
\put(183,118){$-$}
\put(193,104){$+$}
\put(160,104){$-$}
\put(183,78){$+$}
\put(160,100){\line(1,0){40}}
\put(200,100){\vector(-1,0){3}}
\put(160,100){\vector(1,0){3}}
\put(180,80){\vector(0,1){40}}
\put(180,80){\vector(0,-1){3}}

\put(23,38){$-$}
\put(0,24){$-$}
\put(33,24){$-$}
\put(23,00){$-$}
\put(0,20){\vector(1,0){40}}
\put(0,20){\vector(1,0){3}}
\put(20,00){\vector(0,1){40}}
\put(20,00){\vector(0,1){3}}

\put(103,38){$-$}
\put(113,24){$+$}
\put(80,24){$+$}
\put(103,00){$-$}
\put(120,20){\vector(-1,0){40}}
\put(120,20){\vector(-1,0){3}}
\put(100,00){\vector(0,1){40}}
\put(100,00){\vector(0,1){3}}

\put(183,38){$+$}
\put(193,24){$-$}
\put(160,24){$+$}
\put(183,00){$-$}
\put(200,20){\line(-1,0){40}}
\put(200,20){\vector(1,0){3}}
\put(160,20){\vector(-1,0){3}}
\put(180,00){\line(0,1){40}}
\put(180,00){\vector(0,1){3}}
\put(180,40){\vector(0,-1){3}}
\end{picture}
\medskip
\end{center}
\caption{\footnotesize  Graphical presentation of the Boltzmann weights.}
\label{D_figd1}
\end{figure}
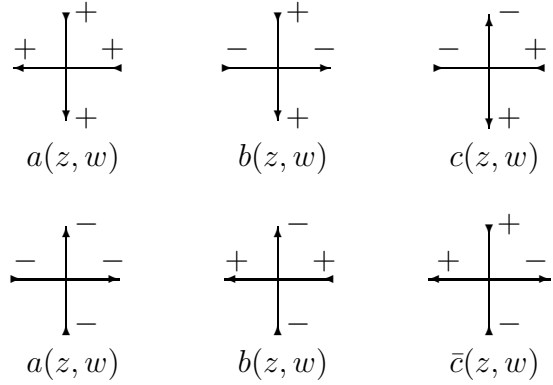

Let us associate the sign `$+$' to the upward  and left directed arrows,
while the  sign `$-$' is associated to the downward and right directed arrows as
shown in the Fig.~\ref{D_figd1}. The Boltzmann weights~\eqref{D_Bw} are gathered in the matrix
\begin{equation}\label{D_Bw1}
R(z,w)=\left(\begin{array}{cccc}a(z,w)&0&0&0\\ 0&b(z,w)& \bar c(z,w)&0\\
0&c(z,w)&b(z,w)&0\\0&0&0&a(z,w)
\end{array}\right)
\end{equation}
acting in the space $\mathbb C^2\otimes\mathbb C^2$ with the basis $e_{\alpha}\otimes e_{\beta}$,
$\alpha,\beta=\pm$. The entry $R(z,w)^{\alpha\beta}_{\gamma\delta}$,
$\alpha,\beta,\gamma,\delta=\pm$ coincides with the Boltzmann weight corresponding to
the Fig.~\ref{D_figd2}:

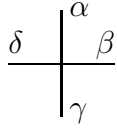
\begin{figure}[h]
\begin{center}
\begin{picture}(40,40)
\put(23,28){$\alpha$}
\put(0,14){$\delta$}
\put(33,14){$\beta$}
\put(23,-10){$\gamma$}
\put(40,10){\line(-1,0){40}}
\put(40,10){\line(-1,0){3}}
\put(20,-10){\line(0,1){40}}
\put(20,-10){\line(0,1){3}}
\end{picture}
\end{center}
\caption{\footnotesize  The Boltzmann weight $R(z,w)^{\alpha\beta}_{\gamma\delta}$.}
\label{D_figd2}
\end{figure}

Different repartitions of the arrows on the edges form different configurations $\{C\}$.
A Boltzmann weight of the lattice is a product of the Boltzmann weights in each vertex.
We define the partition function of the model summing the Boltzmann weights of the lattice over all possible configurations subjected
to some boundary conditions:
\begin{equation}\label{D_sts1}
Z(\{z\},\{w\})=\sum_{\{C\}}\prod_{i,j=1}^n R(z_i,w_j)^{\alpha_{ij}\beta_{ij}}_{\gamma_{ij}\delta_{ij}}.
\end{equation}
Here $\alpha_{ij}$, $\beta_{ij}$, $\gamma_{ij}$, $\delta_{ij}$ are corresponding signs
around $(i,j)$-th vertex. We consider an inhomogeneous model when the Boltzmann
weights depend on the column via the variable $z_i$ and on the row via the variable $w_j$
(see Fig.~\ref{D_figd3}).

\begin{figure}[b]
\begin{center}
\begin{picture}(180,100)
\put(50,70){\line(1,0){80}}
\put(50,70){\vector(1,0){3}}
\put(130,70){\vector(-1,0){3}}
\put(15,70){$w_4$} \put(35,70){$+$} \put(135,70){$-$}
\put(50,50){\line(1,0){80}}
\put(50,50){\vector(1,0){3}}
\put(130,50){\vector(-1,0){3}}
\put(15,50){$w_3$} \put(35,50){$+$} \put(135,50){$-$}
\put(50,30){\line(1,0){80}}
\put(50,30){\vector(1,0){3}}
\put(130,30){\vector(-1,0){3}}
\put(15,30){$w_2$} \put(35,30){$+$} \put(135,30){$-$}
\put(50,10){\line(1,0){80}}
\put(50,10){\vector(1,0){3}}
\put(130,10){\vector(-1,0){3}}
\put(15,10){$w_1$} \put(35,10){$+$} \put(135,10){$-$}

\put(60,0){\vector(0,1){80}}
\put(60,0){\vector(0,-1){3}}
\put(58,-10){$-$} \put(58,-20){$z_4$} \put(58,83){$+$}
\put(80,0){\vector(0,1){80}}
\put(80,0){\vector(0,-1){3}}
\put(78,-10){$-$} \put(78,-20){$z_3$} \put(78,83){$+$}
\put(100,0){\vector(0,1){80}}
\put(100,0){\vector(0,-1){3}}
\put(98,-10){$-$} \put(98,-20){$z_2$} \put(98,83){$+$}
\put(120,0){\vector(0,1){80}}
\put(120,0){\vector(0,-1){3}}
\put(118,-10){$-$} \put(118,-20){$z_1$} \put(118,83){$+$}
\end{picture}
\medskip
\end{center}
\caption{\footnotesize Inhomogeneous lattice with domain wall boundary conditions.}
\label{D_figd3}
\end{figure}
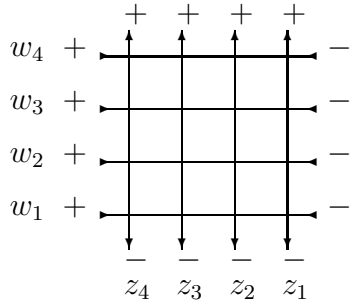

We choose so-called domain wall boundary conditions (DWBC) that fix the boundary arrows (signs)
like it shown in the Fig.~\ref{D_figd3}. In other words, the arrows are entering on the left and
right boundaries and leaving on the lower and upper ones.

In \cite{I87}, A.\,G.\,Izergin has found a determinant presentation for the partition
function of the lattice with DWBC
\begin{multline}
 Z(\{z\},\{w\})= (q-q^{-1})^n \prod_{m=1}^n w_m \times  \\
\qquad \times\frac{\displaystyle \prod_{i,j=1}^n(z_i-w_j)(qz_i-q^{-1}w_j)}
{\displaystyle \prod_{n\ge i>j\ge1}(z_i-z_j)(w_j-w_i)}
\det\left|\left|\frac{1}{(z_i-w_j)(qz_i-q^{-1}w_j)}\right|\right|_{i,j=1,\ldots,n}\, .\label{D_stat-s}
\end{multline}
The Izergin's idea was to use a symmetry of the polynomial~\eqref{D_sts1} and Korepin's recurrent relations for the quantity $Z(\{z\},\{w\})$  to observe that these
recurrent relations allow to reconstruct $Z(\{z\},\{w\})$ in unique way and that the
same recurrent relations are valid for the determinant formula~\eqref{D_stat-s}.

On the other hand it was observed that the integral kernel of the projection of $n$ currents is a
polynomial of the same degree and it satisfies the same Korepin's recurrent relations~\cite{KhP}.
It means that this integral kernel coincides with the partition function for $n\times n$ lattice. Moreover,
the theory of projections gives another expression for the partition function:
\begin{multline}
 Z(\{z\},\{w\}) = (q-q^{-1})^n \prod_{m=1}^n w_m \prod_{n\ge i>j\ge1}\frac{q^{-1}w_i-qw_j}
 {w_i-w_j}\times \\
\qquad \times \sum_{\sigma\in S_n}\prod_{1\le i<j\le n \atop \sigma(i)>\sigma(j)}
  \frac{qw_{\sigma(i)}-q^{-1}w_{\sigma(j)}}{q^{-1}w_{\sigma(i)}-qw_{\sigma(j)}}
  \prod_{n\ge i>k\ge1}(qz_i-q^{-1}w_{\sigma(k})\prod_{1\le i<k\le n} (z_i-w_{\sigma(k)}),
  \label{D_stat-s_pr}
\end{multline}
where $S_n$ is a permutation group. Here the factor
$\frac{qw_{\sigma(i)}-q^{-1}w_{\sigma(j)}}{q^{-1}w_{\sigma(i)}-qw_{\sigma(j)}}$
appears in the product if the both conditions $i<j$ and $\sigma(i)>\sigma(j)$ are
satisfied simultaneously.

\section{Partition function for the SOS model}
\label{D_secd3}

\subsection{Description of the SOS model}
\label{D_secd31}

The SOS model is a face model. We introduce it as usually in terms of heights, but then represent
it in terms of $R$-matrix formalism as in~\cite{FS}. This language is more convenient to
generalize the results reviewed in
the section~\ref{D_secd2} and to prove the symmetry of the partition function.

\begin{figure}
\begin{center}
\begin{picture}(150,90)
\put(50,70){\line(1,0){80}}
\put(35,76){$n$}
\put(50,50){\line(1,0){80}}
\put(35,60){$\ldots$}
\put(50,30){\line(1,0){80}}
\put(5,38){$j=$}
\put(35,38){$2$}
\put(50,10){\line(1,0){80}}
\put(35,19){$1$}
\put(35,00){$0$}
\put(60,0){\line(0,1){80}}
\put(48,-15){$n$}
\put(80,0){\line(0,1){80}}
\put(68,-15){$\ldots$}
\put(100,0){\line(0,1){80}}
\put(88,-15){$2$}
\put(120,0){\line(0,1){80}}
\put(108,-15){$1$}
\put(123,-15){$0$}
\put(20,-15){$i=$}
\end{picture}
\end{center}
\caption{\footnotesize The numeration of faces.}
\label{D_figd5}
\end{figure}
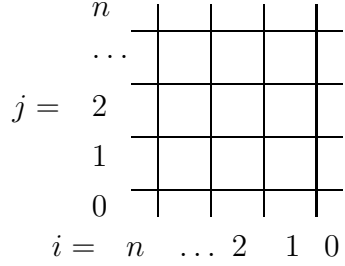

Consider a square $n\times n$ lattice with the vertices enumerated by index $i=1,\ldots,n$ like
in the previous case. It has $(n+1)\times(n+1)$ faces enumerated by pairs $(i,j)$, $i,j=0,\ldots,n$
(see Fig.~\ref{D_figd5}). To each face we put a complex number called height in such a way that the
differences of the heights corresponding to the neighbor faces are $\pm1$. Let us denote
by $d_{ij}$ the height corresponding to the face $(i,j)$ placed to the up-left from the
vertex $(i,j)$. Then the last condition can be written in the from $|d_{ij}-d_{i-1,j}|=1$,
for $i=1,\ldots,n$, $j=0,\ldots,n$, and $|d_{ij}-d_{ij-1}|=1$ for $i=0,\ldots,n$, $j=1,\ldots,n$.
Each distribution of heights $d_{ij}$ ($i,j=0,\ldots,n$) subjected to these conditions and also to
some boundary conditions defines a configuration of the model. It means that the partition
function of this model can be presented in the form
\begin{equation}
 Z=\sum_{C}\prod_{i,j=1}^n W_{ij}(d_{i,j-1},d_{i-1,j-1},d_{i-1,j},d_{ij}), \label{D_pf_def}
\end{equation}
where $W_{ij}(d_{i,j-1},d_{i-1,j-1},d_{i-1,j},d_{ij})$ are Boltzmann weights of $(i,j)$-th vertex
depending on the configuration via connected heights as follows~\cite{J}
\begin{align}\label{D_BWa}
 &W_{ij}(d+1,d+2,d+1,d)=a(u_i-v_j)=\theta(u_i-v_j+\hbar),  \\
 &W_{ij}(d-1,d-2,d-1,d)=a(u_i-v_j)=\theta(u_i-v_j+\hbar), \\
 &W_{ij}(d-1,d,d+1,d)=b(u_i-v_j;\hbar d)=\frac{\theta(u_i-v_j)\theta(\hbar d+\hbar)}{\theta(\hbar d)}, \\
 &W_{ij}(d+1,d,d-1,d)=\bar b(u_i-v_j;\hbar d)=\frac{\theta(u_i-v_j)\theta(\hbar d-\hbar)}{\theta(\hbar d)}, \\
 &W_{ij}(d-1,d,d-1,d)=c(u_i-v_j;\hbar d)=\frac{\theta(u_i-v_j+\hbar d)\theta(\hbar)}{\theta(\hbar d)}, \\
 &W_{ij}(d+1,d,d+1,d)=\bar c(u_i-v_j;\hbar d)=\frac{\theta(u_i-v_j-\hbar d)\theta(\hbar)}{\theta(-\hbar d)}.
\end{align}
As in the 6-vertex case the variables $u_i$, $v_j$ are attached to the $i$-th vertical and $j$-th horizontal
lines respectively, $\hbar$ is a nonzero anisotropy parameter~\footnote{In the elliptic case we
use additive variables $u_i$, $v_j$ and an additive anisotropy parameter $\hbar$ instead of the
multiplicative variables $z_i=e^{2\pi iu_i}$, $w_i=e^{2\pi iv_i}$ and the multiplicative parameter
$q=e^{\pi i\hbar}$.}. The weights are expressed via the ordinary odd theta-function defined by the conditions
\begin{align}
 \theta(u+1)&=-\theta(u), &  \theta(u+\tau)&=-e^{-2\pi iu-\pi i\tau}\theta(u), &   \theta'(0)&=1.
 \label{theta}
\end{align}

\begin{figure}
\begin{center}
 \begin{picture}(240,170)
\put(00,30){\line(1,0){60}}\put(30,00){\line(0,1){60}}
\put(-10,30){\bfseries $-$}\put(60,30){$-$}\put(30,-10){\bfseries $-$}\put(30,60){\bfseries $-$}
\put(0,10){\small $d-1$}\put(10,40){\small $d$}\put(35,40){\small $d-1$}\put(35,10){\small $d-2$}
\put(10,-20){\small  $a(u_i-v_j)$}
\put(00,130){\line(1,0){60}}\put(30,100){\line(0,1){60}}
\put(-10,130){\bfseries $+$}\put(60,130){\bfseries $+$}\put(30,90){\bfseries $+$}\put(30,160){\bfseries $+$}
\put(0,110){\small $d+1$}\put(10,140){\small $d$}\put(35,140){\small $d+1$}\put(35,110){\small $d+2$}
\put(10,77){\small  $a(u_i-v_j)$}
\put(100,30){\line(1,0){60}}\put(130,00){\line(0,1){60}}
\put(90,30){\bfseries $+$}\put(160,30){$+$}\put(130,-10){\bfseries $-$}\put(130,60){\bfseries $-$}
\put(100,10){\small $d+1$}\put(110,40){\small $d$}\put(135,40){\small $d-1$}\put(145,10){\small $d$}
\put(105,-20){\small  $\bar b(u_i-v_j;d)$}
\put(100,130){\line(1,0){60}}\put(130,100){\line(0,1){60}}
\put(90,130){\bfseries $-$}\put(160,130){\bfseries $-$}\put(130,90){\bfseries $+$}\put(130,160){\bfseries $+$}
\put(100,110){\small $d-1$}\put(110,140){\small $d$}\put(135,140){\small $d+1$}\put(145,110){\small $d$}
\put(105,77){\small  $b(u_i-v_j;d)$}
\put(200,30){\line(1,0){60}}\put(230,00){\line(0,1){60}}
\put(190,30){\bfseries $+$}\put(260,30){$-$}\put(230,-10){\bfseries $-$}\put(230,60){\bfseries $+$}
\put(200,10){\small $d+1$}\put(210,40){\small $d$}\put(235,40){\small $d+1$}\put(245,10){\small $d$}
\put(205,-20){\small  $\bar c(u_i-v_j;d)$}
\put(200,130){\line(1,0){60}}\put(230,100){\line(0,1){60}}
\put(190,130){\bfseries $-$}\put(260,130){\bfseries $+$}\put(230,90){\bfseries $+$}\put(230,160){\bfseries $-$}
\put(200,110){\small $d-1$}\put(210,140){\small $d$}\put(235,140){\small $d-1$}\put(245,110){\small $d$}
\put(205,77){\small  $c(u_i-v_j;d)$}
  \end{picture}
  \bigskip
 \end{center}
\caption{\footnotesize The Boltzmann weights for the SOS model.}
\label{D_figd4}
\end{figure}

Let us introduce the notations
\begin{equation}
 \alpha_{ij}=d_{i-1,j}-d_{ij},\  \beta_{ij}=d_{i-1,j-1}-d_{i-1,j},\
  \gamma_{ij}=d_{i-1,j-1}-d_{i,j-1},\  \delta_{ij}=d_{i,j-1}-d_{ij}. \label{D_albegade}
\end{equation}
The differences~\eqref{D_albegade} take the values $\pm1$ and we attach them to the corresponding edges
as in Fig.~\ref{D_figd2}: $\gamma_{i,j+1}=\alpha_{ij}$ is the sign attached to the vertical edge
connecting $(i,j)$-th vertex with the $(i,j+1)$-st one, $\beta_{i,j+1}=\delta_{ij}$ is the sign
attached to the horizontal edge connecting $(i,j)$-th vertex with the $(i+1,j)$-th one.
The configuration can be considered as a distribution of these signs on the internal edges
which are subjected to the conditions $\alpha_{ij}+\beta_{ij}=\gamma_{ij}+\delta_{ij}$, $i,j=1,\ldots,n$.
In terms of signs on external edges the DWBC are the same as shown in Fig.~\ref{D_figd3}.
Additionally we have to fix one of boundary heights, for example, $d_{nn}$.

The Boltzmann weights~\eqref{D_BWa} can be presented as entries of dynamical elliptic
$R$-matrix~\cite{FS}:
\begin{align}
 W_{ij}(d_{i,j-1},d_{i-1,j-1},d_{i-1,j},d_{ij})
   &=R(u_i-v_j;\hbar d_{ij})^{\alpha_{ij}\beta_{ij}}_{\gamma_{ij}\delta_{ij}}, \notag \\[7pt]
 R(u;\lambda)&=\left(\begin{array}{cccc}
   a(u) & 0              & 0              & 0              \\
  0              & b(u;\lambda) & \bar c(u;\lambda)      & 0              \\
  0              & c(u;\lambda) & \bar b(u;\lambda) & 0              \\
  0              & 0              & 0              & a(u)
 \end{array}\right). \label{D_Rz}
\end{align}

Let $\T^{\alpha_{in}}_{\gamma_{i1}}(u_i,\{v\},\lambda_i)$ be {\itshape column transfer matrices}.
It is a  matrix-valued
functions of $u_i$, all spectral parameters
$v_j$, $j=1,\ldots,n$ and the parameters $\lambda_i$ related to the heights:
\begin{align} \label{D_colTM} 
 &\T^{\alpha_{in}}_{\gamma_{i1}}(u_i,\{v\},\lambda_i)^{\beta_{in}\ldots\beta_{i1}}_{\delta_{in}\ldots
 \delta_{i1}} =\\
 = \Big(&R^{(n+1,n)}(u_i-v_n;\lambda_{in})R^{(n+1,n-1)}(u_i-v_{n-1};\lambda_{i,n-1})\cdots
    R^{(n+1,1)}(u_i-v_1;\lambda_{i1})\Big)^{\alpha_{in};\beta_{in}\ldots\beta_{i1}}_{\gamma_{i1};
    \delta_{in}\ldots\delta_{i1}} \notag \\
 =  \Big(&R^{(n+1,n)}(u_i-v_n;\Lambda_{in})R^{(n+1,n-1)}(u_i-v_{n-1};\Lambda_{i,n-1})\cdots
    R^{(n+1,1)}(u_i-v_1;\Lambda_{i1})\Big)^{\alpha_{in};\beta_{in}\ldots\beta_{i1}}_{\gamma_{i1};
    \delta_{in}\ldots\delta_{i1}}, \notag
\end{align}
where $\lambda_{ij}=\hbar d_{ij}=\lambda_i+\hbar\sum\limits_{l=j+1}^n\delta_{il}$,
$\lambda_i=\hbar d_{in}=\lambda+\hbar\sum\limits_{l=i+1}^n\alpha_{ln}$,
$\Lambda_{ij}=\lambda_i+\hbar\sum\limits_{l=j+1}^n H^{(l)}$.
The matrix $H^{(l)}$ acts in the $l$-th two-dimensional space  $V_{l}\cong\mathbb C^2$
as   ${\rm diag}(1,-1)$  and  $R$-matrix $R^{(a,b)}$ acts nontrivially in the
         tensor product of $V_a\otimes V_b$.
The superscript $n+1$ in the $R$-matrices  is regarded to an auxiliary space
         $V_{n+1}\cong\mathbb C^2$. The partition function~\eqref{D_pf_def} corresponding
          to DWBC ($\alpha_{in}=+1$,  $\beta_{1i}=-1$, $\gamma_{i1}=-1$, $\delta_{ni}=+1$,
$i=1,\ldots,n$) can be represented through the column transfer matrices:
\begin{equation}
 Z^{+-}_{-+}(\{u\},\{v\};\lambda)
=\Big(\T^+_-(u_1,\{v\},\lambda_1)\cdots\T^+_-(u_n,\{v\},\lambda_n)\Big)^{-\ldots-}_{+\ldots+}, \label{D_ZcT}
\end{equation}
where $\lambda_i=\lambda+\hbar(n-i)$. Similarly one can define the row transfer matrix.

\subsection{Analytical properties of the partition function}
\label{D_secd32}

Here we describe the analytical properties of the SOS model partition function analogical to those used by
A.G.Izergin to restore the partition function of  the 6-vertex model. These properties
uniquely define this partition function.

\begin{prop} \label{D_propZsym}
 The partition function with DWBC $Z^{+-}_{-+}(\{u\},\{v\};\lambda)$ is a symmetric function in
 both sets of the variables $u_i$ and $v_j$.
\end{prop}

\noindent
Proof is based on the Dynamical Yang-Baxter Equation (DYBE) for the $R$-matrix~\cite{FS}
\begin{multline}
 R^{(12)}(t_1-t_2;\lambda)R^{(13)}(t_1-t_3;\lambda+\hbar H^{(2)})R^{(23)}(t_2-t_3;\lambda)=\\
=R^{(23)}(t_2-t_3;\lambda+\hbar H^{(1)})R^{(13)}(t_1-t_3;\lambda)R^{(12)}(t_1-t_2;\lambda+\hbar H^{(3)}).\notag
\end{multline}
In order to prove the symmetry of partition function $Z^{+-}_{-+}(\{u\},\{v\};\lambda)$ under
permutation $v_j\leftrightarrow v_{j-1}$ we rewrite DYBE in the form
\begin{multline}
 R^{(n+1,j)}(u_i-v_j;\Lambda_{ij})R^{(n+1,j-1)}(u_i-v_{j-1};\Lambda_{ij}+\hbar H^{(j)})\times \\
 \quad \times R^{(j,j-1)}(v_j-v_{j-1};\Lambda_{ij})= R^{(j,j-1)}(v_j-v_{j-1};\Lambda_{ij}+\hbar H^{(n+1)}) \times \\
 \quad \times R^{(n+1,j-1)}(u_i-v_{j-1};\Lambda_{ij})
 R^{(n+1,j)}(u_i-v_j;\Lambda_{ij}+\hbar H^{(j-1)}). \label{D_DYBEuse} 
\end{multline}
 Multiplying $i$-th the
column matrix~\eqref{D_colTM} by $R^{(j,j-1)}(v_j-v_{j-1};\Lambda_{ij})$ from right and moving it
to the left using~\eqref{D_DYBEuse}, the relation $[H_1+H_2,R(u,\lambda)]=0$ and the equality
$\Lambda_{ij}+\hbar\alpha_{in}=\Lambda_{i-1,j}$ we obtain
\begin{multline}
 \T^{\alpha_{in}}_{\gamma_{i1}}(u_i,\{v\},\lambda_i)R^{(j,j-1)}(v_j-v_{j-1};\Lambda_{ij})
 =R^{(j,j-1)}(v_j-v_{j-1};\Lambda_{i-1,j})\times \\
 \times \sigmap^{(j,j-1)}
 \T^{\alpha_{in}}_{\gamma_{i1}}(u_i,\{v_j\leftrightarrow v_{j-1}\},\lambda_i)
 \cdots \T^{\alpha_{nn}}_{\gamma_{n1}}(u_n,\{v_j\leftrightarrow v_{j-1}\},\lambda_n)
 \sigmap^{(j,j-1)}, \label{D_TR_RT}
\end{multline}
where $\sigmap\in\End(\mathbb C^2\otimes\mathbb C^2)$ is a permutation matrix:
$\sigmap(e_1\otimes e_2)=e_2\otimes e_1$ for all $e_1,e_2\in\mathbb C^2$ and
notation $\{v_j\leftrightarrow v_{j-1}\}$ means the set of the parameters $\{v\}$ with
$v_{j-1}$ and $v_j$ are interchanged. Multiplying the product
of the column matrix by $R^{(j,j-1)}(v_j-v_{j-1};\Lambda_{nj})$ from right and moving it to the
left using~\eqref{D_TR_RT} one yields
\begin{multline}
 \T^{\alpha_{1n}}_{\gamma_{11}}(u_1,\{v\},\lambda_1)\cdots\T^{\alpha_{nn}}_{\gamma_{n1}}(u_n,\{v\},
 \lambda_n)R^{(j,j-1)}(v_j-v_{j-1};\Lambda_{nj}) =R^{(j,j-1)}(v_j-v_{j-1};\Lambda_{0,j})\times \\
 \quad\times
  \sigmap^{(j,j-1)}
 \T^{\alpha_{1n}}_{\gamma_{11}}(u_1,\{v_j\leftrightarrow v_{j-1}\},\lambda_1)\cdots
 \T^{\alpha_{nn}}_{\gamma_{n1}}(u_n,\{v_j\leftrightarrow v_{j-1}\},\lambda_n)\sigmap^{(j,j-1)},  \label{D_RTTT}
\end{multline}
where $\Lambda_{0,j}=\lambda+\hbar\sum\limits_{i=1}^n\alpha_{in}+\hbar
\sum\limits_{l=j+1}^n H^{(l)}$, $\Lambda_{nj}=\lambda_n=\lambda$. Eventually,
comparing the matrix element $\big(\cdot\big)^{-,\ldots,-}_{+,\ldots,+}$ of both hand
sides of~\eqref{D_RTTT}, taking into account the formula~\eqref{D_ZcT} and the identities
\begin{equation*}
  R(u,\lambda)^{--}_{\gamma\delta}=a(u)\delta^-_\gamma\delta^-_\delta,
\quad R(u,\lambda)^{\alpha\beta}_{++}=a(u)\delta^\alpha_+\delta^\beta_+,
 \quad  \sigmap^{--}_{\gamma\delta}=\delta^-_\gamma\delta^-_\delta,
  \quad \sigmap^{\alpha\beta}_{++}=\delta^\alpha_+\delta^\beta_+,
\end{equation*}
(where $\delta^\alpha_\gamma$ is a Kronecker's symbol) and substituting $\alpha_{in}=+1$,
$\gamma_{i1}=-1$ one derives
\begin{equation}
 Z^{+-}_{-+}(\{u\},\{v\};\lambda)=Z^{+-}_{-+}(\{u\},\{v_j\leftrightarrow v_{j-1}\};\lambda). \label{D_Z_perm_w}
\end{equation}
Similarly, using the row transfer matrix one can obtain the following equality from DYBE:
\begin{equation}
  Z^{+-}_{-+}(\{u\},\{v\};\lambda)=Z^{+-}_{-+}(\{u_j\leftrightarrow u_{j-1}\},\{v\};\lambda). \label{D_Z_perm_z}
\end{equation}
The partition function with DWBC satisfies the relations~\eqref{D_Z_perm_w}, \eqref{D_Z_perm_z}
for each $j=1,\ldots,n$, what is sufficient to establish the symmetry under an arbitrary permutation. \qed

\begin{prop} \label{D_propZep}
 The partition function with DWBC~\eqref{D_ZcT} is an elliptic polynomial
 \footnote{The notion of elliptic polynomials and their properties are
given in Appendix~\ref{D_Appendix_ep}.} of degree $n$
 with the character $\chi$ in each variable $u_i$, where
\begin{equation}
 \chi(1)=(-1)^n,\qquad
 \chi(\tau)=(-1)^n\exp\Big(2\pi i\big(\lambda+\sum_{j=1}^n v_j\big)\Big). \label{D_char_chi}
\end{equation}
\end{prop}

\noindent
Due to the symmetry with respect to the variables $\{u\}$ it is
sufficient to prove the proposition for variable $u_n$. To present explicitly a
dependence of $Z^{+-}_{-+}(\{u\},\{v\};\lambda)$ on $u_n$ we consider all possibilities for
states of edges attached to the vertices located in the $n$-th column. First consider the $(n,n)$-th vertex.
Due to the boundary conditions $\alpha_{nn}=\delta_{nn}=+1$ and to the condition
$\alpha_{nn}+\beta_{nn}=\gamma_{nn}+\delta_{nn}$ we have two possibilities: either
$\beta_{nn}=\gamma_{nn}=-1$ or $\beta_{nn}=\gamma_{nn}=+1$. In the first case one has a
unique possibility for the whole residual part of $n$-th column: $\gamma_{nj}=-1$,
$\beta_{nj}=+1$, $j=1,\ldots,n-1$; in the second case there are two possibilities for the $(n,n-1)$-st
vertex: either $\beta_{n,n-1}=\gamma_{n,n-1}=-1$ or $\beta_{n,n-1}=\gamma_{n,n-1}=+1$, etc.
Finally the partition function is represented in the from
\begin{eqnarray}
 Z^{+-}_{-+}(\{u\},\{v\};\lambda)=
  \sum_{k=1}^n \prod_{j=k+1}^n  a(u_n-v_j)\;\bar c(u_n-v_k;\lambda+(n-k)\hbar) \notag \\
\quad  \times\prod_{j=1}^{k-1} \bar b(u_n-v_j;\lambda+(n-j)\hbar) g_k(u_{n-1},\ldots,u_1,\{v\};\lambda), \notag
\end{eqnarray}
where $g_k(u_{n-1},\ldots,u_1,\{v\};\lambda)$ are functions not depending on $u_n$.
Each term of this sum is an elliptic polynomial of degree $n$ with the same
character~\eqref{D_char_chi} in the variable $u_n$. \qed
\medskip

\begin{Rem}\label{D_rem1} Similarly one can prove that the function
$Z^{+-}_{-+}(\{u\},\{v\};\lambda)$ is an elliptic polynomial of degree $n$ with the
character $\tilde\chi$ in each variable $v_i$, where $\tilde\chi(1)=(-1)^n$, \\
$\tilde\chi(\tau)=(-1)^n e^{2\pi i(-\lambda+\sum_{i=1}^n u_i)}$.
\end{Rem}

\begin{prop} \label{D_propZrec}
The $n$-th partition function with DWBC~\eqref{D_ZcT} restricted to the condition
$u_n=v_n-\hbar$ can be expressed through the $(n-1)$-st partition function:
\begin{align}
 &Z^{+-}_{-+}(u_n=v_n-\hbar,u_{n-1},\ldots,u_1;v_n,v_{n-1},\ldots,v_1;\lambda)= \label{D_Zrec} \\
 =\frac{\theta(\lambda+n\hbar)\theta(\hbar)}{\theta(\lambda+(n-1)\hbar)}
  &\prod_{m=1}^{n-1}\Big(\theta(v_n-v_m-\hbar)\theta(u_m-v_n)\Big)\,
 Z^{+-}_{-+}(u_{n-1},\ldots,u_1;v_{n-1},\ldots,v_1;\lambda). \notag
\end{align}
\end{prop}

\noindent
Considering the $n$-th column and the $n$-th row
and taking into account that $a(u_n-v_n)\big|_{u_n=v_n-\hbar}=a(-\hbar)=0$ we
conclude that the unique possibility to have a non-trivial contribution is:
$\beta_{nn}=\gamma_{nn}=-1$, $\gamma_{nj}=-1$,
$\beta_{nj}=+1$, $j=1,\ldots,n-1$, $\beta_{in}=-1$, $\gamma_{in}=+1$, $i=1,\ldots,n-1$. The last formulae
impose the same DWBC for the $(n-1)\times(n-1)$ sublattice: $\delta_{n-1,j}=\beta_{nj}=+1$, $j=1,\ldots,n-1$,
$\alpha_{i,n-1}=\gamma_{in}=+1$, $i=1,\ldots,n-1$, $d_{n-1,n-1}=d_{nn}$. Thus the substitution $u_n=v_n-\hbar$
to the partition function for the whole lattice gives us
\begin{multline}
 Z^{+-}_{-+}(u_n=v_n-\hbar,u_{n-1},\ldots,u_1;v_n,v_{n-1},\ldots,v_1;\lambda)= \\
\quad =\bar c(-\hbar;\lambda)\prod_{j=1}^{n-1}\bar b(v_n-v_j-\hbar;\lambda+(n-j)\hbar)
 \prod_{i=1}^{n-1}b(u_i-v_n;\lambda+(n-i)\hbar) \\
\quad \times Z^{+-}_{-+}(u_{n-1},\ldots,u_1;v_{n-1},\ldots,v_1;\lambda). \label{D_proofZrec}
\end{multline}
Using the explicit expressions for the Boltzmann weights~\eqref{D_BWa} one can rewrite
 the last formula in the form~\eqref{D_Zrec}. \qed

\begin{Rem}\label{D_rem2}
From the formula~\eqref{D_proofZrec}
 we see that the following transformation of the $R$-matrix
\begin{equation} \label{D_remZrec}
 b(u,v;\lambda)\to\rho\,b(u,v;\lambda), \qquad \bar b(u,v;\lambda)\to\rho^{-1}\, \bar b(u,v;\lambda)
\end{equation}
does not change the recurrent relation~\eqref{D_Zrec}, where $\rho$ is non-zero constant not
depending on $u$, $v$ and $\lambda$.
\end{Rem}

\begin{lem} \label{D_theor_3prop}
If the set of functions $\{Z^{(n)}(u_n,\ldots,u_1;v_n,\ldots,v_1;\lambda)\}_{n\ge1}$ satisfies
the conditions of the propositions~\ref{D_propZsym}, \ref{D_propZep}, \ref{D_propZrec} and the initial condition
\begin{equation}
 Z^{(1)}(u_1;v_1;\lambda)=\bar c(u_1-v_1)=\frac{\theta(u_1-v_1-\lambda)
 \theta(\hbar)}{\theta(-\lambda)} \label{D_Z0c}
\end{equation}
then
\begin{equation}
 Z^{+-}_{-+}(u_n,\ldots,u_1;v_n,\ldots,v_1;\lambda)=Z^{(n)}(u_n,\ldots,u_1;v_n,\ldots,v_1;\lambda).
 \label{D_Z_Z}
\end{equation}
\end{lem}

\noindent  Due to the~\eqref{D_Z0c} this lemma can be proved by induction
over $n$. Let the equality~\eqref{D_Z_Z} be valid for $n-1$. Consider the functions $Z^{+-}_{-+}(u_n,
\ldots,u_1;v_n,\ldots,v_1;\lambda)$ and $Z^{(n)}(u_n,\ldots,u_1;v_n,\ldots,v_1;\lambda)$ as functions
of $u_n$. Both are elliptic polynomials of degree $n$ with the character~\eqref{D_char_chi}. They have the same
value in the point $u_n=v_n-\hbar$, and due to the symmetry of these functions with respect to the
parameters $\{v_j\}_{j=1}^n$ they coincide in all the points $u_n=v_j-\hbar$, $j=1,\ldots,n$. Then, due
to the Lemma~\ref{D_ep_coin} (see Appendix~\ref{D_Appendix_ep}) these functions are identical. \qed

\begin{Rem}\label{D_rem3}
As we see from the proof of the lemma~\ref{D_theor_3prop}
it is sufficient to establish the symmetry with respect to only the variables $v_j$.
\end{Rem}

\begin{Rem}\label{D_rem4}
The transformation~\eqref{D_remZrec} of the $R$-matrix does
not change the partition function with DWBC.
\end{Rem}

\section{Elliptic projections of currents}
\label{D_secd4}

 Let $\lfK_0=\mathbb C[u^{-1}][[u]]$
be a completed set of complex-valued meromorphic  functions defined in the neighborhood of origin which
have only simple poles at this point.  Let
 $\{\epsilon^i\}$ and $\{\epsilon_i\}$ be two dual bases in $\lfK_0$
 $\oint\frac{du}{2\pi i} \epsilon^i(u)\,\epsilon_j(u)=\delta^i_j$.

\subsection{Current description of the elliptic algebra}
\label{D_secd41}

Let $\A$ be a Hopf algebra generated by elements $\hat h[s]$, $\hat e[s]$, $\hat f[s]$,
$s\in\lfK_0$, which are subjected to linear relations
\begin{equation*}
\hat x[\alpha_1s_1+\alpha_2s_2]=\alpha_1\hat x[s_1]+\alpha_2\hat x[s_2], \qquad \alpha_1,\alpha_2\in\mathbb C,
\quad s_1,s_2\in\lfK_0,
\end{equation*}
where $x\in\{h,e,f\}$ and some commutation relations. These relations will be written in term of currents
\begin{align}
 h^+(u)&=\sum_{i\ge0}\hat h[\epsilon_{i;0}]\epsilon^{i;0}(u),&
 h^-(u)&=-\sum_{i\ge0}\hat h[\epsilon^{i;0}]\epsilon_{i;0}(u), \notag \\
 f(u)&=\sum_i\hat f[\epsilon_i]\epsilon^i(u),&  e(u)&=\sum_i\hat e[\epsilon_i]\epsilon^i(u). \label{D_ftdef}
\end{align}
The currents $e(u)$ and $f(u)$ are called total currents. They are defined by dual bases of
$\lfK_0$ and their definition does not depend on the choice
of these dual bases (see \cite{ER1,S21}).
 The currents $h^+(u)$ and $h^-(u)$ are called Cartan currents and they are defined by special basis
\begin{equation*}
 \epsilon^{k;0}(u)=\frac1{k!}\left(\frac{\theta'(u)}{\theta(u)}\right)^{(k)}, \quad k\ge0;
 \qquad 
 \epsilon_{k;0}(u)=(-u)^k,\quad k\ge0.   
\end{equation*}
The commutation relations read as follows~\cite{EF}:
\begin{equation}
[K^\pm(u),K^\pm(v)]=0, \qquad [K^+(u),K^-(v)]=0,
\end{equation}
\begin{equation}
K^\pm(u)e(v)K^\pm(u)^{-1}=\frac{\theta(u-v+\hbar)}{\theta(u-v-\hbar)}e(v),
\end{equation}
\begin{equation} \label{D_cr_EFR0Kf}
K^\pm(u)f(v)K^\pm(u)^{-1}=\frac{\theta(u-v-\hbar)}{\theta(u-v+\hbar)}f(v), 
\end{equation}
\begin{equation}
\theta(u-v-\hbar)e(u)e(v)=\theta(u-v+\hbar)e(v)e(u),\label{D_cr_EFR0ee}
\end{equation}
\begin{equation}
\theta(u-v+\hbar)f(u)f(v)=\theta(u-v-\hbar)f(v)f(u),\label{D_cr_EFR0ff}
\end{equation}
\begin{equation}
[e(u),f(v)]=\hbar^{-1}\delta(u,v)\Big(K^+(u)-K^-(v)\Big),
\end{equation}
where
$K^+(u)=\exp\Big(\frac{e^{\hbar\partial_u}-e^{-\hbar\partial_u}}{2\partial_u}
h^+(u)\Big)$, $K^-(u)=\exp\big(\hbar h^-(u)\big)$ and $\delta(u,v)=\sum\limits_{n\in\mathbb Z}
\frac{u^n}{v^{n+1}}$ is a delta-function~\footnote{One can find more details about distributions acting on
$\lfK_0$ and their significance in the theory of current algebras in our previous papers~\cite{S21,S22}.} for $\lfK_0$.
The algebra $\A$ is a non-central version of the algebra $A(\tau)$ introduced in~\cite{ER1}.
This algebra is equipped with the co-product and co-unit:
\begin{equation}
 \Delta K^\pm(u)=K^\pm(u)\otimes K^\pm(u),
 \end{equation}
\begin{equation}
\Delta e(u)=e(u)\otimes1+K^-(u)\otimes e(u), \label{D_Delta_e}
\end{equation}
\begin{equation}
 \Delta f(u)=f(u)\otimes K^+(u)+1\otimes f(u), \label{D_Delta_f}
 \end{equation}
\begin{equation}
\varepsilon(K^\pm(u))=1, \qquad \varepsilon(e(u))=0, \qquad \varepsilon(f(u))=0.
\end{equation}

Let $\A_F$ and $\A_E$ be subalgebras of $\A$ generated by the generators $\hat h[\epsilon_{i;0}]$, $\hat f[s]$,
and $\hat h[\epsilon^{i;0}]$, $\hat e[s]$, respectively, $s\in\lfK_0$. The subalgebra $\A_F$ is
described by currents $K^+(u)$, $f(u)$, and the subalgebra $\A_E$ by $K^-(u)$, $e(u)$. Also,
we introduce notations $H^+$ for the subalgebras of $\A$ generated by $\hat h[\epsilon_{i;0}]$.
As it was stated in~\cite{EF} the bialgebras $(\A_F,\Delta^{op})$ and $(\A_E,\Delta)$ are dual to
each other with respect to the Hopf pairing $\langle\cdot,\cdot\ra\colon\A_F\times\A_E\to\mathbb C$ defined
in terms of currents as follows: $\la f(u),K^-(v)\ra=\la K^+(u),e(v)\ra=0$ and
\begin{equation}
 \La f(u),e(v)\Ra=\hbar^{-1}\delta(u,v), \qquad
 \La K^+(u),K^-(v)\Ra=\frac{\theta(u-v-\hbar)}{\theta(u-v+\hbar)}. \label{D_HPdef}
\end{equation}
These formulae uniquely defines the Hopf pairing on  $\A_F\times\A_E$. In particular, one can
derive the following formula (see Appendix~\ref{D_Appendix_HP})
\begin{multline}
 \La f(t_n)\cdots f(t_1),e(v_n)\cdots e(v_1)\Ra
  =  \\
  =\hbar^{-n}\sum_{\sigma\in S_n}\prod_{{l<l' \atop \sigma(l)>\sigma(l')}}
   \frac{\theta(v_{\sigma(l)}-v_{\sigma(l')}+\hbar)}{\theta
   (v_{\sigma(l)}-v_{\sigma(l')}-\hbar)}\prod_{m=1}^n\delta(t_m,v_{\sigma(m)}).
    \label{D_la_ff_ee_ra}
\end{multline}

\subsection{Projections of currents}
\label{D_secd42}

We define the projections as linear maps acting in the subalgebra $\A_F$. In the subalgebra
$\A_E$ dual projections act, which we do not consider here. The main tool to work with the projections
are {\itshape half-currents} $f^+_\lambda(u)$ and $f^-_\lambda(u)$ (defined further).
By this reason we define the projections in terms of the half-currents.
The last ones are usually defined as parts of the sum~\eqref{D_ftdef} (with the corresponding sign)
such that $f(u)=f^+_\lambda(u)-f^-_\lambda(u)$. Here $\lambda$ is a parameter of the decomposition
of the total current into the difference of half-currents.
Elliptic half-currents are investigated in details on the classical level in~\cite{S21}.
We will  introduce the half-currents by their representations via  integral
 transforms of the total current $f(u)$:
\begin{equation}
 \f^+(u)=\oint\limits_{|v|<|u|}\frac{dv}{2\pi i}\frac{\theta(u-v-\lambda)}
 {\theta(u-v)\theta(-\lambda)}f(v), \quad
 \f^-(u)=\oint\limits_{|v|>|u|}\frac{dv}{2\pi i}\frac{\theta(u-v-\lambda)}
 {\theta(u-v)\theta(-\lambda)}f(v),\label{D_fpmht}
\end{equation}
where $\lambda\notin \Gamma=\mathbb Z+\mathbb Z\tau$.
The half-current $f^+_\lambda(u)$ is called positive
 and $f^-_\lambda(u)$ -- negative.

The corresponding positive and negative projections are also parameterized by $\lambda$
and they are defined on the half-currents as follows:
\begin{eqnarray}
 P^+_\lambda\big(f^+_\lambda(u)\big)=f^+_\lambda(u),\qquad P^-_\lambda\big(f^+_\lambda(u)\big)=0,
 \label{D_Pfpdef} \\
 P^+_\lambda\big(f^-_\lambda(u)\big)=0,\qquad  P^-_\lambda\big(f^-_\lambda(u)\big)=f^-_\lambda(u).
 \label{D_Pfmdef}
\end{eqnarray}

Let first define the projections in subalgebra $\A_f$ generated by currents $f(u)$.
As a linear space this subalgebra is spanned by the products $f(u_n)f(u_{n-1})\cdots f(u_1)$,
$n=0,1,2,\ldots$. It means that any element of $\A_f$ can be presented as a sum (maybe infinite)
of the integrals~\footnote{The integral $\oint$ without limits means a formal integral -- a continuous
extension of the integral over the unit circle.}
\begin{equation}
 \oint\frac{du_n\cdots du_1}{(2\pi i)^n}f(u_n)\cdots f(u_1)\,s_n(u_n)\cdots s_1(u_1), \qquad
 s_n,\ldots, s_1\in\lfK_0. \label{D_tutyaubralzvezdochki}
\end{equation}
It follows from PBW theorem proved in~\cite{EF} that any element of $\A_f$ can be also represented
as a sum of the integrals
\begin{equation*}
\oint\limits\frac{du_n\cdots du_1}{(2\pi i)^n}
 f^-_{\lambda+2(n-1)\hbar}(u_n)\cdots f^-_{\lambda+2m\hbar}(u_{m+1})
 f^+_{\lambda+2(m-1)\hbar}(u_m)\cdots f^+_{\lambda}(u_1)s_n(u_n)\cdots s_1(u_1),
\end{equation*}
$s_n\cdots s_1\in\lfK_0$, $0\le m\le n$, and, therefore, it is sufficient to define the projections
on these products of half-currents:
\begin{equation}
 P^+_\lambda(x^-x^+)=\varepsilon(x^-)x^+,\qquad P^-_\lambda(y^-y^+)=y^-\varepsilon(y^+), \label{D_Pxy_def}
\end{equation}
where
\begin{eqnarray}
 x^-=f^-_{\lambda+2(n-1)\hbar}(u_n)\cdots f^-_{\lambda+2m\hbar}(u_{m+1}),\qquad
y^-=f^-_{\lambda}(u_n)\cdots f^-_{\lambda-2(n-m-1)\hbar}(u_{m+1}),  \notag \\
 x^+=f^+_{\lambda+2(m-1)\hbar}(u_m)\cdots f^+_{\lambda}(u_1), \qquad
y^+=f^+_{\lambda-2(n-m)\hbar}(u_m)\cdots f^+_{\lambda-2(n-1)\hbar}(u_1). \notag
\end{eqnarray}
The product of zero number of currents is identified with $1$ and in this case: $\varepsilon(1)=1$.
The counit $\varepsilon$ of nonzero number of half-currents is always zero. So, this definition
generalizes the formulae~\eqref{D_Pfpdef}, \eqref{D_Pfmdef}. We complete the definition of the projections
on all the subalgebra $\A_F=\A_f\cdot H^+$ by formulae
\begin{equation*}
P^+_\lambda(at^+)=P^+_\lambda(a)t^+, \qquad P^-_\lambda(at^+)=P^-_\lambda(a)\varepsilon(t^+),
\end{equation*}
where $a\in\A_f$, $t^+\in H^+$.

\subsection{The projections and the universal elliptic weight function}
\label{D_secd43}

Consider the expressions of the form
\begin{equation}
 P^+_{\lambda-(n-1)\hbar}\big(f(u_n)f(u_{n-1})\cdots f(u_2)f(u_1)\big), \label{D_Pfn0}
\end{equation}
where the parameter $\lambda-(n-1)\hbar$ is chosen for some symmetry reasons. Let us begin with the
case $n=1$. The formula~\eqref{D_Pfpdef} implies that the projection in this case is equal to the positive
half-current, which can be represented as an integral transform of the total current:
\begin{equation*}
 P^+_{\lambda}\big(f(u_1)\big)=f^+_{\lambda}(u_1)=
 \oint\limits_{|u_1|>|v_1|}\frac{dv_1}{2\pi i}\frac{\theta(u_1-v_1{-\lambda})}{\theta(u_1-v_1)
 \theta(-\lambda)} f(v_1).
 \end{equation*}
The integral kernel of this transform gives the initial condition for the partition function with a factor:
\begin{equation}
 Z^{(1)}(u_1;v_1;{\lambda})=
  {\theta(\hbar)\theta(u_1-v_1)}\frac{\theta(u_1-v_1{-\lambda})}{\theta(u_1-v_1)\theta({-\lambda})}.
  \label{D_icZ0}
\end{equation}

The projections~\eqref{D_Pfn0} can be calculated by generalizing the method proposed  in~\cite{KhP}
for the
algebra $U_q(\hat\slt)$. The method use a recursion over $n$. Let us first present the last total
current in~\eqref{D_Pfn0} as the difference of half-currents:
\begin{align}
 P^+_{\lambda-(n-1)\hbar}\big(f(u_n)\cdots f(u_2)f(u_1)\big)&=   \label{D_Pfn0pm}\\
  =P^+_{\lambda-(n-3)\hbar}\big(f(u_n)\cdots f(u_2)\big)f^+_{\lambda-(n-1)\hbar}(u_1)
   &-P^+_{\lambda-(n-1)\hbar}\big(f(u_n)\cdots f(u_2)f^-_{\lambda-(n-1)\hbar}(u_1)\big) . \notag
\end{align}
In the first term we move out the positive half-current from the projection and, therefore,
calculation of this term reduced to the calculation of $(n-1)$-st projection.
In the second term in~\eqref{D_Pfn0pm} we move the negative half-current to the left step by step
using the following commutation relation~\cite{EF}
\begin{equation*}
  f(v)f^-_{\lambda}(u_1)
   =\frac{\theta(v-u_1-\hbar)}{\theta(v-u_1+\hbar)}f^-_{\lambda+2\hbar}(u_1)f(v)
     +\frac{\theta(v-u_1+\lambda+\hbar)}{\theta(v-u_1+\hbar)}F_{\lambda}(v), \\
\end{equation*}
where
\begin{equation*}
   F_\lambda(v)=\frac{\theta(\hbar)}{\theta(\lambda+\hbar)}
        \big(f^+_{\lambda+2\hbar}(v)f^+_{\lambda}(v)
                               -f^-_{\lambda+2\hbar}(v)f^-_{\lambda}(v)\big).
\end{equation*}

In each step we obtain an additional term containing $F_\lambda(u)$ and in the last step the negative
half-current is annihilated by the projection:
\begin{equation}
   P^+_{\lambda-(n-1)\hbar}\big(f(u_n)\cdots f(u_2)f^-_{\lambda-(n-1)\hbar}(u_1)\big)
       =\sum_{j=2}^n Q_j(u_1) X_j,  \label{D_Pfn0m}
\end{equation}
where
\begin{eqnarray*}
   Q_j(u)=\frac{\theta(u_j-u+\lambda-(n-2j+2)\hbar)}{\theta(u_j-u+\hbar)}
        \prod_{k=2}^{j-1}\frac{\theta(u_k-u-\hbar)}{\theta(u_k-u+\hbar)}, \\
   X_j=P^+_{\lambda-(n-1)\hbar}\big(f(u_n)\cdots f(u_{j+1})
            F_{\lambda-(n-2j+3)\hbar}(u_j)f(u_{j-1})\cdots f(u_2)\big).
\end{eqnarray*}
Putting $u_1=u_i$ in~\eqref{D_Pfn0m} we can substitute the negative half-current to the positive one due
to the commutation relation for the total currents $f(u)$ and the equality $f(u)f(u)=0$.
Moving out the positive half-current to left one obtains a linear system of equations for
$X_i$,  $i=2,\ldots,n$:
\begin{equation}
   P^+_{\lambda-(n-3)\hbar}\big(f(u_n)\cdots f(u_2)\big)f^+_{\lambda-(n-1)\hbar}(u_i)
                =\sum_{j=2}^n Q_j(u_i)X_j. \label{D_PfQX}
\end{equation}
Multiplying each equation~\eqref{D_PfQX} by
\begin{equation*}
 \frac{\theta(u_i-u+\lambda)}{\theta(\lambda)}
  \prod_{k=2}^n\frac{\theta(u_k-u_i+\hbar)}{\theta(u_k-u+\hbar)}
  \prod_{{k=2 \atop k\ne i}}^n\frac{\theta(u_k-u)}{\theta(u_k-u_i)},
\end{equation*}
summing it over $i=2,\ldots,n$ and using the interpolation formula (see Appendix~\ref{D_Appendix_ep})
\begin{equation}
 Q_j(u)=\sum_{i=2}^n Q_j(u_i)\frac{\theta(u_i-u+\lambda)}{\theta(\lambda)}
  \prod_{k=2}^n\frac{\theta(u_k-u_i+\hbar)}{\theta(u_k-u+\hbar)}
  \prod_{{k=2 \atop k\ne i}}^n\frac{\theta(u_k-u)}{\theta(u_k-u_i)} \label{D_interpol_Q_p}
\end{equation}
one yields
\begin{align}
  &P^+_{\lambda-(n-3)\hbar}\big(f(u_n)\cdots f(u_2)\big)\times \label{D_Pn1fQX} \\
	&\times \sum_{i=2}^n \frac{\theta(u_i-u+\lambda)}{\theta(\lambda)}
  \prod_{k=2}^n\frac{\theta(u_k-u_i+\hbar)}{\theta(u_k-u+\hbar)}
  \prod_{{k=2 \atop k\ne i}}^n\frac{\theta(u_k-u)}{\theta(u_k-u_i)} f^+_{\lambda-(n-1)\hbar}(u_i)=
    \sum_{j=2}^n Q_j(u)X_j. \notag
\end{align}
Comparing~\eqref{D_Pn1fQX} with~\eqref{D_Pfn0m} we conclude
\begin{multline}
   P^+_{\lambda-(n-1)\hbar}\big(f(u_n)\cdots f(u_2)f^-_{\lambda-(n-1)\hbar}(u_1)\big)
  =P^+_{\lambda-(n-3)\hbar}\big(f(u_n)\cdots f(u_2)\big)\times \\
\quad   \times\sum_{i=2}^n\frac{\theta(u_i-u_1+\lambda)}{\theta(\lambda)}
  \prod_{k=2}^n\frac{\theta(u_k-u_i+\hbar)}{\theta(u_k-u_1+\hbar)}
  \prod_{{k=2 \atop k\ne i}}^n\frac{\theta(u_k-u_1)}{\theta(u_k-u_i)}f^+_{\lambda-(n-1)\hbar}(u_i).
\end{multline}
Eventually, returning to the formula~\eqref{D_Pfn0pm} we derive the following expression for
the projection
\begin{equation}
   P^+_{\lambda-(n-1)\hbar}\big(f(u_n)\cdots f(u_2)f(u_1)\big)
  =P^+_{\lambda-(n-3)\hbar}\big(f(u_n)\cdots f(u_2)\big)f^+_{\lambda-(n-1)\hbar}(u_1;u_n,\ldots,u_2),
  \label{D_Pfn0_Pfn1fp}
\end{equation}
where we introduce the linear combination of the  currents:
\begin{multline}
 f^+_{\lambda-(n-2m+1)\hbar}(u_m;u_n,\ldots,u_{m+1})= f^+_{\lambda-(n-2m+1)\hbar}(u_m)
 -\sum_{i=m+1}^n\frac{\theta(u_i-u_m+\lambda+(m-1)\hbar)}{\theta(\lambda+(m-1)\hbar)}\times \\
 \times\prod_{k=m+1}^n \frac{\theta(u_k-u_i+\hbar)}{\theta(u_k-u_m+\hbar)}
  \prod\limits_{{k=m+1 \atop k\ne i}}^n\frac{\theta(u_k-u_m)}{\theta(u_k-u_i)}
   f^+_{\lambda-(n-2m+1)\hbar}(u_i). \label{D_fpm_nm}
\end{multline}
Continuing this calculation by the induction we obtain an expression for the projections in terms of the
half-currents~\eqref{D_fpm_nm}:
\begin{equation}
   P^+_{\lambda-(n-1)\hbar}\big(f(u_n)\cdots f(u_2)f(u_1)\big)
   =\prod_{n\ge m\ge1}^{\longleftarrow}f^+_{\lambda-(n-2m+1)\hbar}(u_m;u_n,\ldots,u_{m+1}). \label{D_Pfn0_prod_f}
\end{equation}

To represent the projections~\eqref{D_Pfn0_prod_f} in integral form we first rewrite the expression~\eqref{D_fpm_nm} in the form
\begin{align}
 &f^+_{\lambda-(n-2m+1)\hbar}(u_m;u_n,\ldots,u_{m+1})= \prod_{k=m+1}^n\frac{\theta(u_k-u_m)\theta(\hbar)}{\theta(u_k-u_m+\hbar)}\bigg(\prod_{k=m+1}^n G_{\hbar}(u_k-u_m)\times \label{D_fpm_nm_G} \\
 &\times f^+_{\lambda-(n-2m+1)\hbar}(u_m)
 +\sum_{i=m+1}^n G_{-\lambda-(m-1)\hbar}(u_m-u_i) \prod\limits_{{k=m+1 \atop k\ne i}}^n G_{\hbar}(u_k-u_i)
   f^+_{\lambda-(n-2m+1)\hbar}(u_i)\bigg). \notag
\end{align}
where $G_{\lambda}(u-v)=\frac{\theta(u-v+\lambda)}{\theta(u-v)\theta(\lambda)}$.
Substituting~\eqref{D_fpmht} to~\eqref{D_fpm_nm_G} and using the addition formula~\footnote{This formula can be proved using Lemma~\ref{D_ep_coin} of Appendix~\ref{D_Appendix_ep}.}
\begin{equation}
\sum_{i=1}^N\prod_{{j=1 \atop j\ne i}}^N
              G_{\lambda_j}(u_j-u_i)G_{\lambda_0}(u_i-v)=\prod_{i=1}^N G_{\lambda_i}(u_i-v), \label{D_GGG}
\end{equation}
where $\lambda_0=\sum\limits_{i=1}^N\lambda_i$, one can represent the half-currents~\eqref{D_fpm_nm}
as an integral transforms of the total current:
\begin{align}
 &f^+_{\lambda-(n-2m+1)\hbar}(u_m;u_n,\ldots,u_{m+1})= \label{D_fpm_nm_int} \\
 &=\prod_{k=m+1}^n\frac{\theta(u_k-u_m)}{\theta(u_k-u_m+\hbar)}
  \oint\limits_{|u_i|>|v|}\frac{dv}{2\pi i}
   \frac{\theta(u_m-v-\lambda-(m-1)\hbar)}
                      {\theta(u_m-v)\theta(-\lambda-(m-1)\hbar)}
     \prod_{k=m+1}^n\frac{\theta(u_k-v+\hbar)}{\theta(u_k-v)}
                                            f(v). \notag
\end{align}
Replacing each combination of the half-currents~\eqref{D_fpm_nm}  in~\eqref{D_Pfn0_prod_f}
by their integral form we obtain
\begin{eqnarray}
 P^+_{\lambda-(n-1)\hbar}\big(f(u_n)\cdots f(u_2)f(u_1)\big)
 =\prod_{n\ge k>m\ge1}\frac{\theta(u_k-u_m)}{\theta(u_k-u_m+\hbar)}
   \oint\limits_{|u_i|>|v_j|}\frac{dv_n\cdots dv_1}{(2\pi i)^n}\times \notag \\
    \times \prod_{n\ge k>m\ge1}\frac{\theta(u_k-v_m+\hbar)}{\theta(u_k-v_m)}
    \prod_{m=1}^n \frac{\theta(u_m-v_m-\lambda-(m-1)\hbar)}{\theta(u_m-v_m)\theta(-\lambda-(m-1)\hbar)}
           f(v_n)\cdots f(v_1). \label{D_Pfn0_prod_f_int}
\end{eqnarray}

The formulae \eqref{D_Pfn0_prod_f} and  \eqref{D_Pfn0_prod_f_int}
yield expressions for the universal elliptic weight
functions in terms of the current generators of the algebra $\mathcal{A}$.

\subsection{Universal weight function and SOS model partition function}

To extract the integral kernel from the expression \eqref{D_Pfn0_prod_f_int} and derive a formula for the partition function
we use the Hopf pairing~\eqref{D_HPdef}. Let us calculate the following expression generalizing~\eqref{D_icZ0}:
\begin{multline}
 Z^{(n)}(u_n,\ldots,u_1;v_n,\ldots,v_1;\lambda)=\prod_{i,j=1}^n\theta(u_i-v_j)
  \prod_{n\ge k>m\ge1}\frac{\theta(u_k-u_m+\hbar)\theta(v_k-v_m-\hbar)}{\theta(u_k-u_m)
  \theta(v_k-v_m)}\times \\
\quad \times\big(\hbar\theta(\hbar)\big)^n\La P^+_{\lambda-(n-1)\hbar}\big(f(u_n)\cdots f(u_1)\big),
e(v_n)\cdots e(v_1)\Ra. \label{D_Zn_def}
\end{multline}
Using the expression for the projection of the product of the total currents~\eqref{D_Pfn0_prod_f_int}
and the formula~\eqref{D_la_ff_ee_ra} we obtain
\begin{gather}
 Z^{(n)}(u_n,\ldots,u_1;v_n,\ldots,v_1;\lambda)= \notag \\
\quad   =\theta(\hbar)^n\prod_{i,j=1}^n\theta(u_i-v_j)
   \prod_{k>m}\frac{\theta(v_k-v_m-\hbar)}{\theta(v_k-v_m)}
  \sum_{\sigma\in S_n}\prod_{{l<l' \atop \sigma(l)>\sigma(l')}}
   \frac{\theta(v_{\sigma(l)}-v_{\sigma(l')}+\hbar)}{\theta(v_{\sigma(l)}-v_{\sigma(l')}-\hbar)}\times \notag \\
\qquad  \times\prod_{k>m}\frac{\theta(u_k-v_{\sigma(m)}+\hbar)}{\theta(u_k-v_{\sigma(m)})}
   \prod_{m=1}^n\frac{\theta(u_m-v_{\sigma(m)}-\lambda-(m-1)\hbar)}
                      {\theta(u_m-v_{\sigma(m)})\theta(-\lambda-(m-1)\hbar)}= \label{D_Zn_theta} \\
\quad =\prod_{k>m}\frac{\theta(v_k-v_m-\hbar)}{\theta(v_k-v_m)}
  \sum_{\sigma\in S_n}\prod_{{l<l' \atop \sigma(l)>\sigma(l')}}
   \frac{\theta(v_{\sigma(l)}-v_{\sigma(l')}+\hbar)}{\theta(v_{\sigma(l)}-v_{\sigma(l')}-\hbar)}\times \notag \\
\qquad  \times\prod_{k>m}\theta(u_k-v_{\sigma(m)}+\hbar)\prod_{k<m}\theta(u_k-v_{\sigma(m)})
   \prod_{m=1}^n\frac{\theta(u_m-v_{\sigma(m)}-\lambda-(m-1)\hbar)\theta(\hbar)}
                      {\theta(-\lambda-(m-1)\hbar)}. \notag
\end{gather}
We see from this formula that the expression~\eqref{D_Zn_theta} defines a holomorphic functions of
the variables $u_i$.

\begin{theor}
The set of functions $\big\{Z^{(n)}(u_n,\ldots,u_1;v_n,\ldots,v_1;\lambda)\big\}_{n\ge1}$
defined by the formula~\eqref{D_Zn_def} satisfies the conditions of the
propositions~\ref{D_propZsym}, \ref{D_propZep}, \ref{D_propZrec} and the initial condition~\eqref{D_Z0c}.
Therefore, by virtue of the theorem~\ref{D_theor_3prop} they coincide with the partition functions of
the SOS model with DWBC:
\begin{multline}
 Z^{+-}_{-+}(u_n,\ldots,u_1;v_n,\ldots,v_1;\lambda)= \\
\quad =\prod_{n\ge k>m\ge1}\frac{\theta(v_k-v_m-\hbar)}{\theta(v_k-v_m)}
  \sum_{\sigma\in S_n}\prod_{{l<l' \atop \sigma(l)>\sigma(l')}}
   \frac{\theta(v_{\sigma(l)}-v_{\sigma(l')}+\hbar)}{\theta(v_{\sigma(l)}-v_{\sigma(l')}-\hbar)}
     \prod_{1\le k<m\le n}\theta(u_k-v_{\sigma(m)})\times \\
\quad\times\prod_{n\ge k>m\ge1}\theta(u_k-v_{\sigma(m)}+\hbar)\prod_{m=1}^n
\frac{\theta(u_m-v_{\sigma(m)}-\lambda-(m-1)\hbar)\theta(\hbar)}
                      {\theta(-\lambda-(m-1)\hbar)}. \label{D_Z_Z_}
\end{multline}
\end{theor}

\noindent
The initial condition~\eqref{D_Z0c} is verified by checking the formula~\eqref{D_icZ0}.
The first factor in the right hand side of~\eqref{D_Zn_def} is symmetric with respect to both sets
of variables. Then the symmetry with respect to the variables $\{u\}$ and the variables $\{v\}$ follows
from the commutation relations~\eqref{D_cr_EFR0ff} and \eqref{D_cr_EFR0ee} respectively.
The formula~\eqref{D_Zn_theta} implies that~\eqref{D_Zn_def} are elliptic polynomials of degree $n$
with character~\eqref{D_char_chi} in variables $u_i$, particularly in $u_n$. Now let us substitute
$u_n=v_n-\hbar$ to~\eqref{D_Zn_theta}. The non-vanishing terms in the right hand side correspond
to the permutations $\sigma\in S_n$ satisfying $\sigma(n)=n$. Substituting $u_n=v_n-\hbar$ into these
 terms one obtains the recurrent relation~\eqref{D_Zrec}. \qed \\

\section{Trigonometric limit}
\label{D_secd5}

In this section we investigate the trigonometric degenerations of the formulae obtained
in the elliptic case. In particular, taking the corresponding trigonometric limit in the
expression for the SOS model partition function~\eqref{D_Z_Z_} we reproduce the expression for the
6-vertex partition function~\eqref{D_stat-s_pr}.

First we consider the degeneration of $R$-matrix --  the matrix of the Boltzmann weights,
which defines the model. To do it we need the formula of the trigonometric degeneration
($\tau\to i\infty$) of the odd theta function defined by the conditions~\eqref{D_theta}:
\begin{equation*}
 \lim_{\tau\to i\infty}\theta(u)=\frac{\sin\pi u}{\pi}.
\end{equation*}
In terms of the multiplicative variables $z=e^{2\pi iu}$, $w=e^{2\pi iv}$ this formula can be
rewritten as follows:
\begin{equation*}
2\pi i e^{\pi i(u+v)}\lim_{\tau\to i\infty}\theta(u-v)=z-w.
\end{equation*}
Multiplying the $R$-matrix~\eqref{D_Rz} by $2\pi i e^{\pi i(u+v)}$ and taking the limit we
obtain the following matrix depending rationally on the multiplicative variables $z$, $w$
and multiplicative parameters $q=e^{\pi i\hbar}$, $\mu=e^{2\pi i\lambda}$:
\begin{eqnarray}
 R(z,w;\mu)=2\pi i e^{\pi i(u+v)}\lim_{\tau\to i\infty}R(u-v;\lambda)= \notag \\
 =\left(\begin{array}{cccc}
  zq-wq^{-1}  & 0              & 0              & 0              \\
  0 & \frac{(z-w)(\mu q-q^{-1})}{(\mu-1)} &\frac{(z-w\mu)(q-q^{-1})}{(1-\mu)}  & 0              \\
  0 & \frac{(z\mu-w)(q-q^{-1})}{(\mu-1)} &\frac{(z-w)(\mu q^{-1}-q)}{(\mu-1)}  & 0   \\
  0              & 0              & 0              & (zq-wq^{-1})
 \end{array}\right).  \label{D_Rzmu}
\end{eqnarray}
The matrix~\eqref{D_Rzmu} inherits the property to satisfy the dynamical Yang-Baxter equation
and it defines statistical model which is called {\it trigonometric SOS model}.

To obtain the non-dynamical trigonometric case we need to implement the additional
limit $\lambda\to-i\infty$ implying $\mu\to\infty$ (or $\lambda\to i\infty$ implying $\mu\to0$):
\begin{equation}
 \tilde R(z,w)=\lim_{\mu\to\infty}R(z,w;\mu)=\left(\begin{array}{cccc}
  zq-wq^{-1}  & 0              & 0              & 0              \\
  0 & q(z-w) &(q-q^{-1})w  & 0              \\
  0 & (q-q^{-1})z & q^{-1}(z-w)  & 0   \\
  0              & 0              & 0              & zq-wq^{-1}
 \end{array}\right). \label{D_Rzmund}
\end{equation}
The matrix~\eqref{D_Rzmund} differs from the matrix of 6-vertex Boltzmann weights~\eqref{D_Bw1}
by the transformation~\eqref{D_remZrec}. Taking into account the Remark~\ref{D_rem4} (from Subsec.3.2)
we conclude that both matrices~\eqref{D_Bw1} and \eqref{D_Rzmund} define
the same partition function $Z(\{z\},\{w\})$ with DWBC~\footnote{The matrix~\eqref{D_Rzmund} is a
limit of a matrix which differs from~\eqref{D_Rz} by the transformation~\eqref{D_remZrec} with $\rho=q$.}.

To obtain the partition function with DWBC for the trigonometric SOS model
 one should multiply the partition function with DWBC for the
elliptic SOS model by certain factor  and take the trigonometric limit:
\begin{gather}
 Z^{+-}_{-+}(\{z\},\{w\};\mu)
 =\prod_{k,j=1}^n\big(2\pi i e^{\pi i(u_k+v_j)}\big)\lim_{\tau\to i\infty}
 Z^{+-}_{-+}(\{u\},\{v\};\lambda)= \notag \\
\quad  =\prod_{n\ge k>m\ge1}\frac{w_k q^{-1}-w_m q}{w_k-w_m}
  \sum_{\sigma\in S_n}\prod_{{l<l' \atop \sigma(l)>\sigma(l')}}
   \frac{w_{\sigma(l)}q-w_{\sigma(l')}q^{-1}}{w_{\sigma(l)}q^{-1}-w_{\sigma(l')}q}\times \label{D_Z_Z_t} \\
  \times\prod_{n\ge k>m\ge1}\big(z_k q-w_{\sigma(m)}q^{-1}\big)
\prod_{1\le k<m\le n}\big(z_k-w_{\sigma(m)}\big)
   \prod_{m=1}^n\frac{\big(z_m-w_{\sigma(m)}\mu q^{2(m-1)}\big)(q-q^{-1})}
                      {\big(1-\mu q^{2(m-1)}\big)}. \notag
\end{gather}
It easy to check that the formula~\eqref{D_stat-s_pr} is obtained from the formula~\eqref{D_Z_Z_t}
by taking the limit: $Z(\{z\},\{w\})=\lim\limits_{\mu\to\infty} Z^{+-}_{-+}(\{z\},\{w\};\mu)$.

\section*{Acknowledgements}
This paper is a part of PhD thesis
of A.S. which he is preparing in co-direction of S.P. and V.R. in
the Bogoliubov  Laboratory of Theoretical Physics, JINR, Dubna and in LAREMA,
D\'epartement de Math\'ematics, Universit\'e d'Angers. He is
grateful to the CNRS-Russia exchange program on mathematical physics  and  personally to J.-M. Maillet
for financial and general support of this thesis project. V.R. are
thankful to Ph.~Di~Franchesco and T.~Miwa for their stimulating lectures and their interest during
ENIGMA School "Quantum Integrability" held in Lalonde les Maures on October 14-19, 2007.
He thanks O.~Babelon and M.~Talon for the invitation to this school.
V.R. had used during the
project a partial financial support by ANR GIMP
and a support of INFN-RFBR "Einstein" grant (Italy-Russia).
S.P. was supported in part by RFBR grant 06-02-17383. Both V.R. and S.P. were supported in part
by the Grant for support of scientific schools NSh-8065.2006.2

\section*{Appendix}

\setcounter{section}{0} \setcounter{subsection}{0}
\renewcommand{\thesection}{\Alph{section}}

\section{Interpolation formula for the elliptic polynomials}
\label{D_Appendix_ep}

We will call by a character a group homomorphism $\chi\colon\Gamma\to\mathbb C^\times$,
where $\Gamma=\mathbb Z+\tau\mathbb Z$ and $\mathbb C^\times$ is a multiplicative group of nonzero
complex numbers. Each character $\chi$ and an integer number
$n$ define a space $\Theta_n(\chi)$ consisting of the holomorphic functions on $\mathbb C$
with the translation properties
\begin{equation*}
 \phi(u+1)=\chi(1)\phi(u), \qquad
 \phi(u+\tau)=\chi(\tau)e^{-2\pi inu-\pi in\tau}\phi(u).
\end{equation*}
If $n>0$ then $\dim\Theta_n(\chi)=n$ (and $\dim\Theta_n(\chi)=0$ if $n<0$). The elements
of the space $\Theta_n(\chi)$ are called elliptic polynomials (or theta-functions) of degree $n$
with the character $\chi$ (see~\cite{FS}).

\begin{prop}
Let $\{\phi_j\}_{j=1}^n$ be a basis of $\Theta_n(\chi)$, with the character
$\chi(1)=(-1)^n$, $\chi(\tau)\hm=(-1)^n e^{2\pi i \alpha}$, then the determinant of
the matrix $||\phi_j(u_i)||_{1\leq i,j \leq n}$ is equal to
\begin{equation}
  \det||\phi_j(u_i)||=C\cdot\theta(\sum_{k=1}^n u_k-\alpha)\prod_{i<j}\theta(u_i-u_j), \label{D_ell_Vand}
\end{equation}
where $C$ is nonzero constant.
\end{prop}
\noindent
Consider the ratio
\begin{equation}
  \frac{\det||\phi_j(u_i)||}{\theta(\sum_{k=1}^n u_k-\alpha)\prod_{i<j}\theta(u_i-u_j)}. \label{D_ell_Vand_ratio}
\end{equation}
This is an elliptic function of each $u_i$ with only simple pole in any fundamental domain
(the points $u_i$ satisfying $\sum_{k=1}^n u_k-\alpha\in\Gamma$).
Therefore, it is a  constant function of each $u_i$.
Thus this ratio does not
depend on each $u_i$ and we have to prove that it does not vanish, that is to
 prove that the determinant $\det||\phi_j(u_i)||$ is not identically zero. Let us denote by
 $\Delta^{i_1,\ldots,i_k}_{j_1,\ldots,j_k}$ the minor of this determinant corresponding to
 the $i_1$-th, $\ldots$, $i_k$-th rows and the $j_1$-th, $\ldots$, $j_k$-th columns. Suppose
 that this determinant is identically zero and consider the following decomposition
\begin{equation}
 \det||\phi_j(u_i)||=\sum_{k=1}^n(-1)^{k+1}\phi_k(y_1)\Delta^{2,\ldots,n}_{1,\ldots,k-1,k+1,\ldots,n}.
 \label{D_ell_Vand_decomp}
\end{equation}
Since the functions $\phi_k(y_1)$ are linearly independent the minors
$\Delta^{2,\ldots,n}_{1,\ldots,k-1,k+1,\ldots,n}$ are identically zero. Decomposing
the minor $\Delta^{2,\ldots,n}_{2,\ldots,n}$ we conclude that the minors
 $\Delta^{3,\ldots,n}_{2,\ldots,k-1,k+1,\ldots,n}$ are identically zero, and so on.
 Finally, we obtain that $\Delta^{n}_{n}=\phi_n(y_n)$ is identically zero what can not be true. \qed

\begin{lem} \label{D_ep_coin}
 Let us consider two elliptic polynomials $P_1,P_2\in\Theta_n(\chi)$,
where $\chi(1)=(-1)^n$, $\chi(\tau)\hm=(-1)^n e^{\alpha}$, and $n$ points
$u_i$, $i=1,\ldots,n$, such that $u_i-u_j\not\in\Gamma$, $i\ne j$,
and $\sum_{k=1}^n u_k-\alpha\not\in\Gamma$. If the values of these polynomials
coincide at these points, $P_1(u_i)=P_2(u_i)$, then these polynomials coincide identically:
$P_1(u)=P_2(u)$.
\end{lem}

\noindent  Decomposing the considering polynomials as
$P_a(u)=\sum_{i=1}^n p_a^i\phi_i(u)$, $a=1,2$, we have the system of equations
\begin{equation*}
 \sum_{i=1}^n p_{12}^i\phi_i(u)=0,
\end{equation*}
with respect to the variables $p_{12}^i=p_{1}^i-p_2^i$. As we have proved,
the determinant of this system is equal to~\eqref{D_ell_Vand}
and therefore is not zero. Hence, this system has only trivial solution $p_{12}^i=0$,
but this implies $P_1(u)=P_2(u)$. \qed

Let $P\in\Theta_n(\chi)$ be an elliptic polynomial, where $\chi(1)=(-1)^n$,
$\chi(\tau)\hm=(-1)^n e^{2\pi i\alpha}$, and $u_i$, $i=1,\ldots,n$, be $n$ points
such that $u_i-u_j\not\in\Gamma$, $i\ne j$,
and $\sum_{k=1}^n u_k-\alpha\not\in\Gamma$. This polynomial can be restored by the values at these points:
\begin{equation}
\qquad P(u)=\sum_{i=1}^n P(u_i)
   \frac{\theta(u_i-u+\alpha-\sum_{m=1}^n u_m)}{\theta(\alpha-\sum_{m=1}^n u_m)}
  \prod_{{k=1 \atop k\ne i}}^n\frac{\theta(u_k-u)}{\theta(u_k-u_i)}. \label{D_interp_P}
\end{equation}
Indeed, the right hand side belongs to $\Theta_n(\chi)$, this equality holds at points $u=u_i$.
Using the lemma~\ref{D_ep_coin} we conclude that~\eqref{D_interp_P} holds at all $u\in\mathbb C$.

Consider the meromorphic functions
\begin{equation*}
 \qquad  Q_j(u)=\frac{\theta(u_j-u+\lambda-(n-2j+2)\hbar)}{\theta(u_j-u+\hbar)}
        \prod_{k=2}^{j-1}\frac{\theta(u_k-u-\hbar)}{\theta(u_k-u+\hbar)}.
\end{equation*}
It is easy to check that the functions
\begin{multline*}
 \quad  P_j(u)=\prod_{k=2}^n\theta(u_k-u+\hbar)Q_j(u)= \\
  \qquad       =\theta(u_j-u+\lambda-(n-2j+2)\hbar)\prod_{k=2}^{j-1}\theta(u_k-u-\hbar)
           \prod_{k=j+1}^n\theta(u_k-u+\hbar)
\end{multline*}
belong to $\Theta_{n-1}(\chi)$, where $\chi(1)=(-1)^{n-1}$, $\chi(\tau)\hm=(-1)^{n-1}
e^{2\pi i\alpha}$, $\alpha=\lambda+\sum_{k=2}^n u_k$. Since $\lambda\not\in\Gamma$,
the polynomials $P_j(u)$ can be restored by its values $P_j(u_i)$ via the interpolation
 formula~\eqref{D_interp_P}. Taking into account the relation between $Q_j(u)$ and
 $P_j(u)$ we obtain the formula~\eqref{D_interpol_Q_p}.~\footnote{We
 can suppose the condition $u_i-u_j\not\in\Gamma$,
 because $u_i$ in the formula~\eqref{D_interpol_Q_p} are formal variables.}

\section{Proof of the formula~\eqref{D_la_ff_ee_ra}}
\label{D_Appendix_HP}

Let $\A^n_f$ be a subspace spanned by $\hat f[s_n]\cdots\hat f[s_1]$ (that is by~\eqref{D_tutyaubralzvezdochki}), where $s_n,\ldots,s_1\in\lfK_0$. Notice first that due to the Hopfness of the pairing~\eqref{D_HPdef} the current $e(u)$ annihilates the spaces $\A^k_fH^+$ for $k\ge2$. Indeed, using the formulae $\la xy,z\ra=\la x\otimes y,\Delta(z)\ra$, $\la x,1\ra=\varepsilon(x)$ and~\eqref{D_Delta_e} we obtain $\la\A^k_fH^+,e(u)\ra\subset\la\A^1_f\otimes\A^{k-1}_fH^+,e(u)\otimes1+K^-(u)\otimes e(u)\ra=\{0\}$. It follows from the formula~\eqref{D_Delta_f} that the opposite coproduct on $\A^n_f$ has the form
\begin{multline}
 \Delta^{op}\big(f(t_n)\cdots f(t_1)\big)\in H^+\otimes\A^n_f+\sum_{k=2}^n\A^k_fH^+\otimes\A^{n-k}_f+ \\
 +\sum_{j=1}^n\prod_{i=j+1}^nK^+(t_i)f(t_j)\prod_{i=1}^{j-1}K^+(t_i)\otimes f(t_n)\cdots f(t_{j+1})f(t_{j-1})\cdots f(t_1).
\end{multline}
Now let us prove the formula~\eqref{D_la_ff_ee_ra} by induction. For $n=1$ it coincides with the definition of the pairing. Suppose the it holds for $n-1$. Then using the formulae $\la x,yz\ra=\la\Delta^{op}(x),y\otimes z\ra$, $\la1,x\ra=\varepsilon(x)$ and the commutation relation~\eqref{D_cr_EFR0Kf} we obtain
\begin{gather}
 \La f(t_n)\cdots f(t_1),e(v_n)\cdots e(v_1)\Ra =   
 \La \Delta^{op}\big(f(t_n)\cdots f(t_1)\big),e(v_n)\otimes e(v_{n-1})\cdots e(v_1)\Ra =   \notag \\
=\hbar^{-n}\sum_{j=1}^n\prod_{l'=j+1}^n\frac{\theta(v_n-v_{\sigma(l'-1)}+\hbar)}{\theta(v_n-v_{\sigma(l'-1)}-\hbar)}\delta(t_j,v_n)\times \label{D_la_ff_ee_ra_proof} \\
	 \times\sum_{\sigma\in S_{n-1}}\prod_{{1\le l<l'\le n-1 \atop \sigma(l)>\sigma(l')}}
   \frac{\theta(v_{\sigma(l)}-v_{\sigma(l')}+\hbar)}{\theta
   (v_{\sigma(l)}-v_{\sigma(l')}-\hbar)}\prod_{m=1}^{j-1}\delta(t_m,v_{\sigma(m)})\prod_{m=j+1}^{n}\delta(t_m,v_{\sigma(m-1)}). \notag
\end{gather}
Taking into account the identity $\sum\limits_{\sigma\in S_n}X(\sigma)=\sum\limits_{j=1}^n\sum\limits_{\phantom{xx} \sigma\in S_{n-1}}X\left(1\phantom{xx} \ldots \phantom{xx} j-1 \phantom{xx} j\phantom{x}  j+1\phantom{x} \ldots \phantom{xx} n\phantom{xx}
\atop  \sigma(1)\phantom{x} \ldots\phantom{x} \sigma(j-1)\phantom{x} n\phantom{x}  \sigma(j) \phantom{x}\ldots\phantom{x} \sigma(n-1)\right)$ one yields the formula~\eqref{D_la_ff_ee_ra}.

\renewcommand{\thesection}{\arabic{section}}

\pagestyle{empty}

\end{document}